\documentclass[graybox,envcountchap,sectrefs]{svmono}

\usepackage{type1cm}
\usepackage{makeidx}
\usepackage{graphicx}
\usepackage{multicol}
\usepackage[bottom]{footmisc}
\usepackage{newtxtext}
\usepackage[varvw]{newtxmath}
\usepackage{hyperref}
\usepackage{bbm}
\usepackage{tikz}
\usepackage{physics}
\usepackage{enumitem}
\usepackage{amsfonts}
\usepackage{epsfig}
\usepackage{epstopdf}
\usepackage[all]{xy}
\usepackage{dcolumn}
\usepackage{url}
\usepackage{dsfont}
\usepackage{slashed}
\usepackage{mathrsfs}
\usepackage{soul}
\usepackage{float}
\usepackage[T1]{fontenc}

\usepackage{mdframed}
\usepackage{xpatch}

\makeatletter
\xpatchcmd{\endmdframed}
    {\aftergroup\endmdf@trivlist\color@endgroup}
    {\endmdf@trivlist\color@endgroup\@doendpe}
    {}{}
\makeatother

\hypersetup{
	colorlinks=true,
	linkcolor=blue,
	citecolor=blue,
	urlcolor=blue
}

\DeclareMathAlphabet{\mathpzc}{OT1}{pzc}{m}{it}

\def\inbar{\,\vrule height1.5ex width.4pt depth0pt}
\def\IR{\relax{\rm I\kern-.18em R}}
\def\IC{\relax\hbox{$\inbar\kern-.3em{\rm C}$}}

\theoremstyle{plain}
\newtheorem{Proposition}{Proposition}[chapter]

\surroundwithmdframed[
    linewidth=0.4pt,
    linecolor=gray,
    backgroundcolor=gray!19,
    roundcorner=0pt,
    innerleftmargin=6pt,
    innerrightmargin=6pt,
    innertopmargin=4pt,
    innerbottommargin=4pt,
    skipabove=8pt,
    skipbelow=12pt
]{remark}

\makeindex

\makeatletter
\def\ps@headings{%
	\let\@oddfoot\@empty
	\let\@evenfoot\@empty
	\def\@oddhead{\hfil\thepage\hfil}%
	\def\@evenhead{\hfil\thepage\hfil}%
}
\makeatother

\begin{document}
	
	\begin{titlepage}
		\centering
		\vspace*{3cm}
		
		{\fontsize{18pt}{22pt}\selectfont\bfseries Massless Representations in Conformal Space and Their de Sitter Restrictions\par}
		\vspace{0.7cm}
		
		{\fontsize{14pt}{18pt}\selectfont\itshape The Split-Octonionic Fabric of Inflationary Field Dynamics\par}
		\vspace{2cm}
		
		
			{\large Jean-Pierre Gazeau}\\
			Universit\'{e} Paris Cit\'{e}, CNRS, Astroparticule et Cosmologie, F-75013 Paris, France and\\
			Faculty of Mathematics, University of Bia{\l}ystok, 15-245 Bia{\l}ystok, Poland\\
			\texttt{gazeau@apc.in2p3.fr}
			\vspace{0.7cm}
			
			{\large Hamed Pejhan}\\
			Institute of Mathematics and Informatics, Bulgarian Academy of Sciences, Acad. G. Bonchev Str. Bl. 8, 1113 Sofia, Bulgaria\\
			\texttt{pejhan@math.bas.bg}
			\vspace{0.7cm}
			
			{\large Ivan Todorov}\\
			INRNE, Bulgarian Academy of Sciences, Tsarigradsko Chaussee 72, BG-1784 Sofia, Bulgaria

		\vfill
		
		\emph{To be published by Cambridge University Press in the Cambridge Monographs on Mathematical Physics series.}
		
		\vspace*{1cm}	
		{\large \today}
		\vspace*{1cm}
	\end{titlepage}
	
	\frontmatter
	
\cleardoublepage
\thispagestyle{empty}

\vspace*{\fill}

\begin{center}
\itshape
    This monograph is dedicated to the memory of \textbf{Professor Ivan Todorov} (1933-2025), whose profound contributions to theoretical and mathematical physics and invaluable guidance were instrumental in shaping this research. We deeply regret his passing on February 14, 2025, during the preparation of this work. Grounded in his insights and notes, this study reflects his enduring impact on the field. His legacy will continue to inspire generations of physicists and mathematicians.
\end{center}

\vspace*{\fill}

\cleardoublepage
	\chapter*{Preface}
\addcontentsline{toc}{chapter}{Preface}

At the dawn of our Universe, during the inflationary epoch, spacetime underwent a brief but dramatic burst of quasi-exponential expansion. In this extreme regime, all quantum fields were effectively massless, and their dynamics were dictated by the intertwined structures of de Sitter (dS) and conformal symmetries. These symmetries constrained the primordial quantum fluctuations, which ultimately seeded the large-scale structure of the cosmos we observe today. Understanding the mathematical fabric of these massless (conformal) fields is therefore central not only to fundamental physics but also to modern cosmology, shedding light on the quantum origin of cosmic structures and the near-scale invariance of primordial perturbations.

\begin{quote} 
    \textit{``...as time goes on, it becomes increasingly evident that the rules which the mathematician finds interesting are the same as those which Nature has chosen.''}
    \hfill \\\textbf{— Paul Dirac (1939)}
\end{quote}

This profound insight from Paul Dirac perfectly encapsulates the spirit of this volume, which explores the elegant mathematical structures underlying massless elementary systems — ``free'' fields or particles — in dS spacetime. dS spacetime, the maximally symmetric solution of Einstein's equations with a positive cosmological constant $\Lambda$, is a manifold of constant positive curvature that models an exponentially expanding Universe and closely approximates the large-scale geometry of our present cosmos.

This volume is part of a broader research program dedicated to developing a mathematically rigorous, representation-theoretic framework for elementary systems in dS spacetime. In line with Wigner's fundamental principle — that quantum elementary systems correspond to unitary irreducible representations (UIRs) of the underlying spacetime symmetry group — quantum elementary systems in dS spacetime are identified with (projective) UIRs of the dS relativity group $\mathrm{SO}_0(4,1)$ or its universal covering group $\mathrm{Sp}(2,2) \cong \mathrm{Spin}(4,1)$. Wigner originally formulated this classification within the flat Minkowski setting, where the rest mass $m$ and spin $s$ of a system serve as the fundamental invariants labeling the corresponding UIR of the Poincar\'{e} group $\mathrm{ISL}(2,\mathbb{C})$ — the relativity group of flat Minkowski spacetime. Introducing a constant curvature into spacetime is the sole mechanism by which Poincar\'{e} symmetry is deformed, giving rise to the dS and anti-dS symmetry groups; in this curved context, the notion of mass is replaced by a curvature-dependent energy parameter, while spin remains a discrete invariant.

Within this robust representation-theoretic framework, and building on concepts developed in the earlier volume, \textit{The de Sitter (dS) group and its representations (2nd ed., Springer, Cham, Switzerland, 2024)}, this work delves more deeply into the mathematical structures underlying dS massless elementary systems. Fully self-contained, it places particular emphasis on the interplay between massless conformal representations and their restrictions to dS spacetime, with Clifford algebras serving as a unifying algebraic framework throughout.

Beyond Wigner's foundational framework, it has long been recognized that the massless nature of elementary systems is intimately tied to deeper symmetry principles — most notably, conformal symmetry. At the heart of this symmetry lies the conformal group $\mathrm{U}(2,2)$ — the $\mathrm{U}(1)$-extended form of $\mathrm{SU}(2,2)\cong\mathrm{SO}(4,2)$ — whose unitary representations encompass a wide variety of elementary systems. Among these, the massless representations, commonly called ladder representations, form a distinguished subclass, ideally suited for describing physically relevant massless systems. First studied by Mack and Todorov in the late 1960s, these representations exhibit the striking property of remaining irreducible when restricted to the quantum-mechanical Poincar\'{e} subgroup $\mathrm{ISL}(2,\mathbb{C})$. In this limit, they yield the well-known Poincar\'{e} massless representations, characterized by vanishing rest mass and quantized helicity. These representations form the natural framework for describing lightlike — or massless — elementary systems in Minkowski spacetime and serve as a conceptual springboard for investigating their counterparts in more general geometric settings.

This volume undertakes a detailed study of the restriction of massless (ladder) conformal representations to dS spacetime. Specifically, it considers the embedding of $\mathrm{SO}_0(4,1)$ into the conformal group $\mathrm{U}(2,2)$ via its universal covering $\mathrm{Sp}(2,2)$. Within this framework, we analyze the restriction of conformal ladder representations to dS symmetry, emphasizing the persistence of their irreducibility and its implications for the structure of massless fields in curved geometries.

A defining feature of this work, setting it apart from earlier studies, is the central role assigned to the conformal Clifford algebra $\mathfrak{cl}(4,2)$, rather than to the group $\mathrm{U}(2,2)$ itself. This choice reflects a deliberate shift from a purely group-theoretic viewpoint to an algebraic framework in which spinors, symmetry generators, and group actions are treated within a single, unified structure. Working at the level of the Clifford algebra makes it possible to keep both conformal and dS symmetries manifest while retaining direct control over their spinorial realizations.

Within this framework, the spinorial carrier space is fixed from the outset by identifying it with a real $8$-dimensional alternative composition algebra of split-octonions. The realization of $8$-component Majorana spinors then follows canonically from this split-octonionic structure, with the Clifford generators represented as left-regular multiplication operators by a distinguished set of imaginary split-octonion units. As a consequence, the Majorana gamma matrices are real by construction, and the reality of the spinorial carrier space is fixed intrinsically, rather than imposed \emph{a posteriori} through a representation-dependent choice of basis. This intrinsic Majorana structure plays a central conceptual role; it enforces the minimality of fermionic degrees of freedom at the most fundamental level and aligns the spinorial construction with Wigner's notion of elementary systems as minimal carriers of the real symmetry of spacetime.

Although the split-octonionic product is non-associative, its alternativity provides sufficient coherence to allow a consistent realization of the Clifford relations once one passes to left-multiplication operators. The resulting construction is therefore not an algebra homomorphism in the strict sense; nevertheless, when the split-octonion units are represented as linear endomorphisms of the spinor space, operator composition becomes strictly associative, and the defining anti-commutation relations of $\mathfrak{cl}(4,2)$ are satisfied exactly at the operator level. In this way, algebraic, spinorial, and geometric structures are fused into a single framework that exploits the richness of the underlying composition algebra while retaining full associativity precisely where it is required.

Within this Clifford-theoretic setting, the conformal Lie algebra $\mathfrak{su}(2,2)\cong\mathfrak{so}(4,2)$ emerges naturally from commutators of bivectors in the even subalgebra $\mathfrak{cl}^{\mathrm{even}}(4,2)\cong\mathfrak{cl}(4,1)$, while the Clifford pseudoscalar generates the central $\mathfrak{u}(1)$, yielding the full conformal algebra $\mathfrak{u}(2,2)\cong\mathfrak{su}(2,2)\oplus\mathfrak{u}(1)$. Exponentiation of these generators produces elements of the conformal group $\mathrm{U}(2,2)$ entirely within the Clifford algebra. The same formalism accommodates the dS Lie algebra $\mathfrak{so}(4,1)$ and its universal covering group $\mathrm{Sp}(2,2)$, so that exponentiation, inversion, group multiplication, and spinorial action are all realized internally and explicitly within the Clifford framework.

Moreover, this realization naturally provides a unified treatment of spinorial and tensorial representations, allowing invariant bilinear forms, ladder operators, Casimir elements, and spectra to be constructed explicitly for massless fields of arbitrary helicity. By fixing the real (Majorana) structure at the outset and relegating other realizations — such as chiral bases — to unitary reorganizations of the same real spinor space, the framework consistently enforces a principle that runs throughout the book; physical and representation-theoretic minimality should be imposed prior to, and independently of, representation-dependent conventions.

In this sense, situating spinors within the algebraic architecture of split-octonions does more than provide an efficient computational tool. It exposes deep and previously hidden correspondences between conformal and dS symmetries and alternative — indeed exceptional — algebraic structures, and suggests that the organization of massless elementary systems in curved spacetime is governed by a more rigid and intrinsically algebraic logic than is apparent in conventional formulations.

The methods developed in this monograph are accordingly designed to be both rigorous and accessible. Graduate students are guided step by step through the algebraic, geometric, and representation-theoretic tools, while advanced researchers are provided with concrete techniques and fresh perspectives for explicit computations and theoretical explorations. By integrating spinorial, tensorial, and group-theoretic structures within a single algebraic setting, the monograph bridges multiple domains — from higher-spin physics and conformal geometry to quantum field theory (QFT) in curved spacetimes, as well as to string theory, holography, and cosmological modeling.

In this way, the work serves simultaneously as a pedagogical guide, a technical reference, and a conceptual resource. It clarifies how symmetry, curvature, and algebraic structure converge in the theory of massless fields, and it highlights the conceptual role of the split-octonion realization as a unifying structure that connects spinorial geometry, algebraic formalism, and the representation theory of massless conformal systems, pointing toward new directions for both mathematical and physical research.


\medskip
\noindent\textbf{Guide for the Reader}
\medskip

The monograph is organized so as to move from general orientation to explicit algebraic construction and then to field-theoretic realization. A working familiarity with Lie groups and algebras, basic representation theory, and the standard framework of QFT is assumed, although the necessary structures are introduced progressively and with an emphasis on conceptual coherence.

Chap. \ref{Chapter. Introduction} serves as an overview of the main themes, conventions, and scope of the work. It may be read independently as a general introduction to the conformal and dS setting, as well as to the representation-theoretic viewpoint adopted throughout, in which elementary systems are understood, in the sense of Wigner, as irreducible realizations of the underlying spacetime symmetry.

Chaps. \ref{Chapter 2} and \ref{Chapter 3} contain the main algebraic and representation-theoretic development. Chap. \ref{Chapter 2} introduces the split-octonionic and Clifford-algebraic machinery and develops the corresponding Majorana and chiral realizations of $\mathfrak{cl}(4,2)$. Chap. \ref{Chapter 3} builds on this framework to construct and analyze the massless (ladder) representations of $\mathfrak{u}(2,2)$, together with their invariant bilinear forms, spectral structure, and related vertex-algebraic aspects. These chapters form the technical core of the monograph.

Chap. \ref{Chapter 4} turns to the explicit construction of conformal massless fields and their restriction to dS spacetime, with particular emphasis on the scalar case. Readers primarily interested in the field-theoretic realization may wish to read this chapter in parallel with, or following, the main constructions of Chap. \ref{Chapter 3}.

Accordingly, readers whose main interest is conceptual or field-theoretic may begin with Chap. \ref{Chapter. Introduction} and proceed selectively to Chap. \ref{Chapter 4}, referring back to Chaps. \ref{Chapter 2} and \ref{Chapter 3} as needed. Readers seeking a complete algebraic and representation-theoretic development are encouraged to follow the text in sequence.

Throughout the monograph, conventions and notation are introduced at the points where they are needed, and are used consistently thereafter.

For ease of reference, a glossary is provided at the end of the monograph, while the acronyms are listed separately in the front matter.


\vspace{\baselineskip}
\begin{flushright}\noindent
Paris, France and Bia{\l}ystok, Poland \hfill {\it Jean-Pierre Gazeau}\\
Sofia, Bulgaria\hfill {\it Hamed Pejhan}\\
Sofia, Bulgaria\hfill {\it Ivan Todorov}
\end{flushright}
	\cleardoublepage
\chapter*{Acknowledgements}
\addcontentsline{toc}{chapter}{Acknowledgements}

The authors would like to express their sincere gratitude to \textbf{Mrs. Boriana Todorova}, spouse and legal representative of the late Professor Ivan Todorov, for her understanding, support, and cooperation throughout the preparation of this monograph. Her gracious assistance in facilitating the posthumous inclusion of Professor Todorov as a co-author has been invaluable, and we deeply appreciate her encouragement and kindness during this difficult period.\\

The authors would also like to express their sincere thanks to \textbf{Maryam Bajalan} for her careful editorial assistance and meticulous proofreading during the preparation of this manuscript.\\

\textbf{Hamed Pejhan} is supported by:
\begin{enumerate}
    \item{The National Science Fund, Ministry of Education and Science of Bulgaria, under contract KP-06-N92/2.}

    \item{The Bulgarian Ministry of Education and Science, Scientific Programme ``Enhancing the Research Capacity in Mathematical Sciences (PIKOM)'', No. DO1-67/05.05.2022.}
\end{enumerate}
	\tableofcontents
	
\chapter*{Acronyms}
\addcontentsline{toc}{chapter}{Acronyms}

\begin{description}
    \item[\textbf{CFT}]{Conformal Field Theory}
    \item[\textbf{dS}]{de Sitter}
    \item[\textbf{EdS}]{Euclidean de Sitter}
    \item[\textbf{IR}]{Irreducible Representation}
    \item[\textbf{QFT}]{Quantum Field Theory}
    \item[\textbf{UIR}]{Unitary Irreducible Representation}
    \item[\textbf{UPEIR}]{Unitary Positive Energy Irreducible Representation}
\end{description}
	
	\mainmatter
	\renewcommand*\vec{\mathaccent"017E\relax}

\setcounter{equation}{0} \chapter{An Overview of Conformal Massless Elementary Systems and Wigner's Vision in de Sitter (dS) Geometry}\label{Chapter. Introduction}
\chaptermark{An Overview of Conformal Massless Elementary Systems and ...}

\section*{Chapter Overview}

This chapter sets the stage for the monograph by outlining its central motivations, thematic scope, and methodological framework. It places the study of conformal massless elementary systems — ``free'' fields or particles — within both a historical and conceptual context, tracing the influence of Ivan Todorov's pioneering work and the enduring significance of Wigner's principles for the classification of elementary systems. By emphasizing the interplay between representation theory, conformal invariance, and the geometry of de Sitter (dS) spacetime, the chapter provides a conceptual roadmap for analyzing massless elementary systems in cosmologically significant settings. It also establishes key conventions and notational standards, ensuring consistency throughout the text. Serving as both an introduction and a conceptual guide, this chapter prepares readers for the deeper algebraic, geometric, and physical developments in the subsequent chapters.

\section[Conformal Symmetry in Historical Perspective and the Legacy of Todorov]{Conformal Symmetry in Historical Perspective and the Legacy of Ivan Todorov}

Conformal symmetry — the invariance of physical systems under transformations that preserve angles while not necessarily preserving distances — has occupied a foundational role in the development of theoretical physics. Its conceptual origins trace back to the early twentieth century, when mathematicians and physicists began the systematic study of transformations preserving local geometric structures, thereby establishing the algebraic and geometric groundwork for modern symmetry principles in both quantum field theory (QFT) and relativity. The recognition that geometry can enforce structural constraints on physical laws elevated symmetry to the status of a unifying paradigm across diverse domains of physics.

In the 1930s, Cartan's pioneering work \cite{Cartan} on differential geometry provided a rigorous framework for analyzing symmetries on curved manifolds, introducing concepts such as torsion, curvature, and moving frames. While Cartan did not explicitly formulate conformal symmetry in physics, his formalism later supplied essential geometric and algebraic machinery for its systematic study. Concurrently, Weyl's 1918 proposal \cite{Herman Weyl} of scale invariance, further developed during the 1920s and 1930s, highlighted the profound interplay between geometry and physical law, foreshadowing the systematic consideration of local angle-preserving and scale transformations in theoretical physics.

By the 1950s and early 1960s, the mathematical study of unitary representations of Lie groups, particularly the Poincar\'{e} group, had been advanced through the foundational work of Wigner, who classified elementary particles as unitary irreducible representations (UIRs) of spacetime symmetries \cite{Wigner1939}, and by Newton and Wigner's analysis of relativistic localization \cite{Newton/Wigner}. Further progress was made by V. Bargmann in his classification of Lorentz group representations \cite{Bargmann1947}, and by Gel'fand and Naimark in their systematic treatment of representations of classical groups \cite{GelfandNaimark}. Building on these foundations, mathematicians such as Segal \cite{Segal1951} and Mackey \cite{Mackey1950} developed the theory of induced representations and harmonic analysis on groups, while Harish-Chandra \cite{HarishChandra1952} established the general representation theory of semisimple Lie groups. These developments clarified the classification of massive and Poincar\'{e} massless representations — including helicity states and the more exotic continuous-spin cases — and provided representation-theoretic tools that were later extended to larger symmetry groups, notably the conformal group, as in the rigorous treatment by Gel'fand, Minlos, and Shapiro \cite{GelfandMinlosShapiro1963} and later in the field-theoretic constructions of Mack and Salam \cite{Mack1969}.

Building on algebraic and representation-theoretic foundations, Fronsdal's seminal 1978 work \cite{Fronsdal1978} provided a systematic formulation of free massless higher-spin fields in a Poincar\'{e}-invariant framework. His construction introduced the now-standard gauge-invariant Fronsdal equations, giving consistent field equations for totally symmetric particles of arbitrary spin. Although originally formulated in flat spacetime, this framework laid the foundation for later developments that extended higher-spin fields to (anti-)dS backgrounds and to conformal higher-spin theories. These advances were part of a broader trend in the 1970s and 1980s in which conformal symmetry emerged as a unifying principle across diverse areas of theoretical physics, including conformal field theory (CFT), string theory, the study of critical phenomena in statistical mechanics, and higher-spin gauge theories.

As this conceptual and algebraic framework took shape, the late 1960s marked a pivotal phase in the study of conformal symmetry, featuring one of the first systematic investigations of massless representations and the introduction of ladder representations of the conformal group \cite{MT}. Spearheaded by the landmark collaboration between Gerhard Mack and Ivan Todorov, this program demonstrated how Poincar\'{e} massless representations could be consistently extended to the conformal group. The resulting algebraic framework provided a mathematically rigorous characterization of massless elementary systems and established a bridge between conventional flat-space QFT and its conformally extended formulations. This formalism subsequently became foundational for the precise treatment of massless quantum fields in both Minkowski and curved spacetimes, underpinning later developments in CFT, higher-spin theories, and the study of critical phenomena in statistical mechanics.

It was within this evolving historical landscape that Ivan Todorov (1933-2025) emerged as a leading figure whose contributions bridged the realms of rigorous mathematics and theoretical physics. Born in Sofia, Todorov earned his MSc in Physics from Sofia University in 1956. He then pursued his PhD at the Joint Institute for Nuclear Research (JINR) in Dubna, completing his thesis in 1960 under the supervision of the renowned theoretical physicist Nikolai N. Bogoliubov \cite{IT 1933-2025, TodorovMuseum}. His early research at JINR and at the Institute for Advanced Study in Princeton laid the foundation for his seminal contributions. Todorov advanced the axiomatic and algebraic foundations of QFT and developed a mathematically precise understanding of CFTs, making the ladder representation program and its applications in conformal symmetry a central part of his lasting legacy.

\begin{figure}[t]
    \centering
    \includegraphics[width=0.9\textwidth]{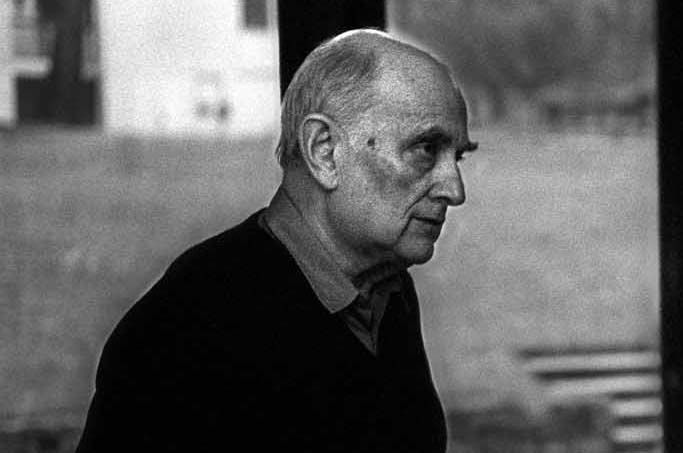} 
    \caption{\textbf{Professor Ivan Todorov (1933-2025)}, whose pioneering contributions to CFT, ladder representations, and massless quantum systems profoundly shaped mathematical physics.}
    \label{fig:todorov}
\end{figure}

Beyond his seminal collaboration with Nikolai N. Bogoliubov and Gerhard Mack, Todorov worked with distinguished scientists such as Valentine Bargmann, Louis Michel, Claude Itzykson, and Victor Kac, thereby enriching the theoretical understanding of symmetries and advancing the dialogue between algebraic structures and physical intuition \cite{TodorovIHES}. Throughout his career, he consistently emphasized the foundational role of group-theoretical methods, forging deep connections between Lie algebra representations, conformal transformations, and the structural analysis of QFT. His extensive body of work includes influential monographs, most notably \emph{General principles of quantum field theory} \cite{Bogoliubov}, \emph{Introduction to axiomatic quantum field theory} \cite{Bogoliubov2}, and \textit{Harmonic analysis: on the $n$-dimensional Lorentz group and its application to conformal quantum field theory} \cite{Dobrev}, co-authored with N.N. Bogoliubov and collaborators, which remain indispensable references for researchers. Equally significant were Todorov's integrity, his dedication to mentorship, and his sustained efforts to foster international collaboration. He played a leading role in organizing conferences and advanced schools that brought together Eastern and Western physicists during politically challenging times, thereby strengthening the global community of mathematical physics \cite{TodorovCERN}.

This monograph is dedicated to the memory of Professor Ivan Todorov, whose profound contributions and guidance were instrumental in shaping this work. Drawing extensively on his insights, notes, and foundational studies, it reflects his enduring impact on the field. Todorov's legacy continues to inspire generations of physicists and mathematicians, exemplifying the integration of rigorous mathematics with deep physical understanding that he championed throughout his career.

\section{Wigner's (Massless) Elementary Systems: Continuity under Curvature-Induced Deformations}
\sectionmark{Wigner's (Massless) Elementary Systems: Continuity under ...}

In modern theories of elementary systems — whether field-theoretic or phenomenological — their formulation, and in particular their interpretation, fundamentally rely on the concepts of energy, momentum, mass, and spin, which emerge from invariance under the Poincar\'{e} group, the relativity group of flat Minkowski spacetime. It is, however, widely acknowledged that such theories cannot ultimately rest solely on this symmetry. A fully consistent framework must respect the broader principle of general covariance, reflecting Einstein's vision of spacetime as a Riemannian manifold.

Beyond the familiar flat Minkowski setting, the absence of non-trivial motion groups in generic curved backgrounds makes the direct extension of these fundamental notions highly non-trivial, if not impossible. In this context, the common suggestion to merely generalize fundamental equations — such as the Klein-Gordon or Dirac equations — to fully covariant forms misses the deeper issue; it addresses only the formal structure, not the conceptual foundation. Indeed, modern theories of elementary systems are not primarily concerned with differential equations per se, but with the deeper symmetry principles that underlie them.

This is precisely where Wigner's principle — particularly, its remarkable continuity under curvature-induced deformations — proves essential, providing a robust and conceptually coherent framework for extending the classification and understanding of elementary systems from flat Minkowski spacetime to curved dS and anti-dS spacetimes.

\begin{figure}[t]
\centering
\begin{tikzpicture}[
    every node/.style={align=center},
    >=latex,
    box/.style={
        draw,
        thick,
        rounded corners=10pt,
        text width=4.35cm,
        minimum height=1.10cm,
        inner sep=4pt
    },
    central/.style={
        draw,
        very thick,
        rounded corners=10pt,
        fill=gray!14,
        text width=8.1cm,
        minimum height=1.30cm,
        inner sep=5pt
    },
    smallbranch/.style={
        draw,
        thick,
        rounded corners=5pt,
        fill=gray!5,
        text width=2.2cm,
        minimum height=0.85cm,
        inner sep=3pt
    },
    branchbox/.style={
        draw,
        thick,
        rounded corners=6pt,
        fill=gray!6,
        text width=4.9cm,
        minimum height=0.98cm,
        inner sep=4pt
    }
]

\node[box] (rep) at (0,4.8)
{\textbf{Representation Theory}\\
{\small unitary irreducible representations}};

\node[box] (phys) at (-3.4,0.0)
{\textbf{Elementary Systems}\\
{\small particles and fields}};

\node[box] (sym) at (3.4,0.0)
{\textbf{Spacetime Symmetry}\\
{\small Poincar\'{e}, de Sitter, and anti-de Sitter}};

\node[central] (wig) at (0,2.4)
{\textbf{Wigner's Paradigm}\\
{\small symmetry $\rightarrow$ unitary irreducible representations $\rightarrow$ elementary systems}\\[0.2em]
{\small \emph{Refinement (massless case):} conformal symmetry}};

\node[smallbranch] (dS) at (1.9,-1.9)
{\textbf{de Sitter}\\
{\small geometry}};

\node[smallbranch] (conf) at (4.9,-1.9)
{\textbf{Conformal}\\
{\small symmetry}};

\node[branchbox] (infl) at (3.45,-3.7)
{\textbf{Inflationary Cosmology}\\
{\small de Sitter geometry + conformal symmetry}};

\draw[thick] (rep) -- (wig);
\draw[thick] (phys) -- (wig);
\draw[thick] (sym) -- (wig);

\draw[thick] (sym.south) -- (3.4,-0.7);

\draw[thick,->] (3.4,-0.7) to[out=-110,in=90] (dS.north);
\draw[thick,->] (3.4,-0.7) to[out=-70,in=90] (conf.north);

\draw[thick,->] (dS.south) to[out=-90,in=155] (infl.north west);
\draw[thick,->] (conf.south) to[out=-90,in=25] (infl.north east);

\end{tikzpicture}
\caption{Wigner's paradigm, conformal refinement, and cosmological relevance. Representation theory, spacetime symmetry, and elementary systems are unified through Wigner's principle that elementary systems are realized as unitary irreducible representations (UIRs) of the spacetime symmetry group. For massless systems, conformal symmetry provides a distinguished refinement of this framework. In the cosmological setting, inflationary models naturally draw on both de Sitter (dS) geometry and conformal symmetry.}
\end{figure}

\subsection{Wigner's Elementary Systems}

In the field-theoretic formulation of elementary systems, the foundational works of Eugene Wigner \cite{Wigner1939, Newton/Wigner} established that the principles of special relativity impose symmetry requirements on the laws of nature — most notably, invariance under the Poincar\'{e} group in flat Minkowski spacetime.

Within the flat Minkowski framework, quantum elementary systems are identified with (projective) UIRs of the Poincar\'{e} group (or one of its coverings) \cite{Wigner1939, Newton/Wigner}. Each such representation is uniquely characterized by two invariant quantities: the rest mass $m$ and the spin $s$ of the system. These invariants not only encode the system's fundamental physical properties but also reflect the symmetry structure of the underlying Minkowski spacetime.

Remarkably, introducing a constant curvature into spacetime yields the unique natural and continuous deformation of the Poincar\'{e} symmetry group. This deformation gives rise to the dS and anti-dS groups of motion. The corresponding dS and anti-dS spacetimes \cite{deSitter, deSitter'} — maximally symmetric solutions of Einstein's equations with, respectively, positive and negative cosmological constant $\Lambda$ — thereby acquire a privileged status as the only curved backgrounds in which Wigner's symmetry-based definition of elementary systems admits a meaningful generalization. In these settings, elementary systems are once again described by UIRs of the relevant spacetime symmetry groups. As in the flat case, these representations are classified by two invariants — the spin and a characteristic mass or energy scale — thus preserving the core structural features of Wigner's original framework.

\begin{figure}[t]
\centering
\begin{tikzpicture}[
    every node/.style={align=center},
    ambient/.style={
        draw,
        very thick,
        rounded corners,
        minimum width=11.8cm,
        minimum height=10.0cm
    },
    group/.style={
        draw,
        thick,
        circle,
        minimum size=2.2cm,
        fill=gray!8
    },
    core/.style={
        draw,
        very thick,
        circle,
        minimum size=1.0cm,
        fill=gray!18
    },
    >=latex
]

\node[ambient] (confbox) at (0,1.0) {};

\node at (0,5.55) {\textbf{Conformal group}};
\node at (0,5.05) {$\mathrm{SU}(2,2)$};
\node at (0,4.63) {\scriptsize $\cong \mathrm{SO}(4,2)$};

\node[group] (ds) at (-4.55,1.55)
{\textbf{$\hspace{0.2cm}$ de Sitter $\hspace{0.2cm}$}\\
$\mathrm{Sp}(2,2)$\\
{\scriptsize $\cong \mathrm{SO}(4,1)$}};

\node[group] (poin) at (0,1.55)
{\textbf{Poincar\'{e}}\\
$\mathrm{ISL}(2,\mathbb{C})$\\
{\scriptsize $\mathbb{R}^{1,3}\rtimes \mathrm{SL}(2,\mathbb{C})$}};

\node[group] (ads) at (4.55,1.55)
{\textbf{anti-de Sitter}\\
$\mathrm{Sp}(4,\mathbb{R})$\\
{\scriptsize $\cong \mathrm{SO}(3,2)$}};

\node[core] (lor) at (0,-2.0)
{\textbf{Lorentz subgroup}\\
$\mathrm{SL}(2,\mathbb{C})$\\
{\scriptsize $\cong \mathrm{SO}(3,1)$}};

\draw[very thick] (0,3.45) -- (0,3.00);
\draw[very thick] (0,0.30) -- (lor.north);

\draw[->, thick] (-1.1,3.66) to[bend right=16] (ds.north);
\draw[->, thick] (0,3.00) -- (poin.north);
\draw[->, thick] (1.1,3.66) to[bend left=16] (ads.north);

\node at (0,3.75) {subgroups};

\draw[->, thick] (ds.south east) to[bend right=18] (lor.north west);
\draw[thick] (poin.south) -- (lor.north);
\draw[->, thick] (ads.south west) to[bend left=18] (lor.north east);

\node at (0,-1.18) {\scriptsize shared subgroup};

\draw[->, thick, dashed]
(ds.east) to[bend left=14]
node[above] {\scriptsize contraction}
node[below] {\scriptsize $\Lambda\to 0$}
(poin.west);

\draw[->, thick, dashed]
(ads.west) to[bend right=14]
node[above] {\scriptsize contraction}
node[below] {\scriptsize $\Lambda\to 0$}
(poin.east);

\end{tikzpicture}
\caption{Schematic illustration of the relationships among the Poincar\'{e}, de Sitter (dS), and anti-de Sitter (AdS) groups as subgroups of the conformal group. Both dS and AdS contract to the Poincar\'{e} group in the flat (zero-curvature) limit $\Lambda\to 0$. Conversely, for non-zero curvature, the dS and AdS groups may be viewed as curved deformations of the Poincar\'{e} group. The conformal group provides a unifying framework encompassing all three.}
\end{figure}

\subsection{Wigner's Massless Elementary Systems}

Beyond Wigner's foundational framework, it has long been established that the massless nature of elementary systems is intimately connected to deeper symmetry principles — most notably, conformal symmetry \cite{Fronsdal1978, MT, D36, Mack1969, Mack}. This symmetry is not merely an auxiliary structure but constitutes an intrinsic aspect of massless field dynamics, manifesting both at the level of representation theory and in the behavior of field equations under spacetime transformations.

In the flat Minkowski setting, massless elementary systems correspond to UIRs of the Poincar\'{e} group with vanishing mass and discrete helicity — commonly referred to as Poincar\'{e} massless representations. A result of fundamental significance, due to Mack and Todorov \cite{MT, Mack}, establishes that each such representation admits a unique extension to a UIR of the conformal group $\mathrm{SU}(2,2)\cong\mathrm{SO}(4,2)$, or more generally its $\mathrm{U}(1)$-extended form $\mathrm{U}(2,2)$; these conformal UIRs are commonly called ladder representations. This framework furnishes a canonical and rigorous link between masslessness and conformal symmetry, so that any quantum system realizing a Poincar\'{e} massless representation necessarily carries a well-defined conformal action.

A parallel and equally deep structure is present in constant-curvature spacetimes. In dS and anti-dS geometries, the natural analogues of Poincar\'{e} massless representations are those UIRs of the respective symmetry groups whose unique conformal extensions are massless (ladder) representations. In representation-theoretic terms, these conformal extensions match the conformal extensions obtained from the massless Poincar\'{e} side, meaning that:
\begin{eqnarray*}
    &\mbox{\textit{The massless} $\Longleftrightarrow$ \textit{conformal correspondence}}&\\[0.2cm]
    &\mbox{\textit{survives passage from flat Minkowski spacetime to dS and anti-dS backgrounds,}}&\\[0.2cm]
    &\mbox{\textit{the only difference being that}}\\[0.2cm] 
    &\mbox{\textit{in dS the energy sign is undefined,}}&\\[0.2cm]
    &\mbox{\textit{while}}\\[0.2cm]
    &\mbox{\textit{in anti-dS helicity is ambiguous.}}&
\end{eqnarray*}

\subsection{Wigner's Theory in de Sitter (dS) Spacetime: From Quantum Structure to Cosmological Dynamics}

This work focuses on massless (ladder) representations in conformal space, and in particular on their restriction to $3+1$-dimensional dS spacetime, whose symmetry structure governs both the earliest \cite{Linde} and the asymptotic \cite{Riess, Perlmutter} stages of cosmic evolution. During the inflationary epoch ($\approx 10^{-36}$ to $10^{-32}$ seconds after the Big Bang), when the Universe underwent rapid quasi-exponential expansion, the Hubble scale vastly exceeded all intrinsic mass scales, rendering all elementary systems effectively massless. In this regime, dS symmetry together with conformal symmetry provides the natural framework for their consistent description. Inflationary (massless) elementary systems are realized through ladder representations of the conformal group $\mathrm{U}(2,2)$, which, upon restriction to the dS subgroup $\mathrm{Sp}(2,2)$, rigorously encode the quantum degrees of freedom of these fields. This representation-theoretic structure underlies the near-scale invariance of primordial fluctuations that seeded cosmic structure and, more generally, furnishes a canonical, symmetry-based classification of inflationary quantum states via UIRs.

Moreover, at late times, observations of distant Type Ia supernovae \cite{Riess, Perlmutter} reveal that the Universe is accelerating, a phenomenon attributed to ``dark energy''. A small, positive cosmological constant $\Lambda$ naturally models this acceleration, making dS geometry both the vacuum state and the asymptotic attractor of the Universe.

Accordingly, dS spacetime remains a central arena for QFT in curved backgrounds, offering deep insights at the interface of quantum physics, gravity, and cosmology. Across both early- and late-time cosmology, massless systems governed by conformal symmetry form a unifying thread. Ladder representations of $\mathrm{U}(2,2)$ provide a coherent and mathematically precise description of massless field dynamics, encompassing both the inflationary quantum fluctuations that seeded cosmic structure and the behavior of fields in the late-time accelerated Universe. This representation-theoretic framework thus directly links quantum symmetries to observable cosmological phenomena, highlighting their continuous and fundamental role throughout cosmic evolution.

\subsection{Conceptual Framework and Scope}\label{Sec. Research Objectives}

The present monograph is part of an ongoing body of research devoted to symmetry-based and representation-theoretic formulations of elementary systems in (anti-)dS spacetime, as developed and reviewed in the companion monograph \cite{Gazeau2022} and references therein. Within this general line of investigation, the present work concentrates on massless elementary systems in dS spacetime and on their intimate relation to conformal symmetry. Although the embedding:
\begin{align}
    \mathrm{SO}_{0}(4,1) \;\hookrightarrow\; \mathrm{U}(2,2) \quad(\mbox{via the covering}\; \mathrm{Sp}(2,2))
\end{align}
is well known at the group-theoretic level, its representation-theoretic and physical consequences — particularly for massless (ladder) representations — call for a focused and algebraically explicit treatment. The aim here is to provide such a treatment, emphasizing both the continuity of conformal structures under curvature-induced deformations and their concrete realization on spinorial carrier spaces.

As already noted, the guiding principle underlying this approach is Wig\-ner's conception of an elementary system as a (``minimal'') irreducible realization of spacetime symmetry; physical degrees of freedom should be fixed as far as possible by symmetry alone, without the introduction of redundant or auxiliary structures. Massless systems are indeed paradigmatic in this respect, as their representations are maximally constrained, admit unique conformal extensions, and exhibit a striking rigidity under deformation from flat spacetime to (anti-)dS geometries. The algebraic framework adopted here is designed to reflect this principle of minimality at every level of the construction, beginning with — and in particular at — the level of the spinorial carrier space.

A central objective is therefore to move beyond purely abstract group-theoretic descriptions and to work instead within a concrete algebraic framework in which spinors, symmetry generators, and group actions are realized explicitly and uniformly. Rather than remaining within a group-representation approach based solely on $\mathrm{U}(2,2)$, as in earlier treatments \cite{MT}, this study adopts the conformal Clifford algebra $\mathfrak{cl}(4,2)$ as its primary language. Far from a change of notation, the Clifford framework furnishes intrinsic spinorial and geometric structures that sharpen both conceptual insight and explicit computation, and in which conformal and dS symmetries are embedded in a unified and canonical way.

In particular, the conformal algebra $\mathfrak{su}(2,2)\cong\mathfrak{so}(4,2)$ arises from commutators of bivectors in the even subalgebra $\mathfrak{cl}^{\mathrm{even}}(4,2)\cong\mathfrak{cl}(4,1)$, while the Clifford pseudoscalar generates the central $\mathfrak{u}(1)$, yielding the full conformal algebra $\mathfrak{u}(2,2)\cong\mathfrak{su}(2,2)\oplus\mathfrak{u}(1)$. Exponentiation of these generators produces invertible elements of $\mathrm{U}(2,2)$. Similarly, the dS algebra $\mathfrak{so}(4,1)$ appears as a distinguished subalgebra whose exponentials generate elements of $\mathrm{Sp}(2,2)$. Since the exponential power series converges within the finite-dimensional Clifford algebra — and, for bivectors, reduces to closed trigonometric or hyperbolic forms — group multiplication, inversion, and action on spinors take place entirely within the Clifford framework.

Within this setting, the choice of spinorial realization becomes conceptually significant. In Wigner's formulation, an elementary system is defined by a (``minimal'') UIR of the spacetime symmetry group, and any redundancy in the underlying spinorial description obscures this minimality. From this perspective, a real (Majorana) realization of the spinor space is preferred whenever it is available, since it encodes the physical degrees of freedom without introducing an \emph{a priori} complex structure not enforced by the real symmetry of spacetime.

In conventional matrix-based approaches, Majorana spinors arise only after selecting a particular representation of the Clifford algebra in which all gamma matrices are real. The Majorana condition is therefore basis-dependent and representation-specific. While mathematically consistent, this procedure leaves the reality structure of the spinor space tied to auxiliary conventions rather than fixed intrinsically by the algebraic framework itself.

The approach adopted in the present work, however, fixes the Majorana structure at the outset. By constructing $\mathfrak{cl}(4,2)$ explicitly on a real spinor space, reality is built into the framework prior to any choice of basis. Other commonly used realizations, such as the chiral basis, then arise as unitary reorganizations of the same real spinor space rather than as defining structures. This ordering reflects the physical requirement that minimality and reality be imposed before representation-dependent conveniences.

To implement this program concretely, the monograph develops an explicit real realization of $\mathfrak{cl}(4,2)$ on $8$-component Majorana spinors, grounded in the alternative composition algebra of split-octonions. The spinor space is identified with the split-octonion algebra itself, and the Clifford generators are realized canonically as left-regular multiplication operators by a distinguished set of imaginary split-octonion units. In this construction, the gamma matrices emerge directly from the intrinsic multiplication structure of the composition algebra and are real by design.

\begin{figure}[t]
\centering
\resizebox{\textwidth}{!}{%
\begin{tikzpicture}[
    >=latex,
    every node/.style={align=center},
    mainbox/.style={
        draw,
        rounded corners,
        thick,
        minimum height=1.15cm,
        text width=7.6cm,
        inner sep=5pt
    },
    centralbox/.style={
        draw,
        rounded corners,
        very thick,
        minimum height=1.35cm,
        text width=8.8cm,
        inner sep=6pt
    },
    sidebox/.style={
        draw,
        rounded corners,
        thick,
        minimum height=1.05cm,
        text width=4.1cm,
        inner sep=4pt
    },
    arrow/.style={->, thick}
]

\node[mainbox] (maj) at (-1.8,7.0)
{\textbf{Wigner minimality + canonical Majorana realization}\\
{\small split-octonions fix a minimal real $8$-component spinor space\\
(reality imposed prior to any representation choice)}};

\node[mainbox] (cliff) at (-1.8,4.0)
{\textbf{Conformal Clifford algebra}\\
$\mathfrak{cl}(4,2)$\\
{\small canonically realized on the spinor space}};

\node[centralbox] (even) at (-1.8,1.0)
{\textbf{Even Clifford algebra as the central algebraic arena}\\
$\mathfrak{cl}^{\mathrm{even}}(4,2)\cong \mathfrak{cl}(4,1)$\\[0.25em]
{\small bivectors $\Rightarrow \mathfrak{su}(2,2)$,\quad
pseudoscalar $\Rightarrow \mathfrak{u}(1)$}};

\node[mainbox] (calg) at (-1.8,-2.0)
{\textbf{Conformal algebra}\\
$\mathfrak{u}(2,2)\cong \mathfrak{su}(2,2)\oplus \mathfrak{u}(1)$};

\node[mainbox] (rep) at (-1.8,-5.0)
{\textbf{Massless (ladder) representations (any helicity)}\\
{\small spectra, ladder operators, invariant forms}};

\node[mainbox] (field) at (-1.8,-8.0)
{\textbf{Field-theoretic realization}\\
{\small conformal fields built on the same spinorial structure}};

\node[sidebox] (cgrp) at (7.5,-2.0)
{\textbf{Conformal group}\\
$\mathrm{U}(2,2)$\\
{\small inside the same Clifford setting}};

\draw[arrow] (maj.south) 
-- node[right,font=\small]{left multiplication} 
(cliff.north);

\draw[arrow] (cliff.south) 
-- node[right,font=\small]{even subalgebra} 
(even.north);

\draw[arrow] (even.south) 
-- node[right,font=\small]{commutators + center} 
(calg.north);

\draw[arrow] (calg.east) 
-- node[above,font=\small]{exponentiation} 
(cgrp.west);

\draw[arrow] (calg.south) 
-- node[right,font=\small]{UIR classification} 
(rep.north);

\draw[arrow] (rep.south) 
-- node[right,font=\small]{field construction} 
(field.north);

\end{tikzpicture}%
}
\caption{Conceptual structure of the conformal construction. Guided by Wigner's principle of minimality, the framework begins with a canonical real (Majorana) spinor realization provided by split-octonions. Left multiplication yields a concrete realization of the conformal Clifford algebra $\mathfrak{cl}(4,2)$, whose even subalgebra generates the conformal Lie algebra $\mathfrak{u}(2,2)$ via bivector commutators and the Clifford pseudoscalar. Exponentiation produces the conformal group $\mathrm{U}(2,2)$ within the same Clifford framework. This unified setting supports the classification of massless (ladder) representations of arbitrary helicity and their field-theoretic realization.}
\label{fig:conformal-roadmap}
\end{figure}

The essential algebraic ingredient underlying this realization is the alternativity of the split-octonion algebra. Although the split-octonionic product is non-associative, alternativity provides sufficient coherence in the multiplication rules to support a consistent implementation of the Clifford relations when one passes to left-multiplication operators. Accordingly, the resulting embedding is not an algebra homomorphism in the strict sense — indeed, in general $L_{AB}\neq L_A L_B$, reflecting the intrinsic non-associativity of the octonionic product — but this non-associativity is controlled and does not obstruct the construction of the Clifford generators.

Crucially, once the split-octonion units are realized as linear endomorphisms of the spinor space, their composition becomes strictly associative. The defining anti-commutation relations of $\mathfrak{cl}(4,2)$ are therefore enforced at the operator level, where associativity is restored, rather than at the level of the underlying composition algebra. In this way, the framework exploits the algebraic richness of split-octonions at the origin of the construction while retaining full associativity where it is required for a faithful realization of the Clifford algebra. This controlled interplay between non-associativity at the algebraic level and strict associativity at the level of operators constitutes a defining structural feature of the approach developed in this monograph.

This unified Clifford framework provides an explicit and fully \linebreak representation-theoretic platform for computing invariant bilinear forms, ladder operators, spectra, and other structural data, while keeping conformal and dS symmetries manifest and tightly intertwined. Moreover, $\mathfrak{cl}(4,2)$ naturally accommodates both spinorial and tensorial representations, allowing the same algebraic machinery to describe fields of arbitrary helicity. With this explicit spinorial structure in place, the representation theory of massless elementary systems can be addressed constructively. The principal objective is the classification of massless ladder (positive-energy) representations of $\mathrm{U}(2,2)$ using the Majorana-split-octonionic realization of $\mathfrak{cl}(4,2)$, with particular attention to ladder operators, spectra, invariant bilinear forms, lowest-weight data, and unitary equivalences between different spinorial bases.

A further objective is to analyze the restriction of these conformal ladder representations to the dS group $\mathrm{Sp}(2,2)\cong\mathrm{SO}_{0}(4,1)$. This restriction provides a precise representation-theoretic realization of Wigner's notion of massless elementary systems in curved spacetime and allows one to track explicitly how invariant quantities and classification data behave under curvature-induced deformation.

Finally, the framework developed here is designed to serve as a foundation for the construction of quantum fields in curved spacetime. By keeping conformal and dS symmetries manifest at the level of spinorial carriers and representations, the Clifford-based approach provides a natural starting point for the field-theoretic realizations developed in subsequent chapters, including the explicit construction of conformally invariant fields and their restriction to dS backgrounds.

\begin{figure}[t]
\centering
\begin{tikzpicture}[every node/.style={align=center}]

\draw[black!80, thick, rounded corners=10pt] (-5.2,-5.5) rectangle (5.2,5.0);

\node[anchor=west] at (-5.0,4.45) {$\mathfrak{cl}(4,2)$};

\draw[black!85, thick]
  (0,4.70)
  .. controls ( 2.55,4.62) and ( 4.65,2.95) .. (4.75,0.18)
  .. controls ( 4.82,-3.05) and ( 3.25,-4.70) .. (0,-4.88)
  .. controls (-3.25,-4.70) and (-4.82,-3.05) .. (-4.75,0.18)
  .. controls (-4.65,2.95) and (-2.55,4.62) .. cycle;

\draw[black!85, thick]
  (0,3.28)
  .. controls ( 2.08,3.23) and ( 3.78,2.02) .. (3.85,-0.06)
  .. controls ( 3.92,-2.48) and ( 2.68,-3.75) .. (0,-3.92)
  .. controls (-2.68,-3.75) and (-3.92,-2.48) .. (-3.85,-0.06)
  .. controls (-3.78,2.02) and (-2.08,3.23) .. cycle;

\draw[black!85, thick]
  (0,1.88)
  .. controls ( 1.55,1.84) and ( 2.82,0.96) .. (2.88,-0.36)
  .. controls ( 2.95,-2.02) and ( 2.05,-3.00) .. (0,-3.15)
  .. controls (-2.05,-3.00) and (-2.95,-2.02) .. (-2.88,-0.36)
  .. controls (-2.82,0.96) and (-1.55,1.84) .. cycle;

\draw[black!85, thick]
  (0,0.35)
  .. controls ( 0.92,0.32) and ( 1.72,-0.22) .. (1.80,-0.98)
  .. controls ( 1.85,-2.00) and ( 1.38,-2.78) .. (0,-2.92)
  .. controls (-1.38,-2.78) and (-1.85,-2.00) .. (-1.80,-0.98)
  .. controls (-1.72,-0.22) and (-0.92,0.32) .. cycle;

\node at (0,4.12)
{$\mathfrak{cl}^{\mathrm{even}}(4,2)\cong \mathfrak{cl}(4,1)$};

\node at (0,2.40)
{$\mathfrak{u}(2,2)=\mathfrak{su}(2,2)\oplus \mathfrak{u}(1)$};

\node at (0,0.78)
{$\mathfrak{su}(2,2)\cong \mathfrak{so}(4,2)$};

\node at (0,-0.86)
{$\mathfrak{sp}(2,2)\cong \mathfrak{so}(4,1)$};

\end{tikzpicture}

\caption{Nested hierarchy of algebraic structures illustrating the relationships among Clifford and Lie algebras relevant to the conformal and spin groups in signature $(4,2)$. The full Clifford algebra $\mathfrak{cl}(4,2)$ contains its even subalgebra $\mathfrak{cl}^{\mathrm{even}}(4,2) \cong \mathfrak{cl}(4,1)$, which in turn includes the Lie algebra $\mathfrak{u}(2,2) \cong \mathfrak{su}(2,2) \oplus \mathfrak{u}(1)$. The special unitary algebra $\mathfrak{su}(2,2)$ is isomorphic to the conformal algebra $\mathfrak{so}(4,2)$ and contains the de Sitter (dS) subalgebra $\mathfrak{sp}(2,2) \cong \mathfrak{so}(4,1)$.}
\end{figure}

\subsection{Layout, Main Notational Conventions, and Comprehensive References}

\subsubsection{Layout}

To achieve our goal, the remainder of this monograph is organized as follows: 
\begin{enumerate}[leftmargin=0.8cm]
    \item{The remainder of this chapter offers a concise overview of dS relativity, including the dS group and its manifold, the corresponding Lie algebra, and the quantum framework formulated in terms of dS UIRs.}

    \item{Chap. \ref{Chapter 2} constructs an explicit Majorana representation of the conformal Clifford algebra $\mathfrak{cl}(4,2)$ using split-octonions, providing a geometrically transparent framework for the internal symmetries of the conformal Lie algebra $\mathfrak{u}(2,2)$:
    \begin{align*}
        \mathfrak{u}(2,2) \cong \mathfrak{su}(2,2) \, \big(\cong \mathfrak{so}(4,2)\big) \oplus \mathfrak{u}(1) \;\subset\; \mathfrak{cl}(4,2)\,.
    \end{align*}
    It introduces a unitary transformation to the chiral basis, in which the Cartan elements of the maximal compact subalgebra of $\mathfrak{su}(2,2)$ are diagonal simultaneously with the $\mathfrak{u}(1)$ generator defined by the $\mathfrak{cl}(4,2)$ volume form, and systematically derives inner-product relations among matrices. These results provide a coherent, geometrically grounded framework for the study of massless conformal systems in flat Minkowski spacetime as well as in curved dS and anti-dS spacetimes.}

    \item{Chap. \ref{Chapter 3} develops a rigorous framework for positive-energy ladder representations of the conformal Lie algebra $\mathfrak{u}(2,2)$, realized as quantized models of massless fields with arbitrary helicity in $4$-dimensional Minkowski, dS, and anti-dS spacetimes. The construction is based on an invariant bilinear form on the Majorana-spinor space of $\mathfrak{cl}(4,2)$ and an explicit $4\times 4$ spinorial matrix realization arising from the isomorphism $\mathfrak{cl}^{\mathrm{even}}(4,2)\cong\mathfrak{cl}(4,1)$, providing a concrete geometric and computational foundation for conformal symmetry, ladder operators, and spectral analysis. The resulting lowest-weight representations are classified by conformal energy and helicity, remain irreducible upon restriction to the dS subalgebra $\mathfrak{sp}(2,2)\cong\mathfrak{so}(4,1)$, and admit an embedding into $4$-dimensional conformal vertex algebras that exposes their analytic structure, Casimir invariants, and a unified treatment of massless conformal systems in both flat and curved settings.}

    \item{Chap. \ref{Chapter 4} constructs low-helicity conformal massless fields, focusing on the zero-helicity case in dS spacetime. Massless scalar fields are formulated via ladder-type vertex operators with Bose operators and zero-modes, ensuring correct lowest-energy states and canonical two-point function, with analytic properties established in a precompact tube domain. The construction restricts fields from $6$-dimensional conformal space to $4$-dimensional dS spacetime using homogeneous light-cone functions, capturing gauge freedom and relating Euclidean and Minkowski-type coordinates via Weyl rescaling. The conformally invariant $6$-dimensional two-point function then projects onto dS spacetime, reproducing the standard massless scalar two-point function.}
\end{enumerate}

\subsubsection{Main Notational Conventions} \label{Sect. Convention}

Our main notational conventions are:
\begin{enumerate}[leftmargin=0.8cm]
    \item{Throughout this monograph, we adopt natural units by setting $c = 1 = \hbar$, where $c$ denotes the speed of light and $\hbar$ the Planck constant.}
    \item{We use the symbol `$\mathbbm{1}$' to denote the identity element within the relevant algebraic framework. In composition algebras, it represents the distinguished unit element that preserves the norm under multiplication; for complex numbers, quaternions, and (split) octonions, `$\mathbbm{1}$' indicates the multiplicative identity. In matrix representations, it corresponds to the identity matrix of the appropriate dimension, with indices added when necessary for clarity. In Clifford algebras, `$\mathbbm{1}$' serves as the unit element, ensuring closure under multiplication. The precise interpretation of `$\mathbbm{1}$' depends on the specific context in which it appears.}

    \item{We use the `$\overline{\phantom{a}}$' symbol to denote the complex conjugate, while `$^\ast$' denotes quaternionic or octonionic conjugation and, depending on context — particularly in matrix representations — also the Hermitian conjugate; for a matrix $\mathfrak{M}$, one has:
    \begin{align*}
        \mathfrak{M}^\ast = \big(\,\overline{\mathfrak{M}}\,\big)^\top \,,
    \end{align*}
    where `$^\top$' denotes transposition.}

    \item{We adopt the following conventions for the split-octonion units and their associated indices:
    \begin{enumerate}[leftmargin=0.8cm]
        \item{The set of split-octonion units is denoted by:
        \begin{align*}
            A \;\text{and/or}\; B = \mathbbm{1}\,,\; \boldsymbol{k}\,,\; \boldsymbol{\imath}\,,\; {\boldsymbol{\jmath}}\,,\; \boldsymbol{\ell}\boldsymbol{\imath}\,,\; {\boldsymbol{\jmath}}\boldsymbol{\ell}\,,\;  \boldsymbol{\ell}\boldsymbol{k}\,,\; \boldsymbol{\ell}\,.
        \end{align*}}

        \item{The subset: 
        \begin{align}
            \nu \;\text{and/or}\; \mu = \boldsymbol{\imath}\,,\;{\boldsymbol{\jmath}}\,,\; \boldsymbol{\ell}\boldsymbol{\imath}\,,\; {\boldsymbol{\jmath}}\boldsymbol{\ell}\,,\; \boldsymbol{\ell}\boldsymbol{k}\,,\; \boldsymbol{\ell}\,.
        \end{align}}
    \end{enumerate}
    It must be emphasized, however, that the symbols $A, B$ (and hence the indices $\nu,\mu$) have a dual interpretation depending on context. They may represent actual split-octonion units in the algebra, or serve merely as labels or indices for the corresponding operators or matrix representations ($L_A,\, S_{AB},\, E_{AB},\, D_{AB},\, \mathbbm{1}_{AB},\, m^{}_\mu,\, m^{}_{\mu\nu},\, \Gamma^{}_\mu,\, \Gamma^{}_{\mu\nu},\, \gamma^{}_\mu,\, \gamma^{}_{\mu\nu},\, \dots$), with no octonionic meaning. In the latter case, whenever it aids readability without compromising clarity, we adopt a streamlined notation by representing the indices $\nu$ and/or $\mu$, respectively, as:
    \begin{align*}
        \nu \;\text{and/or}\; \mu =&\; 0\,(:=\boldsymbol{\imath})\,,\quad 5 \,(:={\boldsymbol{\jmath}})\,,\nonumber\\[0.2cm]
        &\;1\,(:=\boldsymbol{\ell}\boldsymbol{\imath})\,,\quad 2\,(:={\boldsymbol{\jmath}}\boldsymbol{\ell})\,,\quad 3\,(:=\boldsymbol{\ell}\boldsymbol{k})\,,\quad 4\,(:=\boldsymbol{\ell})\,.
    \end{align*}
    We denote any subset of indices $\nu$ and/or $\mu$, typically $1,2,3,4$, by $i,j$. Nevertheless, each symbol should be interpreted according to its context. This dual interpretation is discussed in Sect. \ref{Sect. 8uniii}, Eqs. \eqref{8uni}-\eqref{8uni'*}.}
        
    \item{In the context of the chiral representation, we denote:
    \begin{align*}
        a^{}_\pm \;\text{and/or}\; b^{}_\pm = {{\widehat{\mathbbm{1}\boldsymbol{k}}_\pm\,,\; \widehat{05}_\pm\,,\; \widehat{12}_\pm\,,\; \widehat{34}}}_\pm\,.
    \end{align*}
    For clarity and conciseness, we shall, where appropriate, employ a simplified numeric labeling in Chaps. \ref{Chapter 3} and \ref{Chapter 4} as:
    \begin{align*}
        a := a^{}_+ = 1\, \big(:={{\widehat{\mathbbm{1}\boldsymbol{k}}_+\big)\,,\quad 2\, \big(:= \widehat{05}_+\big)\,,\quad 3\, \big(:=\widehat{12}_+\big)\,,\quad 4\, \big(:=\widehat{34}}}_+\big)\,.
    \end{align*}
    Note that $a,b$ have a fundamentally different meaning from $i,j$.}

    \item{We denote indices $0,1,2,3,4$ by Greek letters $\alpha,\beta,\dots$, and indices $0,1,2,3$ by dotted letters $\dot{\mu}, \dot{\nu},\dots\;$.}

    \item{Repeated indices indicate summation over their corresponding index ranges, unless otherwise specified. The relevant ranges and, where necessary, the manner of summation are specified in each case. In particular, for antisymmetric index pairs, the summation may either be restricted to the independent pairs (e.g., $\mu<\nu$) or, equivalently, be taken over the full index range with the appropriate factor of $1/2$ to account for double counting.}    
    
    \item{The framework of this monograph rests on the mostly positive metric signature, which will serve as our convention from the outset.}

    \item{Throughout this monograph, the symbol `$\cong$' denotes an isomorphism (or homomorphism) between algebras or groups. By a slight abuse of notation, the same symbol is also used to indicate a homeomorphism or continuous homomorphism between groups or topological spaces. Whenever the precise nature of the correspondence is relevant, it will be explicitly stated in the text.}

    \item{In this monograph, various operators are represented as block-diagonal matrices assembled from finite-dimensional matrix blocks (typically $2\times2$ or $4\times4$). To describe such objects unambiguously, we use the following notation:
    \begin{enumerate}[leftmargin=0.8cm]
        \item{\textbf{\textit{Direct-sum notation}:} For any ordered list of matrices \linebreak $(M_{1}, M_{2}, \dots, M_{n})$, we write:
        \begin{align}
            \bigoplus\big(M_{1} + M_{2} + \dots + M_{n}\big) := M_{1} \oplus M_{2} \oplus \dots \oplus M_{n}\,,
        \end{align}
        for the block-diagonal matrix whose diagonal blocks, in the order listed, are $M_{1}, M_{2}, \dots, M_{n}$. For example:
        \begin{align}
            \bigoplus\big(M_{1} + M_{2} + M_{3} + M_{4}\big) 
            = 
            \begin{pmatrix}
                M_{1} & 0 & 0 & 0 \\
                0 & M_{2} & 0 & 0 \\
                0 & 0 & M_{3} & 0 \\
                0 & 0 & 0 & M_{4}
            \end{pmatrix}\,.
        \end{align}
        A minus sign in the list indicates literal multiplication of the corresponding block by $-1$; for instance, $\bigoplus\big(M_{1} -M_{2} + M_{3} + M_{4}\big)$ places $-M_{2}$ in the second block. This convention allows expressions such as:
        \begin{align}
            \bigoplus\big( S_{\mathbbm{1} (\boldsymbol{\jmath}\boldsymbol{\ell})} + S_{\boldsymbol{\jmath}\,\boldsymbol{\ell}} -S_{\boldsymbol{\imath} (\boldsymbol{\ell}\boldsymbol{k})} + S_{\boldsymbol{k} (\boldsymbol{\ell}\boldsymbol{\imath})} \big)
        \end{align}
        to be written compactly and without ambiguity.}

        \item{\textbf{\textit{Commutativity up to isomorphism}:} For finite-dimensional matrices $M$ and $N$, the direct sum is not strictly commutative as a matrix identity. The two block arrangements differ only by a permutation of components, so that:
        \begin{align}
            \bigoplus \big( M+N \big) := M \oplus N\;\;\cong\;\; N \oplus M =: \bigoplus\big( N+M \big) \,,
        \end{align}
        where ``$\cong$'' denotes isomorphism via conjugation by an appropriate permutation matrix.

        In the present context, however, the blocks are arranged such that their order does not affect the resulting operator. That is, for the matrices considered here, the direct sum is literally commutative:
        \begin{align}\label{Oplus M+N=N+M}
             \bigoplus \big( M+N \big) := M \oplus N\;\;=\;\; N \oplus M =: \bigoplus\big( N+M \big) \,.
        \end{align}
        and blocks may be reordered freely without changing the outcome.}
    \end{enumerate}}
\end{enumerate}

\subsubsection{Comprehensive List of References}

We conclude this section with a comprehensive list of references pertinent to the discussion above:
\begin{enumerate}[leftmargin=0.8cm]
    \item{For extensions of Wigner's seminal idea to Galilean systems, see In\"{o}n\"{u} \cite{Wigner1952}, L\'{e}vy-Leblond \cite{Levy-Leblond}, and Voisin \cite{Voisin}. Subsequent extensions to dS and anti-dS systems were developed by G\"{u}rsey \cite{Gursey1963} and Fronsdal \cite{Fronsdal 1, Fronsdal 2}, respectively. More recent studies in this context include \cite{Gazeau2022, Aldrovandi, Gazeau2023}.}

    \item{For the classification of all ten possible relativities connected via three types of contraction, based on symmetry principles and physically motivated assumptions — namely, isotropy of space, parity and time-reversal as group automorphisms, and non-compactness of boosts — see Ref. \cite{Bacry} by Bacry and L\'{e}vy-Leblond. In Ref. \cite{Levy}, emphasizing abstract group-theoretic methods, these results were rigorously established through inverse contraction, i.e., deformation.}

    \item{For the notion of ``masslessness'' and its interplay with conformal symmetry, see, in addition to the references cited above, Refs. \cite{BB, Flato, Binegar1983, DobrevPetkova, Flato', Branson1987, GazeauMurenzi, Francesco, EEMM, Gazeaus1, Jadczyk123, dSgravity2, dSgravity1, Bamba1, Massless2'', Massless2, Massless2', PMG, PMGLetsios, Gezim2023, GaPe2025-Dirac, dSTachyon, Gizem2025, Saha2025, Huguet2026, Bajalan2026}, which provide further details and, in particular, elucidate the structure in dS and anti-dS spacetimes.}

    \item{For the pivotal role of dS spacetime in the formulation of consistent QFTs in curved backgrounds, as well as its fundamental significance for quantum gravity and cosmology, providing a natural laboratory for exploring horizon effects, vacuum structure, and semiclassical phenomena, see Refs. \cite{001, 002, 003, 004, 005, 006, 008, 009, 0010, 0012', 0011, 0012, 0013, 0014, 0015, 0016, 0017}.}

    \item{For a deeper understanding of composition algebras — including the (split-)octonion algebra — encompassing their structural properties, classification, and broad range of applications in mathematics and theoretical physics, particularly in connection with division algebras, spinor constructions, and symmetry structures, the reader is referred to Refs. \cite{R08, Dubois2016, Todorov2017, Todorov2018, S2004, B, B20, T23}.}

    \item{A rigorous treatment of free quantum fields and their vacuum two-point functions is developed in the work of Segal; see in particular Ref. \cite{Segal123}. Analytic and harmonic perspectives related to tube-domain realizations and the kernel structures underlying two-point correlators can be found in the monograph of Hua \cite{Hua123}. Connections between representation theory and analytic function theory on symmetric spaces, which provide a broader geometric framework for such correlation functions, are explored in the seminal work of Piatetski-Shapiro \cite{Shapiro123}.}
\end{enumerate}

\sectionmark{Brief on de Sitter (dS) Relativity: Symmetry Group, Lie Algebra, and ...}
\section{Brief on de Sitter (dS) Relativity: Symmetry Group, Lie Algebra, and Unitary Irreducible Representations (UIRs)}\label{Sect. Brief on dS}
\sectionmark{Brief on de Sitter (dS) Relativity: Symmetry Group, Lie Algebra, and ...}

For future reference and to ensure internal consistency, this section provides a concise overview of dS relativity, covering the dS group and its manifold, the corresponding Lie algebra, and their quantum counterparts — the dS UIRs that characterize elementary systems in dS spacetime. For a more comprehensive and detailed treatment, readers are directed to the companion monograph \cite{Gazeau2022} and the references therein.

\subsection{Symmetry Group and Lie Algebra}

As noted earlier, dS spacetime is the unique maximally symmetric solution to Einstein's equations with a positive cosmological constant $\Lambda$. Its topology is given by $\mathbb{R}^1 \times \mathbb{S}^3$, reflecting that the manifold admits a global foliation by Cauchy hypersurfaces diffeomorphic to $\mathbb{S}^3$. No canonical choice of temporal direction exists; the $\mathbb{R}^1$ factor merely parametrizes the leaves of this foliation. Geometrically, it is a globally hyperbolic Lorentzian manifold with constant positive curvature, uniquely characterized by a curvature radius $R$, which sets the fundamental length scale of the spacetime and governs both its geometric and group-theoretic properties.

The dS manifold admits a convenient representation as a hyperboloid embedded in a $4+1$-dimensional Minkowski spacetime $\mathbb{R}^{4,1}$, and is defined by:
\begin{align}\label{dS-M_R}
    \text{dS} := \Big\{x = (x^\alpha) \in \mathbb{R}^{4,1} \;;\; (x)^2 := x \cdot x = \eta^{}_{\alpha\beta} x^\alpha x^\beta = R^2 \Big\}\,,
\end{align}
where $\alpha, \beta = 0, 1, 2, 3, 4$, the coordinates $x^\alpha$ are Cartesian, and the ambient Minkowski metric $\eta^{}_{\alpha\beta}$ has the signature $(-,+,+,+,+)$.

\begin{remark}
    {\emph{(adopted metric signature and its rationale).} In this monograph, as already mentioned, we adopt a mostly positive metric signature, departing from the convention used in the previous monograph \cite{Gazeau2022}. Far from being a mere change of notation, this choice is motivated by a deeper conceptual reason; it allows for a consistent embedding of all subsequent group and algebraic constructions into the alternative split-octonion algebra, a point that will be clarified naturally at the appropriate places throughout the exposition.}
\end{remark} 

\begin{remark}
    {\emph{(cosmological interpretation of the curvature).} From a cosmological perspective, the curvature radius $R$ is typically associated with the inverse of the Hubble parameter ${\texttt{H}}$, modulo a factor accounting for physical dimensions:
    \begin{align}
        R = \sqrt{\frac{3}{\Lambda}} = \frac{1}{\texttt{H}}\,.
    \end{align}
    The Hubble parameter ${\texttt{H}}$ quantifies the exponential expansion rate of the spatial sections in dS spacetime.}
\end{remark}

The dS relativity group $\mathrm{SO}_0(4,1)$ — or its universal covering group $\mathrm{Sp}(2,2)$ — is a ten-parameter group of linear transformations acting on the ambient $4+1$-dimensional Minkowski spacetime $\mathbb{R}^{4,1}$. These transformations preserve the quadratic form $(x)^2 = \eta^{}_{\alpha\beta} x^\alpha x^\beta$, maintain unit determinant, and preserve the orientation of the ``time'' coordinate $x^0$. A standard realization of the corresponding Lie algebra $\mathfrak{so}(4,1)\cong\mathfrak{sp}(2,2)$ is given by the linear span of ten Killing vector fields:
\begin{align}\label{Killing}
    K_{\alpha\beta} = x_\alpha \partial_\beta - x_\beta \partial_\alpha \,,
\end{align}
where $\partial_\alpha:= {\partial}/{\partial x^\alpha}$.

\subsection{Unitary Irreducible Representations (UIRs)}

At the representation level, the ten dS Killing vectors \eqref{Killing} correspond to (essentially) self-adjoint operators ${\mathcal{J}}_{\alpha\beta}$ acting on a Hilbert space of (spinor-)tensor-valued functions on the dS manifold \eqref{dS-M_R}, square-integrable with respect to a dS-invariant inner product of Klein-Gordon type (or a suitable alternative):
\begin{align}\label{LLLLLLL}
    K_{\alpha\beta} \quad\longmapsto\quad {\mathcal{J}}_{\alpha\beta} = M_{\alpha\beta} + S_{\alpha\beta}\,,
\end{align}
where the orbital part is defined by: 
\begin{align}
    M_{\alpha\beta} = - \mathrm{i} (x_\alpha\partial_\beta - x_\beta\partial_\alpha) \,,
\end{align}
while the spinorial part $S_{\alpha\beta}$ acts on the indices of the (spinor-)tensor-valued functions according to a specific permutation scheme \cite{Gazeau2022}.

\begin{remark}\label{Remark J convention}
    {\emph{(Convention for the generators).} Throughout this chapter, ${\mathcal{J}}_{\alpha\beta}$ denotes the (essentially) self-adjoint generators representing the dS Killing vector fields $K_{\alpha\beta}$. This notation is adopted to distinguish these generators from the anti-Hermitian generators $J_{\alpha\beta}$ (or, more generally, $J_{\mu\nu}$, with $\mu,\nu=0,5,1,2,3,4$) employed in the oscillator realization of Chap. \ref{Chapter 3}. For the dS generators, the two conventions are related by:
    \begin{align}
        {\mathcal{J}}_{\alpha\beta} = \mathrm{i}J_{\alpha\beta}\,.
    \end{align}
    More generally, upon extending the self-adjoint convention to the full conformal algebra, one has ${\mathcal{J}}_{\mu\nu}=\mathrm{i}J_{\mu\nu}$.}
\end{remark}

In this context, there are two Casimir operators:
\begin{align}\label{Casimir 2}
    \mbox{quadratic}\;;\quad {\mathcal{C}}^{\mathfrak{sp}(2,2)}_2 = - \frac{1}{2} {\mathcal{J}}^{\alpha\beta} {\mathcal{J}}_{\alpha\beta}\,,
\end{align}
\begin{align}
    \mbox{quartic}\;;\quad {\mathcal{C}}^{\mathfrak{sp}(2,2)}_4 = - W^\alpha W_\alpha\,,
\end{align}
where the $W_\alpha$, as the dS counterpart of the Pauli-Lubanski operator, is given by: 
\begin{align}
    W_\alpha = - \frac{1}{8} {\varepsilon}_{\tiny{\alpha\beta\gamma\rho\delta}} {\mathcal{J}}^{\beta\gamma} {\mathcal{J}}^{\rho\delta} \,,
\end{align}
while ${\varepsilon}_{\tiny{\alpha\beta\gamma\rho\delta}}$ is the five-dimensional totally anti-symmetric Levi-Civita symbol. These Casimir operators commute with all generator representatives ${\mathcal{J}}_{\alpha\beta}$, and hence, act like constants on all states in a certain dS UIR. Therefore, the eigenvalues assumed by the Casimir operators can be utilized to classify UIRs of the dS group. In the Dixmier notation \cite{Dixmier}, these eigenvalues are expressed in terms of two dS-invariant parameters, $p\in \mathbb{Z}_{\geq 0}/2$ and $q\in\mathbb{C}$, as follows:
\begin{align} \label{Casimir rank 2}
    \big\langle {\mathcal{C}}^{\mathfrak{sp}(2,2)}_2 \big\rangle = -p(p+1) - (q+1)(q-2) \,, 
\end{align}
\begin{align} \label{Casimir rank 4}
    \big\langle {\mathcal{C}}^{\mathfrak{sp}(2,2)}_4 \big\rangle = -p(p+1)q(q-1) \,.
\end{align}
Thus, the pair of dS-invariant parameters $(p,q)$ provides a complete characterization of the UIRs of the dS group. The classification naturally decomposes according to the admissible domains of these parameters, leading to three distinguished series of dS UIRs — principal, complementary, and discrete — as reviewed below (see Refs. \cite{Dixmier, Takahashi, Gazeau2022}).

\begin{remark}
    {\emph{(invariance of the Casimir eigenvalues).} The Casimir eigenvalues remain invariant under the transformation:
    \begin{align}
        q \longmapsto 1-q \,.
    \end{align}
    In other words, the representations labeled by the pairs $(p, q)$ and $(p, 1-q)$ possess identical Casimir spectra. By definition, such representations are said to be Weyl equivalent.}
\end{remark}

\subsubsection{Principal Series Representations}

Principal series representations ${U}^{\mbox{\small{ps}}}_{s,\nu}$ are specified by $(p=s,q=\frac{1}{2}+\mathrm{i}\nu)$, where the label $p=s$ is to be understood as the spin. One must distinguish two separate families:
\begin{enumerate}[leftmargin=0.8cm]
    \item{Integer-spin representations, with $s=0,1,2,\dots$ and $\nu\in \mathbb{R}$.}
    
    \item{Half-integer-spin representations, with $s=\frac{1}{2},\frac{3}{2},\frac{5}{2},\dots$ and $\nu\in \mathbb{R}-\{0\}$.}
\end{enumerate}

The principal UIRs ${U}^{\mbox{\small{ps}}}_{s,\nu}$, in the Poincar\'{e} contraction limit ($R\to\infty$, or equivalently $\Lambda\to 0$), reduce to the positive- and negative-energy Wigner massive UIRs of the Poincar\'{e} group, respectively denoted by ${\mathscr{P}}^\gtrless_{s,m}$ and characterized by spin $s$ and mass $m > 0$ \cite{Mickelsson}; symbolically:
\begin{align}\label{massive contraction}
    {U}^{\mbox{\small{ps}}}_{s,\nu} \;\; \underset{R \rightarrow \infty}{\longrightarrow} \;\; {\mathscr{P}}^>_{s,m} \; \oplus \; {\mathscr{P}}^<_{s,m} \,.
\end{align}
In this sense, the dS principal UIRs are identified as dS ``massive'' representations.

\begin{remark}
    {\emph{(on the contraction of the dS principal series to the Poincar\'{e} UIRs).} The apparent splitting of a dS principal (massive) UIR ${U}^{\mbox{\small{ps}}}_{s,\nu}$ into a direct sum of two Poincar\'{e} massive UIRs ${\mathscr{P}}^\gtrless_{s,m}$ (corresponding to positive- and negative-energy sectors) can be circumvented either by choosing suitable global dS modes — namely, the dS plane waves defined in the relevant analyticity tube domains \cite{Garidi} (see also \cite{Gazeau2022}) — or by carrying out the contraction within the framework of a causality-preserving dS semigroup \cite{Mizony}. On this basis, one obtains:
    \begin{align}\label{massive contraction'}
        {U}^{\mbox{\small{ps}}}_{s,\nu} \;\; \underset{R \rightarrow \infty}{\longrightarrow} \;\; {\mathscr{P}}^>_{s,m}\,.
    \end{align}}
\end{remark}

\subsubsection{Complementary Series Representations}

Complementary series representations ${U}^{\mbox{\small{cs}}}_{s,\nu}$ are specified by $(p=s,q=\frac{1}{2}+\nu)$, where, once again, the parameter $p=s$ is interpreted as the spin. Two distinct subclasses arise:
\begin{enumerate}[leftmargin=0.8cm]
    \item{Scalar representations, with $s=0$ and $\nu\in \mathbb{R}$ satisfying $0<|\nu|<\frac{3}{2}$.}
    
    \item{Spinorial representations, with $s=1,2,3,\dots$ and $\nu\in \mathbb{R}$ satisfying $0<|\nu|<\frac{1}{2}$.}
\end{enumerate}

From the standpoint of a Minkowskian observer, the only complementary series UIR of genuine physical significance is the scalar representation ${U}^{\mbox{\small{cs}}}_{s=0, \nu=\frac{1}{2}}$, i.e., $(p=0,q=1)$. This representation admits a unique extension (denoted by `$\hookrightarrow$') to the conformal group massless UIRs ${\mathscr{C}}^{\gtrless}_{\pm1,0,0}$. This extension is equivalent to the conformal extension of the massless scalar UIRs of the Poincar\'{e} group ${\mathscr{P}}^{\gtrless}_{0,0}$, corresponding respectively to positive and negative energies \cite{BB, Mack}:
\begin{align}\label{rtyt}
    \left. \begin{array}{ccccccc}
    & & {\mathscr{C}}^{>}_{1,0,0} & & {\mathscr{C}}^{>}_{1,0,0} & \hookleftarrow & {\mathscr{P}}^{>}_{0,0} \\[0.3cm]
    {U}^{\mbox{\small{cs}}}_{0,\frac{1}{2}} & \hookrightarrow & \oplus & \underset{R\rightarrow \infty}{\longrightarrow} & \oplus & &\oplus \\[0.3cm]
    & & {\mathscr{C}}^{<}_{-1,0,0} & & {\mathscr{C}}^{<}_{-1,0,0} & \hookleftarrow & {\mathscr{P}}^{<}_{0,0} \,.
    \end{array} \right.
\end{align}
In this sense, the scalar representation ${U}^{\mbox{\small{cs}}}_{0, \frac{1}{2}}$ is identified as ``massless''.

\begin{remark}
    {\emph{(strict conformal invariance and the relevant UIRs).} As already mentioned, conformal invariance, in the strict group-theoretical sense, involves massless (ladder) series representations (and their lowest limits) of (the universal covering of) the conformal group, or its double covering $\mathrm{SO}(4,2)$, or its fourfold covering $\mathrm{SU}(2,2)$. The relevant conformal UIRs are denoted by ${\mathscr{C}}^{\gtrless}_{\mathscr{E}_\circ,\texttt{j}_L, \texttt{j}_R}$, where $\mathscr{E}_\circ$ represents the lowest positive or highest negative conformal energy, and $(\texttt{j}_L,\texttt{j}_R) \in \mathbb{Z}_{\geq 0}/2 \times \mathbb{Z}_{\geq 0}/2$ label the $\mathrm{SU}(2) \times \mathrm{SU}(2)$ representations. This will be discussed in detail in Chap. \ref{Chapter 3}; in particular, see the arguments given in Remark \ref{Remark parity start}.}
\end{remark}

\subsubsection{Discrete Series Representations}

Discrete series representations are denoted by $\Pi_{p,q}^\pm$, with each representation specified by the parameter pair $(p,q)$. In the symmetric case $\Pi^{\pm}_{p=s, q=s}$, with $s = \frac{1}{2}, 1, \frac{3}{2},\dots$, the parameter $p = s$ labels the absolute value of the helicity while the superscript `$\pm$' indicates its sign, corresponding to the two helicity states `$\pm s$'; for later reference, we denote the helicity by $\lambda:= \pm s = \pm\frac{1}{2}, \pm 1, \pm\frac{3}{2},\dots\;$. [For an in-depth examination of the concept of helicity, see Chap. \ref{Chapter 3}; in particular, consult the arguments presented in Remark \ref{Remark parity start} and Appendix \ref{appendix: dS UIRs}.] Within the discrete series, two subclasses can be identified:
\begin{enumerate}[leftmargin=0.8cm]
    \item{Nonsquare-integrable scalar cases, with $p=1,2,\dots$ and $q=0$.}
    
    \item{Spinorial cases, with $p = \frac{1}{2},1,\frac{3}{2},\dots$ and $q = p,p-1,\dots,1$ or $\frac{1}{2}$ ($q>0$); the subclass with $q = \frac{1}{2}$ is not square integrable.}
\end{enumerate}

From the Minkowskian viewpoint, the physically relevant discrete series UIRs are precisely the symmetric cases $\Pi^{\pm}_{p=s,\, q=s}$, with $s = \frac{1}{2},1,\frac{3}{2},\dots$, that lie at the lower limit of the discrete series. These representations admit a unique extension to massless (ladder) UIRs of the conformal group; this conformal extension coincides, in a precise representation-theoretic sense, with the conformal uplift of the massless Poincar\'{e} UIRs of helicity $\lambda = \pm s$ \cite{BB, Mack}:
\begin{align}\label{rtyt2}
    &\left. \begin{array}{ccccccc}
    & & {\cal{C}}^{>}_{s+1,s,0} & & {\cal{C}}^{>}_{s+1,s,0} & \hookleftarrow & {\mathscr{P}}^{>}_{s,0} \\[0.3cm]
    \Pi^+_{s,s} & \hookrightarrow & \oplus & \underset{R\rightarrow \infty}{\longrightarrow} & \oplus & & \oplus \\[0.3cm]
    & & {\cal{C}}^{<}_{-s-1,s,0} & & {\cal{C}}^{<}_{-s-1,s,0} & \hookleftarrow & {\mathscr{P}}^{<}_{s,0}\,,
    \end{array} \right. \\[0.7cm]
    \label{rtyt1}
    &\left. \begin{array}{ccccccc}
    & & {\cal{C}}^{>}_{s+1,0,s} & & {\cal{C}}^{>}_{s+1,0,s} & \hookleftarrow & {\mathscr{P}}^{>}_{-s,0} \\[0.3cm]
    \Pi^-_{s,s} & \hookrightarrow & \oplus & \underset{R\rightarrow \infty}{\longrightarrow} & \oplus & & \oplus \\[0.3cm]
    & & {\cal{C}}^{<}_{-s-1,0,s} & & {\cal{C}}^{<}_{-s-1,0,s} & \hookleftarrow & {\mathscr{P}}^{<}_{-s,0} \,.
    \end{array} \right.
\end{align}
Accordingly, the UIRs $\Pi^{\pm}_{p=s,\, q=s}$, with $s = \frac{1}{2},1,\frac{3}{2},\dots$, are conventionally referred to as ``massless'' representations with helicity $\lambda = \pm s$.

\begin{remark}{
    \emph{(Casimir eigenvalues for the dS massless UIRs).} For later reference — and by including the scalar (zero-helicity; $\lambda=0$) massless case \eqref{rtyt} — the eigenvalues of the quadratic Casimir of the dS Lie algebra $\mathfrak{sp}(2,2)$ associated with the dS massless UIRs, with helicity $\lambda = \pm s= 0, \pm\tfrac{1}{2}, \pm 1, \pm\tfrac{3}{2},\dots$, are:
    \begin{align} \label{Casimir rank 2 massless}
        \big\langle {\mathcal{C}}^{\mathfrak{sp}(2,2)}_2 \big\rangle_{p=s=q} = -2\big( s^2 - 1 \big) = -2\big( \lambda^2 - 1 \big)\,. 
    \end{align}
}\end{remark}

	\renewcommand*\vec{\mathaccent"017E\relax}

\setcounter{equation}{0} 

\chaptermark{Majorana Representation of the Conformal Clifford Algebra $\mathfrak{cl}(4,2)$ and ...}
\chapter{Majorana Representation of the Conformal Clifford Algebra $\mathfrak{cl}(4,2)$ and its Chiral Counterpart}\label{Chapter 2}
\chaptermark{Majorana Representation of the Conformal Clifford Algebra $\mathfrak{cl}(4,2)$ and ...}

\section*{Chapter Overview}

This chapter presents an explicit Majorana representation of the conformal Clifford algebra $\mathfrak{cl}(4,2)$ on $8$-component real spinors, drawing on the rich algebraic structure of the split-octonions — an $8$-dimensional composition algebra constructed via the Cayley-Dickson process. The discussion begins with a systematic overview of the number systems, proceeding from the complex numbers and quaternions to the octonions through this construction. Particular emphasis is placed on the Moufang identities, which capture alternativity, a relaxed form of associativity that remains fully compatible with spinor calculus. In the split-octonion setting, alternativity coexists with non-associativity, creating a subtle but well-controlled framework in which spinorial multiplications retain internal coherence. By exploiting left multiplication by a suitably chosen set of $6$ of the $7$ imaginary split-octonion units, the chapter develops a real-matrix representation of $\mathfrak{cl}(4,2)$ that is both geometrically transparent and ideally suited to the Majorana formalism. This construction provides a unified algebraic framework for the internal symmetries of the conformal Lie algebra $\mathfrak{u}(2,2)$ and naturally realizes the embedding chain:
\begin{align*}
    \mathfrak{u}(2,2) \cong \mathfrak{su}(2,2) \, \big(\cong \mathfrak{so}(4,2)\big) \oplus \mathfrak{u}(1) \;\subset\; \mathfrak{cl}(4,2)\,,
\end{align*}
where the $\mathfrak{cl}(4,2)$ pseudoscalar (volume form) generates the $\mathfrak{u}(1)$ center of $\mathfrak{u}(2,2)$.

Building on this foundation, the chapter investigates the complex structure inherent in $\mathfrak{cl}(4,2)$ and its action on spinor space, using the simpler $\mathfrak{cl}(2)$ case as an instructive prototype. We emphasize that the Majorana (real) basis and the chiral basis are merely two alternative bases of the same real Clifford algebra; no complexification is required. Exploiting this insight, a unitary transformation is constructed that maps the real-matrix basis of $\mathfrak{cl}(4,2)$ to the chiral basis, in which the Cartan elements of the maximal compact subalgebra of $\mathfrak{su}(2,2)$, together with the $\mathfrak{u}(1)$ generator defined by the $\mathfrak{cl}(4,2)$ volume form, assume a diagonal form. These results furnish a coherent, geometrically grounded framework for the study of massless conformal systems, applicable in flat Minkowski spacetime as well as in curved de Sitter (dS) and anti-dS spacetimes.

\sectionmark{Composition Algebras, Cayley-Dickson Construction, and ...}
\section{Composition Algebras, Cayley-Dickson Construction, and Moufang Identities}
\sectionmark{Composition Algebras, Cayley-Dickson Construction, and ...}

A composition algebra $\mathfrak{A}$, defined over the real numbers $\mathbb{R}$, is a finite-dimensional real vector space endowed with the following structures:
\begin{enumerate}[leftmargin=0.8cm]
    \item{\textbf{\emph{Norm form}:} A non-degenerate symmetric bilinear form $\langle\cdot , \cdot \rangle \;;\; \mathfrak{A} \times \mathfrak{A} \;\longmapsto\; \mathbb{R}$, which satisfies the following properties:
    \begin{enumerate}[leftmargin=0.8cm]
        \item{Symmetry, meaning that $\langle x , y\rangle = \langle y , x \rangle$, for all $x, y \in \mathfrak{A}$.}
        \item{Non-degeneracy, implying that the map $x \longmapsto \langle x , \cdot \rangle$ is injective. [In other words, if $\langle \cdot, \cdot \rangle$ is a non-degenerate bilinear form on $\mathfrak{A}$ and $\langle x, y \rangle = \langle x^\prime, y \rangle$, for all $y \in \mathfrak{A}$, then $x = x^\prime$. To see this, let $z = x-x^\prime$. Then, for all $y \in \mathfrak{A}$, we have $\langle z, y \rangle = \langle x, y \rangle - \langle x^\prime, y \rangle = 0$. By the non-degeneracy of $\langle \cdot, \cdot \rangle$, we obtain $z = 0$, and hence, $x = x^\prime$.]}
    \end{enumerate}
    This bilinear form induces a quadratic form:
    \begin{align}
    	N(x) := \langle x, x \rangle \,,
    \end{align}
    referred to as the ``norm'' on $\mathfrak{A}$.}
    
    \item{\textbf{\emph{Bilinear product}:} A bilinear product $\diamond \;;\; \mathfrak{A} \times \mathfrak{A} \;\longmapsto\; \mathfrak{A}$, denoted by $x \diamond y =: xy$, satisfying the norm composition property:
    \begin{align}\label{normdecom}
        N(xy) = N(x)N(y)\,,
    \end{align} 
    for all $x, y \in \mathfrak{A}$. [This multiplicative property guarantees that every element of non-zero norm in $\mathfrak{A}$ has a well-defined inverse, making composition algebras particularly valuable in physics and geometry for describing geometric symmetries and transformations. In particular, multiplication by an element of unit norm defines an invertible norm-preserving transformation, whose inverse restores the original element. More generally, for any $x, y \in \mathfrak{A}$, the identity $N(xy) = N(x)N(y)$ shows that the norm of a product is determined by the product of the individual norms. Consequently, multiplication by elements of unit norm preserves the quadratic norm, providing a natural algebraic framework for constructing rotations and other norm-preserving symmetry transformations.]}

    \item{\textbf{\emph{Unit element}:} A distinguished element $\mathbbm{1} \in \mathfrak{A}$ (the unit), satisfying $ \mathbbm{1}x = x\mathbbm{1} = x$, for all $x \in \mathfrak{A} $. From the above, it is clear that $N(\mathbbm{1}) = 1$. Subsequently, for each $\varrho\in\mathbb{R}$ (i.e., the members of the field), we have $N(\varrho\mathbbm{1}) = \varrho^2$.}
\end{enumerate}

Considering the above, one can readily verify the following identity:
\begin{align}\label{muhum}
    \langle x , y \rangle = \frac{1}{2} \big( N(x+y) - N(x) - N(y) \big) \,, 
\end{align}
for all $x,y \in \mathfrak{A}$, which expresses the bilinear form $\langle \cdot, \cdot \rangle$ in terms of the quadratic norm $N$. As a direct consequence of this identity, together with the norm composition property \eqref{normdecom}, it follows that:
\begin{align}\label{muhumvv}
    \langle z x , z y \rangle = \langle x z , y z \rangle = N(z)\langle x , y \rangle\,,
\end{align}
again, for all $x,y,z \in \mathfrak{A}$. The identity \eqref{muhumvv} then immediately yields:
\begin{align}\label{multi}
    \langle x , y \rangle \langle z , w \rangle = \frac{1}{2} \big( \langle xz , yw \rangle + \langle xw , yz \rangle \big)\,,
\end{align}
for all $x,y,z,w \in \mathfrak{A}$. To establish the latter identity, using Eq. \eqref{muhumvv}, it suffices to expand the right-hand side:
\begin{align}
    \langle xz , yw \rangle + \langle xw , yz \rangle &=  \langle x(z+w) , y(z+w) \rangle - \langle xz , yz \rangle - \langle xw , yw \rangle \nonumber\\[0.2cm]
    &= N(z+w) \langle x , y \rangle - N(z) \langle x , y \rangle - N(w) \langle x , y \rangle \nonumber\\[0.2cm]
    &= \langle x , y \rangle \big( N(z+w) - N(z) - N(w) \big) \nonumber\\[0.2cm]
    &= 2\langle x , y \rangle \langle z , w \rangle \,.
\end{align}

On the other hand, substituting $y$ by $\mathbbm{1}$ and $z$ by $x$ into Eq. \eqref{multi}, and applying \eqref{muhumvv}, we obtain:
\begin{align}
    \langle x , \mathbbm{1} \rangle \langle x , w \rangle = \frac{1}{2} \big( \langle x^2 , w \rangle + N(x) \langle w , \mathbbm{1} \rangle \big)\,,
\end{align}
or equivalently:
\begin{align}
    \left\langle x^2 + N(x)\mathbbm{1} - 2\langle x , \mathbbm{1} \rangle\, x \;,\; w \right\rangle = 0\,.
\end{align}
Then, the non-degeneracy of the bilinear form $\langle \cdot , \cdot\rangle$ ensures that every element $x$ in the composition algebra $\mathfrak{A}$ satisfies the following second-order equation:
\begin{align}\label{secondOrder}
    x^2 - \mathrm{Tr}(x)\, x + N(x)\mathbbm{1} = 0\,,
\end{align}
where $\mathrm{Tr}(x)$ denotes the trace in the composition algebra $\mathfrak{A}$, defined by:
\begin{align} \label{trace}
    \mathrm{Tr}(x) = 2\langle x , \mathbbm{1} \rangle \; \big( \in\mathbb{R}\big)\,.
\end{align}
By employing Eq. \eqref{multi} alongside the symmetric property of the bilinear form $\langle \cdot, \cdot \rangle$, it follows that the trace satisfies the symmetry property:
\begin{align}\label{xy=yx}
    \mathrm{Tr}(xy) &= 2 \langle xy, \mathbbm{1} \rangle \nonumber\\[0.2cm]
    &= 4\langle x, \mathbbm{1} \rangle \langle y, \mathbbm{1} \rangle - 2\langle x, y \rangle \nonumber\\[0.2cm]
    &= 4\langle y, \mathbbm{1} \rangle \langle x, \mathbbm{1} \rangle - 2\langle y, x \rangle \nonumber\\[0.2cm]
    &= 2 \langle yx, \mathbbm{1} \rangle = \mathrm{Tr}(yx) \,,
\end{align}
for all $x,y \in \mathfrak{A}$. 

Within this framework, the Cayley-Dickson conjugate of an element $x \in \mathfrak{A}$ is defined by:
\begin{align}\label{CD-conjugate}
    x^\ast := \mathrm{Tr}(x)\mathbbm{1} - x\,.
\end{align}
Utilizing this identity, one readily verifies that:
\begin{align}\label{x+y=x+y}
    \mathrm{Tr}(x+y) = \mathrm{Tr}(x) + \mathrm{Tr}(y)\,,
\end{align}
for all $x,y \in \mathfrak{A}$.

By substituting $\mathrm{Tr}(x)$ from Eq. \eqref{CD-conjugate} into \eqref{secondOrder}, we derive the following explicit form for the norm on $\mathfrak{A}$:
\begin{align}\label{010203}
    N(x) = x^\ast x\,.
\end{align}
Using the materials mentioned above — specifically Eqs. \eqref{muhum} and \eqref{xy=yx}-\eqref{010203}, one readily verifies that the Cayley-Dickson conjugation satisfies the following key properties:
\begin{align}
    \text{Additivity ;}\quad (x+y)^\ast = x^\ast + y^\ast\,, 
\end{align}
\begin{align}
    \text{Involution ;}\quad (x^\ast)^\ast = x\,, 
\end{align}
\begin{align}
    \text{Norm preservation ;}\quad N(x) = x^\ast x = xx^\ast = N(x^\ast)\,, 
\end{align}
\begin{align}
    \label{Multiplication reversal} \text{Multiplication reversal ;}\quad (xy)^\ast = y^\ast x^\ast\,.
\end{align}
While the first three identities follow almost immediately from the relations mentioned above, establishing the latter requires further clarification:
\begin{align}
    (xy)^\ast &= \mathrm{Tr}(xy) \mathbbm{1} - xy \nonumber\\[0.2cm]
    &= 4\langle y, \mathbbm{1} \rangle \langle x, \mathbbm{1} \rangle \mathbbm{1} - 2\langle y, x \rangle \mathbbm{1} - xy \nonumber\\[0.2cm]
    &= \mathrm{Tr}(y) \mathrm{Tr}(x)\mathbbm{1} - \big( N(y+x) - N(y) - N(x) \big)\mathbbm{1} - xy \nonumber\\[0.2cm]
    &= \mathrm{Tr}(y) \mathrm{Tr}(x)\mathbbm{1} - \big( (y+x)(y+x)^\ast - yy^\ast - xx^\ast \big)\mathbbm{1} - xy \nonumber\\[0.2cm]
    &= \mathrm{Tr}(y) \mathrm{Tr}(x)\mathbbm{1} - x \big(y^\ast + y\big) - yx^\ast \nonumber\\[0.2cm]
    &= \mathrm{Tr}(y) \mathrm{Tr}(x)\mathbbm{1} - x \big(\mathrm{Tr}(y)\mathbbm{1}\big) - y\big(\mathrm{Tr}(x)\mathbbm{1} - x\big) \nonumber\\[0.2cm]
    &= \big( \mathrm{Tr}(y)\mathbbm{1} - y \big) \big( \mathrm{Tr}(x)\mathbbm{1} - x \big) = y^\ast x^\ast \,.
\end{align}

A fundamental result, known as the Hurwitz theorem, constrains the possible dimensions of composition algebras over $\mathbb{R}$. These algebras can only exist in dimensions $1, 2, 4, $ and $8$. The corresponding cases are:
\begin{enumerate}[leftmargin=0.8cm]
    \item{The real numbers $\mathbb{R}$, which are both commutative and associative.}
    \item{The complex numbers $\mathbb{C}$, which are also commutative and associative.}
    \item{The quaternions $\mathbb{H}$, which are non-commutative but retain associativity.}
    \item{The octonions $\mathbb{O}$, which are neither commutative nor associative.}
\end{enumerate}
As already mentioned, these algebras serve as essential tools in mathematics and physics, particularly in areas requiring norm-preserving transformations, such as symmetry operations and rotations. 

To ensure this manuscript is self-contained, the following subsections will delve deeper into the structure and properties of the complex numbers, quaternions, and octonions, providing a comprehensive overview of their defining characteristics.

\subsection{Complex Numbers}\label{Sect. complex numbers}

The complex numbers $\mathbb{C} \ni z = \mathbbm{1} z_{}^{0} + \mathrm{i} z_{}^{1} =: z_{}^{0} + \mathrm{i} z_{}^{1}$ are defined in terms of their real components $z_{}^{0}, z_{}^{1} \in \mathbb{R}$. The real and imaginary units are denoted by $\mathbbm{1} = \mathbbm{1}1 + \mathrm{i} 0 =: 1 $ and $ \mathrm{i}$, respectively, with the property that $\mathrm{i}^2 = -1$. This structure forms a (commutative and associative) composition algebra over the real numbers when equipped with:  
\begin{enumerate}[leftmargin=0.8cm]
    \item{\textbf{\emph{Real bilinear form}:}
    \begin{align}
        \langle z , w \rangle = \left\langle z_{}^{0} + \mathrm{i} z_{}^{1} , w_{}^{0} + \mathrm{i} w_{}^{1} \right\rangle = z_{}^{0} w_{}^{0} + z_{}^{1} w_{}^{1}\,,
    \end{align}
    for any $z,w \in \mathbb{C}$ ($ z_{}^{0}, z_{}^{1}, w_{}^{0}, w_{}^{1} \in \mathbb{R} $). The above bilinear form induces the quadratic norm $ N(z) = (z^{0})^2 + (z^{1})^2 $. Moreover, considering Eqs. \eqref{trace} and \eqref{CD-conjugate}, the above bilinear form yields:
    \begin{align}
        \overline{z} = \overline{ z_{}^{0} + \mathrm{i} z_{}^{1} } = z_{}^{0} - \mathrm{i} z_{}^{1}\,,
    \end{align}
    where the `$\overline{\phantom{a}}$' symbol denotes the complex conjugate.}  
    \item{\textbf{\emph{Multiplication operation}:}
    \begin{align}
        z w &= \big(z_{}^{0} + \mathrm{i} z_{}^{1}\big) \big(w_{}^{0} + \mathrm{i} w_{}^{1}\big) = \big(z_{}^{0} w_{}^{0} - z_{}^{1} w_{}^{1}\big) + \mathrm{i} \big(z_{}^{0} w_{}^{1} + z_{}^{1} w_{}^{0}\big)\,,
    \end{align}
    again, for any $z,w \in \mathbb{C}$ ($ z_{}^{0}, z_{}^{1}, w_{}^{0}, w_{}^{1} \in \mathbb{R} $). Employing this identity, it is easy to show that $N(z) = \overline{z}\, z = z\,\overline{z} = N(\overline{z})$.}
\end{enumerate}

The product of two (arbitrary) complex numbers $z$ and $w$ satisfies the norm composition property:
\begin{align}
    N(zw) = N(z) N(w)\,.
\end{align}
To explicitly verify this, observe that:
\begin{align}
    N(zw) 
    &= \big(z_{}^{0} w_{}^{0} - z_{}^{1} w_{}^{1}\big)^2 + \big(z_{}^{0} w_{}^{1} + z_{}^{1} w_{}^{0}\big)^2 \nonumber\\[0.2cm]
    &= \big((z_{}^{0})^2 + (z_{}^{1})^2\big) \big((w_{}^{0})^2 + (w_{}^{1})^2\big) = N(z) N(w)\,.
\end{align}  
This multiplicative property establishes that the complex numbers form a composition algebra, where the multiplication is both commutative and associative. 

\subsection{Quaternions: Extending the Complex Numbers}\label{Sect. quat}

Building on the properties of complex numbers, quaternions are a $4$-dimensio\-nal extension of the number system. Introduced by William Rowan Hamilton in 1843, quaternions are defined as numbers of the form:
\begin{align}\label{quat def}
    \mathbb{H} = \Big\{ q &= \mathbbm{1} q^{0} + {{\boldsymbol{\imath}}} q^{1} + {{\boldsymbol{\jmath}}} q^{2} + {\boldsymbol{k}} q^{3} \nonumber\\[0.1cm]
    &=: q^{0} + {{\boldsymbol{\imath}}} q^{1} + {{\boldsymbol{\jmath}}} q^{2} + {\boldsymbol{k}} q^{3} \;\; ; \;\; q^{0},q^{1},q^{2},q^{3} \in \mathbb{R} \Big\}\,,
\end{align}
where $\mathbbm{1} = \mathbbm{1}1 + {\boldsymbol{\imath}} 0 + {{\boldsymbol{\jmath}}} 0 + {\boldsymbol{k}} 0 =: 1$ is the real unit element, and ${\boldsymbol{\imath}}$, ${\boldsymbol{\jmath}}, \boldsymbol{k}$ are the quaternion imaginary units. These units obey the following defining multiplication rules:
\begin{align}
    {{\boldsymbol{\imath}}} {{\boldsymbol{\jmath}}} = - {{\boldsymbol{\jmath}}} {{\boldsymbol{\imath}}} = {\boldsymbol{k}}\,, \quad {{\boldsymbol{\jmath}}} {\boldsymbol{k}} = - {\boldsymbol{k}} {{\boldsymbol{\jmath}}} = {{\boldsymbol{\imath}}}\,, \quad {\boldsymbol{k}} {{\boldsymbol{\imath}}} = - {{\boldsymbol{\imath}}} {\boldsymbol{k}} = {{\boldsymbol{\jmath}}}\,,
\end{align}
\begin{align}
    \mathbbm{1} {{\boldsymbol{\imath}}} = {{\boldsymbol{\imath}}} \mathbbm{1} = {{\boldsymbol{\imath}}} \,, \quad \mathbbm{1} {{\boldsymbol{\jmath}}} = {{\boldsymbol{\jmath}}} \mathbbm{1} = {{\boldsymbol{\jmath}}} \,, \quad \mathbbm{1} {\boldsymbol{k}} = {\boldsymbol{k}} \mathbbm{1} = {\boldsymbol{k}}\,,
\end{align}
\begin{align}
    ({{\boldsymbol{\imath}}})^2 = ({{\boldsymbol{\jmath}}})^2 = ({\boldsymbol{k}})^2 = - \mathbbm{1} =: -1 \,.
\end{align}
Then, the multiplication of two quaternions $q$ and $p$ reads as:
\begin{align} \label{quat times}
    q p =&\, \left( \mathbbm{1} q^{0} + {{\boldsymbol{\imath}}} q^{1} + {{\boldsymbol{\jmath}}} q^{2} + {\boldsymbol{k}} q^{3} \right) \left( \mathbbm{1} p^{0} + {{\boldsymbol{\imath}}} p^{1} + {{\boldsymbol{\jmath}}} p^{2} + {\boldsymbol{k}} p^{3} \right) \nonumber\\[0.1cm]
    =&\, \mathbbm{1} \left( q^{0} p^{0} - q^{1} p^{1} - q^{2} p^{2} - q^{3} p^{3} \right) \nonumber\\[0.1cm]
    &+ {{\boldsymbol{\imath}}} \left( q^{0} p^{1} + q^{1} p^{0} + q^{2} p^{3} - q^{3} p^{2} \right) \nonumber\\[0.1cm]
    &+ {{\boldsymbol{\jmath}}} \left( q^{0} p^{2} + q^{2} p^{0} + q^{3} p^{1} - q^{1} p^{3} \right) \nonumber\\[0.1cm]
    &+ {\boldsymbol{k}} \left( q^{0} p^{3} + q^{3} p^{0} + q^{1} p^{2} - q^{2} p^{1} \right) \,,
\end{align}
where $q^{0}, \ldots, q^{3}, p^{0}, \ldots, p^{3} \in \mathbb{R}$. The above explanation makes it evident that quaternion multiplication, while associative, is not commutative.

Under conjugation, the imaginary units of the quaternion behave as:
\begin{align}
    ({\boldsymbol{\imath}})^{\ast} = -{\boldsymbol{\imath}}\,, \quad ({\boldsymbol{\jmath}})^{\ast} = -{\boldsymbol{\jmath}}\,, \quad &(\boldsymbol{k})^{\ast} = -\boldsymbol{k}\,, \nonumber\\[0.2cm] 
    \mbox{and, of course,} \quad (\mathbbm{1})^{\ast} = &\mathbbm{1} =: 1\,.
\end{align}
Then, the squared norm of a quaternion $q$ is given by:
\begin{align}
    N(q) = q^\ast q = qq^{\ast} = \big(q^{0}\big)^2 + \big(q^{1}\big)^2 + \big(q^{2}\big)^2 + \big(q^{3}\big)^2\,,
\end{align}
and, subsequently, the inverse of a non-zero quaternion $q$ by: 
\begin{align}
    q^{-1} = \frac{q^{\ast}}{N(q)}\,.
\end{align}
The latter identity is a clear consequence of the multiplicative property of the composition algebras we discussed earlier.

\subsection{(Split-)Octonions: Beyond Associativity}

The octonions $\mathbb{O}$ and split-octonions $\mathbb{O}_{\mathbb{S}}$ are $8$-dimensional, non-commutative, and non-associative composition algebras over the real numbers $\mathbb{R}$. They are constructed via the Cayley-Dickson process by extending the quaternions $ \mathbb{H} $ with an additional imaginary unit $\boldsymbol{\ell}$. The defining property of this unit distinguishes the octonions $\mathbb{O}$ from their split counterpart $\mathbb{O}_{\mathbb{S}}$:
\begin{align} \label{oct def} 
    \mathbb{O}_{(\mathbb{S})} = \Big\{ o = q + \boldsymbol{\ell} p \;;&\;q, p \in \mathbb{H}\,, \quad (\boldsymbol{\ell})^2 = 
    \begin{cases} 
    -\mathbbm{1} =: -1 & \text{for } o \in \mathbb{O} \\[0.2cm]
    +\mathbbm{1} =: +1 & \text{for } o \in \mathbb{O}_{\mathbb{S}} 
    \end{cases} 
    \Big\}\,,
\end{align}
where $\mathbb{O}_{(\mathbb{S})}$ denotes either $\mathbb{O}$ or $\mathbb{O}_{\mathbb{S}}$. Note that, by definition, the term ``imaginary unit'' refers to an element whose square is $-\mathbbm{1}$. However, in the context of the split-octonions $\mathbb{O}_{\mathbb{S}}$, the extra unit $\boldsymbol{\ell}$ — satisfying $(\boldsymbol{\ell})^2 = +\mathbbm{1}$ — is also commonly (though imprecisely) referred to as an imaginary unit, by abuse of notation.

The (split-)octonion algebras are endowed with the following operations for all $q^{}_1, p^{}_1, q^{}_2, p^{}_2 \in \mathbb{H}$:
\begin{enumerate}[leftmargin=0.8cm]
    \item{\textbf{\emph{Addition}}: 
    \begin{align}
        \left(q^{}_1 + \boldsymbol{\ell} p^{}_1\right) + \left(q^{}_2 + \boldsymbol{\ell} p^{}_2\right) = \left(q^{}_1 + q^{}_2 \right) + \boldsymbol{\ell} \left(p^{}_1 + p^{}_2 \right)\,.
    \end{align}}
    \item{\textbf{\emph{Multiplication}}: Obeying the norm composition property \eqref{normdecom}, the product of two octonions reads as:\footnote{For a mnemonic matrix rule for (split-)octonionic multiplication, see Appendix \ref{Appendix Mnemonic}.}
    \begin{align}\label{Multi-oct}
        \left(q^{}_1 + \boldsymbol{\ell} p^{}_1\right) \left(q^{}_2 + \boldsymbol{\ell} p^{}_2\right) = q^{}_1 q^{}_2 + (\boldsymbol{\ell})^2 p^{}_2 {p^{\ast}_1} + \boldsymbol{\ell} \left( {q^{\ast}_1} p^{}_2 + q^{}_2 p^{}_1 \right)\,.
    \end{align}}
    \item{\textbf{\emph{Conjugation}}: 
    \begin{align}\label{Conju-oct}
        \left( q + \boldsymbol{\ell} p \right)^\ast = {q}^\ast - \boldsymbol{\ell} p \,, \quad (\mbox{hence}\; \boldsymbol{\ell}^\ast = -\boldsymbol{\ell})\,.
    \end{align}
    Note that, considering the anti-commutativity of conjugation with respect to multiplication (see Eq. \eqref{Multiplication reversal}), this identity directly implies the following essential relations:  
    \begin{align}
        p^\ast \boldsymbol{\ell} = \boldsymbol{\ell} p\,,
    \end{align}  
    for all $p\in\mathbb{H}$, which highlight the interplay between the quaternionic components and the additional imaginary unit $\boldsymbol{\ell}$.}
\end{enumerate}

Using the octonionic conjugation \eqref{Conju-oct} and the multiplication rule \eqref{Multi-oct}, the norm of a (split-)octonion can be expressed as:
\begin{align}
    N(q + \boldsymbol{\ell} p) = (q + \boldsymbol{\ell} p) (q + \boldsymbol{\ell} p)^\ast = N(q) - (\boldsymbol{\ell})^2 N(p) \,,
\end{align}
where, as we already discussed, $N(q) = qq^\ast = q^\ast q$ is the quaternionic norm. This leads to the following key distinctions:
\begin{enumerate}[leftmargin=0.8cm]
    \item{For the octonions $\mathbb{O}$, where $(\boldsymbol{\ell})^2 = -\mathbbm{1} =: -1$, the norm is positive-definite:
    \begin{align}
        N\left( q + \boldsymbol{\ell} p \right) = N\left( q \right) + N\left( p \right)\,.
    \end{align}
    This makes $\mathbb{O}$ a division algebra, as every non-zero element is invertible.}
    \item{For the split-octonions $\mathbb{O}_{\mathbb{S}}$, where $(\boldsymbol{\ell})^2 = +\mathbbm{1} =: +1$, the norm is indefinite (with the signature $ (4,4) $):
    \begin{align}
        N\left( q + \boldsymbol{\ell} p \right) = N\left( q \right) - N\left( p \right)\,.
    \end{align}
    The indefinite norm allows for zero divisors, preventing the algebra from forming a division algebra.}
\end{enumerate}

\subsubsection{Alternativity}

It must be emphasized that the resulting (split-)octonion algebras are non-associative and, more precisely, are classified as alternative algebras. While (split-)octonions are not associative, their alternativity ensures that certain partial forms of associativity are preserved when elements are repeated. In particular:
\begin{align} \label{M1} 
    o (\boldsymbol{n}\, o ) = (o \,\boldsymbol{n}) o \,, 
\end{align}
\begin{align}\label{M1H}
    o^{\ast} (\boldsymbol{n}\, o^{}) = (o^{\ast} \boldsymbol{n} ) o^{} \,,
\end{align}
for all $\boldsymbol{n}, o \in \mathbb{O}_{(\mathbb{S})}$. Furthermore, the (split-)octonions adhere to the key left, right, and central Ruth Moufang identities \cite{R08}, which are given respectively as:
\begin{align} \label{M2} 
    o \big(\boldsymbol{m}\,(o\, \boldsymbol{n})\big) = (o\, \boldsymbol{m}\, o ) \boldsymbol{n}\,, 
\end{align}
\begin{align}\label{M2H}
    \big((\boldsymbol{m}\, o)\, \boldsymbol{n}\, o\big) = \boldsymbol{m} (o\, \boldsymbol{n}\, o)\,, 
\end{align}
\begin{align}\label{M2HH}
    (o\, \boldsymbol{m})(\boldsymbol{n}\, o) = o (\boldsymbol{m}\, \boldsymbol{n}) o\,,
\end{align}
for all $\boldsymbol{m}, \boldsymbol{n}, o \in \mathbb{O}_{(\mathbb{S})}$.

To complete this picture, we invoke the quaternionic imaginary units ${\boldsymbol{\imath}}, {\boldsymbol{\jmath}}$, and $\boldsymbol{k}$ and revisit the (split-)octonion definition \eqref{oct def}. The structure naturally extends to include $7$ distinct imaginary units, which we can explicitly enumerate as: 
\begin{align}
    {\boldsymbol{\imath}},\;\; {\boldsymbol{\jmath}},\;\; \boldsymbol{k},\;\; \boldsymbol{\ell},\;\; \boldsymbol{\ell}{\boldsymbol{\imath}},\;\; {\boldsymbol{\jmath}}\boldsymbol{\ell}\; \big( = -\boldsymbol{\ell}{\boldsymbol{\jmath}}\big),\;\; \boldsymbol{\ell}\boldsymbol{k}\,.
\end{align}
These units, together with the real unit $\mathbbm{1} = 1 + {\boldsymbol{\imath}}0 + {\boldsymbol{\jmath}}0 + \ldots + \boldsymbol{\ell}\boldsymbol{k} 0 =: 1$, form the full basis of the (split-)octonions. Their interplay reflects the non-commutative and non-associative nature of the algebra, while satisfying the alternativity properties given in Eqs. \eqref{M1}-\eqref{M2HH}. Specifically, the multiplication rules governing these units can be systematically derived from the general structure provided in Eq. \eqref{Multi-oct}. For the split-octonion algebra $\mathbb{O}_{\mathbb{S}}$ relevant to our construction, these multiplication rules are summarized in Table \ref{table1}.

\begin{table}
    \centering
    \Large
    \renewcommand{\arraystretch}{1.6}
    \begin{tabular}{c|c|c|c|c|c|c|c|c}
        $x\diamond y := xy$ \;&\; $\mathbbm{1}$ \;&\; ${\boldsymbol{\imath}}$ \;&\; ${\boldsymbol{\jmath}}$ \;&\; $\boldsymbol{k}$ \;&\; $\boldsymbol{\ell}$ \;&\; $\boldsymbol{\ell}{\boldsymbol{\imath}}$ \;&\; ${\boldsymbol{\jmath}}\boldsymbol{\ell}$ \;&\; $\boldsymbol{\ell}\boldsymbol{k}$ \\
        \hline
        $\mathbbm{1}$ & {$ \mathbbm{1} $} & $ {\boldsymbol{\imath}} $ \;&\; $ {\boldsymbol{\jmath}} $ \;&\; $ \boldsymbol{k} $ \;&\; $ \boldsymbol{\ell} $ \;&\; $ \boldsymbol{\ell}{\boldsymbol{\imath}} $ \;&\; $ {\boldsymbol{\jmath}}\boldsymbol{\ell} $ \;&\; $ \boldsymbol{\ell}\boldsymbol{k} $ \\
        \hline
        ${\boldsymbol{\imath}}$ \;&\; $ {\boldsymbol{\imath}} $ & {$ -\mathbbm{1} $} & $ \boldsymbol{k} $ \;&\; $ -{\boldsymbol{\jmath}} $ \;&\; $ -\boldsymbol{\ell}{\boldsymbol{\imath}} $ \;&\; $ \boldsymbol{\ell} $ \;&\; $ \boldsymbol{\ell}\boldsymbol{k} $ \;&\; $ -{\boldsymbol{\jmath}}\boldsymbol{\ell} $ \\
        \hline
        ${\boldsymbol{\jmath}}$ \;&\; $ {\boldsymbol{\jmath}} $ \;&\; $ -\boldsymbol{k} $ & {$ -\mathbbm{1} $} & $ {\boldsymbol{\imath}} $ \;&\; $ {\boldsymbol{\jmath}}\boldsymbol{\ell} $ \;&\; $ \boldsymbol{\ell}\boldsymbol{k} $ \;&\; $ -\boldsymbol{\ell} $ \;&\; $ -\boldsymbol{\ell}{\boldsymbol{\imath}} $ \\
        \hline
        $\boldsymbol{k}$ \;&\; $ \boldsymbol{k} $ \;&\; $ {\boldsymbol{\jmath}} $ \;&\; $ -{\boldsymbol{\imath}} $ & {$ -\mathbbm{1} $} & $ -\boldsymbol{\ell}\boldsymbol{k} $ \;&\; $ {\boldsymbol{\jmath}}\boldsymbol{\ell} $ \;&\; $ -\boldsymbol{\ell}{\boldsymbol{\imath}} $ \;&\; $ \boldsymbol{\ell} $ \\
        \hline
        $\boldsymbol{\ell}$ \;&\; $ \boldsymbol{\ell} $ \;&\; $ \boldsymbol{\ell}{\boldsymbol{\imath}} $ \;&\; $ -{\boldsymbol{\jmath}}\boldsymbol{\ell} $ \;&\; $ \boldsymbol{\ell}\boldsymbol{k} $ & {$ \mathbbm{1} $} & $ {\boldsymbol{\imath}} $ \;&\; $ -{\boldsymbol{\jmath}} $ \;&\; $ \boldsymbol{k} $ \\
        \hline
        $\boldsymbol{\ell}{\boldsymbol{\imath}}$ \;&\; $ \boldsymbol{\ell}{\boldsymbol{\imath}} $ \;&\; $ -\boldsymbol{\ell} $ \;&\; $ -\boldsymbol{\ell}\boldsymbol{k} $ \;&\; $ -{\boldsymbol{\jmath}}\boldsymbol{\ell} $ \;&\; $ -{\boldsymbol{\imath}} $ \;&\; {$ \mathbbm{1} $} & $ -\boldsymbol{k} $ \;&\; $ -{\boldsymbol{\jmath}} $ \\
        \hline
        ${\boldsymbol{\jmath}}\boldsymbol{\ell}$ \;&\; $ {\boldsymbol{\jmath}}\boldsymbol{\ell} $ \;&\; $ -\boldsymbol{\ell}\boldsymbol{k} $ \;&\; $ \boldsymbol{\ell} $ \;&\; $ \boldsymbol{\ell}{\boldsymbol{\imath}} $ \;&\; $ {\boldsymbol{\jmath}} $ \;&\; $ \boldsymbol{k} $ \;&\; {$ \mathbbm{1} $} & $ -{\boldsymbol{\imath}} $ \\
        \hline
        $\boldsymbol{\ell}\boldsymbol{k}$ \;&\; $ \boldsymbol{\ell}\boldsymbol{k} $ \;&\; $ {\boldsymbol{\jmath}}\boldsymbol{\ell} $ \;&\; $ \boldsymbol{\ell}{\boldsymbol{\imath}} $ \;&\; $ -\boldsymbol{\ell} $ \;&\; $ -\boldsymbol{k} $ \;&\; $ {\boldsymbol{\jmath}} $ \;&\; $ {\boldsymbol{\imath}} $ \;&\; {$ \mathbbm{1} $} \\
    \end{tabular}
    \caption{Multiplication table for the basis units of the split-octonion algebra $\mathbb{O}_{\mathbb{S}}$.}
    \label{table1}
\end{table}

Keeping the focus specifically on the split-octonion algebra relevant to our study, three complementary remarks are in order:
\begin{remark}\label{Remark aval}
    {\emph{(anti-commutativity and generation of the imaginary split-octonion units).} We observe that:
    \begin{enumerate}[leftmargin=0.8cm]
        \item{The $7$ imaginary split-octonion units are mutually anti-commutative.}
        \item{These anti-commutative imaginary split-octonion units can be expressed in terms of the product of the three basic elements $\boldsymbol{\imath}, {\boldsymbol{\jmath}}$, and $\boldsymbol{\ell}$.}
    \end{enumerate}}
\end{remark}

\begin{remark}\label{Remark 2.2222}
    {\emph{(associativity of pair-generated subalgebras).} Analogous to the identities given in Eqs. \eqref{M1}-\eqref{M2HH}, any pair of different imaginary split-octonion units generates a $4$-dimensional associative subalgebra, containing the identity element $\mathbbm{1}$:
    \begin{enumerate}[leftmargin=0.8cm]
        \item{The pair $ ({\boldsymbol{\imath}}, {\boldsymbol{\jmath}}) $ generates the quaternion subalgebra: 
        \begin{align}
            \big\{\pm{\boldsymbol{\imath}},\, \pm{\boldsymbol{\jmath}},\, \pm\boldsymbol{k},\, \pm \mathbbm{1}\big\} \,,
        \end{align}
        which is associative by construction. For example:
        \begin{align}
            \boldsymbol{\imath} (\boldsymbol{\jmath}\boldsymbol{k}) = (\boldsymbol{\imath} \boldsymbol{\jmath})\boldsymbol{k} \,.
        \end{align}}
        
        \item{Any other pair of distinct imaginary units likewise generates an associative subalgebra of split-octonions. For instance, the pair $({\boldsymbol{\imath}}, {\boldsymbol{\jmath}}\boldsymbol{\ell})$ generates the subalgebra:
        \begin{align}
            \big\{\pm{\boldsymbol{\imath}},\, \pm{\boldsymbol{\jmath}}\boldsymbol{\ell},\, \pm\boldsymbol{\ell}\boldsymbol{k},\, \pm \mathbbm{1}\big\} \,,
        \end{align}
        which is again associative. For instance:
        \begin{align}
            {\boldsymbol{\imath}} \big(({\boldsymbol{\jmath}}\boldsymbol{\ell})(\boldsymbol{\ell}\boldsymbol{k})\big) = \big({\boldsymbol{\imath}} ({\boldsymbol{\jmath}}\boldsymbol{\ell})\big)(\boldsymbol{\ell}\boldsymbol{k})\,.
        \end{align}}
    \end{enumerate}
    \textit{\textbf{Note}:} Three imaginary units are said to be independent if they do not belong to an associative subalgebra of the Moufang loop of $\mathbb{O}_{\mathbb{S}}$ units. For example, the aforementioned three basic elements $\boldsymbol{\imath}, {\boldsymbol{\jmath}}$, and $\boldsymbol{\ell}$ are independent as:
    \begin{align}\label{12neq21}
        \big((\boldsymbol{\imath})({\boldsymbol{\jmath}})\big)\,\boldsymbol{\ell} \neq \boldsymbol{\imath} \big(({\boldsymbol{\jmath}})(\boldsymbol{\ell})\big) \,.
    \end{align}}
\end{remark}

\begin{remark}\label{Remark 2.202022}
    {\emph{(``modified'' associativity for independent split-octonion \linebreak triples).} As an extension of Moufang's theorem (see Ref. \cite{S16}), the loop generated by any three independent elements — for example, $\boldsymbol{\imath}, {\boldsymbol{\jmath}},$ and $\boldsymbol{\ell}$ — forms a ``modified'' associative finite subalgebra. The corresponding ``modified'' associative law differs from the standard one by an overall sign, which encodes the non-associative (in the usual sense) character of the three independent units. For instance, the non-equality in Eq. \eqref{12neq21} explicitly reads:
    \begin{align}
        \big((\boldsymbol{\imath})({\boldsymbol{\jmath}})\big)\,\boldsymbol{\ell} = -\boldsymbol{\imath} \big(({\boldsymbol{\jmath}})(\boldsymbol{\ell})\big) \,.
    \end{align}
    Thus, for any triple of independent (and therefore non-associative) imaginary split-octonion units, Moufang's theorem guarantees that seemingly non-intuitive identities can be derived systematically via the ``modified'' associative law. As a further illustrative case, let us examine the independent generators $\boldsymbol{\ell}, \boldsymbol{k}$, and $\boldsymbol{\imath}$:
    \begin{align}
        \underbrace{\big((\boldsymbol{\ell})(\boldsymbol{k})\big)\boldsymbol{\imath} = -\boldsymbol{\ell}\big((\boldsymbol{k})(\boldsymbol{\imath})\big)}_{\text{by ``modified'' associative law}} = -\boldsymbol{\ell} {\boldsymbol{\jmath}} = {\boldsymbol{\jmath}} \boldsymbol{\ell} \,,
    \end{align}
    which is fully consistent with the multiplication table in \eqref{table1}.

    \textbf{\textit{Summary}:} Any triple of imaginary split-octonion units is either non-independent, and therefore fully associative, or independent, and thus ``modified'' associative, with the modification differing from standard associativity by an overall sign. This precisely expresses the principle of alternativity that characterizes the split-octonion algebra.
}\end{remark}

\begin{Proposition}\label{proposition 2.1}
    \textbf{(alternativity-induced anti-commutation for split-octonion left multiplications).} As a direct consequence of the preceding Remarks, for any imaginary split-octonion basis units $\boldsymbol{m},\boldsymbol{n}$ and any $o\in\mathbb{O}_{\mathbb{S}}$, one has:
    \begin{align}\label{mn=-nm}
        \boldsymbol{m} ( \boldsymbol{n}\, o ) =
        \begin{cases}
            - \boldsymbol{n} ( \boldsymbol{m}\, o ) &\quad\mbox{for}\quad \boldsymbol{m}\neq\boldsymbol{n} \vspace{0.2cm} \\
            \boldsymbol{m}^2\, o &\quad\mbox{for}\quad \boldsymbol{m}=\boldsymbol{n}
        \end{cases}\,.
    \end{align}    
\end{Proposition}

\begin{proof} 
    \textbf{— Step I (first case, with non-independent (associative) units).} 
    If $\boldsymbol{m}, \boldsymbol{n}, o$ (with $\boldsymbol{m}\neq\boldsymbol{n}$) belong to a non-independent subalgebra, we have:
    \begin{align}
        \boldsymbol{m} ( \boldsymbol{n}\, o ) &= (\boldsymbol{m}\,\boldsymbol{n})\,o \qquad \text{(by standard associative law)}\nonumber\\[0.2cm]
        &= (-\boldsymbol{n}\,\boldsymbol{m})\,o \qquad \text{(by anti-commutativity of imaginary units)} \nonumber\\[0.2cm]
        &= - \boldsymbol{n} ( \boldsymbol{m}\, o ) \qquad \text{(again, by standard associative law)}\,.
    \end{align}

    \textbf{— Step II (first case, with independent (non-associative) units).}
    If $\boldsymbol{m}, \boldsymbol{n}, o$ (with $\boldsymbol{m}\neq\boldsymbol{n}$) form an independent triple, invoking the ``modified'' associative law from Remark \ref{Remark 2.202022}, we get:
    \begin{align}
        \boldsymbol{m} ( \boldsymbol{n}\, o ) &= -(\boldsymbol{m}\,\boldsymbol{n})\,o  \qquad \text{(by ``modified'' associative law)} \nonumber\\[0.2cm]
        &= -(-\boldsymbol{n}\,\boldsymbol{m})\,o \qquad \text{(by anti-commutativity of imaginary units)} \nonumber\\[0.2cm]
        &= - \boldsymbol{n} ( \boldsymbol{m}\, o ) \qquad \text{(again, by ``modified'' associative law)}\,.
    \end{align}

    \textbf{— Step III (second case).} When $\boldsymbol{m} = \boldsymbol{n}$, the result follows immediately from alternativity. Indeed, any pair of split-octonion units generates an associative subalgebra — a direct consequence of Remark \ref{Remark 2.2222}. Therefore:
    \begin{align}\label{m m o = m2 o}
        \boldsymbol{m} (\boldsymbol{m}\, o) = (\boldsymbol{m}\,\boldsymbol{m})\, o = \boldsymbol{m}^2\, o \,.
    \end{align}
    This identity generalizes straightforwardly to repeated multiplications:
    \begin{align}\label{m m o = m2 o'}
        \underbrace{\boldsymbol{m}\Big( \dots \big(\boldsymbol{m} (\boldsymbol{m}}_{\text{$r$ times}}\, o) \big) \dots\Big) = \boldsymbol{m}^r\, o \,.
    \end{align}
    As a direct consequence (setting $o=\mathbbm{1}$), powers of a single split-octonion unit are fully associative and therefore well-defined, ensuring that all such repeated products are unambiguous.
    
    From the preceding results, it follows that Eq. \eqref{mn=-nm} is valid.
\end{proof}

\sectionmark{Conformal Clifford Algebra $\mathfrak{cl}(4, 2)$ Generated by Left Multiplication ...}
\section{Conformal Clifford Algebra $\mathfrak{cl}(4, 2)$ Generated by Left Multiplication of Imaginary Split-Octonion Units}\label{Sect. 8uniii}
\sectionmark{Conformal Clifford Algebra $\mathfrak{cl}(4, 2)$ Generated by Left Multiplication ...}

Let us define the operators $L_A$ representing left multiplication by split-octonion units $A$ acting on the split-octonion algebra $\mathbb{O}_{\mathbb{S}}$:
\begin{align}\label{8uni}
     L_A \,x := A \,x \,.
\end{align}
This definition carries the intrinsic split-octonion multiplication (right-hand side) into a corresponding matrix realization (left-hand side). The symbols $A$ and $x$ therefore play a dual role in Eq. \eqref{8uni}; on the right-hand side, they are elements of the split-octonion algebra, while on the left-hand side $A$ functions purely as an index labelling an $8\times 8$ matrix operator and $x$ is treated as an $8$-component column vector representing a Dirac-Majorana spinor. Readers should always interpret each symbol according to its position in \eqref{8uni}. More specifically:
\begin{enumerate}[leftmargin=0.8cm]
    \item{On the right-hand side of \eqref{8uni}, the symbols $A$ and $x$ are to be interpreted strictly as elements of the split-octonion algebra:
    \begin{align}
        x\in\mathbb{O}_{\mathbb{S}} = \mathrm{Span}\Big\{A \;;\; A = \mathbbm{1},\, \boldsymbol{k},\, \nu,\, \;\;\mbox{with}\;\; \nu\in\big\{{\boldsymbol{\imath}},\, {\boldsymbol{\jmath}},\, \boldsymbol{\ell}{\boldsymbol{\imath}},\, {\boldsymbol{\jmath}}\boldsymbol{\ell},\, \boldsymbol{\ell}\boldsymbol{k},\, \boldsymbol{\ell} \big\}\Big\}\,.
    \end{align}}

    \item{On the left-hand side of \eqref{8uni}, by contrast, $A$ solely plays the role of an index labelling the associated $8\times 8$ matrix operator $L_A$; no octonionic interpretation is attached to $A$ in this position. Correspondingly, $x$ is interpreted as an $8$-component Dirac-Majorana spinor, obtained by identifying $\mathbb{O}_{\mathbb S}\cong\mathbb{R}^{4,4}$ with respect to the basis above. Each $L_A$ acts linearly on this spinor space. Accordingly, the operators $L_A$ can be enumerated as:
    \begin{align}\label{8uni'*}
        L_A =&\; L_{\mathbbm{1}},\, L_{\boldsymbol{k}},\, L_\nu,\, \nonumber\\[0.2cm]
        &\text{with}\quad \nu \in \Big\{0:={\boldsymbol{\imath}},\, 5:={\boldsymbol{\jmath}},\, 1:=\boldsymbol{\ell}{\boldsymbol{\imath}},\, 2:={\boldsymbol{\jmath}}\boldsymbol{\ell},\,  3:=\boldsymbol{\ell}\boldsymbol{k},\, 4:=\boldsymbol{\ell}\Big\}\,.
    \end{align}}
\end{enumerate}

\begin{remark}\label{Remark nested}
    {\emph{(non-homomorphic nature of the map $A \;\longmapsto\; L_A$).} We emphasize that the map $A \;\longmapsto\; L_A$ is not an algebra homomorphism. While the product $A\diamond B := AB$, with $A, B = \mathbbm{1},\, \boldsymbol{k},\, \nu \in \left\{ {\boldsymbol{\imath}},\, {\boldsymbol{\jmath}},\, \boldsymbol{\ell}{\boldsymbol{\imath}},\, {\boldsymbol{\jmath}}\boldsymbol{\ell},\,  \boldsymbol{\ell}\boldsymbol{k}\,, \boldsymbol{\ell} \right\}$, is non-associative (more precisely, alternative), the matrix product $L_A L_B$ is associative. 
    
    In fact, all nested products of operators $L_A$ are interpreted as successive left multiplications on $x \in \mathrm{Span}\{A\}$. This convention ensures that, despite the intrinsic non-associativity of the split-octonions, expressions of the form:
    \begin{align}
        L_A L_B L_C\, x &= (L_A L_B) (L_C\, x) \nonumber\\[0.2cm]
        &= (L_A L_B L_C)\, x \nonumber\\[0.2cm]
        &= L_A (L_B L_C)\, x \nonumber\\
        &\;\,\vdots \nonumber\\
        &= L_A\big(L_B (L_C\, x)\big) = A \big(B (C \, x) \big)
    \end{align}
    are fully unambiguous and well-defined. Of course, for brevity, parentheses may occasionally be omitted in such expressions, although their presence is always implied by definition.
    
    Consequently, we have $L_A L_B \neq L_{AB}$, since $L_A L_B\, x = A(B\, x)$ while $L_{AB}\, x = (AB)\, x$. Nevertheless, in the sequel, we adopt the convenient notation $L_A L_B =: L_{AB}$, which is always understood to mean $L_{AB}\, x = A(B\, x)$.}
\end{remark}

Let us now denote $6$ of the $8$ operators $L_A$, when $A = \nu = {\boldsymbol{\imath}},\, {\boldsymbol{\jmath}},\, \boldsymbol{\ell}{\boldsymbol{\imath}},\, {\boldsymbol{\jmath}}\boldsymbol{\ell},\,  \boldsymbol{\ell}\boldsymbol{k}\,, \boldsymbol{\ell}$ or simply $A=\nu={0,5,1,2,3,4}$, by $m^{}_\nu := L_\nu$. By construction, the set of elements $m^{}_\nu$ generates the conformal Clifford algebra $\mathfrak{cl}(4,2)$ — more precisely, its embedding into the alternative-split-octonion algebra. Specifically, these elements satisfy the following anti-commutation relations:
\begin{align}\label{clif}
    &\big\{ m^{}_\mu , m^{}_\nu \big\} \,x \nonumber\\[0.2cm] 
    &= m^{}_\mu m^{}_\nu\, x + m^{}_\nu m^{}_\mu \,x \nonumber\\[0.2cm]
    &= \mu (\nu\, x) + \nu (\mu\, x) \nonumber\\[0.2cm]
    &= \begin{cases}
        \mu (\nu\, x) - \mu (\nu\, x) = 0 \quad\quad& \text{for}\quad \mu\neq\nu \vspace{0.2cm}\\
        2\mu^2 \, x \quad\quad& \text{for}\quad \mu=\nu
    \end{cases} \quad (\text{from Proposition \ref{proposition 2.1}})\nonumber\\[0.2cm]
    &= 2 \eta_{\mu\nu}\, x \,, \quad (\text{from Table \eqref{table1}})\,,
\end{align}
where $x\in\mathrm{Span}\{A\}$ (see Eqs. \eqref{8uni}-\eqref{8uni'*}), while the metric $\eta_{\mu\nu}$ (with $\mu, \nu = {0,5,1,2,3,4}$) has the signature $(-,-,+,+,+,+)$, corresponding to the $4+2$-dimensional Minkowski-like space $\mathbb{R}^{4,2}$. 

Notably, the generators $m^{}_\nu$ admit a matrix representation in which each generator can be expressed as a direct sum of either (only) $2\times 2$ symmetric matrices $S_{AB}$ ($S_{AB} = S_{BA}$) or $2\times 2$ anti-symmetric matrices $E_{AB}$ ($E_{AB} = -E_{BA}$). Specifically: 
\begin{enumerate}[leftmargin=0.8cm]
    \item{The four generators $m^{}_{i}$, with $i = \boldsymbol{\ell}{\boldsymbol{\imath}},\, {\boldsymbol{\jmath}}\boldsymbol{\ell},\,  \boldsymbol{\ell}\boldsymbol{k}\,, \boldsymbol{\ell}$ (or more conveniently, as indices, $i = {1,2,3,4}$), each satisfying (according to Proposition \ref{proposition 2.1}, Remark \ref{Remark nested}, and Table \eqref{table1}):
    \begin{align}
        (m^{}_{i})^2\, x = (m^{}_{i}\, m^{}_{i})\, x = i \big(i (x)\big) = (i)^2\,x = +\mathbbm{1}\, x\,,
    \end{align}
    can be expressed in terms of the $2 \times 2$ symmetric matrices $S_{AB}$ as follows: 
    \begin{align} \label{11} 
        m^{}_{1}\, x \; \big( := m^{}_{\boldsymbol{\ell}{\boldsymbol{\imath}}}\, x \big) = \boldsymbol{\ell}{\boldsymbol{\imath}} \, x =&\, \bigoplus\left(S_{\mathbbm{1}(\boldsymbol{\ell}{\boldsymbol{\imath}})} - S_{{\boldsymbol{\imath}}\, \boldsymbol{\ell}} - S_{{\boldsymbol{\jmath}} (\boldsymbol{\ell}\boldsymbol{k})} - S_{\boldsymbol{k}({\boldsymbol{\jmath}}\boldsymbol{\ell})}\right) x \,,\nonumber\\[0.2cm]
        =&\, \bigoplus\left(S_{\mathbbm{1}{1}} - S_{{04}} - S_{{53}} - S_{\boldsymbol{k}{2}}\right) x \,,
    \end{align}
    \begin{align}
        m^{}_{2}\, x \; \big( := m^{}_{{\boldsymbol{\jmath}}\boldsymbol{\ell}}\, x \big) = {\boldsymbol{\jmath}}\boldsymbol{\ell} \, x =&\, \bigoplus\left(S_{\mathbbm{1}({\boldsymbol{\jmath}}\boldsymbol{\ell})} + S_{{\boldsymbol{\jmath}}\, \boldsymbol{\ell}} - S_{{\boldsymbol{\imath}} (\boldsymbol{\ell}\boldsymbol{k})} + S_{\boldsymbol{k}(\boldsymbol{\ell}{\boldsymbol{\imath}})}\right) x \,,\nonumber\\[0.2cm]
        =&\, \bigoplus\left(S_{\mathbbm{1}{2}} + S_{{54}} - S_{{03}} + S_{\boldsymbol{k}{1}}\right) x \,,
    \end{align}
    \begin{align}
        m^{}_{3}\, x \; \big( := m^{}_{\boldsymbol{\ell}\boldsymbol{k}}\, x \big) = \boldsymbol{\ell}\boldsymbol{k} \, x =&\, \bigoplus\left(S_{\mathbbm{1}(\boldsymbol{\ell}\boldsymbol{k})} - S_{\boldsymbol{k}\, \boldsymbol{\ell}} + S_{{\boldsymbol{\jmath}} (\boldsymbol{\ell}{\boldsymbol{\imath}})} + S_{{\boldsymbol{\imath}}({\boldsymbol{\jmath}}\boldsymbol{\ell})}\right) x \,,\nonumber\\[0.2cm]
        =&\, \bigoplus\left(S_{\mathbbm{1}{3}} - S_{\boldsymbol{k}{4}} + S_{{51}} + S_{{02}}\right) x \,,
    \end{align}
    \begin{align} \label{11''} 
        m^{}_{4}\, x \; \big( := m^{}_{\boldsymbol{\ell}}\, x \big) = \boldsymbol{\ell} \, x =&\, \bigoplus\left(S_{\mathbbm{1} \boldsymbol{\ell}} + S_{{\boldsymbol{\imath}} (\boldsymbol{\ell}{\boldsymbol{\imath}})} - S_{{\boldsymbol{\jmath}} ({\boldsymbol{\jmath}}\boldsymbol{\ell})} + S_{\boldsymbol{k}(\boldsymbol{\ell}\boldsymbol{k})}\right) x \,,\nonumber\\[0.2cm]
        =&\, \bigoplus\left(S_{\mathbbm{1}{4}} + S_{{01}} - S_{{52}} + S_{\boldsymbol{k}{3}}\right) x \,,
    \end{align}
    where $x\in\mathrm{Span}\{A\}$ (see Eqs. \eqref{8uni}-\eqref{8uni'*}), for the convention $\bigoplus (\dots)$, see Sect. \ref{Sect. Convention}, and:
    \begin{align}
        S_{AB}\, x \; \big(= S_{BA}\, x \big) := A\, \delta_{Bx} + B\, \delta_{Ax}\,,
    \end{align}
    for all $A,B = \mathbbm{1},\, \boldsymbol{k},\, \nu \in \left\{ {\boldsymbol{\imath}},\, {\boldsymbol{\jmath}},\, \boldsymbol{\ell}{\boldsymbol{\imath}},\, {\boldsymbol{\jmath}}\boldsymbol{\ell},\,  \boldsymbol{\ell}\boldsymbol{k}\,, \boldsymbol{\ell} \right\}$ or more conveniently, as indices, $\nu\in\{0,5,1,2,3,4\}$ (see Eqs. \eqref{8uni}-\eqref{8uni'*}). Evidently, the matrices $ m^{}_{i} $ ($ i = {1,2,3,4} $), of square $ +\mathbbm{1} $, are symmetric by their very construction. Note that the above equations can be easily verified by substituting $x$ with any of the split-octonion units (and in fact with all of them) and consulting the multiplication table \eqref{table1}.}
    
    \item{The two remaining generators $m^{}_{0}$ and $m^{}_{5}$, each satisfying $(m^{}_{0})^2 = (m^{}_{5})^2 = -\mathbbm{1}$, can be written in terms of $2\times 2$ anti-symmetric matrices $E_{AB}$ as follows:  
    \begin{align} \label{22} 
        m^{}_{0}\, x \; \big( := m^{}_{\boldsymbol{\imath}}\, x \big) = {\boldsymbol{\imath}} \, x =&\, \bigoplus\left(E_{{\boldsymbol{\imath}}\mathbbm{1}} + E_{\boldsymbol{k}{\boldsymbol{\jmath}}} + E_{\boldsymbol{\ell}(\boldsymbol{\ell}{\boldsymbol{\imath}})} + E_{(\boldsymbol{\ell}\boldsymbol{k})({\boldsymbol{\jmath}}\boldsymbol{\ell})}\right) x \nonumber\\[0.2cm]
        =&\, \bigoplus\left(E_{{0} \mathbbm{1}} + E_{\boldsymbol{k}{5}} + E_{{4}{1}} + E_{{3}{2}}\right) x \,, 
    \end{align}
    \begin{align} \label{22''} 
        m^{}_{5}\, x \; \big( := m^{}_{\boldsymbol{\jmath}}\, x \big) = {\boldsymbol{\jmath}} \, x =&\, \bigoplus\left(E_{{\boldsymbol{\jmath}}\mathbbm{1}} + E_{{\boldsymbol{\imath}}\boldsymbol{k}} - E_{(\boldsymbol{\ell}{\boldsymbol{\imath}})(\boldsymbol{\ell}\boldsymbol{k})} + E_{({\boldsymbol{\jmath}}\boldsymbol{\ell})\boldsymbol{\ell}}\right) x \nonumber\\[0.2cm]
        =&\, \bigoplus\left(E_{{5} \mathbbm{1}} + E_{{0} \boldsymbol{k}} - E_{{1}{3}} + E_{{2}{4}}\right) x \,, 
    \end{align}
    where, again, $x\in\mathrm{Span}\{A\}$ (see Eqs. \eqref{8uni}-\eqref{8uni'*}), for the convention $\bigoplus (\dots)$, see Sect. \ref{Sect. Convention}, and:
    \begin{align}
        E_{AB}\, x \; \big(= - E_{BA}\, x\big) := A\, \delta_{Bx} - B\, \delta_{Ax}\,,
    \end{align}
    for all $A,B = \mathbbm{1},\, \boldsymbol{k},\, \nu \in \left\{ {\boldsymbol{\imath}},\, {\boldsymbol{\jmath}},\, \boldsymbol{\ell}{\boldsymbol{\imath}},\, {\boldsymbol{\jmath}}\boldsymbol{\ell},\,  \boldsymbol{\ell}\boldsymbol{k}\,, \boldsymbol{\ell} \right\}$ or more conveniently, as indices, $\nu\in\{0,5,1,2,3,4\}$ (see Eqs. \eqref{8uni}-\eqref{8uni'*}). Evidently, the matrices $m^{}_{0}$ and $m^{}_{5}$, of square $-\mathbbm{1}$, are skew-symmetric by construction. Again, the above equations can be verified by substituting $x$ with any of the split-octonion units (and in fact with all of them).}
\end{enumerate}

\begin{remark}\label{Remark E,D,S,1-AB}
    {\emph{(Weyl matrix basis representation of $E_{AB}$ and $S_{AB}$).} In terms of the standard $2\times 2$ Weyl matrix basis $e_{AB}$, defined as the matrix with entry `$1$' at the intersection of the $A$-th row and $B$-th column and `$0$' elsewhere (for all indices $A, B$ independently ranging over the $8$ values $\mathbbm{1},\, \boldsymbol{k},\, \nu \in \left\{ {\boldsymbol{\imath}},\, {\boldsymbol{\jmath}},\, \boldsymbol{\ell}{\boldsymbol{\imath}},\, {\boldsymbol{\jmath}}\boldsymbol{\ell},\,  \boldsymbol{\ell}\boldsymbol{k}\,, \boldsymbol{\ell} \right\}$; see Eqs. \eqref{8uni}-\eqref{8uni'*}), and characterized by the multiplication rule:
    \begin{align}
        e_{AB} \, e_{CD} = \delta_{BC}\, e_{AD}\,,
    \end{align}
    the above $E_{AB}$ and $S_{AB}$ ($A\neq B$) assume, respectively, the following forms:
    \begin{align}
        E_{AB} = e_{AB} - e_{BA}\,, 
    \end{align}
    \begin{align}
        S_{AB} = e_{AB} + e_{BA}\,.
    \end{align}
    For the latter use, we additionally define:
    \begin{align}
        \label{1def} S_{AB}\, S_{AB} = - E_{AB}\, E_{AB} = e_{AA} + e_{BB} =: \mathbbm{1}_{AB} \,,
    \end{align}
    \begin{align}
        \label{ddef} E_{AB}\, S_{AB} = e_{AA} - e_{BB} =: D_{AB} \,.
    \end{align}
    We also notice that the following relations hold: 
    \begin{align}
        S_{AB} = D_{AB} \, E_{AB} = - E_{AB} \, D_{AB}\,, 
    \end{align}
    \begin{align}
        E_{AB} = D_{AB} \, S_{AB} = - S_{AB} \, D_{AB} \,.
    \end{align}}
\end{remark}

\begin{remark} 
    {\emph{(Cayley conjugation and Hermiticity properties of $m^{}_\nu$).} From the above discussion, it is straightforward to verify that the mapping $\nu \;\longmapsto\; L_\nu\; \big(=: m^{}_\nu\big)$ does not respect the Cayley conjugation \eqref{CD-conjugate}. While according to \eqref{CD-conjugate}, $\nu^\ast = - \nu$ for all $\nu = {\boldsymbol{\imath}},\, {\boldsymbol{\jmath}},\, \boldsymbol{\ell}{\boldsymbol{\imath}},\, {\boldsymbol{\jmath}}\boldsymbol{\ell},\,  \boldsymbol{\ell}\boldsymbol{k}\,, \boldsymbol{\ell}$ or simply, as indices, $\nu = {0,5,1,2,3,4}$ (see Eqs. \eqref{8uni}-\eqref{8uni'*}), only matrices $m^{}_{0}$ and  $m^{}_{5}$ are anti-Hermitian; specifically, we have:
    \begin{align}\label{anti-Hermitian}
        m^\ast_{0} \;\left(= \big(\overline{m^{}_0}\big)^\top = m^{\top}_{0}\right) = - m^{}_{0}\,, \quad m^\ast_{5} = - m^{}_{5}\,, \quad\mbox{while}\quad m^\ast_{i} = m^{}_{i}\,,
    \end{align}
    for all $i={1,2,3,4}$. Equivalently, in a more concise form:
    \begin{align}\label{anti-Hermitian'}
        m^\ast_\nu = m_\nu^2\, m^{}_\nu \,.
    \end{align}}
\end{remark}

Let us introduce the following skew-symmetric product:
\begin{align}\label{6 6}
    m^{}_{\mu\nu} \,x \; \big( = - m^{}_{\nu\mu}\,x\big) :=&\, \frac{1}{2} \big[m^{}_\mu, m^{}_\nu\big]\,x \nonumber\\[0.2cm]
    =&\, \frac{1}{2} \big(\mu (\nu\,x) - \nu (\mu\,x) \big) \nonumber\\[0.2cm]
    =&\, \mu (\nu\,x) \qquad(\text{from Proposition \ref{proposition 2.1}}) \nonumber\\[0.2cm]
    =&\, m^{}_\mu m^{}_\nu \,x \,,  
\end{align}
where $x\in\mathrm{Span}\{A\}$ and $\mu, \nu = {\boldsymbol{\imath}},\, {\boldsymbol{\jmath}},\, \boldsymbol{\ell}{\boldsymbol{\imath}},\, {\boldsymbol{\jmath}}\boldsymbol{\ell},\,  \boldsymbol{\ell}\boldsymbol{k}\,, \boldsymbol{\ell}$ or simply, as indices, $\mu, \nu = {0,5,1,2,3,4}$ (see Eqs. \eqref{8uni}-\eqref{8uni'*}), such that $\mu\neq\nu$ since $m^{}_{\mu\mu} = 0$. These $15$ operators $m^{}_{\mu\nu}$ span the $8$-dimensional realization of the isometry Lie algebra associated with the Clifford algebra $\mathfrak{cl}(4,2)$, namely, $\mathfrak{su}(2,2) \cong\mathfrak{so}(4,2)$.\footnote{For the extension of this construction to the group level, see Appendix \ref{Appendix U(2,2) in the Clifforrd}.} Specifically, the generators $ m^{}_{\mu\nu} $ satisfy the following commutation relations:
\begin{align}
    \big[ m^{}_{\mu\nu}, m^{}_{\lambda\sigma}\big]\, x = 2 \big(\eta_{\nu\lambda} m^{}_{\mu\sigma} - \eta_{\mu\lambda} m^{}_{\nu\sigma} - \eta_{\nu\sigma} m^{}_{\mu\lambda} + \eta_{\mu\sigma} m^{}_{\nu\lambda} \big)\, x\,,
\end{align}  
where again $x\in\mathrm{Span}\{A\}$ and $\mu,\nu,\lambda,\sigma=0,5,1,2,3,4$ (see Eqs. \eqref{8uni}-\eqref{8uni'*}), such that $\mu\neq\nu$ and $\lambda\neq\sigma$. This confirms that the generators $ m^{}_{\mu\nu} $ form a well-defined representation of the Lie algebra $\mathfrak{su}(2,2) \cong\mathfrak{so}(4,2)$ within the Clifford framework. The above commutation relation can be readily verified by extending Proposition \ref{proposition 2.1} and referring to the multiplication table \eqref{table1}:
\begin{align}\label{8888888888888}
    &\big[ m^{}_{\mu\nu}, m^{}_{\lambda\sigma}\big]\, x \nonumber\\[0.2cm] 
    &= m^{}_{\mu\nu} m^{}_{\lambda\sigma}\, x - m^{}_{\lambda\sigma} m^{}_{\mu\nu} \, x \nonumber\\[0.2cm]
    &= m^{}_{\mu}m^{}_{\nu} m^{}_{\lambda}m^{}_{\sigma}\, x - m^{}_{\lambda}m^{}_{\sigma} m^{}_{\mu}m^{}_{\nu} \, x \nonumber\\[0.2cm]
    &= {\mu} \Big({\nu} \big({\lambda} ({\sigma} \, x) \big)\Big) - {\lambda} \Big({\sigma} \big({\mu} ({\nu} \, x) \big)\Big) \nonumber\\[0.2cm]
    &= \begin{cases}
        \mu \Big(\nu \big(\lambda (\sigma\, x) \big) \Big) - \mu \Big(\nu \big(\lambda (\sigma\, x) \big) \Big) = 0 \quad\quad& \text{for}\quad \mu\neq\nu\neq\lambda\neq\sigma \vspace{0.3cm}\\
        \nu^2 \big(\mu (\sigma\, x) \big) + \nu^2 \big(\mu (\sigma\, x) \big) = 2 \nu^2 \big(\mu (\sigma\, x) \big) \quad\quad& \text{for}\quad \mu\neq\nu=\lambda\neq\sigma \vspace{0.3cm}\\
        -\mu^2 \big(\nu (\sigma\, x) \big) - \mu^2 \big(\nu (\sigma\, x) \big) = -2\mu^2 \big(\nu (\sigma\, x) \big) \quad\quad& \text{for}\quad \nu\neq\mu=\lambda\neq\sigma \vspace{0.3cm}\\
        -\nu^2 \big(\mu (\lambda\, x) \big) - \nu^2 \big(\mu (\lambda\, x) \big) = -2\nu^2 \big(\mu (\lambda\, x) \big) \quad\quad& \text{for}\quad \mu\neq\nu=\sigma\neq\lambda \vspace{0.3cm}\\
        \mu^2 \big(\nu (\lambda\, x) \big) + \mu^2 \big(\nu (\lambda\, x) \big) = 2\mu^2 \big(\nu (\lambda\, x) \big) \quad\quad& \text{for}\quad \nu\neq\mu=\sigma\neq\lambda
    \end{cases} \nonumber\\[0.3cm]
    &= 2\eta_{\nu\lambda} \big(\mu(\sigma\, x)\big) - 2\eta_{\mu\lambda} \big(\nu(\sigma\, x)\big) - 2\eta_{\nu\sigma} \big(\mu(\lambda \, x)\big) + 2\eta_{\mu\sigma} \big(\nu(\lambda\, x)\big) \nonumber\\[0.3cm]
    &= 2\eta_{\nu\lambda} \big(m^{}_{\mu}m^{}_{\sigma}\, x\big) - 2\eta_{\mu\lambda} \big(m^{}_{\nu}m^{}_{\sigma}\, x\big) - 2\eta_{\nu\sigma} \big(m^{}_{\mu}m^{}_{\lambda}\,x\big) + 2\eta_{\mu\sigma} \big(m^{}_{\nu} m^{}_{\lambda}x \big) \nonumber\\[0.3cm]
    &= 2 \big(\eta_{\nu\lambda} m^{}_{\mu\sigma} - \eta_{\mu\lambda} m^{}_{\nu\sigma} - \eta_{\nu\sigma} m^{}_{\mu\lambda} + \eta_{\mu\sigma} m^{}_{\nu\lambda} \big)\, x\,.
\end{align}

\begin{remark}\label{Remark Trace orthogonal}
    {\emph{(trace orthogonality of the generators $m^{}_{\mu\nu}$).} The generators $m^{}_{\mu\nu}$ are orthogonal, under the trace product, to all $L_A$:
    \begin{align}
        \mathrm{Tr}(m^{}_{\mu\nu} m^{}_\lambda) = 0\,, 
    \end{align}
    \begin{align}
        \mathrm{Tr}(m^{}_{\mu\nu} L_{\boldsymbol{k}}) = 0\,,  
    \end{align}
    \begin{align}
        \mathrm{Tr}(m^{}_{\mu\nu} L_1) = 0\,,  
    \end{align}  
    where $A = \mathbbm{1},\, \boldsymbol{k},\, \nu$ and $\mu, \nu, \lambda = {0,5,1,2,3,4}$ (see Eqs. \eqref{8uni}-\eqref{8uni'*}). These results follow directly from the standard properties \eqref{xy=yx} and \eqref{x+y=x+y} of the trace.}
\end{remark}

\begin{remark}{
    \emph{(the $\mathfrak{cl}(4)$ subalgebra and its $\mathfrak{so}(4)$ generators).} The matrices $m^{}_i$, with $i = \boldsymbol{\ell}{\boldsymbol{\imath}},\, {\boldsymbol{\jmath}}\boldsymbol{\ell},\,  \boldsymbol{\ell}\boldsymbol{k}\,, \boldsymbol{\ell}$ or more conveniently, as indices, $i = {1,2,3,4}$ (see Eqs. \eqref{8uni}-\eqref{8uni'*}), form the Clifford subalgebra $\mathfrak{cl}(4) \subset \mathfrak{cl}(4,2)$, whose associated isometry Lie algebra is $\mathfrak{so}(4)$. The generators of this Lie algebra are given by:
    \begin{align}
        m^{}_{ij} \; \big( = - m^{}_{ji} \big) := \frac{1}{2} \big[m^{}_i, m^{}_j\big] = m^{}_i m^{}_j \,,
    \end{align}
    where $i,j = {1,2,3,4}$, while $i \neq j$ (since $m^{}_{ii}=0$). They satisfy the following commutation relations:
    \begin{align}
        \big[ m^{}_{ij}, m^{}_{lk} \big] = 2 \big(\eta_{jl} m^{}_{ik} - \eta_{il} m^{}_{jk} - \eta_{jk} m^{}_{il} + \eta_{ik} m^{}_{jl} \big)\,,
    \end{align}
    where the metric $\eta_{ij}$ has signature $(+,+,+,+)$, corresponding to the $4$-dimensional Euclidean space $\mathbb{R}^{4}$. Additionally, employing Proposition \ref{proposition 2.1}, Remark \ref{Remark nested}, and the multiplication table \eqref{table1}, it is straightforward to verify that:
    \begin{align}\label{mij222}
        \big(m^{}_{ij}\big)^2 \, x &= \big(m^{}_{i}\, m^{}_{j}\big)^2 \, x \nonumber\\[0.2cm] 
        &= \big( m^{}_{i}\, m^{}_{j} \, m^{}_{i}\, m^{}_{j}\big) \,x \nonumber\\[0.2cm] 
        &= \bigg( {i}\, \Big({j} \,\big({i}\, ({j} \,x)\big)\Big)\bigg) \nonumber\\[0.2cm] 
        &= - \eta_{ii}\, \eta_{jj} \, x = -\mathbbm{1} \, x\,,
    \end{align}
    where again $i,j = {1,2,3,4}$, while $i \neq j$.   
}\end{remark}

\begin{remark}\label{Remark associativity of E}{
    \emph{(the Clifford pseudoscalar and the remaining split-octonion unit).} The six imaginary $\mathbb{O}_\mathbb{S}$ units $\nu = {\boldsymbol{\imath}},\, {\boldsymbol{\jmath}},\, \boldsymbol{\ell}{\boldsymbol{\imath}},\, {\boldsymbol{\jmath}}\boldsymbol{\ell},\, \boldsymbol{\ell}\boldsymbol{k},\, \boldsymbol{\ell}$ used in our realization of $\mathfrak{cl}(4,2)$ leave $\boldsymbol{k}$ outside the chosen generating set. This remaining unit nevertheless appears, up to sign, both in the pairwise products:
    \begin{align}
        (\boldsymbol{\ell}{\boldsymbol{\imath}})
        ({\boldsymbol{\jmath}}\boldsymbol{\ell})
        =
        (\boldsymbol{\ell}\boldsymbol{k})\boldsymbol{\ell}
        =
        {\boldsymbol{\jmath}}{\boldsymbol{\imath}}
        =
        -\boldsymbol{k} \,,
    \end{align}
    and in the following explicitly parenthesized product of all six units:
    \begin{align}\label{112}
        {\boldsymbol{\imath}}
        \Bigg(
            {\boldsymbol{\jmath}}
            \bigg(
                (\boldsymbol{\ell}{\boldsymbol{\imath}})
                \Big(
                    ({\boldsymbol{\jmath}}\boldsymbol{\ell})
                    \big(
                        (\boldsymbol{\ell}\boldsymbol{k})\boldsymbol{\ell}
                    \big)
                \Big)
            \bigg)
        \Bigg)
        =
        -\boldsymbol{k} \,.
    \end{align}
    In the Clifford representation, the ordered composition of the corresponding left-multiplication operators is itself left multiplication by $-\boldsymbol{k}$:
    \begin{align}\label{221}
        L_{\boldsymbol{\imath}}\,
        L_{\boldsymbol{\jmath}}\,
        L_{\boldsymbol{\ell}{\boldsymbol{\imath}}}\,
        L_{{\boldsymbol{\jmath}}\boldsymbol{\ell}}\,
        L_{\boldsymbol{\ell}\boldsymbol{k}}\,
        L_{\boldsymbol{\ell}}
        \; \Big(
            =: m^{}_{0}\,m^{}_{5}\,m^{}_{1}\,
            m^{}_{2}\,m^{}_{3}\,m^{}_{4}
        \Big)
        =
        L_{-\boldsymbol{k}}
        =: M(4,2)
        =
        E \,.
    \end{align}
    Thus, the remaining imaginary unit represents the real $\mathfrak{cl}(4,2)$ pseudoscalar (or volume element) $M(4,2)$. Equation \eqref{221} is an identity of linear operators on the entire split-octonion space; applying it to the identity element gives Eq. \eqref{112}.

    Three further points deserve attention:
    \begin{enumerate}[leftmargin=0.8cm]
        \item{The ordering of the six Clifford generators in Eq. \eqref{221} fixes an orientation in $\mathbb{R}^{4,2}$. Any odd permutation reverses the sign of the pseudoscalar, whereas any even permutation preserves it.}

        \item{Since $E^2 \; \left(= M^2(4,2)\right) = -\mathbbm{1}$, the pseudoscalar defines a complex structure on the $8$-dimensional real spinor space, a topic we will explore in more detail later.}

        \item{Although the three pairwise split-octonionic products displayed above all equal $-\boldsymbol{k}$, the corresponding matrices $m^{}_{12}$, $m^{}_{34}$, and $m^{}_{50} = -m^{}_{05}$ are distinct: they commute with one another and are mutually trace-orthogonal, as well as trace-orthogonal to $E$. This distinction reflects the fact that composition of left-multiplication operators need not coincide with left multiplication by the product of their defining units. We return to these relations later.}
    \end{enumerate}
}\end{remark}

\sectionmark{Diagonalizing the Complex Structure of $\mathfrak{cl}(4, 2)$ in Spinor Space ...}
\section{Diagonalizing the Complex Structure of $\mathfrak{cl}(4, 2)$ in Spinor Space: Insight from the $\mathfrak{cl}(2)$ Case}
\sectionmark{Diagonalizing the Complex Structure of $\mathfrak{cl}(4, 2)$ in Spinor Space ...}

We observe that the $\mathfrak{cl}(4, 2)$ pseudoscalar (volume form) $E=M(4,2)$ generates the $\mathfrak{u}(1)$ center of:\footnote{For the extension of this construction to the group level, see Appendix \ref{Appendix U(2,2) in the Clifforrd}.}
\begin{align}
    \mathfrak{u}(2,2) \cong \mathfrak{su}(2,2) \, \big(\cong\mathfrak{so}(4,2)\big)\oplus\mathfrak{u}(1) \; \subset \, \mathfrak{cl}(4,2) \,,
\end{align}
meaning that it commutes with every generator $m^{}_{\mu\nu} \; \big(=-m^{}_{\nu\mu}\big)$ of $\mathfrak{su}(2,2) \cong\mathfrak{so}(4,2)$:
\begin{align}\label{tootii}
    \left[ m^{}_{\mu\nu}\; \big( = m^{}_\mu m^{}_\nu\big) \;,\; E \right] = 0 \,,
\end{align}
for all $\mu,\nu = {0,5,1,2,3,4}$ (see Eqs. \eqref{8uni}-\eqref{8uni'*}); to verify this commutation relation, using Proposition \ref{proposition 2.1}, it suffices to follow a procedure analogous to that described in Eq. \eqref{8888888888888}. Thus, any representation space ${\mathfrak{S}}$ of $\mathfrak{su}(2,2) \cong\mathfrak{so}(4,2)$, spanned by the fifteen $8\times 8$ matrices $m^{}_{\mu\nu} \; \big(=-m^{}_{\nu\mu}\big)$, necessarily decomposes into eigenspaces of $E$. 

As an anti-Hermitian — and in particular real skew-symmetric — matrix, the $\mathfrak{cl}(4,2)$ volume form $E=M(4,2)$ reads as:
\begin{align}\label{U(1)generator}
    E \, x =&\, m^{}_{0} \Bigg(m^{}_{5} \bigg(m^{}_{1} \Big(m^{}_{2} \big(m^{}_{3} (m^{}_{4} \, x ) \big) \Big) \bigg) \Bigg) \nonumber\\[0.2cm] 
    =&\, {\boldsymbol{\imath}}\Bigg({\boldsymbol{\jmath}} \bigg( \boldsymbol{\ell}{\boldsymbol{\imath}} \Big({\boldsymbol{\jmath}}\boldsymbol{\ell} \Big(\boldsymbol{\ell}\boldsymbol{k} (\boldsymbol{\ell} \, x) \big)\Big)\bigg)\Bigg) \nonumber\\[0.2cm]
    =&\, \bigoplus\left( E_{\mathbbm{1}\boldsymbol{k}} + E_{05} + E_{12} + E_{34} \right) \, x \,,
\end{align}
where $x\in\mathrm{Span}\{A\}$ (see Eqs. \eqref{8uni}-\eqref{8uni'*}). For the convention $\bigoplus (\dots)$, see Sect. \ref{Sect. Convention}. 

\begin{remark}\label{Remark Ex=(E)x}{
    \emph{(identification of the volume form with left multiplication by the remaining imaginary unit).} We now show that the Clifford volume form $E=M(4,2)$ acts on the split-octonion space as left multiplication by $-\boldsymbol{k}$. The starting point is Eq. \eqref{112}, which gives:
    \begin{align}
        E\,\mathbbm{1} = -\boldsymbol{k} \,.
    \end{align}
    To determine the action of $E$ on the remaining basis elements, recall that the Clifford pseudoscalar in even dimension anticommutes with each generating vector. Thus:
    \begin{align}
        E\,L_{\nu} = -L_{\nu}\,E \,,
        \quad
        \nu =
        {\boldsymbol{\imath}},\,
        {\boldsymbol{\jmath}},\,
        \boldsymbol{\ell}{\boldsymbol{\imath}},\,
        {\boldsymbol{\jmath}}\boldsymbol{\ell},\,
        \boldsymbol{\ell}\boldsymbol{k},\,
        \boldsymbol{\ell} \,.
    \end{align}
    Applying this relation to $\mathbbm{1}$ yields:
    \begin{align}
        E\,\nu
        =
        -L_{\nu}\big(E\,\mathbbm{1}\big)
        =
        \nu\boldsymbol{k}
        =
        -\boldsymbol{k}\nu \,,
    \end{align}
    where the last equality follows from the anticommutation of distinct imaginary basis units.

    To determine the action on $\boldsymbol{k}$, observe that the left-multiplication anticommutation relations also give:
    \begin{align}
        L_{\boldsymbol{k}}\,L_{\nu}
        =
        -L_{\nu}\,L_{\boldsymbol{k}} \,,
    \end{align}
    for each of the six units $\nu$ listed above. Moving $L_{\boldsymbol{k}}$ through their ordered product therefore introduces six minus signs, so that:
    \begin{align}
        E\,L_{\boldsymbol{k}}
        =
        L_{\boldsymbol{k}}\,E \,.
    \end{align}
    Applying this identity to $\mathbbm{1}$, as above, yields:
    \begin{align}
        E\,\boldsymbol{k}
        =
        L_{\boldsymbol{k}}\big(E\,\mathbbm{1}\big)
        =
        \boldsymbol{k}(-\boldsymbol{k})
        =
        -\boldsymbol{k}\boldsymbol{k} \,.
    \end{align}
    Consequently, $E$ acts as left multiplication by $-\boldsymbol{k}$ on all eight basis elements and, by real linearity, on the entire split-octonion space:
    \begin{align}
        E\,x
        =
        -\boldsymbol{k}x
        =
        \bigoplus\left(
            E_{\mathbbm{1}\boldsymbol{k}}
            + E_{05} + E_{12} + E_{34}
        \right)x \,,
    \end{align}
    where $x\in\mathrm{Span}\{A\}$ (see Eqs. \eqref{8uni}-\eqref{8uni'*}). For the convention $\bigoplus (\dots)$, see Sect. \ref{Sect. Convention}. This identifies the Clifford volume form with the left-multiplication operator associated with the remaining imaginary unit:
    \begin{align}
        E = M(4,2) = L_{-\boldsymbol{k}} \,.
    \end{align}
}\end{remark}

We note that $E=M(4,2)$ is diagonalizable only over a complex basis, with eigenvalues $\pm \mathrm{i}$, since (see also Appendix \ref{footnote E2=-1}):
\begin{align}\label{E2=-1}
    E^2 \,x = M^2(4,2) \, x = -{\boldsymbol{k}} (-{\boldsymbol{k}} \,x) = ({\boldsymbol{k}})^2 \,x = - \mathbbm{1} \, x\,,
\end{align}
according to the preceding discussion and Proposition \ref{proposition 2.1} and, of course, Remark \ref{Remark nested}. Then, in view of Eq. \eqref{tootii}, the representation space ${\mathfrak{S}}$ of $\mathfrak{su}(2,2) \cong\mathfrak{so}(4,2)$ splits into two $4$-dimensional chiral subspaces:
\begin{align}
    {\mathfrak{S}} = {\mathfrak{S}}_+ \oplus {\mathfrak{S}}_- \,,
\end{align}
where the complex-conjugate semispinor chiral subspaces ${\mathfrak{S}}_\pm \; \big( =\overline{{\mathfrak{S}}_\mp} \big)$ consist of eigenvectors of $E$ with eigenvalues $\pm\mathrm{i}$, respectively; this point will be discussed explicitly later.

\begin{remark}\label{Remark IR u(2,2)}
    {\emph{(central role of $E$ and the chiral decomposition of spinor space).} Completing the discussion begun by Eq. \eqref{tootii}, note first that the pseudoscalar $E$ lies in the center of $\mathfrak{u}(2,2) \cong \mathfrak{su}(2,2)\; \big(\cong \mathfrak{so}(4,2)\big) \oplus \mathfrak{u}(1)$. In particular, it commutes with every generator $m^{}_{\mu\nu}\; \big(=-m^{}_{\nu\mu}\big)$ of $\mathfrak{su}(2,2)$. As a result, the eigenspaces $\mathfrak{S}_\pm$ of $E$, respectively with eigenvalues $\pm\mathrm{i}$ (to be shown explicitly later), are separately invariant under $\mathfrak{su}(2,2)$. We will see that each $\mathfrak{S}_\pm$ in fact carries an irreducible $\mathfrak{su}(2,2)$ representation. 

    On the other hand, because $E$ itself also generates the central $\mathfrak{u}(1)$ factor in $\mathfrak{u}(2,2)$ and acts on each $\mathfrak{S}_\pm$ by a fixed scalar, the full algebra $\mathfrak{u}(2,2)$ likewise preserves these subspaces. Thus, the $8$-dimensional $\mathfrak{u}(2,2)$-module splits cleanly into two invariant $4$-dimensional pieces, $\mathfrak{S} = \mathfrak{S}_+ \oplus \mathfrak{S}_-$. Each $\mathfrak{S}_\pm$ furnishes an irreducible semispinor representation of $\mathfrak{u}(2,2)$.}
\end{remark}

In the sequel, we will systematically present the similarity transformation that simultaneously diagonalizes the $\mathfrak{u}(1)$ generator $E$ \eqref{U(1)generator} along with a \textit{trace-orthogonal basis}\footnote{See Remark \ref{Remark Trace orthogonal}.} of the Cartan elements of the maximal compact Lie subalgebra $\mathfrak{so}(2) \oplus \mathfrak{so}(4) \subset \mathfrak{so}(4,2)$:\footnote{For the structure of the maximal compact subalgebra of $\mathfrak{su}(2,2)\cong\mathfrak{so}(4,2)$, see Appendix \ref{Appendix Maximal}.}
\begin{align}
    \label{10} m^{}_{05} \, x = m^{}_{0} (m^{}_{5} \, x) = {\boldsymbol{\imath}} ({\boldsymbol{\jmath}} \, x) = \bigoplus\left( E_{\boldsymbol{k}\mathbbm{1}} - E_{05} + E_{12} + E_{34} \right) \, x \,, 
\end{align}
\begin{align}
    \label{20} m^{}_{12} \, x = m^{}_{1} (m^{}_{2} \, x)
    = \boldsymbol{\ell}{\boldsymbol{\imath}} ({\boldsymbol{\jmath}}\boldsymbol{\ell} \, x) = \bigoplus\left( E_{\mathbbm{1}\boldsymbol{k}} - E_{05} + E_{12} - E_{34} \right) \, x \,, 
\end{align}
\begin{align}
    \label{30} m^{}_{34} \, x = m^{}_{3} (m^{}_{4} \, x) = \boldsymbol{\ell}\boldsymbol{k} (\boldsymbol{\ell} \, x) = \bigoplus\left ( E_{\mathbbm{1}\boldsymbol{k}} - E_{05} - E_{12} + E_{34} \right) \, x \,, 
\end{align}
where $x\in\mathrm{Span}\{A\}$ (see Eqs. \eqref{8uni}-\eqref{8uni'*}). Again, for the convention $\bigoplus (\dots)$, see Sect. \ref{Sect. Convention}. Note that: 
\begin{enumerate}[leftmargin=0.8cm]
    \item{These equations can be readily checked by substituting $x$ with any (indeed, with each) of the split-octonion units and referring to the multiplication table \eqref{table1}.}
    
    \item{The product of the three Cartan elements \eqref{10}-\eqref{30} reconstructs the Clifford pseudoscalar \eqref{U(1)generator}, which serves as the generator of the $\mathfrak{u}(1)$ center of $\mathfrak{u}(2,2)$.}
    
    \item{Since the matrices $m^{}_{\mu\nu}$ \eqref{10}-\eqref{30} are skew-symmetric, with \linebreak $\big(m^{}_{\mu\nu}\big)^2 = -\mathbbm{1}$,\footnote{This identity can be verified by a procedure analogous to that used in Eq. \eqref{mij222}.} they have pure imaginary eigenvalues.}
    
    \item{These matrices can be diagonalized simultaneously with the complex structure $E$ \eqref{U(1)generator}, in a complex chiral basis, since:
    \begin{align}
        \big[m^{}_{05}, m^{}_{12}\big] = \big[m^{}_{12}, m^{}_{34}\big] = \dots = \big[m^{}_{05},E\big] = \dots = \big[m^{}_{34}, E\big] = 0\,.
    \end{align}
    We shall discuss this diagonalization procedure in detail in the next section.}    
\end{enumerate}

\subsection{The $\mathfrak{cl}(2)$ Case as a Prototype}\label{Sect. cl(2)}

Since all four (commuting) $8\times 8$ skew-symmetric matrices $E$ \eqref{U(1)generator}, $m^{}_{05}$ \eqref{10}, $m^{}_{12}$ \eqref{20}, and $m^{}_{34}$ \eqref{30} can be expressed as direct sums of $2\times 2$ skew-symmetric matrices $\pm E_{AB}$ for the same set of index pairs $A, B = \mathbbm{1},\, \boldsymbol{k},\, {\boldsymbol{\imath}},\, {\boldsymbol{\jmath}},\, \boldsymbol{\ell}{\boldsymbol{\imath}},\, {\boldsymbol{\jmath}}\boldsymbol{\ell},\,  \boldsymbol{\ell}\boldsymbol{k}\,, \boldsymbol{\ell}$ or more conveniently $A, B = \mathbbm{1}, \boldsymbol{k}, 0,5,1,2,3,4$ (see Eqs. \eqref{8uni}-\eqref{8uni'*}), the simultaneous diagonalization problem of $E$, $m^{}_{05}$, $m^{}_{12}$, and $m^{}_{34}$ reduces to a $2$-dimensional case within the framework of the Clifford algebra $\mathfrak{cl}(2):= \mathfrak{cl}(2,0)$.\footnote{This algebra represents the lowest-dimensional member of the signature-two Clifford algebras $\mathfrak{cl}(n+1, n-1)$, $n = 1,2, \dots$, in Cartan's classification.} The $\mathfrak{cl}(2)$ Majorana (or real) basis $\{m^{}_1, m^{}_2\}$ and the chiral one $\{\sigma^{}_1, \sigma^{}_2\}$ can be chosen, respectively, as:
\begin{align}
    m^{}_1 = \sigma^{}_3 = 
    \begin{pmatrix}
        1 & 0 \\
        0 & -1
    \end{pmatrix}\,, &\quad
    m^{}_2 = \sigma^{}_1 = 
    \begin{pmatrix}
        0 & 1 \\
        1 & 0
    \end{pmatrix} \nonumber\\ 
    \Longrightarrow \quad&
    \underbrace{{\epsilon} = M(2) := m^{}_1 m^{}_2}_{\text{corresponding volume form}} = \mathrm{i} \sigma^{}_2 = 
    \begin{pmatrix}
        0 & 1 \\
        -1 & 0
    \end{pmatrix}\,,  
\end{align}
\begin{align}
    \sigma^{}_1\,, \quad \sigma^{}_2 = -\mathrm{i}{\epsilon} \quad \Longrightarrow \quad \underbrace{\omega(2) := \sigma^{}_1 \sigma^{}_2}_{\text{corresponding volume form}} = \mathrm{i} \sigma^{}_3 = 
    \begin{pmatrix}
        \mathrm{i} & 0 \\
        0 & - \mathrm{i}
    \end{pmatrix}\,.
\end{align}
Here, we must emphasize that $\{m^{}_1, m^{}_2\}$ and $\{\sigma^{}_1, \sigma^{}_2\}$ are two bases in the same real Clifford algebra $\mathfrak{cl}(2)$; no complexification is needed. As a matter of fact, the following unitary similarity transformation relates them:
\begin{align}\label{tur}
    \sigma^{}_{\dot{\nu}} = V m^{}_{\dot{\nu}} V^{-1}\,, \quad {\dot{\nu}} = 1, 2\,,
\end{align}
with (see Appendix \ref{Appendix Pauli}):
\begin{align}\label{shi}
    \quad V = e^{-\frac{\pi}{4}{\epsilon}}\, e^{-\mathrm{i}\frac{\pi}{4}\sigma^{}_3} &= \frac{1}{2} \left( \mathbbm{1} - \mathrm{i}\sigma^{}_1 - {\epsilon} - \mathrm{i}\sigma^{}_3 \right) = \frac{1}{2}
    \begin{pmatrix}
        1 - \mathrm{i}  &  - 1 - \mathrm{i}  \\
        1 - \mathrm{i}  &  1 + \mathrm{i}
    \end{pmatrix}\,, \nonumber\\[0.2cm]
    V V^\ast = V^\ast V &= \mathbbm{1} \quad \Longrightarrow\quad V^\ast = V^{-1}\,,
\end{align}
where, in the present context, the superscript `$^\ast$' denotes the Hermitian conjugate, defined as $V^\ast = \overline{V}^\top = \overline{V^\top}$, with the `$\overline{\phantom{a}}$' symbol representing the complex conjugate, and the superscript `$^\top$' indicating the transpose.

Let $\xi = \begin{pmatrix} \xi^+ \\ \xi^- \end{pmatrix}$ denote the Majorana (real) spinors, and $\psi = \begin{pmatrix} \psi^+ \\ \psi^- \end{pmatrix} = V \, \xi$ the chiral spinors in the context of the $\mathfrak{cl}(2)$. We then observe that the right- and left-handed components, denoted by $\psi^\pm$, of the chiral spinor satisfy:
\begin{align}
    \psi^\pm = \left(\frac{\mathbbm{1} \pm \sigma^{}_3}{2}\right) \psi \,,
\end{align}
and, more importantly:
\begin{align}
    \omega(2) \, \psi^\pm \; \left( = \mathrm{i} \sigma^{}_3 \, \psi^\pm \right) = \pm \mathrm{i} \psi^\pm \,.
\end{align}
Then, $\omega(2)$ naturally plays the role of the ``chirality operator'', distinguishing the right- from the left-handed spinor components.

The reality of the Majorana spinor $\xi$ is encapsulated in the requirement that any unitary transformation, $\psi = V\xi$, preserves its invariance under charge conjugation — denoted by the superscript `$^c$' — so that the charge conjugate of the transformed spinor coincides with itself, ${\psi}^c = \psi$. By definition, the charge conjugate of a spinor is given by:
\begin{align}\label{charge}
    {\psi}^c := C \,\overline{\psi} \,,
\end{align}
where the charge conjugation matrix $C$ naturally swaps $\psi^+$ and $\psi^-$, preserving the chirality structure. To ensure consistency, $C$ must satisfy the following condition:
\begin{align}
    C\,\overline{\sigma}_{\dot{\nu}} = \sigma^{}_{\dot{\nu}} \, C \,, \quad {\dot{\nu}} = 1,2\,.
\end{align}
The former identity implies that $C$ should commute with the real $\mathfrak{cl}(2)$ generators and anti-commute with the imaginary ones. We notice that the matrix $C= \varrho^2\sigma^{}_1$ satisfies this condition for any choice of the phase factor $\varrho^2$; the reason for the exponent `$2$' in $\varrho^2$ will be clarified in the sequel (see the convention given above Eq. \eqref{shi'}).

Considering the unitary basis transformation defined by the matrix $V$, such that $\psi^{}_V = V \psi$ and $\sigma^V_{\dot{\nu}} = V\sigma^{}_{\dot{\nu}} V^{-1}$, Eq. \eqref{charge} becomes:
\begin{align}
    \mbox{Eq. \eqref{charge}}:&\;\; {\psi}^c_{} = C\, \overline{\psi} \nonumber\\[0.2cm]
    &\quad\Longrightarrow\quad {\psi}^c_{V} = C_V\, \overline{\psi}_V = C_V\, \overline{V}\; \overline{\psi} = C_V\left(V^{-1}\right)^\top \overline{\psi} \,.
\end{align}
On the other hand, as previously mentioned, the reality condition of the Majorana spinor $\xi$ ensures that under any unitary transformation $ \psi = V \xi $, we have:
\begin{align}
    \psi^c\; \big( = C\, \overline{\psi} \big) = \psi\,, &\quad\mbox{and similarly}\quad {\psi}^c_{V} = {\psi}^{}_{V} \; \big( = V\,\psi \big) \nonumber\\[0.2cm]
    &\Longrightarrow\quad {\psi}^c_{V} = V \big(C\, \overline{\psi}\big) \,.
\end{align}
Equating the two expressions for ${\psi}^c_V$, we obtain:
\begin{align}
    C_V\left(V^{-1}\right)^\top \overline{\psi} = V\, C\, \overline{\psi} \,.
\end{align}
Since this must hold for all $\psi$, the corresponding transformation law for the charge conjugation matrix follows:
\begin{align}\label{CV}
    C_V = V\, C\, V^\top \,.
\end{align}

In the original (Majorana or real) basis $\xi$, however, the charge conjugation matrix $C$ is simply the identity matrix, so that: 
\begin{align}
    \xi^c = \overline{\xi} \,.
\end{align}
We then apply the transformation law \eqref{CV} for the charge conjugation matrix to determine $C_V$ for $\psi = V \xi$. This yields:
\begin{align}
    C_V = V\, V^\top = e^{-\frac{\pi}{4}{\epsilon}} e^{-\mathrm{i}\frac{\pi}{2}\sigma^{}_3}\, e^{\frac{\pi}{4}{\epsilon}} = -\mathrm{i} \sigma^{}_1 \,,
\end{align}
which is consistent with the structure of the Majorana spinor under the new basis. 

Notably, for any $\varrho \in \mathbb{C}$ satisfying $\varrho \overline{\varrho} = 1$, the matrix $U = \varrho V$ retains the same intertwining properties as $V$ in Eqs. \eqref{tur} and \eqref{shi}. Choosing $\varrho = e^{\mathrm{i}\pi/4}$, we obtain:
\begin{align}\label{shi'}
    U = e^{\mathrm{i}\frac{\pi}{4}} V \quad\Longrightarrow\quad U\xi = \frac{1}{\sqrt{2}}
    \begin{pmatrix}
        1  &  -\mathrm{i} \\
        1  &  \mathrm{i}
    \end{pmatrix}
    \begin{pmatrix}
        \xi^+ \\
        \xi^-
    \end{pmatrix} =
    \begin{pmatrix}
        \psi^+ \\
        \psi^-
    \end{pmatrix}\,,
\end{align}
with:
\begin{align}\label{1.23}
    \psi^\pm = \frac{1}{\sqrt{2}} \left( \xi^+ \mp \mathrm{i}\xi^- \right)\,, \quad\mbox{and}\quad C_U = U U^\top = \sigma^{}_1 \,.
\end{align}
More generally, the involutivity property $ (\psi^c)^c = \psi $ implies that the charge conjugation matrix must satisfy $ C \overline{C} = \mathbbm{1} $.

\begin{remark}
    {\emph{(reduced U$(1)$ representation as a plane rotation).} The reduced U$(1)$ representation $\left(e^{\mathrm{i}\varphi} \in \mathrm{U}(1), \xi^\pm\right) \,\longmapsto\, e^{\pm\mathrm{i}\varphi/2} \xi^\pm $ corresponds to a real rotation in the $(\sigma^{}_1, \sigma^{}_2)$-plane:
    \begin{align}
        e^{\mathrm{i}\sigma^{}_3 \frac{\varphi}{2}} \;\sigma^{}_1\; e^{-\mathrm{i}\sigma^{}_3 \frac{\varphi}{2}} = 
        \begin{pmatrix}
        0  &  e^{\mathrm{i}\varphi}  \\
        e^{-\mathrm{i}\varphi}  &  0
        \end{pmatrix} =
        \cos{\varphi}\,\sigma^{}_1 - \sin{\varphi}\,\sigma^{}_2 \,, 
    \end{align}
    \begin{align}
        e^{\mathrm{i}\sigma^{}_3 \frac{\varphi}{2}} \;\sigma^{}_2\; e^{-\mathrm{i}\sigma^{}_3 \frac{\varphi}{2}} = 
        \begin{pmatrix}
        0  &  -\mathrm{i}e^{\mathrm{i}\varphi} \\
        \mathrm{i}e^{-\mathrm{i}\varphi}  &  0
        \end{pmatrix} =
        \sin{\varphi}\,\sigma^{}_1 + \cos{\varphi}\,\sigma^{}_2 \,.
    \end{align}}
\end{remark}

\begin{remark}
    {\emph{(finite cyclic structure of the basis transformation).} It can be visualized that \eqref{shi} and hence \eqref{shi'} generates a finite cyclic group. Indeed, we have:
    \begin{align}
        V^2 = - V^\ast = - V^{-1} \quad\Longrightarrow\quad V^3 = - \mathbbm{1} \quad \mbox{and} \quad U^{24} = \mathbbm{1}\,.
    \end{align}}
\end{remark}

We shall see that the transformation relating the Majorana (real) and chiral basis of $\mathfrak{cl}(4,2)$, the case of interest, obeys the same relations.

\sectionmark{Unitary Transformation from Real to a Diagonal Compact Cartan ...}
\section{Unitary Transformation from Real to a Diagonal Compact Cartan Basis of $\mathfrak{su}(2,2)\cong\mathfrak{so}(4,2)$}
\sectionmark{Unitary Transformation from Real to a Diagonal Compact Cartan ...}
\sectionmark{Unitary Transformation from Real to a Diagonal Compact Cartan ...}

\textbf{\textit{Outline of this section}:} Diagonalizing the complex structure generated by the volume form $E = M(4,2)$ \eqref{U(1)generator} of $\mathfrak{cl}(4,2)$, we split the spinor space $\mathfrak{S}\; (\cong \mathbb{R}^{4,4})$ into two complex conjugate, chiral subspaces $\mathfrak{S}_+$ and $\mathfrak{S}_-$, each carrying a $4$-dimensional (complex) irreducible representation (IR) of the conformal Lie algebra $\mathfrak{su}(2,2)\cong\mathfrak{so}(4,2) \, \big(\subset \mathfrak{cl}(4,2)\big)$ and thus also the corresponding IR of $\mathfrak{u}(2,2)$; see Remark \ref{Remark IR u(2,2)} for this latter point, while a more systematic discussion of these representations will be given in the next chapter. The resulting chiral basis of $8\times 8$ complex matrices — the complex counterpart of the $8\times 8$ real matrices $m^{}_0, m^{}_5, m^{}_1, \dots, m^{}_4$, namely:
\begin{align}
    &\Gamma^{}_{0} := \Gamma^{}_{{\boldsymbol{\imath}}}\,, \quad \Gamma^{}_{5} := \Gamma^{}_{{\boldsymbol{\jmath}}}\,, \quad \Gamma^{}_{1} := \Gamma^{}_{\boldsymbol{\ell}{\boldsymbol{\imath}}}\,, \nonumber\\[0.2cm]
    &\Gamma^{}_{2} := \Gamma^{}_{{\boldsymbol{\jmath}}\boldsymbol{\ell}}\,, \quad \Gamma^{}_{3} := L_{\boldsymbol{\ell}\boldsymbol{k}}\,, \quad \Gamma^{}_{4} := \Gamma^{}_{\boldsymbol{\ell}}\,,
\end{align}
can be further restricted by diagonalizing a Cartan basis among the $\mathfrak{su}(2,2)\cong\mathfrak{so}(4, 2)$ generators:
\begin{align}\label{GammaMuNu-JPG}
    \Gamma^{}_{\mu\nu} := \frac{1}{2} \big[\Gamma^{}_\mu, \Gamma^{}_\nu\big]\,,
\end{align}
where $\mu,\nu = {0,5,1,2,3,4}$ (see Eq. \eqref{8uni}-\eqref{8uni'*}). There are two types of such bases: 
\begin{enumerate}[leftmargin=0.8cm]
    \item{One, that belongs to the Lie algebra $\mathfrak{so}(2) \oplus \mathfrak{so}(4)\; \big(\cong \mathfrak{u}(1) \oplus \mathfrak{su}(2) \oplus \mathfrak{su}(2)\big)$ of the maximal compact subgroup $\mathrm{S}\big(\mathrm{U}(2) \times \mathrm{U}(2)\big)$ of SU$(2,2)$.}
    \item{Another one that involves generators like $\Gamma^{}_{{03}}$ and $\Gamma^{}_{{45}}$ on non-compact boost transformations.}
\end{enumerate}
Here, we shall restrict attention to the compact Cartan basis with generators $\Gamma^{}_{05}$, $\Gamma^{}_{12}$, and $\Gamma^{}_{34}$ — the chiral counterpart of $m^{}_{05}$ \eqref{10}, $m^{}_{12}$ \eqref{20}, and $m^{}_{34}$ \eqref{30}, respectively. 

Let us begin by noting that the product of each pair of the real matrices $m^{}_{05}$, $m^{}_{12}$, and $m^{}_{34}$ gives rise to the (Hermitian) volume forms of either $\mathfrak{cl}(2,2)$ or $ \mathfrak{cl}(4,0) $. These volume forms are diagonal and commute with each individual $ E_{AB} $ appearing in \eqref{U(1)generator} and \eqref{10}-\eqref{30}. Indeed, we have:
\begin{align}\label{99}
    M_{12}(2,2) \, x &= \big( m^{}_{05} (m^{}_{12} \, x) \big) \nonumber\\[0.2cm]
    &= {\boldsymbol{\imath}} \Big( {\boldsymbol{\jmath}} \big( \boldsymbol{\ell}{\boldsymbol{\imath}} ({\boldsymbol{\jmath}}\boldsymbol{\ell} \, x) \big) \Big) \nonumber\\[0.2cm]
    &= \bigoplus(\mathbbm{1}_{\mathbbm{1}\boldsymbol{k}} - \mathbbm{1}_{05} - \mathbbm{1}_{12} + \mathbbm{1}_{34}) \, x \,,  
\end{align}
\begin{align}\label{88}
    M_{34}(2,2) \, x &= \big( m^{}_{05} (m^{}_{34} \, x ) \big) \nonumber\\[0.2cm]
    &= {\boldsymbol{\imath}} \Big( {\boldsymbol{\jmath}} \big( \boldsymbol{\ell}\boldsymbol{k} (\boldsymbol{\ell} \, x) \big) \Big) \nonumber\\[0.2cm]
    &= \bigoplus(\mathbbm{1}_{\mathbbm{1}\boldsymbol{k}} - \mathbbm{1}_{05} + \mathbbm{1}_{12} - \mathbbm{1}_{34}) \, x \,,  
\end{align}
\begin{align}\label{77}
    M(4,0) \, x &= \big( m^{}_{12} (m^{}_{{43}} \, x ) \big) \nonumber\\[0.2cm]
    &= \boldsymbol{\ell}{\boldsymbol{\imath}} \Big( {\boldsymbol{\jmath}}\boldsymbol{\ell} \big( \boldsymbol{\ell} (\boldsymbol{\ell} \boldsymbol{k}\, x) \big) \Big) \nonumber\\[0.2cm]
    &= \bigoplus(\mathbbm{1}_{\mathbbm{1}\boldsymbol{k}} + \mathbbm{1}_{05} - \mathbbm{1}_{12} - \mathbbm{1}_{34}) \, x \,,
\end{align}
where $x\in\mathrm{Span}\{A\}$ (see Eqs. \eqref{8uni}-\eqref{8uni'*}) and the $\mathbbm{1}_{AB}$ represents the unit $2\times 2$ matrices defined in Eq. \eqref{1def}. Again, for the convention $\bigoplus (\dots)$, see Sect. \ref{Sect. Convention}. The above results can be verified by substituting $x$ with any of the split-octonion units (and in fact with all of them) and consulting the multiplication table \eqref{table1}. 

Regarding the above construction, we observe that: 
\begin{enumerate}[leftmargin=0.8cm]
    \item{The following relation holds:
    \begin{align}
        M(4,0)\, x = M_{12}(2,2) \big(M_{34}(2,2) \,x\big)\,.
    \end{align}
    To verify this relation, it is sufficient — by Proposition \ref{proposition 2.1} — to follow a procedure analogous to that in Eq. \eqref{8888888888888}.}
    \item{The Lorentz $\mathfrak{cl}(3,1)$ volume form is:
    \begin{align}
        E_L\, x = M(3,1) \, x &= \big( m^{}_{{03}} (m^{}_{12} \, x) \big) \nonumber\\[0.2cm]
        &= {\boldsymbol{\imath}} \Big( \boldsymbol{\ell}\boldsymbol{k} \big( \boldsymbol{\ell}{\boldsymbol{\imath}} ({\boldsymbol{\jmath}}\boldsymbol{\ell} \, x) \big) \Big) \nonumber\\[0.2cm]
        \label{1.25} &= \bigoplus(E_{\mathbbm{1} {1}} - E_{{0 4}} + E_{{5 3}} - E_{{2} \boldsymbol{k}}) \, x \,, 
    \end{align}
    where, once again, it can be verified by substituting $x$ with any of the split-octonion units (see Eqs. \eqref{8uni}-\eqref{8uni'*}) and consulting the multiplication table \eqref{table1}. We observe that the aforementioned volume forms $M_{12}(2,2)$ and $M(4,0)$ do not commute with the Lorentz one $M(3,1)$:
    \begin{align}
        &\big[ M_{12}(2,2) , M(3,1) \big] x \neq 0 \,, \nonumber\\[0.2cm]
        &\big[ M(4,0) , M(3,1) \big] x \neq 0 \,.
    \end{align}
    Once again, verification of these results reduces — by Proposition \ref{proposition 2.1} — to following a procedure analogous to that in Eq. \eqref{8888888888888}.}
\end{enumerate}

Happily, there exists a \emph{third-order}\footnote{We refer informally to a Clifford basis monomial involving $k$ distinct generators as an element of \emph{order $k$}; equivalently, it belongs to the degree-$k$ component $\Lambda^k$ of the standard exterior-algebra grading.} diagonal element $D \in\mathfrak{cl}(4,2)$ that anti-commutes with the Cartan elements $m^{}_{05}$, $m^{}_{12}$, and $m^{}_{34}$ of the maximal compact Lie subalgebra $\mathfrak{so}(2) \oplus \mathfrak{so}(4) \subset \mathfrak{so}(4,2)$, and consequently commutes with the fourth-order diagonal volume forms given in Eq. \eqref{99}-\eqref{77}:
\begin{align} \label{D}
    D \,x &= m^{}_{5}\big(m^{}_{4} (m^{}_{2} \,x)\big) \nonumber\\[0.2cm]
    &= {\boldsymbol{\jmath}}\, \big({\boldsymbol{\ell}}\, ({\boldsymbol{\jmath}}\boldsymbol{\ell} \,x)\big) \nonumber\\[0.2cm]
    &= \bigoplus(D_{\mathbbm{1}\boldsymbol{k}} - D_{05} - D_{12} - D_{34})\, x \nonumber\\[0.2cm]
    &= \underbrace{\bigoplus(\mathbbm{1}_{\mathbbm{1}, {5}, {2}, {4}} - \mathbbm{1}_{\boldsymbol{k}, {0}, {1}, {3}})\,x}_{\text{as a by-product of Sect. \ref{Sect. Convention} (see Eq. \eqref{Oplus M+N=N+M})}} \,, \nonumber\\
\end{align}
such that:
\begin{align}\label{1.26}
    \big\{ D, m^{}_{05} \big\}\,x = \big\{ D, m^{}_{12} \big\}\,x = \big\{ D, m^{}_{34} \big\}\,x = \big\{ D, E \big\}\,x = 0 \,,
\end{align}
\begin{align}\label{1000.26}
    \big[D, M_{12}(2,2)\big]\,x = \big[D, M_{34}(2,2)\big]\,x = \big[D, M(2,2)\big]\,x = 0\,,
\end{align}
where $x\in\mathrm{Span}\{A\}$ (see Eqs. \eqref{8uni}-\eqref{8uni'*}), $D_{AB}$ denotes the $2 \times 2$ matrices defined in Eq. \eqref{ddef}, while ${\mathbbm{1}}_{A,B,C,D}$ represents the $4 \times 4$ identity matrix corresponding to its indices. To verify relation \eqref{D}, one simply substitutes $x$ with any of the split-octonion units (and in fact with all of them) and refers to the multiplication table \eqref{table1}. The verification of the (anti-)commutation relations, however, relies on Proposition \ref{proposition 2.1} and proceeds analogously to the method outlined in Eq. \eqref{8888888888888}.

\begin{remark} 
    {\emph{(necessity of the operator $D$ and resolution of the degeneracy).} Let $\mathfrak{S} \cong \mathbb{R}^{4,4}$ be the spinor representation space of $\mathfrak{u}(2,2) \subset \mathfrak{cl}(4,2)$. We have already introduced the four mutually commuting $8 \times 8$ skew-symmetric operators:
    \begin{align}
        E \; (\text{see Eq. \eqref{U(1)generator}})\,, &\quad 
        m^{}_{05} \; (\text{Eq. \eqref{10}})\,, \nonumber\\[0.2cm]
        m^{}_{12} \; (\text{Eq. \eqref{20}})\,, &\quad 
        m^{}_{34} \; (\text{Eq. \eqref{30}})\,,
    \end{align}
    where $E$ generates the $\mathfrak{u}(1)$ center of $\mathfrak{u}(2,2)$, and $m^{}_{05}, m^{}_{12}$, and $m^{}_{34}$ form a trace-orthogonal Cartan basis of the maximal compact subalgebra $\mathfrak{so}(2)\oplus\mathfrak{so}(4) \subset \mathfrak{so}(4,2)$, such that:
    \begin{enumerate}[leftmargin=0.8cm]
        \item{Each operator can be expressed as a direct sum of four $2 \times 2$ skew-symmetric blocks $E_{AB}$ (up to a sign), with exactly the same set of index pairs $A,B = \mathbbm{1}, \boldsymbol{k}, \boldsymbol{\imath}, \boldsymbol{\jmath}, \boldsymbol{\ell}\boldsymbol{\imath}, \boldsymbol{\jmath}\boldsymbol{\ell}, \boldsymbol{\ell}\boldsymbol{k}, \boldsymbol{\ell}$ or, more conveniently, $A, B = \mathbbm{1}, \boldsymbol{k},0,5,1,2,3,4$ (see Eqs. \eqref{8uni}-\eqref{8uni'*}); that is:
        \begin{align}
            E_{\mathbbm{1}\boldsymbol{k}}\,,\quad E_{05}\,,\quad E_{12}\,,\quad E_{34}\,,
        \end{align}
        where the corresponding blocks define mutually independent $2$-dimensional real invariant subspaces.}

        \item{The common action of $E$, $m^{}_{05}$, $m^{}_{12}$, and $m^{}_{34}$ decomposes $\mathfrak{S}$ into $2$-dimensional real invariant subspaces (see Appendix \ref{footnote E2=-1}). Indeed, each $2\times 2$ block $E_{AB}$ corresponds to one such invariant subspace. Upon complexification, each of these real invariant subspaces splits into one-dimensional joint eigenspaces corresponding to opposite eigenvalue quadruples \linebreak $(\mu^{}_E,\mu^{}_{05},\mu^{}_{12},\mu^{}_{34})$ and $(-\mu^{}_E,-\mu^{}_{05},-\mu^{}_{12},-\mu^{}_{34})$. In particular, each real invariant subspace is spanned by two linearly independent vectors; for example:
        \begin{align}
            v_1 = \begin{pmatrix} 1 \\ 1 \end{pmatrix}\,, \quad v_2 = \begin{pmatrix} 1 \\ -1 \end{pmatrix}\,.
        \end{align}
        However, there is no canonical way to distinguish these individual vectors (say $v_1$ and $v_2$) inside this real invariant subspace using $E$, $m^{}_{05}$, $m^{}_{12}$, and $m^{}_{34}$ alone, producing an intrinsic twofold degeneracy.}
    \end{enumerate}

    To resolve this degeneracy, one requires an extra operator that acts within each $2$-dimensional real invariant subspace and distinguishes its two components. This operator is indeed the third-order Clifford element $D \in \mathfrak{cl}(4,2)$ (see Eq. \eqref{D}), which satisfies the anti-commutation relations \eqref{1.26}. Since $D$ anti-commutes with $E, m^{}_{05}, m^{}_{12}$, and $m^{}_{34}$, its action on a joint eigenvector $v$ in the complexification of such a subspace, with eigenvalues $\mu^{}_X$, satisfies:
    \begin{align}
        X (D v) = - D (X v) = - \mu^{}_X (D v)\,, \quad\mbox{for}\quad X = E\,, \; m^{}_{05}\,, \; m^{}_{12}\,, \; m^{}_{34}\,,
    \end{align}
    where $\mu^{}_X$ are the eigenvalues of $v$ under $X$. Thus, $D$ maps the joint eigenspace corresponding to the eigenvalue quadruple $(\mu^{}_E,\mu^{}_{05},\mu^{}_{12},\mu^{}_{34})$ onto that corresponding to $(-\mu^{}_E,-\mu^{}_{05},-\mu^{}_{12},-\mu^{}_{34})$. The direct sum of these two opposite joint eigenspaces is the complexification of the corresponding $2$-dimensional real invariant subspace. Consequently, $D$ preserves this real invariant subspace while, upon complexification, interchanging its two opposite eigenvalue sectors. In particular, $Dv$ is linearly independent of $v$; if $Dv$ were proportional to $v$, i.e., $Dv = \alpha v$ for some non-zero scalar $\alpha$, then acting with any $X = E, m^{}_{05}, m^{}_{12}, m^{}_{34}$ would yield $\mu^{}_X = -\mu^{}_X$, forcing $\mu^{}_X = 0$, which does not occur for the non-trivial $2 \times 2$ blocks under consideration. Therefore, $D$ resolves the intrinsic twofold degeneracy within each corresponding real invariant subspace through its two distinct eigenvalues $+1$ and $-1$; equivalently, upon complexification, it interchanges the two opposite joint-eigenvalue sectors. Concretely, $D$ acts as the Pauli-$\sigma_3$ matrix within each $2 \times 2$ block and exchanges the two vectors $v_1$ and $v_2$ introduced above:
    \begin{align}
        D_{AB} \, v_1 &= \begin{pmatrix} 1 & 0 \\ 0 & -1 \end{pmatrix}\begin{pmatrix} 1 \\ 1 \end{pmatrix} = \begin{pmatrix} 1 \\ -1 \end{pmatrix} = v_2\,, \\[0.2cm]
        D_{AB} \, v_2 &= \begin{pmatrix} 1 & 0 \\ 0 & -1 \end{pmatrix}\begin{pmatrix} 1 \\ -1 \end{pmatrix} = \begin{pmatrix} 1 \\ 1 \end{pmatrix} = v_1\,.
    \end{align}
    Equivalently, in the canonical basis of each block, $D_{AB}$ has the distinct eigenvalues $+1$ and $-1$. In addition, $D$ commutes (see Eq. \eqref{1000.26}) with all fourth-order volume forms $M_{12}(2,2)$ \eqref{99}, $M_{34}(2,2)$ \eqref{88}, and $M(4,0)$ \eqref{77}, ensuring that it preserves the previously obtained block decomposition.

    Accordingly, after including $D$, the extended set of mutually commuting or anti-commuting operators:
    \begin{align}
        \Big\{ E,\, m^{}_{05},\, m^{}_{12},\, m^{}_{34},\, D \Big\}
    \end{align}
    provides a complete resolution of the $2$-dimensional degeneracy in the spinor representation space $\mathfrak{S}$.    
}\end{remark}

\begin{remark}{
    \emph{(non-uniqueness of the operator $D$).} In connection with the above, we note that:
    \begin{enumerate}[leftmargin=0.8cm]
        \item{The representation of \eqref{D} as a sum of $2\times 2$ diagonal matrices $D_{AB}$ is not unique. Actually, as a by-product of the discussion in Sect. \ref{Sect. Convention} (in particular, Eq. \eqref{Oplus M+N=N+M}), we may write, for instance:
        \begin{align}
            \bigoplus\big(D_{\mathbbm{1}\boldsymbol{k}} - D_{05}\big) &= \bigoplus\big(D_{\mathbbm{1}{0}} - D_{\boldsymbol{k}{5}}\big)\,, \nonumber\\[0.2cm]
            \bigoplus\big(D_{12} + D_{34}\big) &= \bigoplus\big(D_{14} - D_{23}\big)\,, \quad \mbox{etc.}
        \end{align}}
        
        \item{Besides the representation of \eqref{D}, it is important to emphasize that the choice of the third-order diagonal matrix $D=m^{}_5 m^{}_4 m^{}_2$ in \eqref{D} is not unique either. This choice corresponds to triples of imaginary split-octonion units that generate an associative subalgebra within the larger, non-associative algebra of split-octonions. For a more detailed discussion, see Appendix \ref{Appendix A}.}
    \end{enumerate}
}\end{remark}

\begin{Proposition}\label{proposition S=DE}
    \textbf{($\mathfrak{cl}(2)$ subalgebra generated by $E$, $D$, and $S := D\, E$ in the split-octonionic realization).} The operators $E$ \eqref{U(1)generator}, $D$ \eqref{D}, and their product $S := D\, E$ obey the same multiplication laws as their $2\times 2$ matrix components (see Remark \ref{Remark E,D,S,1-AB}):
    \begin{align}\label{S=DE=-ED}
        S = D\, E = - E\, D = \bigoplus\big(S_{\mathbbm{1}\boldsymbol{k}} - S_{05} - S_{12} - S_{34}\big) \,,
    \end{align}
    \begin{align}\label{DS=-SD=E}
        D\, S = - S\, D = E \,, \quad\mbox{and}\quad D^2 = S^2 = {\mathbbm{1}}_{8} = - E^2\,,
    \end{align}
    generating a representation of the $\mathfrak{cl}(2)$ algebra. Here, ${\mathbbm{1}}_{8}$ denotes the $8\times 8$ identity matrix. [We note that the labels of the indices of $D$ and $S$ will change in the chiral basis.]
\end{Proposition}

\begin{proof}{ 
    Let us rewrite the above identities explicitly, one by one, by substituting the operators $E$ \eqref{U(1)generator} and $D$ \eqref{D} as they act from the left on $x\in\mathrm{Span}\{A\}$ (see Eqs. \eqref{8uni}-\eqref{8uni'*}):
    \begin{enumerate}[leftmargin=0.8cm]
        \item{We begin with Eq. \eqref{S=DE=-ED}:
        \begin{align}
            S \,x = D\, (E \,x) = m^{}_{5}\left[m^{}_{4} \left[m^{}_{2} \Bigg[m^{}_{0} \Bigg(m^{}_{5} \bigg(m^{}_{1} \Big(m^{}_{2} \big(m^{}_{3} (m^{}_{4} \, x ) \big) \Big) \bigg) \Bigg) \;\Bigg]\; \right]\; \right]\,.
        \end{align}
        Applying Proposition \ref{proposition 2.1} in the same way as Eq. \eqref{8888888888888}, and consulting the multiplication table \eqref{table1}, we find:
        \begin{align}
            S \,x = D\, (E \,x) &= - m^{}_{0} \big(m^{}_{1} (m^{}_{3} \, x ) \big) \nonumber\\[0.2cm]
            &= - {\boldsymbol{\imath}} \big( \boldsymbol{\ell}{\boldsymbol{\imath}} ( \boldsymbol{\ell}\boldsymbol{k}\, x ) \big) \nonumber\\[0.2cm]
            &= \bigoplus(S_{\mathbbm{1}\boldsymbol{k}} - S_{05} - S_{12} - S_{34})\, x \,.
        \end{align} 
        For the convention $\bigoplus (\dots)$, see Sect. \ref{Sect. Convention}. Following a similar procedure, we obtain:
        \begin{align}\label{S}
            S \,x = -E\, (D \,x) &= - m^{}_{0} \big(m^{}_{1} (m^{}_{3} \, x ) \big) \nonumber\\[0.2cm]
            &= - {\boldsymbol{\imath}} \big( \boldsymbol{\ell}{\boldsymbol{\imath}} ( \boldsymbol{\ell}\boldsymbol{k}\, x ) \big) \nonumber\\[0.2cm]
            &= \bigoplus(S_{\mathbbm{1}\boldsymbol{k}} - S_{05} - S_{12} - S_{34})\, x \,.
        \end{align}
        Therefore, Eq. \eqref{S=DE=-ED} follows.}

        \item{Proceeding analogously, we find $D\, S = - S\, D = E$; in particular:
        \begin{align}
            D (S\, x) &= - m^{}_{5} \Bigg(m^{}_{4} \bigg(m^{}_{2} \Big(m^{}_{0} \big(m^{}_{1} (m^{}_{3}\, x) \big) \Big) \bigg) \Bigg) \nonumber\\
            &= m^{}_{0} \Bigg(m^{}_{5} \bigg(m^{}_{1} \Big(m^{}_{2} \big(m^{}_{3} (m^{}_{4} \, x ) \big) \Big) \bigg) \Bigg) = E\, x\,,
        \end{align}
        and:
        \begin{align}
            -S (D\, x) &= m^{}_{0} \Bigg(m^{}_{1} \bigg(m^{}_{3} \Big(m^{}_{5} \big(m^{}_{4} (m^{}_{2}\, x) \big) \Big) \bigg) \Bigg) \nonumber\\
            &= m^{}_{0} \Bigg(m^{}_{5} \bigg(m^{}_{1} \Big(m^{}_{2} \big(m^{}_{3} (m^{}_{4} \, x ) \big) \Big) \bigg) \Bigg) = E\, x\,.
        \end{align}}

        \item{Lastly, by a similar argument, the second identity in Eq. \eqref{DS=-SD=E} follows trivially:
        \begin{align}
            D (D\,x) = S(S\,x) = {\mathbbm{1}}_{8} = - E (E\,x)\,.
        \end{align}}
    \end{enumerate}
}\end{proof}

We now have all the ingredients to define the $\mathfrak{cl}(4,2)$ analogue of the similarity transformation \eqref{shi} (or, more generally, \eqref{shi'}) corresponding to the $\mathfrak{cl}(2)$ case, which simultaneously diagonalizes the $\mathfrak{u}(1)$ generator $E$ \eqref{U(1)generator} of the conformal Lie algebra $\mathfrak{u}(2,2)\cong\mathfrak{su}(2,2)\oplus\mathfrak{u}(1)$ and the trace-orthogonal Cartan elements $m^{}_{05}$ \eqref{10}, $m^{}_{12}$ \eqref{20}, and $m^{}_{34}$ \eqref{30} of the maximal compact subalgebra $\mathfrak{so}(2) \oplus \mathfrak{so}(4) \subset \mathfrak{so}(4,2)\cong\mathfrak{su}(2,2)$. The corresponding unitary transformation $U$ explicitly reads as:
\begin{align}\label{1.31}
    &U = \varrho V \,, \quad\varrho\in\mathbb{C}\,,\quad  \varrho\overline{\varrho} = 1\,, \nonumber\\[0.2cm]
    &V = e^{-\frac{\pi}{4}E}\, e^{-\mathrm{i}\frac{\pi}{4} D} = \frac{1}{2} \left( \mathbbm{1} - \mathrm{i}S - E - \mathrm{i}D \right)\,,\nonumber\\[0.2cm]
    &V^\ast \; \big( = V^{-1} \big) = \frac{1}{2} \left( \mathbbm{1} + \mathrm{i}S + E + \mathrm{i}D \right) = -V^2 \quad\Longrightarrow\quad V^3 = - \mathbbm{1}_8\,,
    \end{align}
    where, by construction, we have:
    \begin{align}
        UU^\ast = U^\ast U = \mathbbm{1}_8 = V^\ast V = VV^\ast \,.
\end{align}
Note that the phase factor is chosen as $\varrho = e^{\mathrm{i}\pi/4}$ so that the corresponding charge-conjugation matrix $C$ is real; this point will be clarified later. Furthermore, it is worth noting that, in the above matrix realization, the superscript `$^\ast$' denotes the Hermitian conjugate, defined as $V^\ast = \overline{V}^\top = \overline{V^\top}$, where the `$\overline{\phantom{a}}$' symbol represents the complex conjugate, and the superscript `$^\top$' indicates the transpose.

To achieve our goal of simultaneous diagonalization of the aforementioned operators, there are two (semi-)equivalent approaches:
\begin{enumerate}[leftmargin=0.8cm]
    \item{First, we transform the operators and let the conjugated versions \linebreak $U E U^{-1}$, $U m_{05} U^{-1},\;\dots$ act on the original Majorana (real) basis $A = \mathbbm{1},\, \boldsymbol{k},\, \nu \in \big\{ 0:={\boldsymbol{\imath}},\, 5:={\boldsymbol{\jmath}},\, 1:=\boldsymbol{\ell}{\boldsymbol{\imath}},\, 2:={\boldsymbol{\jmath}}\boldsymbol{\ell},\,  3:=\boldsymbol{\ell}\boldsymbol{k}\,, 4:=\boldsymbol{\ell} \big\}$.}

    \item{Second, we leave the operators in their original form but instead transform the basis, working with $U^{-1}A$, i.e., the chiral basis.}
\end{enumerate}
The two procedures are equivalent in the sense that the resulting matrix realizations are similar; however, the meaning of the indices differs in each case. In what follows, we present both approaches explicitly, beginning with the first.

\subsection{Operator Transformation in the Majorana Basis}

Under the linear operator $U$, the volume form $M(4,2)=E$ \eqref{U(1)generator} transforms as:
\begin{align}\label{omega(4,2)=iD} 
    \omega(4,2) &= U\,M(4,2)\,U^{-1} \; \big(= V\,M(4,2)\,V^{-1}\big) \nonumber\\[0.2cm]
    &= U\,E\,U^{-1} \; \big(= V\,E\,V^{-1}\big) \nonumber\\[0.2cm]   
    &= e^{-\frac{\pi}{4}E}\, e^{-\mathrm{i}\frac{\pi}{4} D} \; E\; e^{\mathrm{i}\frac{\pi}{4} D}\, e^{\frac{\pi}{4}E} \nonumber\\[0.2cm]
    &= E\; e^{-\frac{\pi}{4}E}\, e^{\mathrm{i}\frac{\pi}{2} D}\, e^{\frac{\pi}{4}E} \nonumber\\[0.1cm]
    &= \frac{1}{2}\, E\, (1-E)(\mathrm{i}D) (1+E) \nonumber\\[0.2cm]
    &= \mathrm{i} E\,D\,E = \mathrm{i} D \,.
\end{align}
To obtain the above result, we consider the left multiplication on the split-octonion units and invoke Proposition \ref{proposition S=DE}. Under the action of $\omega(4,2)$ on the split-octonion units — namely, on the Majorana basis — the operator assumes the diagonal matrix form (up to a factor of `$\mathrm{i}$') given in Eq. \eqref{D}:
\begin{align}\label{omega(4,2)=iDMat} 
    \omega(4,2) \,x = \big(U\,E\,U^{-1}\big)\, x = \mathrm{i} D \,x = \mathrm{i}\bigoplus (D_{\mathbbm{1}\boldsymbol{k}} - D_{05} - D_{12} - D_{34})\, x\,,
\end{align}
where $x\in\mathrm{Span}\{A\}$ (see Eqs. \eqref{8uni}-\eqref{8uni'*}). For the convention $\bigoplus(\dots)$, refer to Sect. \ref{Sect. Convention}.

We observe that the operator $D$ \eqref{D} commutes with the factors $m_{\mu^\prime}$, for $\mu^\prime = 5,4,2$, while it anti-commutes with the remaining factors $m_{\nu^\prime}$, for $\nu^\prime = 0,1,3$. By contrast, the operator $E$ anti-commutes with all $m_{\nu}$, for $\nu = 0,5,1,2,3,4$. These properties can once again be readily justified by substituting the operators as they act from the left on $x\in\mathrm{Span}\{A\}$ (see Eqs. \eqref{8uni}-\eqref{8uni'*}), and then invoking Proposition \ref{proposition 2.1} together with an argument analogous to that in Eq. \eqref{8888888888888}; see also the proof of Proposition \ref{proposition S=DE}. With these observations in hand and using Proposition \ref{proposition S=DE}, we obtain:
\begin{align}\label{Gamma542}
    \Gamma^{}_{\mu^\prime} &= U\, m^{}_{\mu^\prime}\, U^{-1} \; \big( = V\, m^{}_{\mu^\prime}\, V^{-1} \big) \nonumber\\[0.2cm]
    &= e^{-\frac{\pi}{4}E}\, e^{-\mathrm{i}\frac{\pi}{4} D} \; m^{}_{\mu^\prime} \; e^{\mathrm{i}\frac{\pi}{4} D}\, e^{\frac{\pi}{4}E} \nonumber\\[0.2cm]
    &= m^{}_{\mu^\prime} \, e^{\frac{\pi}{2}E} \nonumber\\[0.2cm]
    &= m^{}_{\mu^\prime} E\,, \quad\mbox{for}\quad {\mu^\prime} = {5,4,2}\,,
\end{align}
\begin{align}\label{Gamma013}
    \Gamma^{}_{\nu^\prime} &= U\, m^{}_{\nu^\prime}\, U^{-1} \; \big(= V\, m^{}_{\nu^\prime}\, V^{-1}\big) \nonumber\\[0.2cm]
    &= e^{-\frac{\pi}{4}E}\, e^{-\mathrm{i}\frac{\pi}{4} D} \; m^{}_{\nu^\prime} \; e^{\mathrm{i}\frac{\pi}{4} D}\, e^{\frac{\pi}{4}E} \nonumber\\[0.2cm]
    &= m^{}_{\nu^\prime} \; e^{\frac{\pi}{4}E}\, e^{\mathrm{i}\frac{\pi}{2} D}\, e^{\frac{\pi}{4}E} \nonumber\\[0.2cm]
    &= \frac{1}{2} \, m^{}_{\nu^\prime} (1+E) (\mathrm{i}D) (1+E) \nonumber\\[0.2cm]
    &= \mathrm{i} m^{}_{\nu^\prime} D\,, \quad\mbox{for}\quad {\nu^\prime} = {0,1,3}\,, 
\end{align}
and:
\begin{align}\label{GammaMuNu}
    \Gamma^{}_{{\mu^\prime}{\nu^\prime}} &= U\, m^{}_{{\mu^\prime}{\nu^\prime}}\, U^{-1} \; \big(= V\, m^{}_{{\mu^\prime}{\nu^\prime}}\, V^{-1}\big) \nonumber\\[0.2cm]
    &= e^{-\frac{\pi}{4}E}\, e^{-\mathrm{i}\frac{\pi}{4} D} \; m^{}_{{\mu^\prime}{\nu^\prime}} \; e^{\mathrm{i}\frac{\pi}{4} D}\, e^{\frac{\pi}{4}E} \nonumber\\[0.2cm]
    &= m^{}_{{\mu^\prime}{\nu^\prime}} \; e^{-\frac{\pi}{4}E}\, e^{\mathrm{i}\frac{\pi}{2} D}\, e^{\frac{\pi}{4}E} \nonumber\\[0.2cm]
    &= \frac{1}{2}\, m^{}_{{\mu^\prime}{\nu^\prime}} (1-E)(\mathrm{i}D)(1+E) \nonumber\\[0.2cm]
    &= \mathrm{i}m^{}_{{\mu^\prime}{\nu^\prime}}\, S\,, \quad\mbox{while}\quad {\mu^\prime} = {5,4,2}\,, \quad\mbox{and}\quad {\nu^\prime} = {0,1,3}\,,
\end{align}
but, when $\mu^\prime$ and $\nu^\prime$ simultaneously belong to either ${0,1,3}$ or ${5,4,2}$, we have:
\begin{align}\label{GammaMuNu'}
    \Gamma^{}_{{\mu^\prime}{\nu^\prime}} &= U\, m^{}_{{\mu^\prime}{\nu^\prime}}\, U^{-1} \; \big(= V\, m^{}_{{\mu^\prime}{\nu^\prime}}\, V^{-1}\big) \nonumber\\[0.2cm]
    &= e^{-\frac{\pi}{4}E}\, e^{-\mathrm{i}\frac{\pi}{4} D} \; m^{}_{{\mu^\prime}{\nu^\prime}} \; e^{\mathrm{i}\frac{\pi}{4} D}\, e^{\frac{\pi}{4}E} \nonumber\\[0.2cm]
    &= m^{}_{{\mu^\prime}{\nu^\prime}} \,, \quad\mbox{while}\quad {\mu^\prime,\nu^\prime} = {5,4,2}\,, \quad\mbox{or}\quad {\mu^\prime,\nu^\prime} = {0,1,3}\,.
\end{align}

Using Eqs. \eqref{GammaMuNu} and \eqref{S}, together with Proposition \ref{proposition 2.1} and the multiplication table \eqref{table1} — in a manner analogous to the proof of Proposition \ref{proposition S=DE} — one verifies that the transformed versions of the skew-symmetric Cartan matrices $m^{}_{05}, m^{}_{12}$, and $m^{}_{34}$ (see Eqs. \eqref{10}-\eqref{30}) in the Majorana (real) basis assume, respectively, the following diagonal forms:
\begin{align}    
    \label{Cartan'www.} \Gamma^{}_{05} \, x &= \big(U\, m^{}_{05}\, U^{-1}\big) \, x \nonumber\\[0.2cm]
    &= \mathrm{i} m^{}_{05} \, S \, x \nonumber\\[0.2cm]
    &= - \mathrm{i} m^{}_{0} \bigg(m^{}_{5} \Big(m^{}_{0} \big(m^{}_{1} (m^{}_{3}\, x) \big) \Big)\bigg) \quad \nonumber\\[0.2cm]
    &= - \mathrm{i} m^{}_{5} \big(m^{}_{1} (m^{}_{3}\, x) \big) \nonumber\\[0.2cm]
    &= - \mathrm{i} {\boldsymbol{\jmath}}\, \big(\boldsymbol{\ell}\boldsymbol{\imath}\, (\boldsymbol{\ell}{\boldsymbol{k}}\, x) \big) \nonumber\\[0.2cm]
    &= \mathrm{i}\bigoplus \big( - D_{\mathbbm{1}{\boldsymbol{k}}} + D_{05} - D_{12} - D_{34} \big)\, x\,,
\end{align}
\begin{align}
    \label{AfsaneH} \Gamma^{}_{12} \, x &= \big(U\, m^{}_{12}\, U^{-1}\big) \, x \nonumber\\[0.2cm]
    &= \mathrm{i} m^{}_{12} \, S \, x \nonumber\\[0.2cm]
    &= - \mathrm{i} m^{}_{1} \bigg(m^{}_{2} \Big(m^{}_{0} \big(m^{}_{1} (m^{}_{3} \, x) \big) \Big)\bigg) \quad \nonumber\\[0.2cm]
    &= - \mathrm{i} m^{}_{2} \big(m^{}_{0} (m^{}_{3} \, x) \big) \nonumber\\[0.2cm]
    &= - \mathrm{i} {\boldsymbol{\jmath}}\boldsymbol{\ell}\, \big(\boldsymbol{\imath}\, (\boldsymbol{\ell}{\boldsymbol{k}} \, x) \big) \nonumber\\[0.2cm]
    &= \mathrm{i}\bigoplus \big( D_{\mathbbm{1}{\boldsymbol{k}}} + D_{05} - D_{12} + D_{34} \big) \, x\,,
\end{align}
\begin{align}
    \Gamma^{}_{34} \, x &= \big(U\, m^{}_{34}\, U^{-1}\big) \, x \nonumber\\[0.2cm]
    &= \mathrm{i} m^{}_{34} \, S \, x \nonumber\\[0.2cm]
    &= - \mathrm{i} m^{}_{3} \bigg(m^{}_{4} \Big(m^{}_{0} \big(m^{}_{1} (m^{}_{3} \, x) \big) \Big)\bigg) \quad \nonumber\\[0.2cm]
    &= \mathrm{i} m^{}_{4} \big(m^{}_{0} (m^{}_{1} \, x) \big) \nonumber\\[0.2cm]
    &= \mathrm{i} \boldsymbol{\ell}\, \big(\boldsymbol{\imath}\, (\boldsymbol{\ell}\boldsymbol{\imath} \, x) \big) \nonumber\\[0.2cm]
    &= \mathrm{i}\bigoplus \big( D_{\mathbbm{1}{\boldsymbol{k}}} + D_{05} + D_{12} - D_{34} \big) \, x\,, \label{Cartan'}
\end{align}
where $x\in\mathrm{Span}\{A\}$ (see Eqs. \eqref{8uni}-\eqref{8uni'*}). Again, for the convention $\bigoplus (\dots)$, see Sect. \ref{Sect. Convention}.

Moreover, one can show that the product of these diagonal matrices reproduces $\omega(4,2)$ \eqref{omega(4,2)=iD}:
\begin{align}\label{OOmega}
    \big(\Gamma^{}_{05}\, \Gamma^{}_{12}\, \Gamma^{}_{34}\big)\, x = \big(U\, E\, U^{-1}\big)\, x = \omega(4,2)\, x\,,
\end{align}
and, utilizing the standard properties \eqref{xy=yx} and \eqref{x+y=x+y} of the trace, that the four $8\times 8$ matrices \eqref{omega(4,2)=iDMat} and \eqref{Cartan'www.}-\eqref{Cartan'} are trace orthogonal:  
\begin{align}
    T\left(\Gamma^{}_{05} \Gamma^{}_{12}\right) = T\left(\Gamma^{}_{05} \Gamma^{}_{34}\right) = T\left(\Gamma^{}_{12} \Gamma^{}_{34}\right) = \dots = T\big(\Gamma^{}_{12} D\big) = 0\,.
\end{align}
We observe further that the diagonal volume forms \eqref{99}-\eqref{77} of the Clifford subalgebras $\mathfrak{cl}(2,2)$ and $\mathfrak{cl}(4,0)$ remain unchanged (up to the meaning of the indices) under the similarity transformation.

\subsection{Changing the Basis: From Majorana to Chiral}

Now, we turn to the second approach, namely, keeping the operators in their original form while transforming the basis, that is, working with:
\begin{align}\label{tytyty}
    &n^{}_{\widehat{\mathbbm{1}\boldsymbol{k}}_+} \big( = U^{-1}\, \mathbbm{1} \big) = \frac{1}{\sqrt{2}} \left( \mathbbm{1} + \mathrm{i}{\boldsymbol{k}} \right)\,, \nonumber\\[0.2cm]
    & n^{}_{\widehat{\mathbbm{1}\boldsymbol{k}}_-} \big( = U^{-1}\, {\boldsymbol{k}} \big) = \frac{1}{\sqrt{2}} \left( \mathbbm{1} - \mathrm{i}{\boldsymbol{k}} \right) = \overline{n^{}_{\widehat{\mathbbm{1}\boldsymbol{k}}_+}} \,,
\end{align}
\begin{align}\label{zx}
    &n^{}_{\widehat{05}_+} \big( = \mathrm{i}U^{-1}\, {\boldsymbol{\jmath}} \big) = \frac{1}{\sqrt{2}} \left( \boldsymbol{\imath} + \mathrm{i}{\boldsymbol{\jmath}} \right)\,, \nonumber\\[0.2cm]
    &n^{}_{\widehat{05}_-} \big( = \mathrm{i}U^{-1}\, \boldsymbol{\imath} \big) = \frac{1}{\sqrt{2}} \left( \boldsymbol{\imath} - \mathrm{i}{\boldsymbol{\jmath}} \right) = \overline{n^{}_{\widehat{05}_+}} \,,
\end{align}
\begin{align}\label{zc}
    &n^{}_{\widehat{12}_+} \big( = \mathrm{i}U^{-1}\, {\boldsymbol{\jmath}}\boldsymbol{\ell} \big) = \frac{1}{\sqrt{2}} \left( \boldsymbol{\ell}\boldsymbol{\imath} + \mathrm{i}{\boldsymbol{\jmath}}\boldsymbol{\ell} \right)\,, \nonumber\\[0.2cm]
    &n^{}_{\widehat{12}_-} \big( = \mathrm{i}U^{-1}\, \boldsymbol{\ell}\boldsymbol{\imath} \big) = \frac{1}{\sqrt{2}} \left( \boldsymbol{\ell}\boldsymbol{\imath} - \mathrm{i}{\boldsymbol{\jmath}}\boldsymbol{\ell} \right) = \overline{n^{}_{\widehat{12}_+}} \,, 
\end{align}
\begin{align}\label{zv}
    &n^{}_{\widehat{34}_+} \big( = \mathrm{i}U^{-1}\, \boldsymbol{\ell} \big) = \frac{1}{\sqrt{2}} \left( \boldsymbol{\ell}{\boldsymbol{k}} + \mathrm{i}\boldsymbol{\ell} \right)\,, \nonumber\\[0.2cm]
    &n^{}_{\widehat{34}_-} \big( = \mathrm{i} U^{-1}\, \boldsymbol{\ell}{\boldsymbol{k}} \big) = \frac{1}{\sqrt{2}} \left( \boldsymbol{\ell}{\boldsymbol{k}} - \mathrm{i}\boldsymbol{\ell} \right) = \overline{n^{}_{\widehat{34}_+}} \,.
\end{align}
To derive the above results, it suffices to substitute $U=e^{\mathrm{i}\frac{\pi}{4}} V$, taking into account Eqs. \eqref{1.31}, \eqref{S=DE=-ED}, \eqref{D}, and \eqref{U(1)generator}. For the complex rescaling of $n^{}_{\widehat{05}_\pm}, n^{}_{\widehat{12}_\pm}$, and $n^{}_{\widehat{34}_\pm}$ by `$\mathrm{i}$', see Remark \ref{Remark i}. 

The above transformed basis spans the chiral space $\mathfrak{S}$, which, by construction, decomposes into two complex-conjugate semispinor chiral subspaces $\mathfrak{S}_\pm$:
\begin{align}\label{Spm}
    \mathfrak{S}_\pm = \mathrm{Span} \Big\{ n^{}_{a^{}_\pm} \;;\;\; n^{}_{\widehat{\mathbbm{1}\boldsymbol{k}}_\pm}, \;\; n^{}_{\widehat{05}_\pm},\;\; n^{}_{\widehat{12}_\pm},\;\; n^{}_{\widehat{34}_\pm} \Big\}\,, \quad \mathfrak{S}_\pm = \overline{\mathfrak{S}_\mp}\,.
\end{align}
Notably, each chiral subspace is totally isotropic; that is, for all $n^{}_{a^{}_\pm}, n^{}_{b^{}_\pm} \in \mathfrak{S}_\pm$, the bilinear form defined in \eqref{muhum} satisfies:
\begin{align}
    \big\langle n^{}_{a^{}_+}, n^{}_{b^{}_+} \big\rangle = 0 = \big\langle n^{}_{a^{}_-}, n^{}_{b^{}_-} \big\rangle\,,
\end{align}
so that every vector has zero norm and all vectors within the same subspace are mutually orthogonal. It is important to emphasize that this property holds independently for $\mathfrak{S}_+$ and $\mathfrak{S}_-$, while in general the bilinear form between vectors from different subspaces does not need to vanish. Actually, explicit calculations show that:\footnote{Note that the star `$^\ast$' on $n^{}_{a}$ denotes octonionic conjugation, which reverses the sign of the imaginary units of $\mathbb{O}_{\mathbb{S}}$ but leaves the complex scalar $\mathrm{i}\in\mathbb{C}$ unchanged. Accordingly, the resulting norm differs from that with respect to which the similarity transformation $U$ is unitary.}
\begin{align}
    \big\langle n^{}_{\widehat{\mathbbm{1}\boldsymbol{k}}_+}, n^{}_{\widehat{\mathbbm{1}\boldsymbol{k}}_+} \big\rangle = N\big(n^{}_{\widehat{\mathbbm{1}\boldsymbol{k}}_+}\big) &= n^{}_{\widehat{\mathbbm{1}\boldsymbol{k}}_+}\, n^{\ast}_{\widehat{\mathbbm{1}\boldsymbol{k}}_+} \nonumber\\[0.2cm]
    &= \frac{1}{2} \left( 1 + \mathrm{i}{\boldsymbol{k}} \right)\left( 1 - \mathrm{i}{\boldsymbol{k}} \right) = 0 \,,
\end{align}
\begin{align}
    \big\langle n^{}_{\widehat{05}_+}, n^{}_{\widehat{05}_+} \big\rangle = N\big(n^{}_{\widehat{05}_+}\big) &= n^{}_{\widehat{05}_+}\, n^{\ast}_{\widehat{05}_+} \nonumber\\[0.2cm]
    &= -\frac{1}{2} \left( \boldsymbol{\imath} + \mathrm{i}{\boldsymbol{\jmath}} \right)\left( \boldsymbol{\imath} + \mathrm{i}{\boldsymbol{\jmath}} \right) = 0 \,,
\end{align}
\begin{align}
    \big\langle n^{}_{\widehat{12}_+}, n^{}_{\widehat{12}_+} \big\rangle = N\big(n^{}_{\widehat{12}_+}\big) &= n^{}_{\widehat{12}_+}\, n^{\ast}_{\widehat{12}_+} \nonumber\\[0.2cm]
    &= -\frac{1}{2} \left( \boldsymbol{\ell}\boldsymbol{\imath} + \mathrm{i}{\boldsymbol{\jmath}}\boldsymbol{\ell} \right)\left( \boldsymbol{\ell}\boldsymbol{\imath} + \mathrm{i}{\boldsymbol{\jmath}}\boldsymbol{\ell} \right) = 0 \,,
\end{align}
\begin{align}
    \big\langle n^{}_{\widehat{34}_+}, n^{}_{\widehat{34}_+} \big\rangle = N\big(n^{}_{\widehat{34}_+}\big) &= n^{}_{\widehat{34}_+}\, n^{\ast}_{\widehat{34}_+} \nonumber\\[0.2cm]
    &= -\frac{1}{2} \left( \boldsymbol{\ell}{\boldsymbol{k}} + \mathrm{i}\boldsymbol{\ell} \right)\left( \boldsymbol{\ell}{\boldsymbol{k}} + \mathrm{i}\boldsymbol{\ell} \right) = 0 \,,
\end{align}
moreover, employing Eq. \eqref{muhum}, we have:
\begin{align}
    \big\langle n^{}_{\widehat{\mathbbm{1}\boldsymbol{k}}_+} , n^{}_{\widehat{05}_+} \big\rangle &= \frac{1}{2} \left( N\big(n^{}_{\widehat{\mathbbm{1}\boldsymbol{k}}_+} + n^{}_{\widehat{05}_+}\big) -  N\big(n^{}_{\widehat{\mathbbm{1}\boldsymbol{k}}_+}\big) - N\big(n^{}_{\widehat{05}_+}\big) \right) \nonumber\\[0.2cm]
    &= \frac{1}{2} N\big(n^{}_{\widehat{\mathbbm{1}\boldsymbol{k}}_+} + n^{}_{\widehat{05}_+}\big) \nonumber\\[0.2cm]
    &= \frac{1}{2}\big(n^{}_{\widehat{\mathbbm{1}\boldsymbol{k}}_+} + n^{}_{\widehat{05}_+}\big)\big(n^{}_{\widehat{\mathbbm{1}\boldsymbol{k}}_+} + n^{}_{\widehat{05}_+}\big)^\ast \nonumber\\[0.2cm]
    &= \frac{1}{4} \big((\mathbbm{1}+\boldsymbol{\imath}) + \mathrm{i}({\boldsymbol{k}}+{\boldsymbol{\jmath}}) \big) \big((\mathbbm{1}-\boldsymbol{\imath}) - \mathrm{i}({\boldsymbol{k}}+{\boldsymbol{\jmath}}) \big) = 0\,,
\end{align}
\begin{align}
    \big\langle n^{}_{\widehat{12}_+} , n^{}_{\widehat{05}_+} \big\rangle &= \frac{1}{2} \left( N\big(n^{}_{\widehat{12}_+} + n^{}_{\widehat{05}_+}\big) -  N\big(n^{}_{\widehat{12}_+}\big) - N\big(n^{}_{\widehat{05}_+}\big) \right) \nonumber\\[0.2cm]
    &= \frac{1}{2} N\big(n^{}_{\widehat{12}_+} + n^{}_{\widehat{05}_+}\big) \nonumber\\[0.2cm]
    &= \frac{1}{2}\big(n^{}_{\widehat{12}_+} + n^{}_{\widehat{05}_+}\big)\big(n^{}_{\widehat{12}_+} + n^{}_{\widehat{05}_+}\big)^\ast \nonumber\\[0.2cm]
    &= -\frac{1}{4} \big((\boldsymbol{\ell}\boldsymbol{\imath}+\boldsymbol{\imath}) + \mathrm{i}({\boldsymbol{\jmath}}\boldsymbol{\ell}+{\boldsymbol{\jmath}}) \big) \big((\boldsymbol{\ell}\boldsymbol{\imath}+\boldsymbol{\imath}) + \mathrm{i}({\boldsymbol{\jmath}}\boldsymbol{\ell}+{\boldsymbol{\jmath}}) \big) = 0\,,
\end{align}
and so forth:
\begin{align}
    \big\langle n^{}_{\widehat{34}_+} , n^{}_{\widehat{05}_+} \big\rangle = \big\langle n^{}_{\widehat{34}_+} , n^{}_{\widehat{12}_+} \big\rangle = \big\langle n^{}_{\widehat{1{\boldsymbol{k}}}_+} , n^{}_{\widehat{12}_+} \big\rangle = \big\langle n^{}_{\widehat{1{\boldsymbol{k}}}_+} , n^{}_{\widehat{34}_+} \big\rangle = \dots = 0\,.
\end{align}
As naturally expected, the set of basis vectors $\{n^{}_{a^{}_-}\}$ satisfies orthogonality relations analogous to those of its complex-conjugate counterpart $\{n^{}_{a^{}_+}\}$:
\begin{align}
    \big\langle n^{}_{\widehat{\mathbbm{1}\boldsymbol{k}}_-} , n^{}_{\widehat{\mathbbm{1}\boldsymbol{k}}_-} \big\rangle = \big\langle n^{}_{\widehat{05}_-} , n^{}_{\widehat{05}_-} \big\rangle = \big\langle n^{}_{\widehat{12}_-} , n^{}_{\widehat{12}_-} \big\rangle = \big\langle n^{}_{\widehat{34}_-} , n^{}_{\widehat{34}_-} \big\rangle = 0\,,
\end{align}
\begin{align}
    \big\langle n^{}_{\widehat{34}_-} , n^{}_{\widehat{05}_-} \big\rangle = \big\langle n^{}_{\widehat{34}_-} , n^{}_{\widehat{12}_-} \big\rangle = \big\langle n^{}_{\widehat{1{\boldsymbol{k}}}_-} , n^{}_{\widehat{12}_-} \big\rangle = \big\langle n^{}_{\widehat{1{\boldsymbol{k}}}_-} , n^{}_{\widehat{34}_-} \big\rangle = \dots = 0\,.
\end{align}
In this construction, the only non-zero inner products are (see Eq. \eqref{muhum} again):
\begin{align}\label{aram}
    \big\langle n^{}_{\widehat{\mathbbm{1}\boldsymbol{k}}_-} , n^{}_{\widehat{\mathbbm{1}\boldsymbol{k}}_+} \big\rangle &= \frac{1}{2} \left( N\big(n^{}_{\widehat{\mathbbm{1}\boldsymbol{k}}_-} + n^{}_{\widehat{\mathbbm{1}\boldsymbol{k}}_+}\big) -  N\big(n^{}_{\widehat{\mathbbm{1}\boldsymbol{k}}_-}\big) - N\big(n^{}_{\widehat{\mathbbm{1}\boldsymbol{k}}_+}\big) \right) \nonumber\\[0.2cm]
    &= \frac{1}{2} N\big(n^{}_{\widehat{\mathbbm{1}\boldsymbol{k}}_-} + n^{}_{\widehat{\mathbbm{1}\boldsymbol{k}}_+}\big) \nonumber\\[0.2cm]
    &= \frac{1}{2}\big(n^{}_{\widehat{\mathbbm{1}\boldsymbol{k}}_-} + n^{}_{\widehat{\mathbbm{1}\boldsymbol{k}}_+}\big)\big(n^{}_{\widehat{\mathbbm{1}\boldsymbol{k}}_-} + n^{}_{\widehat{\mathbbm{1}\boldsymbol{k}}_+}\big)^\ast \nonumber\\[0.2cm]
    &= \frac{1}{4} \big((\mathbbm{1}+\mathbbm{1}) + \mathrm{i}(-{\boldsymbol{k}}+{\boldsymbol{k}}) \big) \big((\mathbbm{1}+\mathbbm{1}) + \mathrm{i}({\boldsymbol{k}}-{\boldsymbol{k}}) \big) = 1\,,
\end{align}
\begin{align}
    \big\langle n^{}_{\widehat{05}_-} , n^{}_{\widehat{05}_+} \big\rangle &= \frac{1}{2} \left( N\big(n^{}_{\widehat{05}_-} + n^{}_{\widehat{05}_+}\big) -  N\big(n^{}_{\widehat{05}_-}\big) - N\big(n^{}_{\widehat{05}_+}\big) \right) \nonumber\\[0.2cm]
    &= \frac{1}{2} N\big(n^{}_{\widehat{05}_-} + n^{}_{\widehat{05}_+}\big) \nonumber\\[0.2cm]
    &= \frac{1}{2}\big(n^{}_{\widehat{05}_-} + n^{}_{\widehat{05}_+}\big)\big(n^{}_{\widehat{05}_-} + n^{}_{\widehat{05}_+}\big)^\ast \nonumber\\[0.2cm]
    &= \frac{1}{4} \big((\boldsymbol{\imath}+\boldsymbol{\imath}) + \mathrm{i}(0) \big) \big(-(\boldsymbol{\imath}+\boldsymbol{\imath}) + \mathrm{i}(0) \big) = 1 \,,
\end{align}
\begin{align}
    \big\langle n^{}_{\widehat{12}_-} , n^{}_{\widehat{12}_+} \big\rangle &= \frac{1}{2} \left( N\big(n^{}_{\widehat{12}_-} + n^{}_{\widehat{12}_+}\big) -  N\big(n^{}_{\widehat{12}_-}\big) - N\big(n^{}_{\widehat{12}_+}\big) \right) \nonumber\\[0.2cm]
    &= \frac{1}{2} N\big(n^{}_{\widehat{12}_-} + n^{}_{\widehat{12}_+}\big) \nonumber\\[0.2cm]
    &= \frac{1}{2}\big(n^{}_{\widehat{12}_-} + n^{}_{\widehat{12}_+}\big)\big(n^{}_{\widehat{12}_-} + n^{}_{\widehat{12}_+}\big)^\ast \nonumber\\[0.2cm]
    &= \frac{1}{4} \big(2\boldsymbol{\ell}\boldsymbol{\imath} + \mathrm{i}(0)\big) \big(-2\boldsymbol{\ell}\boldsymbol{\imath} + \mathrm{i}(0)\big) = -1 \,,
\end{align}
\begin{align}\label{aram'}
    \big\langle n^{}_{\widehat{34}_-} , n^{}_{\widehat{34}_+} \big\rangle &= \frac{1}{2} \left( N\big(n^{}_{\widehat{34}_-} + n^{}_{\widehat{34}_+}\big) -  N\big(n^{}_{\widehat{34}_-}\big) - N\big(n^{}_{\widehat{34}_+}\big) \right) \nonumber\\[0.2cm]
    &= \frac{1}{2} N\big(n^{}_{\widehat{34}_-} + n^{}_{\widehat{34}_+}\big) \nonumber\\[0.2cm]
    &= \frac{1}{2}\big(n^{}_{\widehat{34}_-} + n^{}_{\widehat{34}_+}\big)\big(n^{}_{\widehat{34}_-} + n^{}_{\widehat{34}_+}\big)^\ast \nonumber\\[0.2cm]
    &= \frac{1}{4} \big(2\boldsymbol{\ell}{\boldsymbol{k}} + \mathrm{i}(0)\big) \big(-2\boldsymbol{\ell}{\boldsymbol{k}} + \mathrm{i}(0) \big) = -1 \,.
\end{align}

\begin{remark}\label{Remark i}
    {\emph{(complex rescaling of the null vectors as a phase rotation).} In light of the foregoing discussion, it is worth emphasizing that the complex rescaling of the null vectors $n^{}_{\widehat{05}_\pm}, n^{}_{\widehat{12}_\pm}$, and $n^{}_{\widehat{34}_\pm}$ by the factor `$\mathrm{i}$' is prescribed in such a way as to ensure the validity of the identities \eqref{aram}-\eqref{aram'}. This rescaling corresponds to a rotation in the complex plane and leaves the underlying geometric structure invariant.}
\end{remark}

Taking the above into account, we now examine the matrix realization of the $\mathfrak{u}(1)$ generator $E$ \eqref{U(1)generator} of the conformal Lie algebra $\mathfrak{u}(2,2)\cong\mathfrak{su}(2,2)\oplus\mathfrak{u}(1)$, expressed with respect to the chiral basis. This yields a form analogous to $\omega(4,2)$ \eqref{omega(4,2)=iDMat}, but with a different interpretation of the indices:
\begin{align}\label{omega(4,2)=iDMat'}
    E\, y &= \mathrm{i} \bigoplus\big(D_{\widehat{\mathbbm{1}\boldsymbol{k}}_+ \widehat{\mathbbm{1}\boldsymbol{k}}_-} + D_{\widehat{05}_+ \widehat{05}_-} + D_{\widehat{12}_+ \widehat{{12}}_-} + D_{\widehat{34}_+ \widehat{{34}}_-}\big)\, y \nonumber\\[0.2cm]
    &\,\underbrace{=: \mathrm{i}D\, y = \omega(4,2)\,y}_{\text{by abuse of notation!}}\,,
\end{align}
where $y\in\mathrm{Span}\{ n^{}_{a^{}_\pm} \}$. For the convention $\bigoplus (\dots)$, see Sect. \ref{Sect. Convention}. This result can be easily checked by substituting $y$ with any of the zero-norm basis vectors $n^{}_{a^{}_\pm}$ (in fact, by substituting each one individually). Just to clarify the point, let us set $y=n^{}_{\widehat{05}_+}$. Then, employing Eqs. \eqref{zx}, \eqref{omega(4,2)=iD}, and \eqref{omega(4,2)=iDMat}, we obtain:
\begin{align}\label{Ey}
    E\, n^{}_{\widehat{05}_+} = E\, \big( \mathrm{i}U^{-1} {\boldsymbol{\jmath}} \big) = \mathrm{i}U^{-1} \big(U\, E\, U^{-1}\big) \,{\boldsymbol{\jmath}} = \mathrm{i}U^{-1} \big( \mathrm{i} {\boldsymbol{\jmath}} \big) = \mathrm{i} \big( n^{}_{\widehat{05}_+} \big)\,,
\end{align}
which is consistent with Eq. \eqref{omega(4,2)=iDMat'}. In this context, one observes that the chiral (isotropic) subspaces $\mathfrak{S}_\pm$ \eqref{Spm} are precisely the eigenspaces of $\omega(4,2)$ \eqref{omega(4,2)=iDMat'} corresponding to the eigenvalues $\pm \mathrm{i}$, respectively. That is:
\begin{align}\label{chirality matrix}
    \omega(4,2) \,\mathfrak{S}_\pm = \pm\mathrm{i}\,\mathfrak{S}_\pm \,.
\end{align}

Similarly, one verifies that the skew-symmetric Cartan matrices $m^{}_{05}, m^{}_{12}$, and $m^{}_{34}$ (see Eqs. \eqref{10}-\eqref{30}) in their original forms assume the following diagonal representations in the chiral basis, analogous to those in Eqs. \eqref{Cartan'www.}-\eqref{Cartan'}, but with a different interpretation of the indices:
\begin{align}\label{za} 
    m^{}_{05}\, y &= \mathrm{i} \bigoplus\big( - D_{\widehat{\mathbbm{1}\boldsymbol{k}}_+\widehat{\mathbbm{1}\boldsymbol{k}}_-} - D_{\widehat{05}_+\widehat{05}_-} + D_{\widehat{12}_+\widehat{{12}}_-} + D_{\widehat{34}_+\widehat{{34}}_-} \big)\,y \nonumber\\[0.2cm]
    &\underbrace{=: \Gamma^{}_{05}\, y}_{\text{by abuse of notation!}}\,,
\end{align}
\begin{align}\label{za'}
    m^{}_{12}\, y &= \mathrm{i} \bigoplus\big( D_{\widehat{\mathbbm{1}\boldsymbol{k}}_+\widehat{\mathbbm{1}\boldsymbol{k}}_-} - D_{\widehat{05}_+\widehat{05}_-} + D_{\widehat{12}_+\widehat{{12}}_-} - D_{\widehat{34}_+\widehat{{34}}_-} \big) \, y \nonumber\\[0.2cm]
    &\underbrace{=: \Gamma^{}_{12}\, y}_{\text{by abuse of notation!}}\,,
\end{align}
\begin{align}\label{za''}
    m^{}_{34}\, y &= \mathrm{i} \bigoplus\big( D_{\widehat{\mathbbm{1}\boldsymbol{k}}_+\widehat{\mathbbm{1}\boldsymbol{k}}_-} - D_{\widehat{05}_+\widehat{05}_-} - D_{\widehat{12}_+\widehat{{12}}_-} + D_{\widehat{34}_+\widehat{{34}}_-} \big) \,y \nonumber\\[0.2cm]
    &\underbrace{=: \Gamma^{}_{34}\, y}_{\text{by abuse of notation!}}\,, 
\end{align}
where $y\in\mathrm{Span}\{ n^{}_{a^{}_\pm} \}$ and see Sect. \ref{Sect. Convention} for the convention on $\bigoplus (\dots)$. These relations can be readily verified by following the same procedure outlined above (see Eq. \eqref{Ey}).

With the above in mind, and analogously to the $\mathfrak{cl}(2)$ case (see Sect. \ref{Sect. cl(2)}), the following relations hold for the chiral spinor $\psi$:
\begin{align}\label{1.33}
    &\psi = 
    \begin{pmatrix}
        \psi^+ \\ 
        \psi^-
    \end{pmatrix} = U\,\xi \,, \quad D_\pm \psi = \left( \frac{\mathbbm{1}\pm D}{2} \right) \psi = \psi^\pm \qquad & \nonumber\\[0.2cm]
    &\Longrightarrow\quad \omega(4,2) \,\psi^\pm = \pm \mathrm{i}\, \psi^\pm \;\;;\quad {\psi}^c = C\, \overline{\psi}\,, \quad\mbox{where}\quad C\,\overline{\Gamma_\nu} = \Gamma^{}_\nu \, C\,,
\end{align}
\begin{align}
    \left(\psi^c\right)^c = \psi \quad\Longrightarrow\quad C\,\overline{C} = \mathbbm{1} \;;\;\;\; C = U U^\top\,, \quad (C_{\text{Majorana}} = \mathbbm{1})\,. 
\end{align}

Finally, repeating the analysis for $\mathfrak{cl}(2)$ leading to \eqref{1.23}, we find the phase factor $\varrho$ relating $U$ and $V$ in \eqref{1.31} can again be fixed in a way to make the charge conjugation matrix $C$ real. Specifically, one obtains:
\begin{align}\label{1.36}
    C_V = VV^\top = e^{-\frac{\pi}{4}E} e^{-\mathrm{i}\frac{\pi}{4}D} e^{-\mathrm{i}\frac{\pi}{4}D} e^{\frac{\pi}{4}E} = - \mathrm{i} DE = - \mathrm{i} S\,.
\end{align}
Note that $D^\top = D$ and $E^\top = - E$. For $\varrho = e^{\mathrm{i}\pi/4}$ (that is, $U = e^{\mathrm{i}\pi/4} V$), this relation yields:
\begin{align}
    C_U = \overline{C_U} = UU^\top = S \,.
\end{align}

\subsection{Adjustment of Notation}\label{Sect. adjusment}

From this point forward, we adopt the compact chiral basis of isotropic vectors, defined as:
\begin{align}
    \label{1.42} \xi^\mathbbm{1} \mathbbm{1} + \xi^{\boldsymbol{k}} {\boldsymbol{k}} =&\, \psi^{\widehat{\mathbbm{1}\boldsymbol{k}}_+}\, n^{}_{\widehat{\mathbbm{1}\boldsymbol{k}}_+} + \psi^{\widehat{\mathbbm{1}\boldsymbol{k}}_-}\, n^{}_{\widehat{\mathbbm{1}\boldsymbol{k}}_-}\,, \nonumber\\[0.2cm]
    \mbox{where}&\quad n^{}_{\widehat{\mathbbm{1}\boldsymbol{k}}_\pm} = \frac{1}{\sqrt{2}} (\mathbbm{1} \pm \mathrm{i}{\boldsymbol{k}})\,,\quad \psi^{\widehat{\mathbbm{1}\boldsymbol{k}}_\pm} = \frac{1}{\sqrt{2}} \big( \xi^\mathbbm{1} \mp \mathrm{i}\xi^{\boldsymbol{k}} \big)\,,
\end{align}
\begin{align}
    \label{1.42'} \xi^{1} \boldsymbol{\ell}\boldsymbol{\imath} + \xi^{2} {\boldsymbol{\jmath}}\boldsymbol{\ell} =&\, \psi^{\widehat{12}_+}\, n^{}_{\widehat{12}_+} + \psi^{\widehat{12}_-}\, n^{}_{\widehat{12}_-} \,, \nonumber\\[0.2cm]
    \mbox{where}&\quad n^{}_{\widehat{12}_\pm} = \frac{1}{\sqrt{2}} \left( \boldsymbol{\ell}\boldsymbol{\imath} \pm \mathrm{i} {\boldsymbol{\jmath}}\boldsymbol{\ell} \right)\,, \quad \psi^{\widehat{12}_\pm} = \frac{1}{\sqrt{2}} \big( \xi^{1} \mp \mathrm{i}\xi^{2} \big)\,,
\end{align}
and similar expression for $n^{}_{\widehat{05}_\pm}$, $n^{}_{\widehat{34}_\pm}$, $\psi^{\widehat{05}_\pm}$, and $\psi^{\widehat{34}_\pm}$. Again, the vectors $n^{}_{a^{}_+}$, with $a^{}_+ = {\widehat{\mathbbm{1}\boldsymbol{k}}_+, \widehat{05}_+, \widehat{12}_+, \widehat{34}}_+$, form an isotropic basis consisting of eigenvectors of $\omega(4,2)$ with eigenvalue $+\mathrm{i}$. Their complex conjugates $n^{}_{a^{}_-} \, \big( =\overline{n^{}_{a^{}_+}} \big)$, with $a^{}_- = {\widehat{\mathbbm{1}\boldsymbol{k}}_-, \widehat{05}_-, \widehat{12}_-, \widehat{34}}_-$, correspondingly form an isotropic basis with eigenvalue $-\mathrm{i}$. That is:
\begin{align}
    \left( \omega(4,2) \mp \mathrm{i}\mathbbm{1}\right) n^{}_{a_\pm} = 0\,,
\end{align}
where:
\begin{align}
    \big\langle n^{}_{a_\pm} , n^{}_{a_\pm} \big\rangle = N\big(n^{}_{a_\pm}\big) = n^{}_{a_\pm} n^{\ast}_{a_\pm} = 0\,, 
\end{align}
\begin{align}
    \big\langle n^{}_{a_-} , n^{}_{b_+} \big\rangle = \epsilon^{}_a \, \delta_{ab}\,,
\end{align}
with $\epsilon^{}_{\widehat{1{\boldsymbol{k}}}} = \epsilon^{}_{\widehat{05}} = 1 = - \epsilon^{}_{\widehat{12}} = - \epsilon^{}_{\widehat{34}}$.

\sectionmark{Appendix: A Mnemonic Matrix Rule for (Split-)Octonionic ...}
\section{Appendix: A Mnemonic Matrix Rule for (Split-)Octonionic Multiplication}\label{Appendix Mnemonic}
\sectionmark{Appendix: A Mnemonic Matrix Rule for (Split-)Octonionic ...}

This appendix develops a mnemonic matrix formulation following the methodology recently outlined in Ref. \cite{Gazeau-Mnemonic}. We consider a (split-)octonion expressed in Cayley-Dickson form:
\begin{align}
    o = q + \boldsymbol{\ell} p\,, \quad q,p \in \mathbb{H}\,, \quad (\boldsymbol{\ell})^2 = \pm 1 \,.
\end{align}
The multiplication of two (split-)octonions is given by Eq. \eqref{Multi-oct}; for convenience, we reproduce it here:
\begin{align}\label{Multi-oct000}
        \left(q^{}_1 + \boldsymbol{\ell} p^{}_1\right) \left(q^{}_2 + \boldsymbol{\ell} p^{}_2\right) = q^{}_1 q^{}_2 + (\boldsymbol{\ell})^2 p^{}_2 {p^{\ast}_1} + \boldsymbol{\ell} \left( {q^{\ast}_1} p^{}_2 + q^{}_2 p^{}_1 \right)\,.
\end{align}
To each octonion, we associate the matrix:
\begin{align} \label{eq:M}
    M(o) =
    \begin{pmatrix}
        q & & & (\boldsymbol{\ell})^2 p^{\ast} \\[1ex]
        p & & & q^{\ast}
    \end{pmatrix} \,.
\end{align}

A direct computation yields:
\begin{align} \label{eq:mnemo-block}
    M(o^{}_1) M(o^{}_2) =
    \begin{pmatrix}
        q^{}_1 q^{}_2 + (\boldsymbol{\ell})^2\, \overset{\curvearrowleft}{p_1^{\ast}p^{}_2} & & & (\boldsymbol{\ell})^2\,\overset{\curvearrowleft}{q^{}_1 p_2^{\ast}}  + (\boldsymbol{\ell})^2 p_1^{\ast} q_2^{\ast} \\[2ex]
        \overset{\curvearrowleft}{p^{}_1 q^{}_2} + q_1^{\ast} p^{}_2 & & & (\boldsymbol{\ell})^2 p^{}_1 p_2^{\ast} + \overset{\curvearrowleft}{q_1^{\ast} q_2^{\ast}}
    \end{pmatrix}\,,
\end{align}
where the symbol `$\curvearrowleft$' indicates reversal of the order of multiplication. Remarkably, this matrix product reproduces the octonionic multiplication once these reversed products are replaced by their adjoint counterparts, thereby recovering Eq. \eqref{Multi-oct000} (or Eq. \eqref{Multi-oct}). 

Inspection of Eq. \eqref{eq:mnemo-block} reveals a clear mnemonic pattern: in each block of the matrix product, the terms naturally split into two contributions — one involving the quaternionic product in standard (right) order, denoted $R$, and one in reversed (left) order, denoted $L$. This separation provides a convenient guide for reconstructing the octonionic multiplication from the matrix representation:
\begin{align}
    \text{upper row}\;; \qquad & R + L \quad \text{and} \quad L + R \,, \label{eq:RL-up} \\[0.2cm]
    \text{lower row}\;; \qquad & L + R \quad \text{and} \quad R + L \,. \label{eq:RL-down}
\end{align}
Equivalently:
\begin{align}
    \text{left column}\;; \qquad & R + L \quad \text{and} \quad L + R \,, \\[0.2cm]
    \text{right column}\;; \qquad & L + R \quad \text{and} \quad R + L \,.
\end{align}
In this way, the octonion multiplication table manifests as a highly structured block arrangement of quaternionic products, with only half of the terms appearing in reversed order. The matrix representation $M(o)$ thus provides a compact and practical mnemonic for the non-associative octonionic product, offering a clear framework for navigating the inherent non-associativity.

\sectionmark{Appendix: Realization of the Conformal Group $\mathrm{U}(2,2)$ ...}
\section{Appendix: Realization of the Conformal Group $\mathrm{U}(2,2)\cong\mathrm{SU}(2,2)\; \big( \cong\mathrm{SO}(4,2)\big)\times\mathrm{U}(1)$ in the Clifford Algebra $\mathfrak{cl}(4,2)$ and its Action on Spinors}\label{Appendix U(2,2) in the Clifforrd}
\sectionmark{Appendix: Realization of the Conformal Group $\mathrm{U}(2,2)$ ...}

Using the Clifford-algebraic framework developed in this chapter, the \linebreak unitary-conformal group $\mathrm{U}(2,2)$ — equivalently $\mathrm{SU}(2,2)\; \big(\cong\mathrm{SO}(4,2) \big)\times\mathrm{U}(1)$ — admits a concrete realization in the real Clifford algebra $\mathfrak{cl}(4,2)$, specifically within its even subalgebra $\mathfrak{cl}^{\text{even}}(4,2)$. In this realization, the bivectors of the even subalgebra generate the conformal Lie algebra $\mathfrak{su}(2,2)\cong\mathfrak{so}(4,2)$, while the Clifford pseudoscalar (volume element) provides the $1$-dimensional central $\mathfrak{u}(1)$ component, so that $\mathfrak{u}(2,2)\cong\mathfrak{su}(2,2)\, \big(\cong\mathfrak{so}(4,2)\big)\oplus\mathfrak{u}(1)$. This unified Clifford-algebraic setting naturally produces both the conformal Lorentz group $\mathrm{SO}(4,2)$ and its spin representation $\mathrm{SU}(2,2)$ directly from the Clifford algebra's multiplicative structure, with the central $\mathrm{U}(1)$ factor realized as internal phase rotations generated by the pseudoscalar. In this way, the full $\mathrm{U}(2,2)$ group emerges naturally from the Clifford-algebraic framework. We make this realization explicit below by giving a step-by-step derivation that yields an explicit construction.

The even Clifford subalgebra $\mathfrak{cl}^{\text{even}}(4,2)$ is given by:
\begin{align}
    \mathfrak{cl}^{\text{even}}(4,2) = \mathrm{Span} \Big\{ \mathbbm{1}\,, \; m^{}_{\mu\nu} := m^{}_\mu\, m{}_\nu\,, \; m^{}_\mu\, m{}_\nu\, m^{}_\varrho\, m{}_\sigma\,,&\nonumber\\[0.2cm]
    \; E := m^{}_0\, m^{}_5\, m^{}_1\, m^{}_2\, m^{}_3\, m^{}_4 \,;&\nonumber\\[0.2cm]
    \; 0\leq\mu<\nu<\varrho<\sigma\leq 5 &\Big\} \,,
\end{align} 
such that:
\begin{enumerate}[leftmargin=0.8cm]
    \item{The $6$ operators $m^{}_\mu$ represent left multiplication by the $6$ imaginary split-octonion units $\big\{0:={\boldsymbol{\imath}},\, 5:={\boldsymbol{\jmath}},\, 1:=\boldsymbol{\ell}{\boldsymbol{\imath}},\, 2:={\boldsymbol{\jmath}}\boldsymbol{\ell},\,  3:=\boldsymbol{\ell}\boldsymbol{k},\, 4:=\boldsymbol{\ell}\big\}$, as defined in Eqs. \eqref{8uni}-\eqref{8uni'*}, and they satisfy the Clifford anti-commutation relations \eqref{clif}.}

    \item{The Lie algebra generated by the $15$ bivectors $m^{}_{\mu\nu} := \frac{1}{2}\big[ m^{}_\mu , m^{}_\nu \big] = m^{}_\mu m^{}_\nu$ coincides with the conformal Lie algebra $\mathfrak{su}(2,2)\cong\mathfrak{so}(4,2)$ (see Eq. \eqref{8888888888888}). We therefore identify the Lie algebra $\mathfrak{su}(2,2)\cong\mathfrak{so}(4,2)$ as embedded in the even Clifford subalgebra $\mathfrak{cl}^{\text{even}}(4,2)$ via:
    \begin{align}
        \mathfrak{su}(2,2) \quad\hookrightarrow\quad \mathrm{Span}\big\{m^{}_{\mu\nu}\big\} \;\subset\; \mathfrak{cl}^{\text{even}}(4,2)\,.
    \end{align}
    Each $m^{}_{\mu\nu}$ acts on the real $8$-dimensional spinor space $\mathfrak{S} = \mathrm{Span}\{A\}$ through left multiplication (see Eqs. \eqref{8uni}-\eqref{8uni'*} and Remark \ref{Remark nested}).}

    \item{The Clifford pseudoscalar (volume form) $E := m^{}_0\, m^{}_5\, m^{}_1\, m^{}_2\, m^{}_3\, m^{}_4$ commutes with every bivector $m^{}_{\mu\nu}$ (see Eq. \eqref{tootii}) and therefore generates the $1$-dimensional center of the algebra, corresponding to the $\mathfrak{u}(1)$ factor in $\mathfrak{u}(2,2)$. Then, combining this central element with the conformal sector generated by $m^{}_{\mu\nu}$, one obtains the full Lie algebra:
    \begin{align}
        \mathfrak{u}(2,2) \cong \mathfrak{su}(2,2)\oplus\mathfrak{u}(1) \;\subset\; \mathfrak{cl}^{\text{even}}(4,2) \,.
    \end{align}}    
\end{enumerate}

Any element of the group $\mathrm{U}(2,2)$ can be expressed naturally by exponentiating a linear combination of the bivectors $m^{}_{\mu\nu}$ and the pseudoscalar $E$:\footnote{Note that the factor `$\tfrac{1}{2}$' ensures the correct spinorial form of conformal transformations. In Clifford algebra, bivectors generate rotations and boosts with twice the physical angle when acting on vectors; including `$\tfrac{1}{2}$' adjusts this so that $e^{\tfrac{1}{2}X}$ gives the proper single-angle rotation on spinors, as required by the double-cover relation $\mathrm{Spin}(4,2)\cong\mathrm{SU}(2,2)$.}
\begin{align}
    \mathrm{U}(2,2) \;\ni\; \tilde{R}(X,\omega_0 \, E) = e^{\tfrac{1}{2}X+\omega_0 E}\,,
\end{align}
with $X\; \big(\in \mathfrak{su}(2,2)\cong\mathfrak{so}(4,2) \big) = \sum_{\mu<\nu} \omega_{\mu\nu}\, m^{}_{\mu\nu}$, while $\omega_{\mu\nu},\omega_0 \in \mathbb{R}$. Since $E$ commutes with all $m^{}_{\mu\nu}$, this expression factorizes:\footnote{For matrices or Lie algebra elements $M$ and $N$, $e^{M+N} = e^M e^N$ if $\big[M,N\big] = 0$; otherwise $e^M e^N = e^L$ with $L$ given by the Baker-Campbell-Hausdorff series.}
\begin{align}
    \tilde{R}(X,\omega_0) = e^{\tfrac{1}{2}X} e^{\omega_0 E} = e^{\omega_0 E} e^{\tfrac{1}{2}X} \,,
\end{align}
thereby realizing $\mathrm{U}(2,2)$ as the image of $\mathrm{SU}(2,2)\times\mathrm{U}(1)$ under the multiplication homomorphism, whose kernel is isomorphic to $\mathbb{Z}_4$.

In this regard, it is important to emphasize that:
\begin{enumerate}[leftmargin=0.8cm]
    \item{Expanding the exponential as a power series yields:
    \begin{align}
        e^{\tfrac{1}{2}X+\omega_0 E} = e^{\tfrac{1}{2}X} e^{\omega_0 E} = \sum_{n=0}^{\infty} \frac{\big(\tfrac{1}{2}X\big)^n}{n!} \sum_{m=0}^{\infty} \frac{\big(\omega_0 E\big)^m}{m!}\,.
    \end{align}  
    Each term in this series is a product of an even number of Clifford generators and therefore belongs to the even subalgebra $\mathfrak{cl}^{\text{even}}(4,2)$. Specifically, since $\mathfrak{cl}^{\text{even}}(4,2)$ is finite-dimensional, the series converges within this subalgebra, ensuring that the exponential (the group element) remains entirely inside $\mathfrak{cl}^{\text{even}}(4,2)$: 
    \begin{enumerate}[leftmargin=0.8cm]
        \item{From Proposition \ref{proposition 2.1} and of course Remark \ref{Remark nested}, we recall that all powers $(E)^n = (L_{-\boldsymbol{k}})^n$ are well-defined and unambiguous. Exploiting this fact together with the identity $E^2=-\mathbbm{1}$ (see Eq. \eqref{E2=-1}), the exponential of the central element can be computed explicitly:
        \begin{align}
            e^{\omega_0 E} &= \sum_{m=0}^{\infty} \frac{(\omega_0 E)^m}{m!} \nonumber\\[0.2cm]
            &= \sum_{m=0}^{\infty} \frac{(-1)^m (\omega_0)^{2m}}{(2m)!} \, \mathbbm{1} + \sum_{m=0}^{\infty} \frac{(-1)^m (\omega_0)^{2m+1}}{(2m+1)!} \, E \nonumber\\[0.2cm]
            &= \cos\omega_0 \,\mathbbm{1} + \sin\omega_0 \,E \,,
        \end{align}
        so that $\omega_0$ plays the role of the internal $\mathrm{U}(1)$ phase angle.}

        \item{Similarly, invoking Proposition \ref{proposition 2.1} and Remark \ref{Remark nested}, all powers of the form $\big(m^{}_\mu m^{}_\nu\big)^n$ are well-defined and unambiguous. Exploiting this property, together with the fact that each bivector $m^{}_{\mu\nu}=m^{}_\mu m^{}_\nu$ squares to:\footnote{This identity can be verified by a procedure analogous to that used in Eq. \eqref{mij222}.}
        \begin{align}
            \big(m^{}_{\mu\nu}\big)^2 = -\eta_{\mu\mu}\, \eta_{\nu\nu}\,  \mathbbm{1} = \pm \mathbbm{1}\,,
        \end{align}
        for a single bivector component, with no summation over $\mu,\nu$, the exponential can be evaluated explicitly as:
        \begin{align}
            &e^{\tfrac{1}{2}\omega_{\mu\nu} m^{}_{\mu\nu}} \nonumber\\[0.2cm]
            &= 
            \begin{cases}
                \cos\tfrac{\omega_{\mu\nu}}{2} \, \mathbbm{1} + \sin\tfrac{\omega_{\mu\nu}}{2} \, m^{}_{\mu\nu}\, \quad& \big(m^{}_{\mu\nu}\big)^2 = -\mathbbm{1} \vspace{0.3cm}\\
                \cosh\tfrac{\omega_{\mu\nu}}{2} \, \mathbbm{1} + \sinh\tfrac{\omega_{\mu\nu}}{2} \, m^{}_{\mu\nu}\, \quad& \big(m^{}_{\mu\nu}\big)^2 = +\mathbbm{1}
            \end{cases}\,,
        \end{align}
        with each term lying in the even subalgebra $\mathfrak{cl}^{\text{even}}(4,2)$.}
    \end{enumerate}}

    \item{The inverse of any group element follows directly from the exponential structure:
    \begin{align}
        \tilde{R}^{-1}(X,\omega_0) = e^{-\tfrac{1}{2}X - \omega_0 E} = e^{-\tfrac{1}{2}X} e^{-\omega_0 E} = e^{-\omega_0 E} e^{-\tfrac{1}{2}X} \,.
    \end{align}}

    \item{Because the even subalgebra is closed under multiplication, the product of any two such elements again yields an element of the same form:
    \begin{align}
        \tilde{R}(X,\omega_0)\, \tilde{R}(X^\prime,\omega^\prime_0) = \tilde{R}(X^{\prime\prime},\omega^{\prime\prime}_0) \;\in\; \mathfrak{cl}^{\mathrm{even}}(4,2) \,.
    \end{align}
    Accordingly, the exponential construction realizes $\mathrm{U}(2,2)$ entirely within $\mathfrak{cl}^{\text{even}}(4,2)$.}

    \item{Finally, the associativity required for the group law is ensured by the associativity of operator composition, which holds by construction as explained in Remark \ref{Remark nested}.}
\end{enumerate}

\begin{remark}{
     \emph{($2\times 2$ matrix realization of $\mathrm{SU}(2,2)$ in complex quaternionic form).} For a more tangible visualization of the $\mathrm{SU}(2,2)$ group structure, we conclude this appendix by exhibiting a $2\times 2$-matrix realization of $\mathrm{SU}(2,2)$ in terms of complex quaternionic components:
    \begin{align}
        \mathrm{SU}(2,2) = \bigg\{ &g= 
        \begin{pmatrix}
            a & b \\
            c & d
        \end{pmatrix} \;;\;\; \det(g) = 1 \,,\;\, J\,g^\dagger J = g^{-1} \,, \nonumber\\[0.2cm]
        &\mbox{with}\;\, J = \begin{pmatrix}
            \mathbbm{1}_2 & 0 \\
            0 & -\mathbbm{1}_2
        \end{pmatrix} \;\,\mbox{and}\;\, a,b,c,d \in \mathbb{H}\otimes\mathbb{C}\cong\mathrm{Mat}(2,\mathbb{C}) \bigg\} \,.
    \end{align}
    For a detailed account of complex quaternions and their matrix realization, we refer to Ref. \cite{Gazeau2023}.    
}\end{remark}

\section{Appendix: Diagonalization of the Real Skew-Symmetric Matrix $E$} \label{footnote E2=-1}
\sectionmark{Appendix: Diagonalization of the Real Skew-Symmetric Matrix $E$}

A real skew-symmetric matrix such as $E$, satisfying $E^\top=-E$, is not diagonalizable over $\mathbb{R}$. However, if the matrix has an even dimension, it can be brought to a block-diagonal form via an orthogonal transformation. This property is clearly illustrated by the matrix representation \eqref{U(1)generator} of $E$:
\begin{align}
    \text{Eq. \eqref{U(1)generator}} =
    \begin{pmatrix}
    0 & 1 &  &  & \\
    -1 & 0 &  &  & \\
       &  & \ddots &  & \\
       &  &  &  0   & 1 \\
       &  &  & -1 & 0
    \end{pmatrix}_{8\times 8}\,,
\end{align}
where each of the four $2\times 2$ blocks corresponds to a pair of imaginary eigenvalues $\pm \mathrm{i}$. In contrast, over $\mathbb{C}$, $E$ is fully diagonalizable via a similarity transformation $T$:
\begin{align}
    T^{-1} E T = \operatorname{diag}(+\mathrm{i}, -\mathrm{i}, +\mathrm{i}, -\mathrm{i}, +\mathrm{i}, -\mathrm{i}, +\mathrm{i}, -\mathrm{i})\,.
\end{align}   
This confirms that full diagonalization is only possible on a complex basis. The matrix remains $8\times 8$ in dimension, but it is now fully diagonal over $\mathbb{C}$. For further details, see Ref. \cite{K}. 

\section{Appendix: Maximal Compact Subalgebra of $\mathfrak{su}(2,2)\cong\mathfrak{so}(4,2)$}\label{Appendix Maximal}
\sectionmark{Appendix: Maximal Compact Subalgebra of $\mathfrak{su}(2,2)$ ...}

The maximal compact subalgebra of $\mathfrak{su}(2,2)$ is given by:
\begin{align}\label{Appen 172}
    \mathfrak{s}\big(\mathfrak{u}(2) \oplus \mathfrak{u}(2)\big) = \Big\{ &(X, Y) \in \mathfrak{u}(2) \oplus \mathfrak{u}(2) \;;\;\; \nonumber\\[0.2cm]
    &\qquad\operatorname{Tr}(X + Y) = \operatorname{Tr}(X) + \operatorname{Tr}(Y) = 0 \Big\}\,.
\end{align}
Each $\mathfrak{u}(2)$ factor admits the decomposition:
\begin{align}
    \mathfrak{u}(2) \cong \mathfrak{su}(2) \oplus \mathfrak{u}(1)\,.
\end{align}
Let $\mathfrak{su}(2)_L$ and $\mathfrak{su}(2)_R$ denote the special unitary parts of the left and right $\mathfrak{u}(2)$ factors, respectively, in \eqref{Appen 172}. The two corresponding $\mathfrak{u}(1)$ generators, denoted $Z_L$ and $Z_R$, represent the scalar (central) components of the two $\mathfrak{u}(2)$ summands.

The tracelessness condition $\operatorname{Tr}(X) + \operatorname{Tr}(Y) = 0$ implies that $Z_L$ and $Z_R$ are not independent, but satisfy the linear constraint:
\begin{align}
    Z_L + Z_R = 0\,.
\end{align}
As a consequence, the center of the maximal compact subalgebra is $1$-dimensional and generated by the single $\mathfrak{u}(1)$ element
\begin{align}
    Z := Z_L - Z_R \,.
\end{align}

Accordingly, the maximal compact subalgebra of $\mathfrak{su}(2,2)$ reads as:
\begin{align}\label{S..Z2}
    \mathfrak{s}\big(\mathfrak{u}(2) \oplus \mathfrak{u}(2)\big) \cong \mathfrak{su}(2) \oplus \mathfrak{su}(2) \oplus \mathfrak{u}(1)\,,
\end{align}
or, equivalently, in terms of real Lie algebras:
\begin{align}
    \mathfrak{s}\big(\mathfrak{u}(2) \oplus \mathfrak{u}(2)\big) \cong \mathfrak{so}(4) \oplus \mathfrak{so}(2)\,.
\end{align}

At the group level, the corresponding maximal compact subgroup of $\mathrm{SU}(2,2)$ is:
\begin{align}
    S\big(\mathrm{U}(2) \times \mathrm{U}(2)\big) = \Big\{ &(g_L, g_R) \in \mathrm{U}(2) \times \mathrm{U}(2) \;;\; \nonumber\\[0.2cm]
    &\qquad\det(g_L\, g_R) = \det(g_L)\det(g_R) = 1 \Big\}\,.
\end{align}
Each factor admits the local decomposition:
\begin{align}
    \mathrm{U}(2) \cong \frac{\mathrm{SU}(2) \times U(1)}{\mathbb{Z}_2}\,,
\end{align}
where the quotient by $\mathbb{Z}_2$ is generated by the element $(-\mathbbm{1}_2, -1)$, which acts trivially on all representations and enforces the determinant constraint globally.

Substituting this decomposition into the product $\mathrm{U}(2)_L \times \mathrm{U}(2)_R$, the defining unimodularity condition $\det(g_L)\det(g_R) = 1$ relates the two central $\mathrm{U}(1)$ phases via:
\begin{align}
    e^{\mathrm{i}(\theta_L + \theta_R)} = 1\,,
\end{align}
where $e^{\mathrm{i}\theta_L}$ and $e^{\mathrm{i}\theta_R}$ denote the phase rotations associated with $\mathrm{U}(1)_L$ and $\mathrm{U}(1)_R$, respectively. This condition removes the overall (diagonal) combination of the two $\mathrm{U}(1)$ phases, leaving a single residual $\mathrm{U}(1)$ generated by the relative phase $e^{\mathrm{i}(\theta_L - \theta_R)}$. Accounting for this shared central identification yields the global structure:
\begin{align}\label{SZ2}
    S\big(\mathrm{U}(2) \times \mathrm{U}(2)\big) \cong \frac{\mathrm{SU}(2)_L \times \mathrm{SU}(2)_R \times \mathrm{U}(1)}{\mathbb{Z}_2}\,,
\end{align}
where the $\mathbb{Z}_2$ quotient is generated by the element $(-\mathbbm{1}_2, -\mathbbm{1}_2, -1)$, which acts trivially on all representations and enforces the determinant constraint globally.

\section{Appendix: Matrix Exponentials of the Pauli Matrices} \label{Appendix Pauli}

The Pauli matrices satisfy the algebraic identity:
\begin{align}
    \sigma^{}_i \sigma^{}_j = \delta_{ij} \mathbbm{1} + \mathrm{i} \varepsilon_{ijk} \sigma^{}_k\,,
\end{align}
where $i,j,k = 1,2,3$, $\delta_{ij}$ is the Kronecker delta, and $\varepsilon_{ijk}$ is the Levi-Civita symbol. From this, it follows that:
\begin{align}
    \sigma^{2n}_j = \mathbbm{1}\,, \quad \sigma^{2n+1}_j = \sigma^{}_j\,,
\end{align}
for all $n \in \mathbb{N}$. Using these properties, the matrix exponential of any Pauli matrix takes the form:
\begin{align}
    e^{-\mathrm{i}\theta\sigma_j} &= \sum_{n=0}^\infty \frac{(-\mathrm{i}\theta)^n \sigma_j^n}{n!} \nonumber\\[0.2cm]
    &= \left(\sum_{k=0}^\infty \frac{(-\mathrm{i}\theta)^{2k}}{(2k)!}\right)\mathbbm{1} + \left(\sum_{k=0}^\infty \frac{(-\mathrm{i}\theta)^{2k+1}}{(2k+1)!}\right)\sigma_j \nonumber\\[0.2cm]
    &= \cos\theta\,\mathbbm{1} - \mathrm{i}\sin\theta\,\sigma_j \,,
\end{align}
which is valid for any $j = 1, 2, 3$ and $\theta \in \mathbb{R}$.

\sectionmark{Appendix: Diagonal Inner Products of Real Matrices Representing ...}
\section{Appendix: Diagonal Inner Products of Real Matrices Representing Imaginary Split-Octonion Units}\label{Appendix A}
\sectionmark{Appendix: Diagonal Inner Products of Real Matrices Representing ...}

\begin{Proposition}
    \textbf{(diagonal structure of triple products for associative split-octonion triples).} The product $m^{}_\mu\, m^{}_\nu\, m^{}_\lambda$, where $\mu,\nu,\lambda \in \big\{ {\boldsymbol{\imath}},\, {\boldsymbol{\jmath}},\, \boldsymbol{\ell}{\boldsymbol{\imath}},\, {\boldsymbol{\jmath}}\boldsymbol{\ell},\,  \boldsymbol{\ell}\boldsymbol{k}\,, \boldsymbol{\ell} \big\}$, is a \emph{diagonal} $8\times 8$ matrix iff $(\mu,\nu,\lambda)$ is an associative triple belonging to a split-octonion subalgebra such that their product is the unit octonion:
    \begin{align}\label{A.1}
        \mu\, \nu\, \lambda = (\mu\, \nu)\, \lambda = \mu\, (\nu\, \lambda) = \mathbbm{1} \quad \iff \quad\mbox{diagonal}\quad m^{}_\mu\, m^{}_\nu\, m^{}_\lambda\,.
    \end{align}
\end{Proposition}
 
\begin{proof}
    For any split-octonion $x$ and $\mu, \nu, \lambda$ satisfying \eqref{A.1}, we have the following identity, up to a sign (see Remark \ref{Remark 2.202022}):
    \begin{align}
        m^{}_\nu ( m^{}_\lambda\, x ) \underbrace{= \nu (\lambda\,x) = (\nu\,\lambda)\,x =}_{\text{clarification omitted henceforth for simplicity!}} ( m^{}_\nu\, m^{}_\lambda ) \,x \,.
    \end{align}
    Expanding this relation, again up to a sign, we obtain:
    \begin{align}
        m^{}_\mu \big(m^{}_\nu ( m^{}_\lambda\, x )\big) = m^{}_\mu \big((m^{}_\nu\, m^{}_\lambda)\, x\big)  = \underbrace{( m^{}_\mu\, m^{}_\nu\, m^{}_\lambda ) \,x = x}_{\text{since $(\mu\, \nu\, \lambda)\,x = \mathbbm{1}x$ by Eq. \eqref{A.1}}}\,.
    \end{align}
    This proves the proposition ($\Longrightarrow$) and also tells us that the diagonal elements of the matrix  $m^{}_\mu\, m^{}_\nu\, m^{}_\lambda$ are $1$, up to a sign. The inverse statement ($\Longleftarrow$) of the above, however, follows from Moufang's theorem (see Remark \ref{Remark 2.2222}). 
    
    Here are two examples that differ from $D = m^{}_{{\boldsymbol{\jmath}}}\, m^{}_{\boldsymbol{\ell}}\, m^{}_{{\boldsymbol{\jmath}}\boldsymbol{\ell}} = m^{}_{5}\, m^{}_{4}\, m^{}_{2} $ \eqref{D}:
    \begin{align}
        D^\prime \,x &= m^{}_{0}\big( m^{}_{2} (m^{}_{3}\, x) \big) \nonumber\\[0.2cm]
        &= {{\boldsymbol{\imath}}}\big({{\boldsymbol{\jmath}}\boldsymbol{\ell}} ({\boldsymbol{\ell}\boldsymbol{k}} \,x)\big) \nonumber\\[0.2cm]
        &= \bigoplus(D_{\mathbbm{1}\boldsymbol{k}} + D_{05} - D_{12} + D_{34})\,x\,, 
    \end{align}
    \begin{align}
        D^{\prime\prime}\, x &= m^{}_{5}\big(m^{}_{1} (m^{}_{3}\, x)\big) \nonumber\\[0.2cm]
        &= {{\boldsymbol{\jmath}}}\big({\boldsymbol{\ell}{\boldsymbol{\imath}}} ({\boldsymbol{\ell}\boldsymbol{k}}\, x)\big) \nonumber\\[0.2cm]
        &= \bigoplus(D_{\mathbbm{1}\boldsymbol{k}} - D_{05} + D_{12} + D_{34})\, x\,,
    \end{align}
    where $x\in\mathrm{Span}\{A\}$ (see Eqs. \eqref{8uni}-\eqref{8uni'*}) and see Sect. \ref{Sect. Convention} for the convention on $\bigoplus (\dots)$.
\end{proof}
	\renewcommand*\vec{\mathaccent"017E\relax}

\setcounter{equation}{0} 

\chapter{Ladder Representation of $\mathfrak{u}(2, 2)$ as a Quantized $4$-dimensional Realization}\label{Chapter 3}

\section*{Chapter Overview}

This chapter develops a rigorous framework for positive-energy ladder representations of the conformal Lie algebra:
\begin{align*}
    \mathfrak{u}(2,2) \cong \mathfrak{su}(2,2) \, \big(\cong \mathfrak{so}(4,2)\big) \oplus \mathfrak{u}(1) \;\subset\; \mathfrak{cl}(4,2)\,,
\end{align*}
realized as quantized models of massless fields with arbitrary helicity in $4$-dimensional Minkowski, de Sitter (dS), and anti-dS spacetimes. The construction is founded on an invariant bilinear form in the Majorana-spinor space of $\mathfrak{cl}(4,2)$, establishing the geometric basis for the representation theory. Through the isomorphism $\mathfrak{cl}^{\text{even}}(4,2) \cong \mathfrak{cl}(4,1) \;\big(\supset \mathfrak{su}(2,2)\big)$, an explicit $4 \times 4$ matrix realization is introduced, encoding conformal symmetry in a computationally tractable spinorial form. This provides the foundation for the systematic construction of ladder operators, spectral analysis, and conformally covariant field operators.

The representations appear as lowest-weight modules characterized by conformal energy and helicity, giving a full classification of massless conformal fields under $\mathfrak{u}(2,2)$. In particular, restriction to the dS subalgebra $\mathfrak{sp}(2,2) \cong \mathfrak{so}(4,1)$ — a central focus of this study — preserves irreducibility, thereby capturing the notion of masslessness in this curved background. Embedding these modules into $4$-dimensional conformal vertex algebras reveals their analytic structure, including operator product expansions and modular properties. The Casimir invariants are explicitly computed, clarifying their role in the spectrum, while the correspondence between $\mathfrak{so}(4,2)$ and its Euclidean form $\mathfrak{so}(5,1)$ completes a unified algebraic framework for the study of massless conformal systems in both flat and curved settings.

\section{Invariant Bilinear Form in the Space of $\mathfrak{cl}(4, 2)$ Spinors}

\textbf{\textit{Note}:} Having clarified in the preceding chapter the Majorana realization of the conformal Clifford algebra $\mathfrak{cl}(4,2)$ and its subalgebras — implemented via left multiplication by imaginary split-octonion units on the split-octonion algebra — along with their chiral realizations, we shall henceforth streamline the notation and present only the resulting formulas.

In the preceding chapter, we showed that the $8$-dimensional Majorana (real) vector space of $\mathfrak{cl}(4,2)$ spinors, spanned by the split-octonion units, naturally carries the structure of a composition algebra — non-associative but alternative. In particular, it can be equipped with a real-valued trace, a Cayley-Dickson conjugation, a factorizable quadratic norm $N(\xi)$, and an associated non-degenerate symmetric bilinear form $\langle\xi, \eta\rangle$ such that:
\begin{align}\label{2.0}
    N(\xi) &= \xi\xi^\ast \, \big(= \xi^\ast \xi \big) \nonumber\\[0.2cm]
    &= (\xi^{}_{\mathbbm{1}})^2 + (\xi^{}_{{\boldsymbol{k}}})^2 + (\xi^{}_{0})^2 + (\xi^{}_{5})^2 - (\xi^{}_{1})^2 - (\xi^{}_{2})^2 -(\xi^{}_{3})^2 - (\xi^{}_{4})^2\,,
\end{align}
where $\xi = \xi_{}^{\mathbbm{1}}\mathbbm{1} + \xi_{}^{\boldsymbol{k}} {\boldsymbol{k}} + \xi_{}^{{0}} \boldsymbol{\imath} + \xi_{}^{{5}} {\boldsymbol{\jmath}} + \xi_{}^{{1}} \boldsymbol{\ell}\boldsymbol{\imath} + \xi_{}^{{2}} {\boldsymbol{\jmath}}\boldsymbol{\ell} + \xi_{}^{{3}} \boldsymbol{\ell}{\boldsymbol{k}} + \xi_{}^{{4}} \boldsymbol{\ell} \in \mathbb{O}_{\mathbb{S}}$, while the star `$^\ast$' changes the sign of the imaginary $\mathbb{O}_{\mathbb{S}}$ units, and according to Eq. \eqref{muhum}:  
\begin{align}\label{2.1}
    \langle \xi , \eta \rangle = \langle \eta, \xi \rangle &= \frac{1}{2} \big( N(\xi+\eta) - N(\xi) - N(\eta) \big) \nonumber\\[0.2cm]
    &= \frac{1}{2} \big( \xi\eta^\ast + \eta\xi^\ast \big) \nonumber\\[0.2cm]
    &= \xi^{}_{\mathbbm{1}}\eta^{}_{\mathbbm{1}} + \xi^{}_{{\boldsymbol{k}}}\eta^{}_{{\boldsymbol{k}}} + \xi^{}_{0}\eta^{}_{0} + \xi^{}_{5}\eta^{}_{5} - \xi^{}_{1}\eta^{}_{1} - \xi^{}_{2}\eta^{}_{2} -\xi^{}_{3}\eta^{}_{3} - \xi^{}_{4}\eta^{}_{4} \,.
\end{align}
Recall that, throughout this work, we employ the simplified index notation $\nu \in \big\{ 0:={\boldsymbol{\imath}},\, 5:={\boldsymbol{\jmath}},\, 1:=\boldsymbol{\ell}{\boldsymbol{\imath}},\, 2:={\boldsymbol{\jmath}}\boldsymbol{\ell},\,  3:=\boldsymbol{\ell}\boldsymbol{k},\, 4:=\boldsymbol{\ell}\big\}$. This construction yields two inequivalent $8$-dimensional spinor representations of $\mathrm{Spin}(4,4)$, arising from left and right multiplication by split-octonion units, together with a third $8$-dimensional vector representation. The three are tied together by an invariant trilinear (``Yukawa'') form, which manifests the triality symmetry of the Lie algebra $D_{4}$ (see Refs. \cite{Y, Br}, for more details).

We now reproduce the inner product \eqref{2.1} in the chiral (complex) basis of the generators $\Gamma^{}_\nu$ of $\mathfrak{cl}(4,2)$. The basis generators satisfy the anti-commutation relations: 
\begin{align}
    \big\{ \Gamma^{}_\mu,  \Gamma^{}_\nu\big\} = 2\eta_{\mu\nu}\,,
\end{align}
which follow directly from Eq. \eqref{clif}. Their behaviour under Hermitian conjugation is given by:
\begin{align}\label{2.2}
    \Gamma^\ast_\nu = \Gamma^2_\nu\, \Gamma^{}_\nu \quad\Longrightarrow\quad \Gamma^\ast_{{0}} &= - \Gamma^{}_{{0}}\,, \nonumber\\[0.2cm]
    \Gamma^\ast_{{5}} &= - \Gamma^{}_{{5}}\,, \nonumber\\[0.2cm] 
    \Gamma^\ast_{i} &=  \Gamma^{}_{i}\,, \qquad (i=1,2,3,4) \,.
\end{align}
This is indeed a direct consequence of Eqs. \eqref{anti-Hermitian} and \eqref{anti-Hermitian'}; in a real basis $\Gamma^{}_\nu = m^{}_\nu$, $m^\ast_\nu$ stands for the transposed $m^\top_\nu$ of $m^{}_\nu$. In light of the preceding discussion, our first objective — whose justification will be given in the following Remark — is therefore to determine a unitary element $B$ of $\mathfrak{cl}(4,2)$ satisfying:
\begin{align}\label{2.3}
    \Gamma^\ast_\nu \,B = - B\, \Gamma^{}_\nu \quad\Longrightarrow\quad&\; \Gamma^{}_0 \,B = B\, \Gamma^{}_0 \,, \nonumber\\[0.2cm]
    &\; \Gamma^{}_5 \,B = B\, \Gamma^{}_5 \,,\nonumber\\[0.2cm]
    \Downarrow \qquad\qquad\qquad \;\;&\; \Gamma^{}_i \,B = - B\, \Gamma^{}_i \,, \qquad (i=1,2,3,4) \,,\nonumber\\[0.2cm]
    \Gamma^{\ast}_{\mu\nu}\, B = - B\, \Gamma^{}_{\mu\nu} \quad\Longrightarrow\quad&\; \Gamma^{}_{05}\, B = B\, \Gamma^{}_{05} \,, \nonumber\\[0.2cm]
    &\; \Gamma^{}_{12}\, B = B\, \Gamma^{}_{12} \,,\nonumber\\[0.2cm]
    &\; \Gamma^{}_{34}\, B = B\, \Gamma^{}_{34} \,,\nonumber\\[0.2cm]
    &\qquad\;\;\Downarrow \nonumber\\[0.2cm]
    &\big[B, \omega(4,2)\big] = 0\,.
\end{align}
Recall that $\omega(4,2)$ is the $\mathfrak{cl}(4,2)$ volume form in the basis $\big\{ \Gamma^{}_\nu \big\}$ (see Eq. \eqref{OOmega}), while the bivector generators $\Gamma_{\mu\nu}$ are defined in Eq. \eqref{GammaMuNu-JPG}.

\begin{remark}{
    \emph{(the $B$-Hermitian form and restoration of anti-Hermiticity in the chiral basis).} Let $(\mathscr{H},\langle\cdot,\cdot\rangle)$ be a complex inner product space and let $X$ be a linear operator on $\mathscr{H}$. The adjoint $X^\dagger$ of $X$ (with respect to $\langle\cdot,\cdot\rangle$) is defined by the relation:
    \begin{align}
        \langle \psi, X \phi \rangle = \langle X^\dagger \psi, \phi \rangle\,, \quad \mbox{for all} \quad \psi,\phi \in \mathscr{H} \,.
    \end{align}
    The operator $X$ is said to be Hermitian if $X^\dagger = X$ and anti-Hermitian if $X^\dagger = -X$. These notions depend explicitly on the chosen inner product. In particular, when the inner product is modified, the corresponding adjoint changes accordingly. Only in the special case of the standard inner product does the adjoint $X^\dagger$ coincide with ordinary Hermitian conjugation $X^\ast:= \overline{X}^\top$.
    
    As noted above, in the chiral basis of $\mathfrak{cl}(4,2)$, the Clifford generators $\Gamma^{}_\nu$ are not all anti-Hermitian with respect to ordinary Hermitian conjugation (see Eq. \eqref{2.2}); this reflects the mixed Hermiticity properties induced by the Clifford relations in signature $(4,2)$. As a result, the standard inner product fails to be compatible with the requirement that the generators $\Gamma^{}_{\mu\nu}$ of the conformal algebra $\mathfrak{so}(4,2)$ act as anti-Hermitian operators. Indeed, with respect to the standard inner product:
    \begin{align}\label{std-InnerProduct}
        \langle \psi , \phi \rangle^{}_{\text{std}} = \psi^\ast \phi \,,
    \end{align}
    one finds:
    \begin{align}
        \langle \psi , \Gamma^{}_{\mu\nu}\phi \rangle^{}_{\text{std}} &= \psi^\ast \, \Gamma^{}_{\mu\nu} \phi \nonumber\\[0.2cm]
        &= \big(\Gamma^{\ast}_{\mu\nu} \psi \big)^\ast\, \phi \nonumber\\[0.2cm]
        &= \langle \Gamma^{\ast}_{\mu\nu} \psi , \phi \rangle^{}_{\text{std}} = \langle \pm\Gamma^{}_{\mu\nu} \psi , \phi \rangle^{}_{\text{std}} \,,
    \end{align}
    where the sign depends on the specific bivector generator $\Gamma^{}_{\mu\nu}$ under consideration. Thus, the generators $\Gamma^{}_{\mu\nu}$ are not uniformly anti-Hermitian with respect to \eqref{std-InnerProduct}, obstructing the construction of a unitary representation in the chiral basis.
    
    \textbf{\textit{Note}:} For a unitary representation on a Hilbert space, it is necessary that the generators $X$ of the Lie algebra be anti-Hermitian with respect to the chosen inner product, ensuring that their exponentials $e^{\alpha X}$ $(\alpha \in \mathbb{R})$ are unitary operators and hence that the group representation is unitary.

    To remedy this, one introduces a unitary element $B \in \mathfrak{cl}(4,2)$ and defines a modified notion of Hermiticity by imposing $\Gamma^\ast_\nu \,B = - B\, \Gamma^{}_\nu$, and hence $\Gamma^{\ast}_{\mu\nu}\, B = - B\, \Gamma^{}_{\mu\nu}$ (see Eq. \eqref{2.3}). This condition allows one to define a non-degenerate Hermitian form:
    \begin{align}\label{B-InnerProduct}
        \langle \psi, \phi \rangle_B = \psi^\ast B\, \phi \,,
    \end{align}
    with respect to which the generators $\Gamma^{}_{\mu\nu}$ are anti-Hermitian. Indeed:
    \begin{align}
        \langle \psi , \Gamma^{}_{\mu\nu}\phi \rangle^{}_B &= \psi^\ast B\, \Gamma^{}_{\mu\nu} \phi \nonumber\\[0.2cm]
        &= \psi^\ast \big(-\Gamma^{\ast}_{\mu\nu}\big)\, B\, \phi \nonumber\\[0.2cm]
        &= \big(-\Gamma^{}_{\mu\nu} \psi \big)^\ast B\, \phi = \langle -\Gamma^{}_{\mu\nu} \psi , \phi \rangle^{}_B \,.
    \end{align}
    Consequently, the associated one-parameter subgroups preserve the Hermitian form \eqref{B-InnerProduct}.

    Finally, requiring that $B$ commute with the Clifford volume form $\omega(4,2)$ ensures that the Hermitian form respects the chirality decomposition of the spinor space. In particular, left- and right-handed sectors remain orthogonal, and the action of the bivector generators $\Gamma^{}_{\mu\nu}$ preserves each chiral sector. In this way, the element $B$ provides a canonical and algebraically consistent Hermitian structure for realizing unitary positive-energy irreducible representations (UPEIRs) of the conformal algebra in the chiral Clifford basis.

    \textbf{\textit{Note}:} For the sake of notational simplicity, we shall henceforth drop the subscripts  
    `$\text{std}$' and `$B$' from the corresponding sesquilinear forms, as the relevant choice will always be clear from the context.
}\end{remark}

Remarkably, there exists such a Hermitian unitary $B\in\mathfrak{cl}(4,2)$ determined by the above properties up to sign. It turns out that the resulting $B$ coincides with the $\mathfrak{cl}(4,0)$ pseudo scalar:
\begin{align}
    B = M(4,0) &= m^{}_{{12}} m^{}_{{43}} \nonumber\\[0.2cm]
    &= U\, m^{}_{{12}} m^{}_{{43}} \, U^{-1} \qquad (\text{since} \; [m^{}_{{12}} m^{}_{{43}} \,,\, U]\,=\,0) \nonumber\\[0.2cm]
    &= U\, m^{}_{{12}} \, U\, U^{-1} m^{}_{{43}} \, U^{-1}  \nonumber\\[0.2cm]
    &= \Gamma^{}_{{12}} \Gamma^{}_{{43}} \,.
\end{align}
The commutation relation $\big[m^{}_{{12}} m^{}_{{43}}, U\big]=0$ can be verified straightforwardly by substituting the explicit form of the unitary transformation $U$ from Eq. \eqref{1.31}, making use of Eqs. \eqref{U(1)generator}, \eqref{D}, and \eqref{S}, and referring to Proposition \ref{proposition 2.1}. The Hermitian unitary operator $B$ admits a diagonal form both in the Majorana basis (see Eq. \eqref{77}):
\begin{align}\label{2.4}
    B\,x \; \big( = B^\ast \,x\big) &= m^{}_{{12}} \,(m^{}_{{43}}\,x)  \nonumber\\[0.2cm]
    &= \Gamma^{}_{{12}} \, \big(\Gamma^{}_{{43}}\,x\big)  \nonumber\\[0.2cm]
    &= \bigoplus(\mathbbm{1}_{\mathbbm{1}, {\boldsymbol{k}}, {5}, {0} } - \mathbbm{1}_{{2}, {1}, {4}, {3}})\,x \,,
\end{align}
and in the compact chiral basis (following the procedure described in Eq. \eqref{Ey}):
\begin{align}\label{2.4H}
    B\,y \; \big( = B^\ast\,y \big) &= m^{}_{{12}} (m^{}_{{43}}\,y)  \nonumber\\[0.2cm]
    &= \Gamma^{}_{{12}} \big(\Gamma^{}_{{43}}\,y\big) \nonumber\\[0.2cm]
    &= \bigoplus\big(\mathbbm{1}_{\widehat{\mathbbm{1}\boldsymbol{k}}_+, \widehat{\mathbbm{1}\boldsymbol{k}}_-, \widehat{{05}}_+, \widehat{{05}}_-} - \mathbbm{1}_{\widehat{{12}}_+, \widehat{{{12}}}_-, \widehat{{34}}_+, \widehat{{{34}}}_-}\big)\,y \,,
\end{align}
where $x\in\mathrm{Span}\{A\}$ (see Eqs. \eqref{8uni}-\eqref{8uni'*}), $y\in\mathrm{Span}\{ n^{}_{a^{}_\pm} \}$ (see Eq. \eqref{Spm} and Sect. \ref{Sect. adjusment}), and consult Sect. \ref{Sect. Convention} for the convention concerning $\bigoplus (\dots)$. We observe that:
\begin{align}
    B^2 \, x = \mathbbm{1}_8 \, x \,, \quad (\mbox{and similarly on $y$})\,.
\end{align}
The sign is fixed by $(B-1) \mathbbm{1}_8 = 0$, to preserve the spinor equal to the real unit $\mathbbm{1}$ of $\mathbb{O}_\mathbb{S}$. 

Remarkably, we see that, for the Majorana (real) spinors $\xi$ and $\eta$, the Hermitian form \eqref{2.4} reproduces the scalar product \eqref{2.1}:
\begin{align}\label{2.5}
    \langle \xi , \eta \rangle = \xi^A\, B_{AB}\, \eta^B\,,
\end{align}
where $A,B= \mathbbm{1},\, \boldsymbol{k},\, \nu\in \big\{ 0:={\boldsymbol{\imath}},\, 5:={\boldsymbol{\jmath}},\, 1:=\boldsymbol{\ell}{\boldsymbol{\imath}},\, 2:={\boldsymbol{\jmath}}\boldsymbol{\ell},\,  3:=\boldsymbol{\ell}\boldsymbol{k},\, 4:=\boldsymbol{\ell}\big\}$ (see Eq. \eqref{8uni}-\eqref{8uni'*}). Similarly, in the space of $4$-component pure chiral spinors\footnote{For reviews on Cartan simple (or Chevalley pure) $8$-spinors, see Refs. \cite{BTSpinor, D}} $\psi^+ = \psi^{\widehat{\mathbbm{1}\boldsymbol{k}}_+} \,n^{}_{\widehat{\mathbbm{1}\boldsymbol{k}}_+} + \psi^{{\widehat{{05}}}_+} \,n^{}_{\widehat{05}_+} + \psi^{\widehat{{12}}_+} \,n^{}_{\widehat{12}_+} + \psi^{\widehat{{34}}_+} \,n^{}_{\widehat{34}_+}$ displayed in the compact chiral basis (see Sect. \ref{Sect. adjusment}), where:
\begin{align}
    \psi^{\widehat{\mathbbm{1}\boldsymbol{k}}_+} \;\left(:= \sqrt{2}\,\psi^{\widehat{\mathbbm{1}\boldsymbol{k}}_+}\right) &= \xi^{\mathbbm{1}} - \mathrm{i} \xi^{{\boldsymbol{k}}}\,, \nonumber\\[0.2cm]
    \psi^{{\widehat{{05}}}_+} \;\left(:= \sqrt{2}\,\psi^{{\widehat{{05}}}_+}\right) &= \xi^{{0}} - \mathrm{i} \xi^{{5}}\,,\nonumber\\[0.2cm]
    \psi^{\widehat{{12}}_+} \;\left(:= \sqrt{2}\,\psi^{\widehat{{12}}_+}\right) &= \xi^{{1}} - \mathrm{i} \xi^{{2}}\,, \nonumber\\[0.2cm]
    \psi^{\widehat{{34}}_+} \;\left(:= \sqrt{2}\,\psi^{\widehat{{34}}_+}\right) &= \xi^{{3}} - \mathrm{i} \xi^{{4}}\,,
\end{align}
the Hermitian form \eqref{2.4H} also allows us to reproduce the quadratic norm \eqref{2.0}:
\begin{align}\label{2.6}
    &\overline{{\psi^{a_+}}}\, B_{a_+ b_+}\, \psi^{b_+}  \nonumber\\[0.2cm]
    &= \overline{{\psi^{\widehat{\mathbbm{1}\boldsymbol{k}}_+}}}\,{\psi^{\widehat{\mathbbm{1}\boldsymbol{k}}_+}} + \overline{{\psi^{\widehat{{05}}_+}}}\, {\psi^{\widehat{{05}}_+}} - \overline{{\psi^{\widehat{{12}}_+}}}\, {\psi^{\widehat{{12}}_+}} - \overline{{\psi^{\widehat{{34}}_+}}}\, {\psi^{\widehat{{34}}_+}} \nonumber\\[0.2cm]
    &= \xi^2_{\mathbbm{1}} + (\xi^{}_{{\boldsymbol{k}}})^2 + (\xi^{}_{0})^2 + (\xi^{}_{5})^2 - (\xi^{}_{1})^2 - (\xi^{}_{2})^2 -(\xi^{}_{3})^2 - (\xi^{}_{4})^2 \nonumber\\[0.2cm]
    &= \xi\xi^\ast = \xi^\ast\xi = N(\xi) = \langle \xi , \xi \rangle\,,
\end{align}
where $a^{}_+, b^{}_+ = {\widehat{\mathbbm{1}\boldsymbol{k}}_+, \widehat{05}_+, \widehat{12}_+, \widehat{34}}_+$, while $\overline{\psi}$ stands for the complex conjugate of $\psi$.

\sectionmark{Isomorphism of $\mathfrak{cl}(4, 1)$ with the Even Part of $\mathfrak{cl}(4, 2)$: ...}
\section{Isomorphism of $\mathfrak{cl}(4, 1)$ with the Even Part of $\mathfrak{cl}(4, 2)$: $4\times 4$ Matrix Realization of $\mathfrak{su}(2,2)\cong\mathfrak{so}(4, 2)$}
\sectionmark{Isomorphism of $\mathfrak{cl}(4, 1)$ with the Even Part of $\mathfrak{cl}(4, 2)$: ...}

It looks natural to view the de Sitter (dS) Clifford algebra $\mathfrak{cl}(4,1)$ as a Clifford subalgebra of the conformal one $\mathfrak{cl}(4,2)$. It turns out, however, to be advantageous to use the known isomorphism between the even part of $\mathfrak{cl}(4,2)$, denoted here by $\mathfrak{cl}^{\text{even}}(4,2)$, and the irreducible realization $\mathbb{C}[4]$ of $\mathfrak{cl}(4,1)$. Specifically, we identify the $4\times 4$-matrix realization of the generators $\gamma^{}_\alpha$ ($\alpha = {0,1,2,3,4}$) of $\mathfrak{cl}(4,1)$ with the projected ``boosts'' $\Gamma^{}_{\alpha{5}}$ of the Lie algebra $\mathfrak{su}(2,2)\cong\mathfrak{so}(4,2)$:
\begin{align}\label{2.7}
    \gamma^{}_i := \gamma^{}_{i{5}} = \Gamma^{}_{i {5}}\, D_+ \,, \quad\mbox{and}\quad \gamma^{{0}}_{} := \gamma^{}_{{05}} = \Gamma^{}_{{05}}\, D_+ \,,
\end{align}
where: 
\begin{align}
    D_+ = \frac{\mathbbm{1}+D}{2}\,, \quad\mbox{and}\quad \Gamma^{}_{\alpha{5}} = \Gamma^{}_{\alpha}\, \Gamma^{}_{{5}}\,,
\end{align}
with $\alpha = {0,1,2,3,4}$ and $i = {1,2,3,4}$, while $D_+$ is the projection operator on $\mathfrak{S}_+$; that is:
\begin{align}
    D_+ \mathfrak{S} =: \mathfrak{S}_+ \,,
\end{align}
where we recall that in the chiral (complex) basis \eqref{Spm} (see also Eqs. \eqref{tytyty}-\eqref{zv}):
\begin{align}
    \mathfrak{S} &= \mathrm{Span}\Big\{ n^{}_{a^{}_\pm} \;;\; n^{}_{a^{}_\pm} = n^{}_{\widehat{\mathbbm{1}\boldsymbol{k}}_\pm}, n^{}_{\widehat{{05}}_\pm}, n^{}_{\widehat{{12}}_\pm}, n^{}_{\widehat{{34}}_\pm} \Big\} \,, \nonumber\\[0.2cm]
    \mathfrak{S}_+ &= \mathrm{Span}\big\{ n^{}_{a^{}_+} \big\} \,,
\end{align}
and in the Majorana (real) basis \eqref{8uni}-\eqref{8uni'*} (for connection with the above basis, see again Eqs. \eqref{tytyty}-\eqref{zv}):
\begin{align}
    \mathfrak{S} &= \mathrm{Span} \big\{A \;:\; A= \mathbbm{1},\, \boldsymbol{k},\, 0,\, 5,\, 1,\, 2,\, 3,\, 4\big\}\,, \nonumber\\[0.2cm]
    \mathfrak{S}_+ &= \mathrm{Span}\big\{\mathbbm{1},\, 5,\, 2,\, 4 \big\}\,.
\end{align}
Moreover, for the definitions of $D$ in the chiral (complex) and Majorana (real) bases, respectively, see Eqs. \eqref{omega(4,2)=iDMat'} and \eqref{D}.

\begin{remark}{
    \emph{(the map $\Gamma^{}_{\alpha 5} \;\longmapsto\; \gamma^{}_\alpha$ and the realization of $\mathfrak{cl}(4,1)$).} Some remarks on the map $\Gamma^{}_{\alpha 5} \;\longmapsto\; \gamma^{}_\alpha$ \eqref{2.7} are as follows:
    \begin{enumerate}[leftmargin=0.8cm]
        \item{The above realization of the generators $\gamma^{}_\alpha$ of the dS Clifford algebra $\mathfrak{cl}(4,1)\cong\mathfrak{cl}^{\text{even}}(4,2)$ follows from the identities:
        \begin{align}
            \big\{\Gamma^{}_{\alpha{5}} , \Gamma^{}_{\beta{5}}\big\} &= \big\{\Gamma^{}_{\alpha} , \Gamma^{}_{\beta}\big\} = 2\eta_{\alpha\beta}\,, \nonumber\\[0.2cm] \big[\Gamma^{}_{\alpha{5}} , \Gamma^{}_{\beta{5}}\big] &= \big[\Gamma^{}_{\alpha} , \Gamma^{}_{\beta}\big]\,,
        \end{align}
        where $\eta_{\alpha\beta} = \mbox{diag}(-,+,+,+,+)$, with $\alpha, \beta = {0,1,2,3,4}$. These relations show that the matrix $\Gamma^{}_{5}$ drops out of both the anti-commutation and commutation relations among the projected generators $\Gamma^{}_{\alpha5}$. Consequently the projected ``boosts'' $\Gamma^{}_{\alpha5}$ obey the same Clifford $\mathfrak{cl}(4,1)$ and Lie algebra $\mathfrak{so}(4,1)$ relations as the $\Gamma^{}_{\alpha}$. Hence, one can consistently regard $\Gamma^{}_{\alpha 5}$ as an $8\times8$-matrix representation of $\mathfrak{cl}(4,1)$, and define the $4\times4$ matrices $\gamma^{}_i := \Gamma^{}_{i5} D_+$ and $\gamma^{0} := \Gamma^{}_{05} D_+$, which provide a $4\times4$-matrix representation of $\mathfrak{cl}(4,1)$ (and therefore of the corresponding Lie algebra $\mathfrak{so}(4,1)$).}

        \item{Obviously, the map $\Gamma^{}_{\alpha 5} \;\longmapsto\; \gamma^{}_\alpha$ \eqref{2.7} does not respect the $\mathbb{Z}_2$ grading of $\mathfrak{cl}(4,2)$. The Clifford algebra $\mathfrak{cl}(4,2)$ is naturally $\mathbb{Z}_2$-graded, meaning it has an even and odd part:
        \begin{align}
            \mathfrak{cl}(4,2) = \mathfrak{cl}^{\text{even}}(4,2) \oplus \mathfrak{cl}^{\text{odd}}(4,2)\,.
    \end{align}
    By definition, the even subalgebra $\mathfrak{cl}^{\text{even}}(4,2)$ consists of elements that can be expressed as products of an even number of generators $\Gamma^{}_\nu$, while the odd subspace $\mathfrak{cl}^{\text{odd}}(4,2)$ consists of elements formed from an odd number of generators. The mapping in question explicitly interchanges even and odd elements, sending even elements of $\mathfrak{cl}(4,2)$ to odd elements of $\mathfrak{cl}(4,1)$ and vice versa. As a result, it fails to preserve the $\mathbb{Z}_2$ grading.}

    \item{As a benefit of this $4\times 4$ matrix realization of $\mathfrak{cl}(4,1) \cong \mathfrak{cl}^{\text{even}}(4,2)$, the $\mathfrak{cl}(4,1)$ is found to have the constant value `$\mathrm{i}$', which corresponds to the eigenvalue of $\omega(4,2)$ on the chiral subspace $\mathfrak{S}_+$:
    \begin{align}\label{Hashomi}
        \texttt{w}(4,1)\, y :=&\, \gamma^{}_{{05}}\, \gamma^{}_{{15}}\, \gamma^{}_{{25}}\, \gamma^{}_{{35}}\, \gamma^{}_{{45}} \, y \nonumber\\[0.2cm]
        =&\,\Gamma^{}_{{05}}\, \Gamma^{}_{{15}}\, \Gamma^{}_{{25}}\, \Gamma^{}_{{35}}\, \Gamma^{}_{{45}} \, D_+\, y \nonumber\\[0.2cm]
        =&\, \Gamma^{}_{{05}}\, \Gamma^{}_{{12}}\, \Gamma^{}_{{34}}\, D_+\, y \nonumber\\[0.2cm]
        =&\, \omega(4,2)\, D_+\, y \nonumber\\[0.2cm]
        =&\, \mathrm{i} \mathbbm{1} D_+\, y\,,
    \end{align}
    where $y\in\mathrm{Span}\{ n^{}_{a^{}_\pm} \}$ and $\mathbbm{1}$ is the unit operator in the totally isotropic space $\mathfrak{S}_+$ spanned by the null vectors $n^{}_{\widehat{\mathbbm{1}\boldsymbol{k}}_+}, n^{}_{\widehat{{05}}_+}, n^{}_{\widehat{{12}}_+}, n^{}_{\widehat{{34}}_+}$.

    \textbf{\textit{Note}:} $\omega(4,2)$ can be identified with the third-order Casimir invariant of $\mathfrak{so}(4,2)$.
    
    The significance of this result, and in particular the sense in which it constitutes a genuine benefit of the present realization, merits a brief clarification. In general, the Clifford pseudoscalar of $\mathfrak{cl}(4,1)$ is central but need not act as a fixed scalar on a spinor module, and its value typically depends on representation-dependent conventions. In the present realization, however, $\texttt{w}(4,1)$ is identified with the conformal volume element $\omega(4,2)$ restricted to the chiral subspace $\mathfrak{S}_+$; since $\omega(4,2)$ is a Casimir invariant of $\mathfrak{so}(4,2)$, its action is fixed, yielding $\texttt{w}(4,1) = \mathrm{i}\mathbbm{1}$. Thus, the $4\times 4$ representation of $\mathfrak{cl}(4,1)$ is automatically irreducible, and its complex structure is canonically inherited from conformal chirality rather than imposed by hand.}
    \end{enumerate}
}\end{remark}

We proceed to displaying the $4\times 4$-matrix realization of the relevant physical operators of $\mathfrak{cl}^{\text{even}}(4,2) \supset \mathfrak{so}(4,2) \supset \mathfrak{so}(4,1)$. To begin with, the $4$-dimensional projection of the Hermitian form $B$ (see Eq. \eqref{2.4} and \eqref{2.4H}) is:
\begin{align}
    {\texttt{B}}\, x := B\, D_+\, x = \bigoplus(\mathbbm{1}_{\mathbbm{1} {5}} - \mathbbm{1}_{{2}{4}})\, D_+\, x \,, 
\end{align}
\begin{align}\label{2.9}
    {\texttt{B}}\, y := B\, D_+\, y = \bigoplus\big(\mathbbm{1}_{\widehat{\mathbbm{1}\boldsymbol{k}}_+ \widehat{{0}{5}}_+} - \mathbbm{1}_{\widehat{{1}{2}}_+ \widehat{{3}{4}}_+}\big)\, D_+\, y\,,
\end{align}
in the Majorana $x\in\mathrm{Span}\{A\}$ (see Eqs. \eqref{8uni}-\eqref{8uni'*}) and chiral $y\in\mathrm{Span}\{ n^{}_{a^{}_\pm} \}$ bases, respectively. Sect. \ref{Sect. Convention} provides the convention for \linebreak $\bigoplus (\dots)$. The (diagonal) Cartan elements of the maximal compact subalgebra $\mathfrak{so}(2)\oplus\mathfrak{so}(4)$ of the $\mathfrak{so}(4,2)$ — specifically, in the chiral basis — are obtained respectively by projecting \eqref{za}, \eqref{za'}, and \eqref{za''}:
\begin{align}\label{2.9+1}
    \gamma^{}_{{05}}\, y = \Gamma^{}_{{05}}\, D_+\, y = \mathrm{i} \bigoplus\big( D_{\widehat{{1}{2}}_+ \widehat{\mathbbm{1}\boldsymbol{k}}_+ } + D_{\widehat{{3}{4}}_+ \widehat{{0}{5}}_+ }\big)\, D_+\, y \; \big(= -\mathrm{i} {\texttt{B}}\, y \big)\,,
\end{align}
\begin{align}\label{pp12}
    \gamma^{}_{{12}}\, y = \Gamma^{}_{{12}}\, D_+\, y = \mathrm{i} \bigoplus\big( D_{\widehat{\mathbbm{1}\boldsymbol{k}}_+ \widehat{{0}{5}}_+ } + D_{\widehat{{1}{2}}_+ \widehat{{3}{4}}_+ }\big)\,D_+\, y\,,
\end{align}
\begin{align}\label{pp34}
    \gamma^{}_{{34}}\, y = \Gamma^{}_{{34}}\, D_+\, y = \mathrm{i} \bigoplus\big( D_{\widehat{\mathbbm{1}\boldsymbol{k}}_+ \widehat{{0}{5}}_+ } - D_{\widehat{{1}{2}}_+ \widehat{{3}{4}}_+}\big)\,D_+ y\,,
\end{align}
reproducing $\omega(4,2) = \mathrm{i}\mathbbm{1}$ (see Eq. \eqref{OOmega}) and similar to Eq. \eqref{2.4H}: 
\begin{align}
    {\texttt{B}}\, y = \gamma^{}_{{12}}\, \gamma^{}_{{43}} \, y = \Gamma^{}_{{12}}\, \Gamma^{}_{{43}} \, D_+ y = \bigoplus\big( \mathbbm{1}_{\widehat{\mathbbm{1}\boldsymbol{k}}_+ \widehat{{0}{5}}_+} - \mathbbm{1}_{\widehat{{1}{2}}_+ \widehat{{3}{4}}_+}\big)\,D_+\, y \,.
\end{align}


\begin{remark}{
    \emph{(chiral coupling and the $\mathfrak{su}(2)$ commutant).} The inner product \eqref{2.6} couples the $4$-dimensional pure conformal chiral spinors $\psi^+$ in $\mathfrak{S}_+$ to their complex conjugate $\psi^- \; \big( = \overline{\psi^+} \big)$ in $\mathfrak{S}_-\; \big(= \overline{\mathfrak{S}_+} \big)$; the $4$-dimensional representations of $\mathfrak{so}(4,1) \subset \mathfrak{so}(4,2)$ in the chiral subspaces $\mathfrak{S}_+$ and $\mathfrak{S}_-$ are equivalent. The covariant transformation of charge conjugation $\psi \longmapsto \psi^c$ \eqref{1.33}, however, requires returning to the reducible $8$-dimensional spinors. In this setting, the commutant of the reducible $\Gamma^{}_{\alpha\beta}$-realization of the dS Lie algebra $\mathfrak{so}(4,1)$ is generated by the volume form $\omega(4,1)$ and $\Gamma^{}_{{5}}$. In the Majorana basis, where $x\in\mathrm{Span}\{A\}$ (see Eqs. \eqref{8uni}-\eqref{8uni'*}), these read as:
    \begin{align}
        \omega(4,1) \,x = \Gamma^{}_{{05}}\, \Gamma^{}_{{15}}\, \Gamma^{}_{{25}}\, \Gamma^{}_{{35}}\, \Gamma^{}_{{45}} \,x = \Gamma^{}_{{05}}\, \Gamma^{}_{{12}}\, \Gamma^{}_{{34}} \,x = \omega(4,2) \,x = \mathrm{i} D \,x \,,
    \end{align}
    \begin{align}
        \mbox{Eq. \eqref{Gamma542}}\;;\quad \Gamma^{}_{{5}} \,x = m^{}_{{5}}\, E \,x &= m^{}_{{0}}\, m^{}_{{1}}\, m^{}_{{2}}\, m^{}_{{3}}\, m^{}_{{4}} \,x \nonumber\\[0.2cm]
        &= \bigoplus(E_{\mathbbm{1}{0}} + E_{5{\boldsymbol{k}}} - E_{{14}} - E_{{23}}) \,x\,.
    \end{align}
    Note that $\mathrm{i}D$ and $\Gamma^{}_{{5}}$ generate a $\mathfrak{su}(2)$ Lie algebra.
}\end{remark}

\sectionmark{Positive Energy Ladder Representations of $\mathfrak{u}(2, 2)$ ...}
\section{Positive Energy Ladder Representations of $\mathfrak{u}(2, 2) \cong \mathfrak{su}(2,2)\, \big(\cong\mathfrak{so}(4,2) \big) \oplus \mathfrak{u}(1) \subset \mathfrak{cl}(4,2)$ Describing Massless Particles of Any Helicity}
\sectionmark{Positive Energy Ladder Representations of $\mathfrak{u}(2, 2)$ ...}

We now present an updated treatment of the $1969$ construction \cite{MT} of irreducible massless representations of the (U$(1)$-extended) conformal group U$(2,2)$, framed within the context of the $\mathfrak{cl}(4,2)$ semispinors developed in the preceding sections. An intermediate step toward this construction was given in Ref. \cite{T19}, where the Lie algebra $\mathfrak{u}(2,2)$ was embedded in $\mathfrak{cl}(4,1)$.

The idea is to view the components of the conformal chiral semispinors $\psi^+$, denoted here by $\varphi$ and its Dirac conjugate $\widetilde{\varphi} = \varphi^\ast {\texttt{B}}$, as bosonic creation and annihilation operators.\footnote{See Ref. \cite{MT} for references to earlier work on such representations, including the 1962 book by B. Kursunoglu, which is now recognized as a key reference in the field alongside more recent sources.} More precisely, we postulate that $\varphi\in\mathfrak{S}_+$ and $\widetilde{\varphi}\in\mathfrak{S}_-$ satisfy the canonical commutation relations:
\begin{align}\label{CCR1}
    \big[\varphi^a, \varphi^{a^\prime}\big] &= 0 \,, \nonumber\\[0.2cm]
    \big[\varphi^a, \widetilde{\varphi}_{a^\prime}\big] &= \delta^a_{a^\prime}\,, \nonumber\\[0.2cm]
    \big[\widetilde{\varphi}_a, \widetilde{\varphi}_{a^\prime}\big] &= 0 \,,
\end{align}
where $\widetilde{\varphi}_{a^\prime} = \varphi^{\ast b} {\texttt{B}}_{b a^\prime}$, with $a,a^\prime,b = {\widehat{\mathbbm{1}\boldsymbol{k}}_+}, {\widehat{{05}}_+}, {\widehat{{12}}_+}, {\widehat{{34}}_+} $ (see Sect. \ref{Sect. adjusment}), and ${\texttt{B}} \; \big(= B\, D_+ \big)$ is the $4$-dimensional (reduced) Hermitian form \eqref{2.9}.  

One can construct two series of unitary irreducible representations (UIRs) of the $\mathfrak{u}(2,2)$ Lie algebra with the ``second-quantized'' ladder operators:
\begin{align}\label{fJf6}
    &J_{\mu\nu} = \frac{1}{2}\, \big(\widetilde{\varphi}\, \gamma^{}_{\mu\nu}\, \varphi\big) \,,  \nonumber\\[0.2cm]
    &\mbox{or, more explicitly,}\quad J_{\mu\nu} = \frac{1}{2}\, \left(\widetilde{\varphi}_a\, (\gamma^{}_{\mu\nu})^a_b\, \varphi^b \right) \,,
\end{align}
where $a,b = {\widehat{\mathbbm{1}\boldsymbol{k}}_+}, {\widehat{{05}}_+}, {\widehat{{12}}_+}, {\widehat{{34}}_+}$, while $\mu,\nu = {0}, {5}, {1}, {2}, {3}, {4}$, $\gamma^{}_{\mu\nu} = - \gamma^{}_{\nu\mu}$, and:
\begin{align}
    \gamma^{}_{i{5}} \, \big(=: \gamma^{}_i\big) = \Gamma^{}_{i{5}}\, D_+\,, \quad\mbox{and}\quad \gamma^{}_{{05}} \,\big(=: \gamma^{0}\big) = \Gamma^{}_{{05}}\, D_+\,,
\end{align}
and hence (as a consequence of Proposition \ref{proposition 2.1}):
\begin{align}\label{i5j5=ij}
    \gamma^{}_{ij} = \gamma^{}_{i5}\, \gamma^{}_{j5}\,, \quad\mbox{and}\quad \gamma^{}_{0i} = \gamma^{}_{05}\, \gamma^{}_{i5}\,,
\end{align}
with $i,j = 1,2,3,4$.

\begin{remark}\label{Remark closed faithful}{    
    \emph{(faithful chiral spinor representation of $\mathfrak{so}(4,2)$ via bilinears).} The bilinear operators $J_{\mu\nu}$ introduced in Eq. \eqref{fJf6} furnish an explicit realization of the Lie algebra $\mathfrak{su}(2,2)\cong\mathfrak{so}(4,2)$:
    \begin{enumerate}[leftmargin=0.8cm]
        \item{\textbf{\textit{Algebra closure and commutation relations:}} The bilinears $J_{\mu\nu} = \frac{1}{2} \widetilde{\varphi}\, \gamma^{}_{\mu\nu}\, \varphi$ form a closed representation of the Lie algebra $\mathfrak{so}(4,2)$, as they satisfy the same commutation relations as the matrix-valued generators $\frac{1}{2} \gamma^{}_{\mu\nu}$:
        \begin{align}\label{Ghalee}
            &\left[J_{\mu\nu} \;,\; J_{\lambda\sigma}\right]  \nonumber\\[0.2cm]
            &= \frac{1}{4} \left[ \widetilde{\varphi}^{}_a\, (\gamma^{}_{\mu\nu})^a_b\, \varphi^b \;,\; \widetilde{\varphi}^{}_c\, (\gamma^{}_{\lambda\sigma})^c_d\, \varphi^d \right] \nonumber\\[0.2cm]
            &= \frac{1}{4}\widetilde{\varphi}^{}_a\, (\gamma^{}_{\mu\nu})^a_b \,\underbrace{\left[\varphi^b_{} \;,\; \widetilde{\varphi}^{}_c \,(\gamma^{}_{\lambda\sigma})^c_d\right]}_{=\; \delta^b_c \,(\gamma^{}_{\lambda\sigma})^c_d} \, \varphi^d + \frac{1}{4} \underbrace{\left[\widetilde{\varphi}^{}_a\, (\gamma^{}_{\mu\nu})^a_b \;,\; \widetilde{\varphi}^{}_c \,(\gamma^{}_{\lambda\sigma})^c_d\right]}_{=\;0}\, \varphi^b_{}\, \varphi^d_{} \nonumber\\
            &\;\;\;+ \frac{1}{4} \widetilde{\varphi}^{}_c \,(\gamma^{}_{\lambda\sigma})^c_d\, \widetilde{\varphi}^{}_a\, (\gamma^{}_{\mu\nu})^a_b\, \underbrace{\left[\varphi^b_{} \;,\; \varphi^d_{}\right]}_{=\;0} + \,\frac{1}{4} \widetilde{\varphi}^{}_c\,(\gamma^{}_{\lambda\sigma})^c_d \,\underbrace{\left[\widetilde{\varphi}^{}_a\, (\gamma^{}_{\mu\nu})^a_b \;,\; \varphi^d_{}\right]}_{=\; - \delta^d_a \,(\gamma^{}_{\mu\nu})^a_b} \,\varphi^b_{} \nonumber\\
            &= \frac{1}{4} \underbrace{\widetilde{\varphi}^{}_a\, \left( (\gamma^{}_{\mu\nu})^a_b\, (\gamma^{}_{\lambda\sigma})^b_d - (\gamma^{}_{\lambda\sigma})^a_b\, (\gamma^{}_{\mu\nu})^b_d \right) \, \varphi^d_{}}_{\text{by reassigning the dummy indices}} \nonumber\\
            &= \frac{1}{4} \widetilde{\varphi}\, \left[\gamma^{}_{\mu\nu} \;,\; \gamma^{}_{\lambda\sigma}\right]\, \varphi \nonumber\\[0.2cm]
            &= \frac{1}{2}\, \widetilde{\varphi}\,\big(\eta_{\nu\lambda} \gamma^{}_{\mu\sigma} - \eta_{\mu\lambda} \gamma^{}_{\nu\sigma} - \eta_{\nu\sigma} \gamma^{}_{\mu\lambda} + \eta_{\mu\sigma} \gamma^{}_{\nu\lambda} \big) \, \varphi \nonumber\\[0.2cm]
            &= \big(\eta_{\nu\lambda} J^{}_{\mu\sigma} - \eta_{\mu\lambda} J^{}_{\nu\sigma} - \eta_{\nu\sigma} J^{}_{\mu\lambda} + \eta_{\mu\sigma} J^{}_{\nu\lambda} \big) \,,
        \end{align}
        where, again, $\eta_{\mu\nu}$ is the mostly `$+$' metric tensor in $\mathbb{R}^{4,2}$. Note that to derive the identity above, we make use of the following formula:
        \begin{align}\label{LubiA}
            [AB,CD] = A[B,C]D + [A,C]BD + CA[B,D] + C[A,D]B\,.
        \end{align}}
    
        \item{\textbf{\textit{Chirality preservation:}} The full set $\{\gamma^{}_{\mu\nu}\}$ preserves chirality, since each $\gamma^{}_{\mu\nu}$ commutes with the ``chirality matrix'' $\omega(4,2)$, that is \linebreak $\big[\gamma^{}_{\mu\nu}, \omega(4,2)\big] = 0$ (see also Eqs. \eqref{chirality matrix} and \eqref{1.33}). Specifically, we have:
        \begin{align}
            \omega(4,2)\big(\gamma^{}_{\mu\nu} \varphi\big) &= \gamma^{}_{\mu\nu} \big(\omega(4,2) \varphi\big) = + \mathrm{i} \big(\gamma^{}_{\mu\nu} \varphi\big) \,,\\[0.2cm]
            \big(\widetilde\varphi\, \gamma^{}_{\mu\nu}\big) \omega(4,2) &= \big(\widetilde\varphi\,\omega(4,2)\big) \gamma^{}_{\mu\nu} = - \mathrm{i} \big(\widetilde\varphi\, \gamma^{}_{\mu\nu} \big) \,,
        \end{align}
        meaning that each chiral subspace $\mathfrak{S}_+ = \mathrm{Span}\{\varphi\}$ or $\mathfrak{S}_- = \mathrm{Span}\{\widetilde{\varphi}\}$ is invariant under the action of $\gamma^{}_{\mu\nu}$. Hence, the restrictions $\gamma^{}_{\mu\nu}\big|_{\mathfrak{S}_\pm}$ yield closed chiral spinor representations of $\mathfrak{so}(4,2)$ on the respective subspaces.}
    
        \item{\textbf{\textit{Faithfulness of the chiral spinor representation:}} The chiral spinor representation of $\mathfrak{so}(4,2)$, obtained by restricting the action of the matrices $\gamma^{}_{\mu\nu}$ to either $\mathfrak{S}_+$ or $\mathfrak{S}_-$ is not only closed under commutation but also faithful. This can be seen from the following observations: 
        \begin{enumerate}[leftmargin=0.8cm]
            \item{Each chiral spinor subspace $\mathfrak{S}_\pm$ is a $4$-dimensional complex vector space. Therefore, the algebra of all linear operators on it, $\mathrm{End}(\mathfrak{S}_\pm)$, is isomorphic to $\mathrm{Mat}(4,\mathbb{C})$ and has complex dimension $16$.}
            
            \item{The $15$ matrices $\gamma^{}_{\mu\nu} = -\gamma^{}_{\nu\mu}$, obtained by restricting to a fixed chiral subspace, are linearly independent and satisfy the $\mathfrak{so}(4,2)$ Lie algebra commutation relations.}
            
            \item{Therefore, since the $15$ restricted matrices are linearly independent and $\mathfrak{so}(4,2)$ has dimension $15$, the kernel of the representation map is trivial. Hence, the representation is injective, establishing a one-to-one correspondence between the elements of $\mathfrak{so}(4,2)$ and their images in $\mathrm{End}(\mathfrak{S}_\pm)$, and thus confirming the faithfulness of the representation.}
        \end{enumerate}
        
        \textbf{\textit{Note}:} This key fact — namely, the faithfulness of the chiral spinor representation on either $\mathfrak{S}_+$ or $\mathfrak{S}_-$ — will be illustrated explicitly in the realization presented in Eqs. \eqref{123456}-\eqref{654321}.
        
        \textbf{\textit{Note}:} If the chiral spinor representation on either $\mathfrak{S}_+$ or $\mathfrak{S}_-$ were not faithful, then distinct elements of $\mathfrak{so}(4,2)$ could act identically on the corresponding subspace. In that case, the image of the representation would be a proper subalgebra of the corresponding operator algebra, isomorphic to the quotient of $\mathfrak{so}(4,2)$ by the kernel of the representation. In other words, part of the Lie algebra would be ``invisible'' to the spinors, and the representation would effectively ``forget'' that portion of the algebraic structure.
        
        \textbf{\textit{Note}:} It is important to note that the Lie algebra $\mathfrak{u}(2, 2) \cong \mathfrak{su}(2,2) \big(\cong\mathfrak{so}(4,2) \big) \oplus \mathfrak{u}(1) \subset \mathfrak{cl}(4,2)$ contains, in addition to the ``second-quantized'' conformal generators $J_{\mu\nu} \in \mathfrak{su}(2,2)\cong\mathfrak{so}(4,2)$, a central $\mathfrak{u}(1)$ generator, which commutes with all elements of the algebra. In the setting involving the ``second-quantized'' ladder operators, we denote this central element by ${\mathcal{C}}_1$. A more detailed discussion of this point will be given later; see Eq. \eqref{C111111} and the subsequent discussion.}
    \end{enumerate}
}\end{remark}

From now on, we focus on a series of the $\mathrm{U}(2,2)$ representations characterized by a lowest weight state, treated as a ``Fock vacuum'' $|0^{\scriptscriptstyle{\mathrm{SO}(4)}}_{}\rangle$ that transforms as a \textit{singlet}\footnote{In physics and mathematics — especially in group theory, quantum mechanics, and representation theory — a singlet means a state or object that transforms trivially under the action of a group.} under the maximal compact subgroup $\mathrm{S}\big(\mathrm{U}(2) \times \mathrm{U}(2)\big) \cong \mathrm{SO}(4)\times\mathrm{U}(1)$ of $\mathrm{SU}(2,2)$, more specifically under the $\mathrm{SO}(4)$ factor, and such that the \textit{conformal Hamiltonian} $H$\footnote{Its significance was recognized early on by Irving Segal (1918-1998), albeit in a rather speculative context \cite{S}.} (i.e., the $\mathrm{U}(1)$ generator of the subgroup) is bounded from below:
\begin{align}\label{2.16}
    H := \mathrm{i} J_{05} = \frac{1}{2} \big( \widetilde{\varphi}\, \mathrm{i}\gamma^{}_{05} \varphi \big) = \frac{1}{2} \big( \widetilde{\varphi}\, \mathrm{i}\gamma^{0}_{} \varphi \big) \geq 1\,, \quad \left(H - 1 \right) |0^{\scriptscriptstyle{\mathrm{SO}(4)}}_{}\rangle = 0\,.
\end{align}
The identification of $|0^{\scriptscriptstyle{\mathrm{SO}(4)}}_{}\rangle$ with a ``Fock vacuum'' presupposes singling out the creation and annihilation:
\begin{align}\label{2.17}
    \mathrm{P}_+ \, \varphi \,|0^{\scriptscriptstyle{\mathrm{SO}(4)}}_{}\rangle = 0 = \widetilde{\varphi} \;\mathrm{P}_- \,|0^{\scriptscriptstyle{\mathrm{SO}(4)}}_{}\rangle\,,
\end{align}
where $\mathrm{P}_\pm = \big(\mathbbm{1} \pm {\texttt{B}}\big)/2$.

\textbf{\textit{Note}:} A second series of $\mathrm{U}(2,2)$ representations arises, corresponding to highest-weight ladder representations, in which the conformal Hamiltonian is negative definite. This situation occurs when the roles of the projectors $\mathrm{P}_+$ and $\mathrm{P}_-$ are interchanged. However, this case will not be considered in the present work.

To make the subsequent discussion more transparent, we shall display a realization of the $4\times 4$ matrices $\gamma^{}_i$ (with $i= {1}, {2}, {3}, {4}$) in accord with the construction \eqref{2.7}. It follows from Eqs. \eqref{1.33}, \eqref{1.36}, and \eqref{S} that:
\begin{align}\label{alberait}
    C\,\overline{\Gamma^{}_\nu} = \Gamma^{}_\nu \, C\,, \quad\mbox{with}\quad C = - \mathrm{i} S = \mathrm{i} m^{}_{0}\, m^{}_{1}\, m^{}_{3}\,,
\end{align}
for $\nu = {0,5,1,2,3,4}$. On the other hand, from Eq. \eqref{Gamma542} and \eqref{Gamma013}, we have, respectively:
\begin{align}
    \Gamma^{}_{\mu^\prime} &= m^{}_{\mu^\prime}\, E \,, \quad\mbox{for}\quad \mu^\prime = {5,4,2} \,, \nonumber\\[0.2cm]
    \Gamma^{}_{\nu^\prime} &= \mathrm{i} m^{}_{\nu^\prime}\, D\,, \quad\mbox{for}\quad \nu^\prime = {0,1,3} \,.
\end{align}
By Proposition \ref{proposition S=DE} and in fact its natural extension, we obtain:
\begin{align}
    \Gamma^{}_{\mu^\prime} \,C = C\, \Gamma^{}_{\mu^\prime} \quad\Longrightarrow\quad\; \overline{\Gamma^{}_{\mu^\prime}} &= \Gamma^{}_{\mu^\prime}\,, \quad\mbox{for}\quad \mu^\prime = {5,4,2} \,, \nonumber\\[0.2cm]
    \Gamma^{}_{\nu^\prime} \,C = - C\, \Gamma^{}_{\nu^\prime} \quad\Longrightarrow\quad\; \overline{\Gamma^{}_{\nu^\prime}} &= -\Gamma^{}_{\nu^\prime}\,, \quad\mbox{for}\quad \nu^\prime = {0,1,3} \,. 
\end{align}
These results imply that the $8 \times 8$ matrices $\Gamma^{}_{{2}}$, $\Gamma^{}_{{4}}$, and $\Gamma^{}_{{5}}$ are real, whereas $\Gamma^{}_{{0}}$, $\Gamma^{}_{{1}}$, and $\Gamma^{}_{{3}}$ are purely imaginary. Correspondingly, we find:
\begin{align}\label{2.18}
    \overline{\gamma^{}_{0}} = -\gamma^{}_{0} \,, \quad\mbox{and}\quad \overline{\gamma^{}_i} = (-1)^i\, \gamma^{}_i\,,
\end{align}
where $i = 1,2,3,4$. 

\begin{Proposition}\label{proposition maghz}
    \textbf{(explicit matrix realization of the $\gamma$- and $\Gamma$- matrices in the twisted quaternionic basis).} With the above in mind, and under the aforementioned convention — namely, the basis ordered as $n^{}_{\widehat{\mathbbm{1}\boldsymbol{k}}_+}, n^{}_{\widehat{05}_+}, n^{}_{\widehat{12}_+}, n^{}_{\widehat{34}_+}$ — we introduce the following $4 \times 4$ matrix representations for the $\gamma$-matrices:
    \begin{align}\label{Bdiagonal}
        \gamma^{}_{{0}} \; \big( = -\gamma^{{0}}_{} \big) := - \gamma^{}_{05} = \mathrm{i} \sigma^{}_3 \otimes \sigma^{}_0 = \mathrm{i} \begin{pmatrix} \sigma^{}_{0} & 0 \\ 0 & -\sigma^{}_{0} \end{pmatrix} = \mathrm{i}{\texttt{B}}\,, 
    \end{align}
    \begin{align}\label{Bdiagonal'}
        \gamma^{}_i :=\gamma^{}_{i5} = \begin{pmatrix} 0 & {\mathcal{E}}^{}_i \\ {\mathcal{E}}^{\ast}_i & 0 \end{pmatrix} \,, \quad i={1,2,3,4}\,,
    \end{align}
    where: 
    \begin{align}\label{E}
        {\mathcal{E}}_{{1}} &= \mathrm{i} \sigma^{}_{1} = \mathrm{i} \begin{pmatrix} 0 & 1 \\ 1 & 0 \end{pmatrix}\,, \nonumber\\[0.2cm]
        {\mathcal{E}}_{{2}} &= \mathrm{i} \sigma^{}_{2} = \mathrm{i} \begin{pmatrix} 0 & - \mathrm{i} \\ \mathrm{i} & 0 \end{pmatrix}\,, \nonumber\\[0.2cm] 
        {\mathcal{E}}_{{3}} &= \mathrm{i} \sigma^{}_{3} = \mathrm{i} \begin{pmatrix} 1 & 0 \\ 0 & -1 \end{pmatrix}\,, \nonumber\\[0.2cm]
        {\mathcal{E}}_{{4}} &= \sigma^{}_{0} = \begin{pmatrix} 1 & 0 \\ 0 & 1 \end{pmatrix} = \mathbbm{1}_2 \,,
    \end{align}
    with ${\mathcal{E}}_i$ representing a set of matrices that mirror the algebraic structure of quaternionic units but with crucial sign differences. Specifically, they obey a twisted quaternionic algebra — sometimes referred to as the ``opposite'' or ``conjugate-quaternionic'' system — in which the products of the imaginary units acquire an overall sign flip compared to the standard quaternion multiplication rules; this reversal, highlighted below using quotation marks ($\mbox{`}{-}\mbox{'}$). This subtle modification leads to the following algebraic relations:
    \begin{align}\label{mahdavi}
        ({\mathcal{E}}^{}_{1})^2 = ({\mathcal{E}}^{}_{2})^2 = ({\mathcal{E}}^{}_{3})^2 = - {\mathcal{E}}^{}_{4} \,,&\nonumber\\[0.2cm]
        {\mathcal{E}}^{}_{1}{\mathcal{E}}^{}_{2} = - {\mathcal{E}}^{}_{2} {\mathcal{E}}^{}_{1} = \mbox{`}{-}\mbox{'} {\mathcal{E}}^{}_{3}\,,& \nonumber\\[0.2cm]
        {\mathcal{E}}^{}_{2}{\mathcal{E}}^{}_{3} = - {\mathcal{E}}^{}_{3} {\mathcal{E}}^{}_{2} = \mbox{`}{-}\mbox{'}{\mathcal{E}}^{}_{1}\,,& \nonumber\\[0.2cm]
        {\mathcal{E}}^{}_{3}{\mathcal{E}}^{}_{1} = - {\mathcal{E}}^{}_{1} {\mathcal{E}}^{}_{3} = \mbox{`}{-}\mbox{'} {\mathcal{E}}^{}_{2}\,,&\nonumber\\[0.2cm]
        {\mathcal{E}}^{}_{1} {\mathcal{E}}^{}_{4} = {\mathcal{E}}^{}_{4} {\mathcal{E}}^{}_{1} = {\mathcal{E}}^{}_{1}\,, &\nonumber\\[0.2cm]
        {\mathcal{E}}^{}_{2} {\mathcal{E}}^{}_{4} = {\mathcal{E}}^{}_{4} {\mathcal{E}}^{}_{2} = {\mathcal{E}}^{}_{2}\,, &\nonumber\\[0.2cm]
        {\mathcal{E}}^{}_{3} {\mathcal{E}}^{}_{4} = {\mathcal{E}}^{}_{4} {\mathcal{E}}^{}_{3} = {\mathcal{E}}^{}_{3}\,, &\nonumber\\[0.2cm]
        {\mathcal{E}}_{1}^\ast = -{\mathcal{E}}^{}_{1}\,, \quad {\mathcal{E}}_{2}^\ast = -{\mathcal{E}}^{}_{2}\,,& \nonumber\\[0.2cm]
        {\mathcal{E}}_{3}^\ast = -{\mathcal{E}}^{}_{3}\,, \quad {\mathcal{E}}_{4}^\ast = {\mathcal{E}}^{}_{4}\,.
    \end{align}
    Correspondingly, the $8\times 8$ $\Gamma$-matrices are set to be:
    \begin{align}\label{poorshirazi}
        \Gamma^{}_0 &= \sigma^{}_1 \otimes \gamma^{0}\,, \nonumber\\[0.2cm]
        \Gamma^{}_{5} &= \epsilon^\ast \otimes \mathbbm{1}_4\,, \quad \epsilon = \mathrm{i}\sigma^{}_2\,, \quad \epsilon^\ast = 
        \begin{pmatrix}
            0 & -1 \\ 1 & 0
        \end{pmatrix} \,, \nonumber\\[0.2cm]
        \Gamma^{}_i &= \sigma^{}_1 \otimes \gamma^{}_i \,, \quad i = {1,2,3,4}\,. 
    \end{align}
\end{Proposition}

Here, and before proceeding any further, it is perhaps worth noting that the above realization of ${\mathcal{E}}_1, \,\dots, \,{\mathcal{E}}_4$ as the ``opposite'' or ``conjugate-quaternionic'' structure is not unique. For example, one may equally identify them with the quaternion basis elements ${\boldsymbol{\imath}},\, -{\boldsymbol{\jmath}},\, {\boldsymbol{k}},\, \mathbbm{1}$, respectively.

\begin{proof}
    {\textbf{(consistency check).} The consistency of the above setup is ensured by the fact that all algebraic relations within each realization — namely, among the $\gamma$ and among the $\Gamma$ matrices — are explicitly satisfied. For example:
    \begin{enumerate}[leftmargin=0.8cm]
        \item{Recalling that the basis used in the matrix representations \eqref{Bdiagonal} and \eqref{Bdiagonal'} is ordered as $n^{}_{\widehat{\mathbbm{1}\boldsymbol{k}}_+}, n^{}_{\widehat{05}_+}, n^{}_{\widehat{12}_+}, n^{}_{\widehat{34}_+}$, one readily verifies that the matrix realization of $\gamma^{}_{05}$ in \eqref{Bdiagonal} precisely coincides with that given in \eqref{2.9+1}.}
        
        \item{The matrix representations given in \eqref{Bdiagonal} and \eqref{Bdiagonal'} for $\gamma^{}_0 := - \gamma^{}_{05}$ and $\gamma^{}_i := \gamma^{}_{i{5}}$ $(i = {1,2,3,4})$ satisfy the commutation relations of the Lie algebra $\mathfrak{su}(2,2) \cong \mathfrak{so}(4,2)$ — more specifically, those of the subalgebra $\mathfrak{so}(4,1) \subset \mathfrak{so}(4,2)$; for instance, we have:
        \begin{align}
            \big[\gamma^{}_{{15}}, \gamma^{}_{{25}}\big] &= 2
            \begin{pmatrix}
                {\mathcal{E}}_3 & 0 \\ 0 & {\mathcal{E}}_3 
            \end{pmatrix} 
            = 2 \gamma^{}_{{12}} \,, \quad (\mbox{see Eq. \eqref{pp12}})\,, \\[0.2cm]
            \big[\gamma^{}_{{35}}, \gamma^{}_{{45}}\big] &= 2
            \begin{pmatrix}
                {\mathcal{E}}_3 & 0 \\ 0 & -{\mathcal{E}}_3
            \end{pmatrix} = 2 \gamma^{}_{{34}} \,, \quad (\mbox{see Eq. \eqref{pp34}})\,.
        \end{align}}

        \item{In view of the above, one can further verify that:
        \begin{align}
            \texttt{w}(4,1) &= \gamma^{}_{{05}}\, \gamma^{}_{{12}}\, \gamma^{}_{{34}} \nonumber\\[0.2cm]
            &= -\mathrm{i} 
            \begin{pmatrix} 
                \sigma^{}_{0} & 0 \\ 0 & -\sigma^{}_{0}
            \end{pmatrix} 
            \begin{pmatrix}
                {\mathcal{E}}_3 & 0 \\ 0 & {\mathcal{E}}_3 
            \end{pmatrix}
            \begin{pmatrix}
                {\mathcal{E}}_3 & 0 \\ 0 & -{\mathcal{E}}_3 
            \end{pmatrix} = \mathrm{i} \mathbbm{1}_4\,,
        \end{align}
        or similarly considering the realization given in \eqref{poorshirazi}: 
        \begin{align}
            \omega(4,2) = \Gamma^{}_{{0}}\, \Gamma^{}_{{5}}\, \Gamma^{}_{{1}}\, \dots\, \Gamma^{}_{{4}} = \mathrm{i} \sigma^{}_3 \otimes \mathbbm{1}_4 = \mathrm{i} D\,, \quad (D = \sigma^{}_3 \otimes \mathbbm{1}_4)\,,
        \end{align}
        which is compatible with \eqref{Hashomi}.}
        
        \item{The $\Gamma$-matrices \eqref{poorshirazi} satisfy the anti-commutation relations of the Clifford algebra $\mathfrak{cl}(4,2)$, namely:
        \begin{align}
            \big\{ \Gamma^{}_\mu, \Gamma^{}_\nu \big\} = 2 \eta_{\mu\nu}\mathbbm{1}\,,
        \end{align}
        where, as before, the indices run over $\mu, \nu = 0,5,1,2,3,4$, and the metric is $\eta_{\mu\nu} = \mathrm{diag}(-1, -1, +1, +1, +1, +1)$.}
        
        \item{Lastly, with the above setup, one verifies that:
        \begin{align}
            \gamma^{}_{i5} &= \Gamma^{}_{i5} D_+ = \Gamma^{}_{i} \Gamma^{}_5 D_+ = \gamma^{}_i \,, \nonumber\\[0.2cm] 
            \gamma^{}_{05} &= \Gamma^{}_{05} D_+ = \Gamma^{}_{0} \Gamma^{}_5 D_+ = \gamma^0 \,.
        \end{align}}
    \end{enumerate}
}\end{proof}

In the ${\texttt{B}}$-diagonal basis, ordered as $n^{}_{\widehat{\mathbbm{1}\boldsymbol{k}}_+}, n^{}_{\widehat{05}_+}, n^{}_{\widehat{12}_+}, n^{}_{\widehat{34}_+}$ (see Eq. \eqref{Bdiagonal}), we implement \eqref{2.17} by introducing two types of creation and annihilation operators $\mathfrak{a}$ (resp. $\mathfrak{a}^\ast$) and $\mathfrak{b}$ (resp. $\mathfrak{b}^\ast$):
\begin{align}
    \varphi = 
    \begin{pmatrix}
        \mathfrak{a} \\ \mathfrak{b}^\ast         
    \end{pmatrix}\,, \quad\mbox{and}\quad \widetilde{\varphi} = \varphi^\ast {\texttt{B}} = 
    \begin{pmatrix}
        \mathfrak{a}^\ast & -\mathfrak{b}         
    \end{pmatrix}\,,
\end{align}
where:\footnote{It is clear that the index $A$ here must be distinguished from the one previously used to label the set of split-octonion units (see Eqa. \eqref{8uni}-\eqref{8uni'*}).}
\begin{align}
    \mathfrak{a} = \big(\mathfrak{a}^A_{} \;;\; A=1,2 \big)\,, \quad\mbox{and}\quad \mathfrak{a}^\ast = \big(\mathfrak{a}^\ast_A \;;\; A=1,2 \big) \,,
\end{align}
\begin{align}
    \mathfrak{b} = \big(\mathfrak{b}^B_{} \;;\; B=1,2 \big) \,, \quad\mbox{and}\quad \mathfrak{b}^\ast = \big(\mathfrak{b}^\ast_B \;;\; B=1,2 \big) \,,
\end{align}
such that the non-vanishing canonical commutation relations, consistent with those in Eq. \eqref{CCR1}, are:
\begin{align}\label{CRoperators}
    \left[\mathfrak{a}^{A^\prime}_{} , \mathfrak{a}^\ast_A \right] = \delta^{A^\prime}_A\,, \quad \left[\mathfrak{b}^{B^\prime}_{} , \mathfrak{b}^\ast_B \right] = \delta^{B^\prime}_B\,.
\end{align}
Moreover, consistent with Eq. \eqref{2.17}, we have:
\begin{align}\label{Vacdef}
    \mathfrak{a}^A |0^{\scriptscriptstyle{\mathrm{SO}(4)}}_{}\rangle = \mathfrak{b}^B |0^{\scriptscriptstyle{\mathrm{SO}(4)}}_{}\rangle = 0 \,, \quad \langle 0^{\scriptscriptstyle{\mathrm{SO}(4)}}_{} |\mathfrak{a}^\ast_A = \langle 0^{\scriptscriptstyle{\mathrm{SO}(4)}}_{} |\mathfrak{b}^\ast_B = 0\,.
\end{align}

On this basis, the conformal Hamiltonian $H$ \eqref{2.16} — namely, the $\mathrm{U}(1)$ generator of the maximal compact subgroup $\mathrm{S}\big(\mathrm{U}(2)\times \mathrm{U}(2)\big) \cong \mathrm{SO}(4)\times \mathrm{U}(1)$ of $\mathrm{SU}(2,2)$ — takes the explicit form:
\begin{align}\label{2.16'}
    H = \mathrm{i} J_{05} &= \frac{1}{2} \big( \widetilde{\varphi}\, \mathrm{i}\gamma^{}_{05}\, \varphi \big)  \nonumber\\[0.2cm]
    &= \frac{1}{2} \big( \widetilde{\varphi}\, \mathrm{i}\gamma^{0}_{}\, \varphi \big)  \nonumber\\[0.2cm]
    &= \frac{1}{2} \big({\mathfrak{a}}^\ast \mathfrak{a} + \mathfrak{b}\, \mathfrak{b^\ast}\big) \nonumber\\[0.2cm]
    &= \frac{1}{2} \big( \mathfrak{a}^\ast_1\, \mathfrak{a}^1 + \mathfrak{a}^\ast_2\, \mathfrak{a}^2 + \mathfrak{b}^1\, \mathfrak{b}^\ast_1 + \mathfrak{b}^2\, \mathfrak{b}^\ast_2 \big)\,,
\end{align}
which explicitly satisfies $\left(H - 1 \right) |0^{\scriptscriptstyle{\mathrm{SO}(4)}}_{}\rangle = 0$.

For an irreducible representation (IR) of $\mathrm{U}(2,2)$, the Lie algebra $\mathfrak{u}(2, 2) \cong \mathfrak{su}(2,2) \big(\cong\mathfrak{so}(4,2) \big) \oplus \mathfrak{u}(1)$ includes, in addition to the conformal generators $J_{\mu\nu} \in \mathfrak{su}(2,2) \cong \mathfrak{so}(4,2)$, a central $\mathfrak{u}(1)$ generator which commutes with all elements of the algebra. We denote this central element by ${\mathcal{C}}_1$, which takes the explicit form:\footnote{At this point, it is important to emphasize a subtle but significant detail:
\begin{align*}
    \big[\mathfrak{b}, \mathfrak{b}^\ast\big] = \big[\mathfrak{b}^1 + \mathfrak{b}^2, \mathfrak{b}^\ast_1 + \mathfrak{b}^\ast_2\big] = \big[\mathfrak{b}^1 , \mathfrak{b}^\ast_1\big] + \big[\mathfrak{b}^2 , \mathfrak{b}^\ast_2\big] = 2\,.
\end{align*}\label{footnote nn2}}  
\begin{align}\label{C111111}
    {\mathcal{C}}_1 = \widetilde{\varphi} \, \varphi = {\mathfrak{a}}^\ast \mathfrak{a} - \mathfrak{b}\, \mathfrak{b^\ast} = {\mathfrak{a}}^\ast \mathfrak{a} - (\mathfrak{b}^\ast \mathfrak{b} + 2) \,.
\end{align}
Following steps analogous to those used in deriving Eq. \eqref{Ghalee}, one readily verifies that: 
\begin{align}
    \big[J_{\mu\nu} , {\mathcal{C}}_1\big] = 0\,,
\end{align}
for all $\mu,\nu = {0,5,1,2,3,4}$. It is important to emphasize that the ``second-quantized'' ladder operator $\mathcal{C}_1$ (which belongs to $\mathfrak{u}(2,2)$ but not to $\mathfrak{su}(2,2)$) serves as the representative of the pseudoscalar (volume form) of $E = M(4,2)$, as defined in Eq. \eqref{U(1)generator}, within this framework.

The normal ordering in \eqref{C111111} ensures that ${\mathcal{C}}_1$ is a Hermitian operator with a well-defined spectrum. Since the bosonic operators $N_{\mathfrak{a}} := \mathfrak{a}^\ast \mathfrak{a}$ and $N_{\mathfrak{b}} := \mathfrak{b}^\ast \mathfrak{b}$, which act on the Fock space, are number operators with non-negative integer eigenvalues, their difference has integer eigenvalues in the Fock space realization. Thus, we can write:
\begin{align}
    {\mathfrak{a}}^\ast \mathfrak{a} - \mathfrak{b}^\ast \mathfrak{b} =: 2\lambda\,, \quad\mbox{with}\quad \lambda \in \frac{1}{2}\mathbb{Z} = \Big\{0, \pm \tfrac{1}{2}, \pm 1, \dots \Big\}\,.
\end{align}
Substituting into Eq. \eqref{C111111}, we obtain:
\begin{align}\label{C1}
    {\mathcal{C}}_1 = \widetilde{\varphi} \, \varphi = {\mathfrak{a}}^\ast \mathfrak{a} - \mathfrak{b}\, \mathfrak{b^\ast} = {\mathfrak{a}}^\ast \mathfrak{a} - \mathfrak{b}^\ast \mathfrak{b} - 2 =: 2 (\lambda - 1) \;\in\;\mathbb{Z}\,.
\end{align}
Notably, this quantized parameter $\lambda$ plays a fundamental role in classifying the UPEIRs of $\mathfrak{u}(2,2)$, and will later be identified as the ``hilicity'' — the conformal analog of helicity — characterizing massless conformal particles in $4$ dimensions \cite{MT}.

\begin{remark}{
    \emph{(single-valuedness of the $\mathrm{U}(2,2)$ representation and integer eigenvalues of ${\mathcal{C}}_1$).} The fact that ${\mathcal{C}}_1$ assumes a fixed integer value in any IR of $\mathfrak{u}(2,2)$ reflects a deeper principle rooted in the representation theory of Lie groups, as outlined in the following steps:
    \begin{enumerate}[leftmargin=0.8cm]
        \item{The generator ${\mathcal{C}}_1$ lies in the center of $\mathfrak{u}(2,2)$, as it commutes with all elements of the algebra.}
        
        \item{By Schur's lemma, in any IR of $\mathfrak{u}(2,2)$, such a central operator must act as a scalar multiple of the identity.}
        
        \item{Therefore, in an IR of $\mathfrak{u}(2,2)$, all states are eigenstates of ${\mathcal{C}}_1$ with a fixed eigenvalue.}
        
        \item{To ensure that this representation lifts to a single-valued representation $U$ of the Lie group $\mathrm{U}(2,2)$ (and not just of its algebra), the exponential of the central generator must be well-defined and single-valued:
        \begin{align}
            U(\theta) = e^{\mathrm{i} \theta {\mathcal{C}}_1}\,, \quad\mbox{with}\quad U(\theta + 2\pi) = U(\theta)\,.
        \end{align}
        Note that the presence of `$\mathrm{i}$' in the exponent, together with the Hermiticity of $\mathcal{C}_1$, ensures that the representation $U(\theta)$ is unitary; $U(\theta) \big(U(\theta)\big)^\ast = \mathbbm{1}$. This periodicity condition requires:
        \begin{align}
            e^{2\mathrm{i}\pi {\mathcal{C}}_1} = \mathbbm{1} \quad \Longrightarrow \quad {\mathcal{C}}_1 = n\,\mathbbm{1}\,,\quad n\in\mathbb{Z}\,.
        \end{align}
        Hence, in any IR of $\mathfrak{u}(2,2)$ that lifts to a single-valued representation of $\mathrm{U}(2,2)$, the eigenvalue of the central generator ${\mathcal{C}}_1$ must be an integer.}
    \end{enumerate}
}\end{remark}

\begin{Proposition}\label{TheormCasimirs}
    \textbf{(polynomial dependence of conformal and dS \linebreak Casimir invariants on the central $\mathfrak{u}(1)$ parameter).} The higher-order Casimir invariants of the ladder-type UPEIRs of $\mathfrak{su}(2,2)$ are polynomial functions of the parameter $\lambda$, which labels the eigenvalue of the central $\mathfrak{u}(1)$ generator $\mathcal{C}_1$ of $\mathfrak{u}(2,2)$. The same property holds for the dS subalgebra $\mathfrak{sp}(2,2)\subset\mathfrak{su}(2,2)$.

    Recall from Remark \ref{Remark J convention} that the essentially self-adjoint generators ${\mathcal{J}}_{\mu\nu}$ employed in Chap. \ref{Chapter. Introduction} and the anti-Hermitian oscillator generators $J_{\mu\nu}$ used in the present chapter are related by:
    \begin{align}
        {\mathcal{J}}_{\mu\nu} = \mathrm{i}J_{\mu\nu}\,.
    \end{align}
    Accordingly,
    \begin{align}
        -\frac{1}{2}{\mathcal{J}}^{\mu\nu}{\mathcal{J}}_{\mu\nu} = \frac{1}{2}J^{\mu\nu}J_{\mu\nu}\,,
    \end{align}
    and similarly for the dS subalgebra.

    In terms of the oscillator generators $J_{\mu\nu}$, one finds:
    \begin{align}\label{C2base}
        J^{\alpha 5}J_{\alpha 5} = -\left(\frac{1}{4}({\mathcal{C}}_1)^2 + {\mathcal{C}}_1\right) = -(\lambda^2-1)\,,
    \end{align}
    together with:
    \begin{align}\label{C2dS}
        {\mathcal{C}}^{\mathfrak{sp}(2,2)}_2 = -\frac{1}{2}{\mathcal{J}}^{\alpha\beta}{\mathcal{J}}_{\alpha\beta} = \frac{1}{2}J^{\alpha\beta}J_{\alpha\beta} = -2(\lambda^2-1)\,,
    \end{align}
    and:
    \begin{align}\label{C2conform}
        {\mathcal{C}}^{\mathfrak{su}(2,2)}_2 = -\frac{1}{2}{\mathcal{J}}^{\mu\nu}{\mathcal{J}}_{\mu\nu} &= \frac{1}{2}J^{\mu\nu}J_{\mu\nu} \nonumber\\[0.2cm]
        &= \underbrace{\frac{1}{2}J^{\alpha\beta}J_{\alpha\beta}}_{={\mathcal{C}}^{\mathfrak{sp}(2,2)}_2} + J^{\alpha 5}J_{\alpha 5} = -3(\lambda^2-1)\,,
    \end{align}
    where $\alpha,\beta=0,1,2,3,4$ and $\mu,\nu=0,5,1,2,3,4$.
\end{Proposition}

This proposition reflects a general feature of IRs, where Casimir operators act as scalars determined by the representation labels. As already pointed out, physically, $\lambda$ plays the role of a helicity-like quantum number, so these Casimir eigenvalues encode invariant properties of massless conformal fields classified by helicity $\lambda$.

\begin{proof}{
    {\textbf{— Step I (proof of Eq. \eqref{C2conform}).} Starting with Eq. \eqref{C2conform}, we observe that the proof rests on the following essential tensor identity:
    \begin{align}\label{sumover}
        (\gamma^{}_{\mu\nu})^a_b\, (\gamma^{\mu\nu})^c_d = {\mathcal{A}}\; \delta^a_d\, \delta^c_b + {\mathcal{B}}\; \delta^a_b\, \delta^c_d \,,
    \end{align}
    where ${\mathcal{A}}$ and ${\mathcal{B}}$ are constants to be determined and $\gamma^{}_{\mu\nu} = \frac{1}{2}\big[\gamma^{}_\mu, \gamma^{}_\nu\big]$. Let us first justify the validity of this ansatz.

    In even-dimensional spaces of dimension $\texttt{d}$ (as in our case, $\texttt{d}=6$), the Dirac spinor representation has dimension $N = 2^{\texttt{d}/2}$, so the space of all $N \times N$ matrices acting on spinors has dimension $N^2 = 2^{\texttt{d}}$. The corresponding Clifford algebra provides a natural basis for this matrix space. Strictly speaking, it consists of the anti-symmetrized products of gamma matrices:
    \begin{align}\label{899921}
        \Big\{ \mathbbm{1},\, \gamma^{}_\mu,\, \gamma^{}_{\mu\nu},\, \gamma^{}_{\mu\nu\rho},\, \dots,\, \gamma^{}_{\mu_1 \dots \mu_\texttt{d}} \Big\}\,,
    \end{align}
    which are linearly independent and span the entire space of operators on the full Dirac spinor space; one observes that the number of independent elements in this set is $\sum_{k=0}^{\texttt{d}} \binom{\texttt{d}}{k} = 2^{\texttt{d}} = N^2$, exactly matching the dimension of the matrix space. Consequently, any matrix $M$ acting on Dirac spinors can be expressed as a linear combination of the corresponding Clifford basis elements:
    \begin{align}\label{M=gamma}
        M = a\,\mathbbm{1} + b^\mu \gamma_\mu + c^{\mu\nu} \gamma_{\mu\nu} + \dots\,,
    \end{align}
    for some coefficients $a, b^\mu, c^{\mu\nu}, \dots \in \mathbb{C}$. 
    
    With these observations in mind, let us now return to the expression:
    \begin{align}
        (\gamma^{}_{\mu\nu})^a_b\, (\gamma^{\mu\nu})^c_d\,,
    \end{align}
    which, by the preceding discussion, must be expressible in terms of the corresponding Clifford basis elements. Since the Lorentz indices $\mu,\nu$ are fully contracted, the resulting object defines a Lorentz-invariant rank-$(2,2)$ spinor tensor. This excludes structures carrying uncontracted Lorentz indices, although contracted combinations of Clifford basis elements, such as $(\gamma_\mu)^a_b\,(\gamma^\mu)^c_d$, may, in general, also define Lorentz-invariant rank-$(2,2)$ spinor tensors. In the chiral realization relevant to our setup,\footnote{Recall that we work with IRs of the conformal algebra realized on the chiral subspace $\mathfrak{S}_+$.} however, odd-degree Clifford elements reverse chirality and therefore do not define endomorphisms of $\mathfrak{S}_+$, while the higher even-degree elements are not independent; in six dimensions, Hodge duality together with the fixed-chirality condition identifies the degree-$4$ elements with the degree-$2$ ones and the degree-$6$ elements with the identity. Consequently, the independent Clifford endomorphisms of the $4$-dimensional chiral space $\mathfrak{S}_+$ are represented by $\mathbbm{1}$ and the $15$ matrices $\gamma_{\mu\nu}$, consistently with
    \begin{align}
        \texttt{d} \left(\operatorname{End}(\mathfrak{S}_+) \right) = 4^2 = 16 = 1 + 15\,.
    \end{align}
    The corresponding completeness (Fierz) relation in this chiral realization therefore reduces the contracted product under consideration to the two independent Kronecker-delta structures:
    \begin{align}
        \delta^a_d\,\delta^c_b\,, \quad \delta^a_b\,\delta^c_d\,.
    \end{align}
    Accordingly, the tensor under consideration can be written as a linear combination of these two structures, yielding the ansatz \eqref{sumover}.

    Moreover, the ansatz \eqref{sumover} is manifestly compatible with the symmetry of the left-hand side under the simultaneous exchange $(a \leftrightarrow c,\ b \leftrightarrow d)$, since
    \begin{align}
        \delta^a_d\, \delta^c_b \;\;\xrightarrow{(a \leftrightarrow c,\, b \leftrightarrow d)}\;\; \delta^c_b \,\delta^a_d = \delta^a_d \, \delta^c_b\,, 
    \end{align}
    \begin{align}
        \delta^a_b\, \delta^c_d \;\;\xrightarrow{(a \leftrightarrow c,\, b \leftrightarrow d)}\;\; \delta^c_d \, \delta^a_b = \delta^a_b \,\delta^c_d \,.
    \end{align}

    With the expansion \eqref{sumover} established, we now aim to determine the constants $\mathcal{A}$ and $\mathcal{B}$. To this end, we make use of standard trace identities of the Clifford algebra in $\texttt{d}$ dimensions (applicable to both odd and even $\texttt{d}$; here $\texttt{d}=6$). The general identity in $\texttt{d}$ dimensions, for spinors of dimension $N$ (here $N=4$\footnote{Since we focus exclusively on the IR of $\mathfrak{so}(4,2)$ realized on the chiral subspace $\mathfrak{S}_+$.}), reads:    
    \begin{align}\label{Golamam}
       \mathrm{Tr}\big(\gamma^{}_{\mu\nu} \gamma^{\mu\nu}\big) = - N \frac{\texttt{d}(\texttt{d} - 1)}{2}\,, \quad\mbox{with}\quad \mu<\nu\,,
    \end{align}
    which, for $\mathfrak{so}(4,2)$, with $\texttt{d}=6$ and $N=4$, gives:
    \begin{align}
        \mathrm{Tr}\big(\gamma^{}_{\mu\nu} \gamma^{\mu\nu}\big) = - 4 \frac{6 \times 5}{2} = - 60 \,, \quad\mbox{with}\quad 0\leq\mu<\nu\leq 5 \,. 
    \end{align}
    
    Keeping this in mind, let us proceed by taking two traces of \eqref{sumover}:
    \begin{enumerate}[leftmargin=0.8cm]
        \item{Contract $a = d$, $b = c$:
            \begin{align}
                - 60 = 16{\mathcal{A}} + 4{\mathcal{B}} \,.
            \end{align}}
        \item{Contract $a = b$, $c = d$:
            \begin{align}
                0 = 4{\mathcal{A}} + 16{\mathcal{B}}\,,
            \end{align}
            since $\mathrm{Tr}\big(\gamma^{}_{\mu\nu}\big) = \mathrm{Tr}\big(\tfrac{1}{2}[\gamma^{}_\mu, \gamma^{}_\nu]\big) = 0$ (see Eqs. \eqref{xy=yx} and \eqref{x+y=x+y}).}
    \end{enumerate}
    Accordingly, we obtain ${\mathcal{A}} = -4$ and $\mathcal{B} = 1$, and hence:
    \begin{align}
        (\gamma^{}_{\mu\nu})^a_b\, (\gamma^{\mu\nu})^c_d = -4\, \delta^a_d\, \delta^c_b + \delta^a_b\, \delta^c_d\,.
    \end{align}

    Equipped with the above expansion, together with Eqs. \eqref{CCR1} and \eqref{fJf6}, we are now in a position to establish Eq. \eqref{C2conform} rigorously:
    \begin{align}\label{C2conform'}        
        &{\mathcal{C}}^{\mathfrak{su}(2,2)}_2  \nonumber\\[0.2cm]
        &= \frac{1}{2} J^{\mu\nu}J_{\mu\nu} = \sum_{{0} \leq\mu<\nu\leq {5}} J^{\mu\nu}J_{\mu\nu} \nonumber\\[0.2cm]       
        &= \frac{1}{4} \big(\widetilde{\varphi}\, \gamma^{\mu\nu} \, \varphi\big)\big(\widetilde{\varphi} \,\gamma^{}_{\mu\nu}\, \varphi\big) \nonumber\\[0.2cm]
        &= \frac{1}{4} \left(\widetilde{\varphi}_a\, (\gamma^{\mu\nu})^a_b \,\varphi^b \right) \left(\widetilde{\varphi}_c \, (\gamma^{}_{\mu\nu})^c_d \, \varphi^d \right) \nonumber\\[0.2cm]
        &= \left(\widetilde{\varphi}_a\, \varphi^b \right) \left(\widetilde{\varphi}_c \, \varphi^d \right) \left(-\delta^a_d \,\delta^c_b + \frac{1}{4} \delta^a_b \,\delta^c_d\right) \nonumber\\[0.2cm]
        &= - \left(\widetilde{\varphi}_a\, \varphi^b \,\widetilde{\varphi}_b \, \varphi^a \right) + \frac{1}{4} \left(\widetilde{\varphi}_a\, \varphi^a \right) \left(\widetilde{\varphi}_c \, \varphi^c \right)\nonumber\\[0.2cm]
        &= - \left(\widetilde{\varphi}_a \, \widetilde{\varphi}_b \,\varphi^b \, \varphi^a\right) - \widetilde{\varphi}_a \underbrace{\left[\varphi^b, \widetilde{\varphi}_b\right]}_{=\, \delta^b_b =\, 4} \varphi^a + \frac{1}{4} \left(\widetilde{\varphi}_a\, \varphi^a \right) \left(\widetilde{\varphi}_c \, \varphi^c \right) \nonumber\\
        &= - \left(\widetilde{\varphi}_b \,\varphi^b \right) \left(\widetilde{\varphi}_a  \varphi^a\right) - \widetilde{\varphi}_b \underbrace{\left[\widetilde{\varphi}_a , \varphi^b\right]}_{=\, -\delta^b_a } \varphi^a - 4 \left(\widetilde{\varphi}_a \, \varphi^a\right) + \frac{1}{4} \left(\widetilde{\varphi}_a\, \varphi^a \right) \left(\widetilde{\varphi}_c \, \varphi^c \right) \nonumber\\  
        &= - \frac{3}{4} \left(\widetilde{\varphi}_b \,\varphi^b \right) \left(\widetilde{\varphi}_a  \varphi^a\right) - 3 \left(\widetilde{\varphi}_a  \varphi^a\right) \nonumber\\[0.2cm]
        &= - 3 \left(\frac{1}{4}({\mathcal{C}}_1)^2 + {\mathcal{C}}_1\right) = - 3 (\lambda^2 - 1)\,.
    \end{align}

    \textbf{— Step II (proof of Eq. \eqref{C2base}).} We next evaluate the complementary contribution $J^{\alpha5}J_{\alpha5}$ directly from the oscillator realization. The completeness relation used in Step I involves all fifteen matrices $\gamma_{\mu\nu}$, which span the space of traceless endomorphisms of the chiral spinor space. By contrast, the five matrices $\gamma_{\alpha5}$ span only a proper subspace, and their contracted tensor product does not, in general, admit the same completeness expansion. In particular, the two tensor structures appearing in Eq. \eqref{sumover} are no longer sufficient to describe the restricted contraction, so taking two traces alone cannot establish the required identity. We therefore use the explicit oscillator realization, which provides a direct evaluation without introducing additional tensor structures.

    From Eq. \eqref{fJf6}, we have:
    \begin{align}\label{fJf6-5}
        J_{\alpha5} = \frac{1}{2} 
        \big(     
        \widetilde{\varphi}\,
        \gamma^{}_{\alpha5}\,
        \varphi
        \big)
        \,, \quad \alpha = 0,1,2,3,4 \,.
    \end{align}
    To evaluate $J^{\alpha5}J_{\alpha5}$, we make use of the explicit oscillator expressions for $J_{05}$ and $J_{i5}$, together with the relevant Pauli-matrix identities, which are established in subsequent sections of this chapter. We invoke these results here to avoid repeating their derivations.

    More precisely, recalling the number operators $N_{\mathfrak{a}}=\mathfrak{a}^{\ast}\mathfrak{a}$ and $N_{\mathfrak{b}}=\mathfrak{b}^{\ast}\mathfrak{b}$, and the convention specified in Footnote \ref{footnote nn2}, we use the definitions \eqref{TTTTTi} and \eqref{C=TAST}, together with Eqs. \eqref{Taaa}, \eqref{Caaa}, and \eqref{Haaa}, to obtain:
    \begin{align}
        J_{05} &= -\frac{\mathrm{i}}{2} 
        \big( 
        N_{\mathfrak{a}}
        +
        N_{\mathfrak{b}}
        +
        2
        \big)
        \,, \quad
        J_{i5} = \frac{1}{2}
        \big(
        \mathfrak{a}^{\ast}\,
        \mathcal{E}_i\,
        \mathfrak{b}^{\ast}
        -
        \mathfrak{b}\,
        \mathcal{E}_i^{\ast}\,
        \mathfrak{a}
        \big)
        \,, \quad
        i=1,2,3,4 \,,
    \end{align}
    where the matrices $\mathcal{E}_i$ are defined in Eq. \eqref{E}.

    To simplify the quadratic products, we use the identities \eqref{T2=0=C2} and \eqref{EE*}, together with the canonical commutation relations \eqref{CRoperators}. In particular, the contributions quadratic in the creation operators, as well as those quadratic in the annihilation operators, vanish upon summation over $i$. The remaining mixed products reduce to expressions involving only the number operators.

    Thus, using $\eta_{00}=\eta_{55}=-1$ and $\eta_{ii}=+1$, we obtain:
    \begin{align}
        J^{\alpha5}J_{\alpha5} 
        &=
        J_{05}^{\,2} - J_{i5}^{\,2} \nonumber\\[0.2cm]
        &=
        -\frac{1}{4}
        \big(
        N_{\mathfrak{a}}
        +
        N_{\mathfrak{b}}
        +
        2
        \big)^2
        -
        \frac{1}{4}
        \big(
        \mathfrak{a}^{\ast}\,
        \mathcal{E}_i\,
        \mathfrak{b}^{\ast}
        -
        \mathfrak{b}\,
        \mathcal{E}_i^{\ast}\,
        \mathfrak{a}
        \big)^2 \nonumber\\[0.2cm]
        &=
        -\frac{1}{4}
        \big(
        N_{\mathfrak{a}}
        +
        N_{\mathfrak{b}}
        +
        2
        \big)^2
        +
        \frac{1}{4}
        \big(
        \mathfrak{a}^{\ast}\,
        \mathcal{E}_i\,
        \mathfrak{b}^{\ast}
        \mathfrak{b}\,
        \mathcal{E}_i^{\ast}\,
        \mathfrak{a}
        +
        \mathfrak{b}\,
        \mathcal{E}_i^{\ast}\,
        \mathfrak{a}
        \mathfrak{a}^{\ast}\,
        \mathcal{E}_i\,
        \mathfrak{b}^{\ast}
        \big) \nonumber\\[0.2cm]
        &=
        -\frac{1}{4}
        \big(
        N_{\mathfrak{a}}
        +
        N_{\mathfrak{b}}
        +
        2
        \big)^2
        +
        \frac{1}{2}
        N_{\mathfrak{a}}N_{\mathfrak{b}}
        +
        \frac{1}{2}
        \big(
        N_{\mathfrak{a}}+2
        \big)
        \big(
        N_{\mathfrak{b}}+2
        \big)
        \nonumber\\[0.2cm]
        &=
        1 - \frac{1}{4}
        \big(
        N_{\mathfrak{a}}
        -
        N_{\mathfrak{b}}
        \big)^2 \,.
    \end{align}
    Finally, using $\mathcal{C}_1 = N_{\mathfrak{a}} - N_{\mathfrak{b}} - 2$ from Eq. \eqref{C111111}, we conclude that: 
    \begin{align}\label{complement-oscillator}
        J^{\alpha5}J_{\alpha5} 
        &=
        -\left(
        \frac{1}{4}
        (\mathcal{C}_1)^2
        +
        \mathcal{C}_1
        \right)
        =
        -(\lambda^2-1) \,,
    \end{align}
    where the last equality holds on the fixed-helicity sector specified by Eq. \eqref{C1}. This establishes Eq. \eqref{C2base}.

    \textbf{— Step III (proof of Eq. \eqref{C2dS}).} The dS generators are the bilinears:
    \begin{align}\label{fJf5}
        J_{\alpha\beta} = \frac{1}{2} \big( \widetilde{\varphi}\, \gamma^{}_{\alpha\beta}\, \varphi \big) \,, \quad \alpha,\beta=0,1,2,3,4 \,.
    \end{align}
    Splitting the independent index pairs in the conformal quadratic Casimir gives:
    \begin{align}
        \mathcal{C}^{\mathfrak{su}(2,2)}_2 &= \frac{1}{2}J^{\alpha\beta}J_{\alpha\beta} + J^{\alpha5}J_{\alpha5} \nonumber\\[0.2cm]
        &= \mathcal{C}^{\mathfrak{sp}(2,2)}_2 + J^{\alpha5}J_{\alpha5} \,.
    \end{align}
    The dS Casimir therefore follows directly from Steps I and II, without requiring a separate completeness relation for the ten matrices $\gamma_{\alpha\beta}$:
    \begin{align}
        \mathcal{C}^{\mathfrak{sp}(2,2)}_2 &= \mathcal{C}^{\mathfrak{su}(2,2)}_2 - J^{\alpha5}J_{\alpha5} \nonumber\\[0.2cm] 
        &= -2\left( \frac{1}{4}(\mathcal{C}_1)^2+\mathcal{C}_1 \right) \nonumber\\[0.2cm]
        &= -2(\lambda^2-1) \,,
    \end{align}
    on the fixed-helicity sector. This establishes Eq. \eqref{C2dS}.}
}\end{proof}

\begin{remark}
    {\emph{(unified treatment of the dS massless representations via \linebreak $\mathfrak{cl}(4,2)$).} It is worth noting that the dS (quadratic) Casimir ${\mathcal{C}}^{\mathfrak{sp}(2,2)}_2 = -2 (\lambda^2 - 1)$ \eqref{C2dS}, derived within a rigorous Clifford-algebraic framework, provides a clear explanation of the observation made by Barut and B\"{o}hm in 1970 \cite{BB} using a conventional group-theoretic analysis of the conformal group. Restricting the ladder-type UPEIRs of the conformal algebra to the dS subalgebra, the Casimir invariant \eqref{C2dS} coincides exactly with the quadratic Casimir eigenvalues of the dS group associated with the so-called massless representations (see Eq. \eqref{Casimir rank 2 massless}). By adopting the conformal Clifford algebra $\mathfrak{cl}(4,2)$ as the primary algebraic framework, this formulation not only reproduces these results but also provides a geometrically enriched, unified treatment of algebraic, spinorial, and conformal structures, facilitating the construction and analysis of massless ladder representations and their restriction to dS symmetry. This picture will be clarified further in the sequel.}
\end{remark}

We now revisit the ``second-quantized'' ladder generators of the $\mathfrak{su}(2,2)$ Lie algebra, defined by $J_{\mu\nu} = \frac{1}{2} (\widetilde{\varphi}\, \gamma^{}_{\mu\nu}\, \varphi)$ \eqref{fJf6}, by providing an explicit and consistent matrix realization of these generators. To this end, we employ Eq. \eqref{i5j5=ij}, along with the matrix representation of the $\gamma$-matrices given in Proposition \ref{proposition maghz}. Our analysis begins with the elements of the (diagonal) Cartan subalgebra:
\begin{align}\label{123456}   
    2\, J_{{12}} = \widetilde{\varphi}\, \gamma^{}_{{12}}\, \varphi = \widetilde{\varphi}\, \left(\gamma^{}_{15} \gamma^{}_{25}\right)\, \varphi &= 
    \begin{pmatrix}
        \mathfrak{a}^\ast & -\mathfrak{b}
    \end{pmatrix}
    \begin{pmatrix}
        {\mathcal{E}^{}_1}{\mathcal{E}^{\ast}_2} & 0 \\
        0 & {\mathcal{E}^{\ast}_1}{\mathcal{E}^{}_2} 
    \end{pmatrix} 
    \begin{pmatrix}
        \mathfrak{a} \\ 
        \mathfrak{b}^{\ast}
    \end{pmatrix} \nonumber\\[0.2cm]
    &= \mathfrak{a}^{\ast} \left( {\mathcal{E}^{}_1}{\mathcal{E}^{\ast}_2} \right) \mathfrak{a} - \mathfrak{b} \left( {\mathcal{E}^{\ast}_1}{\mathcal{E}^{}_2} \right) \mathfrak{b}^\ast\nonumber\\[0.2cm]
    &= \mathrm{i} \big( \underbrace{\mathfrak{a}^{\ast}_1 \mathfrak{a}^1 - \mathfrak{a}^{\ast}_2 \mathfrak{a}^2}_{=:\, H_1} \underbrace{- \mathfrak{b}^1 \mathfrak{b}^{\ast}_1 + \mathfrak{b}^2 \mathfrak{b}^{\ast}_2}_{=:\, H_3} \big)  \nonumber\\[0.2cm]
    &= \mathrm{i} (H_1 + H_3) \,,
\end{align}
\begin{align}
    2\, J_{{34}} = \widetilde{\varphi}\, \gamma^{}_{{34}}\, \varphi = \widetilde{\varphi}\, \left(\gamma^{}_{35} \gamma^{}_{45}\right)\, \varphi &= 
    \begin{pmatrix}
        \mathfrak{a}^\ast & -\mathfrak{b}
    \end{pmatrix}
    \begin{pmatrix}
         {\mathcal{E}^{}_3}{\mathcal{E}^{\ast}_4} & 0 \\
        0 & {\mathcal{E}^{\ast}_3}{\mathcal{E}^{}_4}  
    \end{pmatrix} 
    \begin{pmatrix}
        \mathfrak{a} \\
        \mathfrak{b}^{\ast}
    \end{pmatrix} \nonumber\\[0.2cm]
    &= \mathfrak{a}^{\ast} \left( {\mathcal{E}^{}_3}{\mathcal{E}^{\ast}_4} \right) \mathfrak{a} - \mathfrak{b} \left( {\mathcal{E}^{\ast}_3}{\mathcal{E}^{}_4} \right) \mathfrak{b}^{\ast}  \nonumber\\[0.2cm]
    &= \mathrm{i} \big( \underbrace{\mathfrak{a}^{\ast}_1 \mathfrak{a}^1 - \mathfrak{a}^{\ast}_2 \mathfrak{a}^2}_{=:\, H_1} + \underbrace{\mathfrak{b}^1 \mathfrak{b}^{\ast}_1 - \mathfrak{b}^2 \mathfrak{b}^{\ast}_2}_{=:\, -H_3} \big)  \nonumber\\[0.2cm]
    &= \mathrm{i} (H_1 - H_3) \,,
\end{align}
\begin{align}
    2\, J_{{05}} = \widetilde{\varphi}\, \gamma^{}_{{05}}\, \varphi &= 
    \begin{pmatrix}
        \mathfrak{a}^\ast & -\mathfrak{b} 
    \end{pmatrix}
    \begin{pmatrix}
        -\mathrm{i}\sigma_0 & 0 \\ 
        0 & \mathrm{i}\sigma_0 
    \end{pmatrix} 
    \begin{pmatrix}
        \mathfrak{a} \\
        \mathfrak{b}^{\ast}
    \end{pmatrix} \nonumber\\[0.2cm]
    &= -\mathrm{i} \big(\mathfrak{a}^{\ast} \mathfrak{a} + \mathfrak{b} \mathfrak{b}^{\ast} \big)\nonumber\\[0.2cm]
    &= -\mathrm{i} \big(\mathfrak{a}^{\ast}_1 \mathfrak{a}^1 + \underbrace{\mathfrak{a}^{\ast}_2 \mathfrak{a}^2 + \mathfrak{b}^1 \mathfrak{b}^{\ast}_1}_{=:\, H_2} + \mathfrak{b}^2 \mathfrak{b}^{\ast}_2 \big)  \nonumber\\[0.2cm]
    &= -\mathrm{i} (H_1 + 2H_2 + H_3) \,,
\end{align}
where we recall that:
\begin{align}
    \mathfrak{a} := \begin{pmatrix} \mathfrak{a}^1 \\ \mathfrak{a}^2 \end{pmatrix}\,, &\quad \mathfrak{b}^\ast := \begin{pmatrix} \mathfrak{b}^\ast_1 \\ \mathfrak{b}^\ast_2 \end{pmatrix}\,, \nonumber\\[0.2cm] \mathfrak{a}^\ast := \begin{pmatrix} \mathfrak{a}^\ast_1 & \mathfrak{a}^\ast_2 \end{pmatrix}\,, &\quad\quad \mathfrak{b} := \begin{pmatrix} \mathfrak{b}^1 & \mathfrak{b}^2 \end{pmatrix}\,.
\end{align}
Next, we turn our attention to the remaining elements of the algebra:
\begin{align}
    2\, J_{{01}} = \widetilde{\varphi}\, \gamma^{}_{{01}}\, \varphi = \widetilde{\varphi}\, \left(\gamma^{}_{05} \gamma^{}_{15}\right)\, \varphi &= 
    \begin{pmatrix}
        \mathfrak{a}^\ast & -\mathfrak{b}
    \end{pmatrix}
    \begin{pmatrix}
        0 & -\mathrm{i} {\mathcal{E}^{}_1} \\
        \mathrm{i} {\mathcal{E}^{\ast}_1} & 0 
    \end{pmatrix} 
    \begin{pmatrix}
        \mathfrak{a} \\
        \mathfrak{b}^{\ast}
    \end{pmatrix} \nonumber\\[0.2cm]
    &= -\mathrm{i} \big( \mathfrak{b}\,{\mathcal{E}^{\ast}_1} \mathfrak{a} + \mathfrak{a}^{\ast} {\mathcal{E}^{}_1} \mathfrak{b}^{\ast}\big) \nonumber\\[0.2cm]
    &= -\mathfrak{b}^2 \mathfrak{a}^1 - \mathfrak{b}^1 \mathfrak{a}^2 + \mathfrak{a}^{\ast}_2 \mathfrak{b}^{\ast}_1 + \mathfrak{a}^{\ast}_1 \mathfrak{b}^{\ast}_2 \,,
\end{align}
\begin{align}
    2\, J_{{02}} = \widetilde{\varphi}\, \gamma^{}_{{02}}\, \varphi = \widetilde{\varphi}\, \left(\gamma^{}_{05} \gamma^{}_{25}\right)\, \varphi &= 
    \begin{pmatrix}
        \mathfrak{a}^\ast & -\mathfrak{b}
    \end{pmatrix}
    \begin{pmatrix}
        0 & -\mathrm{i}{\mathcal{E}^{}_2} \\
        \mathrm{i} {\mathcal{E}^{\ast}_2} & 0 
    \end{pmatrix} 
    \begin{pmatrix}
        \mathfrak{a} \\ 
        \mathfrak{b}^{\ast}
    \end{pmatrix} \nonumber\\[0.2cm]
    &= -\mathrm{i} \big( \mathfrak{b}\,{\mathcal{E}^{\ast}_2} \mathfrak{a} + \mathfrak{a}^\ast {\mathcal{E}^{}_2} \mathfrak{b}^\ast\big) \nonumber\\[0.2cm]
    &= -\mathrm{i} \big( \mathfrak{b}^2 \mathfrak{a}^1 - \mathfrak{b}^1 \mathfrak{a}^2 - \mathfrak{a}^{\ast}_2 \mathfrak{b}^{\ast}_1 + \mathfrak{a}^{\ast}_1 \mathfrak{b}^{\ast}_2 \big) \,,
\end{align}
\begin{align}
    2\, J_{{03}} = \widetilde{\varphi}\, \gamma^{}_{{03}}\, \varphi = \widetilde{\varphi}\, \left(\gamma^{}_{05} \gamma^{}_{35}\right)\, \varphi &= 
    \begin{pmatrix}
        \mathfrak{a}^\ast & -\mathfrak{b}
    \end{pmatrix}
    \begin{pmatrix}
        0 & -\mathrm{i} {\mathcal{E}^{}_3} \\
        \mathrm{i} {\mathcal{E}^{\ast}_3} & 0 
    \end{pmatrix}  
    \begin{pmatrix}
        \mathfrak{a} \\
        \mathfrak{b}^{\ast}
    \end{pmatrix} \nonumber\\[0.2cm]
    &= -\mathrm{i} \big( \mathfrak{b}\,{\mathcal{E}^{\ast}_3} \mathfrak{a} + \mathfrak{a}^\ast {\mathcal{E}^{}_3} \mathfrak{b}^\ast\big) \nonumber\\[0.2cm]
    &= -\mathfrak{b}^1 \mathfrak{a}^1 + \mathfrak{b}^2 \mathfrak{a}^2 + \mathfrak{a}^{\ast}_1 \mathfrak{b}^{\ast}_1 - \mathfrak{a}^{\ast}_2 \mathfrak{b}^{\ast}_2 \,,
\end{align}
\begin{align}
    2\, J_{{04}} = \widetilde{\varphi}\, \gamma^{}_{{04}}\, \varphi = \widetilde{\varphi}\, \left(\gamma^{}_{05} \gamma^{}_{45}\right)\, \varphi &= 
    \begin{pmatrix}
        \mathfrak{a}^\ast & -\mathfrak{b}
    \end{pmatrix}
    \begin{pmatrix}
        0 & -\mathrm{i} {\mathcal{E}^{}_4} \\
        \mathrm{i} {\mathcal{E}^{\ast}_4} & 0 
    \end{pmatrix}  
    \begin{pmatrix}
        \mathfrak{a} \\
        \mathfrak{b}^{\ast}
    \end{pmatrix} \nonumber\\[0.2cm]
    &= -\mathrm{i} \big( \mathfrak{b}\, {\mathcal{E}^{\ast}_4} \mathfrak{a} + \mathfrak{a}^\ast {\mathcal{E}^{}_4}\mathfrak{b}^\ast\big) \nonumber\\[0.2cm]
    &= -\mathrm{i} \big( \mathfrak{b}^1 \mathfrak{a}^1 + \mathfrak{b}^2 \mathfrak{a}^2 + \mathfrak{a}^{\ast}_1 \mathfrak{b}^{\ast}_1 + \mathfrak{a}^{\ast}_2 \mathfrak{b}^{\ast}_2 \big) \,,
\end{align}
\begin{align}
    2\, J_{{13}} = \widetilde{\varphi}\, \gamma^{}_{{13}}\, \varphi = \widetilde{\varphi}\, \left(\gamma^{}_{15} \gamma^{}_{35}\right)\, \varphi &= 
    \begin{pmatrix}
        \mathfrak{a}^\ast & -\mathfrak{b}
    \end{pmatrix}
    \begin{pmatrix}
        {\mathcal{E}^{}_1}{\mathcal{E}^{\ast}_3} & 0 \\
        0 & {\mathcal{E}^{\ast}_1}{\mathcal{E}^{}_3} 
    \end{pmatrix} 
    \begin{pmatrix}
        \mathfrak{a} \\
        \mathfrak{b}^{\ast}
    \end{pmatrix} \nonumber\\[0.2cm]
    &= \mathfrak{a}^\ast \left( {\mathcal{E}^{}_1}{\mathcal{E}^{\ast}_3} \right) \mathfrak{a} - \mathfrak{b} \left( {\mathcal{E}^{\ast}_1}{\mathcal{E}^{}_3} \right) \mathfrak{b}^\ast \nonumber\\[0.2cm]
    &= \mathfrak{a}^{\ast}_2 \mathfrak{a}^1 - \mathfrak{a}^{\ast}_1 \mathfrak{a}^2 - \mathfrak{b}^2 \mathfrak{b}^{\ast}_1 + \mathfrak{b}^1 \mathfrak{b}^{\ast}_2 \,,
\end{align}
\begin{align}
    2\, J_{{14}} = \widetilde{\varphi}\, \gamma^{}_{{14}}\, \varphi = \widetilde{\varphi}\, \left(\gamma^{}_{15} \gamma^{}_{45}\right)\, \varphi &= 
    \begin{pmatrix}
        \mathfrak{a}^\ast & -\mathfrak{b}
    \end{pmatrix}
    \begin{pmatrix}
        {\mathcal{E}^{}_1}{\mathcal{E}^{\ast}_4} & 0 \\
        0 & {\mathcal{E}^{\ast}_1}{\mathcal{E}^{}_4} 
    \end{pmatrix} 
    \begin{pmatrix}
        \mathfrak{a} \\
        \mathfrak{b}^{\ast}
    \end{pmatrix} \nonumber\\[0.2cm]
    &= \mathfrak{a}^\ast \left( {\mathcal{E}^{}_1}{\mathcal{E}^{\ast}_4} \right) \mathfrak{a} - \mathfrak{b} \left( {\mathcal{E}^{\ast}_1}{\mathcal{E}^{}_4} \right) \mathfrak{b}^\ast  \nonumber\\[0.2cm]
    &= \mathrm{i} \big(\mathfrak{a}^{\ast}_2 \mathfrak{a}^1 + \mathfrak{a}^{\ast}_1 \mathfrak{a}^2 + \mathfrak{b}^2 \mathfrak{b}^{\ast}_1 + \mathfrak{b}^1 \mathfrak{b}^{\ast}_2 \big) \,,
\end{align}
\begin{align}
    2\, J_{{15}} = \widetilde{\varphi}\, \gamma^{}_{{15}}\, \varphi &= 
    \begin{pmatrix}
        \mathfrak{a}^\ast & -\mathfrak{b}
    \end{pmatrix}
    \begin{pmatrix}
        0 & {\mathcal{E}^{}_1} \\
        {\mathcal{E}^{\ast}_1} & 0 
    \end{pmatrix} 
    \begin{pmatrix}
        \mathfrak{a} \\
        \mathfrak{b}^{\ast}
    \end{pmatrix} \nonumber\\[0.2cm]
    &= -\mathfrak{b}\,{\mathcal{E}^{\ast}_1} \mathfrak{a} + \mathfrak{a}^\ast{\mathcal{E}^{}_1}\mathfrak{b}^\ast \nonumber\\[0.2cm]
    &= \mathrm{i} \big(\mathfrak{b}^2 \mathfrak{a}^1 + \mathfrak{b}^1 \mathfrak{a}^2 + \mathfrak{a}^\ast_2 \mathfrak{b}^{\ast}_1 + \mathfrak{a}^\ast_1 \mathfrak{b}^{\ast}_2 \big) \,,
\end{align}
\begin{align}
    2\, J_{{23}} = \widetilde{\varphi}\, \gamma^{}_{{23}}\, \varphi = \widetilde{\varphi}\, \left(\gamma^{}_{25} \gamma^{}_{35}\right)\, \varphi &= 
    \begin{pmatrix}
        \mathfrak{a}^\ast & -\mathfrak{b}
    \end{pmatrix}
    \begin{pmatrix}
        {\mathcal{E}^{}_2}{\mathcal{E}^{\ast}_3} & 0 \\
        0 & {\mathcal{E}^{\ast}_2}{\mathcal{E}^{}_3} 
    \end{pmatrix} 
    \begin{pmatrix}
        \mathfrak{a} \\
        \mathfrak{b}^{\ast}
    \end{pmatrix} \nonumber\\[0.2cm]
    &= \mathfrak{a}^\ast \left( {\mathcal{E}^{}_2}{\mathcal{E}^{\ast}_3} \right) \mathfrak{a} - \mathfrak{b} \left( {\mathcal{E}^{\ast}_2}{\mathcal{E}^{}_3} \right) \mathfrak{b}^\ast \nonumber\\[0.2cm]
    &= \mathrm{i} \big(\mathfrak{a}^\ast_2 \mathfrak{a}^1 + \mathfrak{a}^\ast_1 \mathfrak{a}^2 - \mathfrak{b}^2 \mathfrak{b}^{\ast}_1 - \mathfrak{b}^1 \mathfrak{b}^{\ast}_2 \big) \,,
\end{align}
\begin{align}
    2\, J_{{24}} = \widetilde{\varphi}\, \gamma^{}_{{24}}\, \varphi = \widetilde{\varphi}\, \left(\gamma^{}_{25} \gamma^{}_{45}\right)\, \varphi&= 
    \begin{pmatrix}
        \mathfrak{a}^\ast & -\mathfrak{b}
    \end{pmatrix}
    \begin{pmatrix}
        {\mathcal{E}^{}_2}{\mathcal{E}^{\ast}_4} & 0 \\
        0 & {\mathcal{E}^{\ast}_2}{\mathcal{E}^{}_4} 
    \end{pmatrix} 
    \begin{pmatrix}
        \mathfrak{a} \\
        \mathfrak{b}^{\ast}
    \end{pmatrix} \nonumber\\[0.2cm]
    &= \mathfrak{a}^\ast \left( {\mathcal{E}^{}_2}{\mathcal{E}^{\ast}_4} \right) \mathfrak{a} - \mathfrak{b} \left( {\mathcal{E}^{\ast}_2}{\mathcal{E}^{}_4} \right) \mathfrak{b}^\ast  \nonumber\\[0.2cm]
    &= -\mathfrak{a}^\ast_2 \mathfrak{a}^1 + \mathfrak{a}^\ast_1 \mathfrak{a}^2 - \mathfrak{b}^2 \mathfrak{b}^{\ast}_1 + \mathfrak{b}^1 \mathfrak{b}^{\ast}_2  \,,
\end{align}
\begin{align}
    2\, J_{{25}} = \widetilde{\varphi}\, \gamma^{}_{{25}}\, \varphi &= 
    \begin{pmatrix}
        \mathfrak{a}^\ast & -\mathfrak{b}
    \end{pmatrix}
    \begin{pmatrix}
        0 & {\mathcal{E}^{}_2} \\
        {\mathcal{E}^{\ast}_2} & 0 
    \end{pmatrix} 
    \begin{pmatrix}
        \mathfrak{a} \\
        \mathfrak{b}^{\ast}
    \end{pmatrix} \nonumber\\[0.2cm]
    &= -\mathfrak{b}\,{\mathcal{E}^{\ast}_2} \mathfrak{a} + \mathfrak{a}^\ast{\mathcal{E}^{}_2}\mathfrak{b}^\ast \nonumber\\[0.2cm]
    &= -\mathfrak{b}^2 \mathfrak{a}^1 + \mathfrak{b}^1 \mathfrak{a}^2 - \mathfrak{a}^\ast_2 \mathfrak{b}^{\ast}_1 + \mathfrak{a}^\ast_1 \mathfrak{b}^{\ast}_2 \,,
\end{align}
\begin{align}
    2\, J_{{35}} = \widetilde{\varphi}\, \gamma^{}_{{35}}\, \varphi &= 
    \begin{pmatrix}
        \mathfrak{a}^\ast & -\mathfrak{b}
    \end{pmatrix}
    \begin{pmatrix}
        0 & {\mathcal{E}^{}_3} \\
        {\mathcal{E}^{\ast}_3} & 0 
    \end{pmatrix} 
    \begin{pmatrix}
        \mathfrak{a} \\
        \mathfrak{b}^{\ast}
    \end{pmatrix} \nonumber\\[0.2cm]
    &= -\mathfrak{b}\,{\mathcal{E}^{\ast}_3} \mathfrak{a} + \mathfrak{a}^\ast{\mathcal{E}^{}_3}\mathfrak{b}^\ast \nonumber\\[0.2cm]
    &= \mathrm{i} \big(\mathfrak{b}^1 \mathfrak{a}^1 - \mathfrak{b}^2 \mathfrak{a}^2 + \mathfrak{a}^\ast_1 \mathfrak{b}^{\ast}_1 - \mathfrak{a}^\ast_2 \mathfrak{b}^{\ast}_2 \big)\,,
\end{align}
\begin{align}\label{654321}
    2\, J_{{45}} = \widetilde{\varphi}\, \gamma^{}_{{45}}\, \varphi &= 
    \begin{pmatrix}
        \mathfrak{a}^\ast & -\mathfrak{b}
    \end{pmatrix}
    \begin{pmatrix}
        0 & {\mathcal{E}^{}_4} \\
        {\mathcal{E}^{\ast}_4} & 0 
    \end{pmatrix} 
    \begin{pmatrix}
        \mathfrak{a} \\
        \mathfrak{b}^{\ast}
    \end{pmatrix} \nonumber\\[0.2cm]
    &= - \mathfrak{b}\,{\mathcal{E}^{\ast}_4} \mathfrak{a} + \mathfrak{a}^\ast{\mathcal{E}^{}_4}\mathfrak{b}^\ast  \nonumber\\[0.2cm]
    &= - \mathfrak{b}^1 \mathfrak{a}^1 - \mathfrak{b}^2 \mathfrak{a}^2 + \mathfrak{a}^\ast_1 \mathfrak{b}^{\ast}_1 + \mathfrak{a}^\ast_2 \mathfrak{b}^{\ast}_2 \,.
\end{align}
The above results render the anti-Hermitian nature of $J_{\mu\nu}$ manifest: 
\begin{align}\label{anti-Hermitian J}
    J_{\mu\nu} = -(J_{\mu\nu})^\ast\,.
\end{align}

To define the unitary ladder representations of $\mathfrak{u}(2,2) \cong \mathfrak{su}(2,2)\oplus\mathfrak{u}(1)$, we introduce the standars Chevalley-Cartan basis of $\mathfrak{su}(2,2)$:
\begin{align}
    E_a = \widetilde{\varphi}_a\, \varphi^{a+1} \;;\quad E_1 = \mathfrak{a}^\ast_1 \mathfrak{a}^2\,, \quad E_2 = \mathfrak{a}^\ast_2 \mathfrak{b}^{\ast}_1\,, \quad E_3 = -\mathfrak{b}^1 \mathfrak{b}^{\ast}_2\,,
\end{align}
\begin{align}
    F_a = \widetilde{\varphi}_{a+1}\, \varphi^{a} \;;\quad F_1 = \mathfrak{a}^\ast_2 \mathfrak{a}^1\,, \quad F_2 = -\mathfrak{b}^1 \mathfrak{a}^2\,, \quad F_3 = -\mathfrak{b}^2 \mathfrak{b}^{\ast}_1\,,
\end{align}
\begin{align}\label{cartoon}
    \underbrace{H_a = [E_a, F_a] = \widetilde{\varphi}_{a}\, \varphi^{a} - \widetilde{\varphi}_{a+1}\, \varphi^{a+1}}_{\mbox{by Eq. \eqref{LubiA}}} \;;\quad H_1 &= \mathfrak{a}^\ast_1 \mathfrak{a}^1 - \mathfrak{a}^\ast_2 \mathfrak{a}^2\,, \nonumber\\
    H_2 &= \mathfrak{a}^\ast_2 \mathfrak{a}^2 + \mathfrak{b}^1 \mathfrak{b}^{\ast}_1\,, \nonumber\\[0.2cm]
    H_3 &= \mathfrak{b}^2 \mathfrak{b}^{\ast}_2 - \mathfrak{b}^1 \mathfrak{b}^{\ast}_1\,,
\end{align}
where $E_a$, $F_a$, and $H_a$ denote the raising, lowering (see Remark \ref{Remark Raising and Lowering}), and (diagonal) Cartan operators of $\mathfrak{su}(2,2)$, respectively. We actually observe that:
\begin{align}
    [H_a , H_b] = 0\,,
\end{align}
\begin{align}
    [E_a, F_b] = \delta_{ab}\, H_a\,,
\end{align}
\begin{align}
    [H_a , E_b] = + A_{ab}\, E_b \;\;;&\quad [H_1 , E_1] = +2\, E_1\,,\nonumber\\[0.2cm]
    & \quad [H_1 , E_2] = -\, E_2\,, \nonumber\\[0.2cm]
    & \quad [H_1 , E_3] = 0\,, \nonumber\\[0.2cm]
    & \quad [H_2 , E_1] = -\, E_1\,, \nonumber\\[0.2cm]
    & \quad [H_2 , E_2] = +2\, E_2\,, \nonumber\\[0.2cm]
    & \quad [H_2 , E_3] = -\, E_3\,, \nonumber\\[0.2cm]
    & \quad [H_3 , E_1] = 0\,, \nonumber\\[0.2cm]
    & \quad [H_3 , E_2] = -\, E_2\,, \nonumber\\[0.2cm]
    & \quad [H_3 , E_3] = +2\, E_3\,,
\end{align}
\begin{align}
    [H_a , F_b] = - A_{ab}\, F_b \;\;;&\quad [H_1 , F_1] = -2\, F_1\,,\nonumber\\[0.2cm]
    & \quad [H_1 , F_2] = +\, F_2\,, \nonumber\\[0.2cm]
    & \quad [H_1 , F_3] = 0\,, \nonumber\\[0.2cm]
    & \quad [H_2 , F_1] = +\, F_1\,, \nonumber\\[0.2cm]
    & \quad [H_2 , F_2] = -2\, F_2\,, \nonumber\\[0.2cm]
    & \quad [H_2 , F_3] = +\, F_3\,, \nonumber\\[0.2cm]
    & \quad [H_3 , F_1] = 0\,, \nonumber\\[0.2cm]
    & \quad [H_3 , F_2] = +\, F_2\,, \nonumber\\[0.2cm]
    & \quad [H_3 , F_3] = -2\, F_3\,.
\end{align}
Note that $A_{ab}$ is called the Cartan matrix of the algebra:
\begin{align}\label{Cartan matrix}
    A_{ab} = 
    \begin{pmatrix}
        +2 & -1 & 0 \\
        -1 & +2 & -1 \\
        0 & -1 & +2 
    \end{pmatrix}\,.
\end{align}

\begin{remark}
    {\emph{(non-inclusion of the central generator ${\mathcal{C}}_1$ in $\mathfrak{su}(2,2)$).} In the Chevalley-Cartan basis of $\mathfrak{su}(2,2)$ introduced above, it is manifest that the central $\mathfrak{u}(1)$ generator $\mathcal{C}_1 = \mathfrak{a}^\ast \mathfrak{a} - \mathfrak{b}\,\mathfrak{b}^\ast$ \eqref{C111111} of the full $\mathfrak{u}(2,2)$ algebra lies outside the $\mathfrak{su}(2,2)$ subalgebra, since it cannot be expressed as a linear combination of the Chevalley-Cartan generators of $\mathfrak{su}(2,2)$.}
\end{remark}

\begin{remark}\label{Remark maximal compact subalgebra}{
    \emph{(maximal compact subalgebra of $\mathfrak{su}(2,2)$ and its oscillator realization).} The maximal compact subalgebra $\mathfrak{s}\big(\mathfrak{u}(2) \oplus \mathfrak{u}(2)\big) \cong \mathfrak{su}(2) \oplus \mathfrak{su}(2) \oplus \mathfrak{u}(1)$ of $\mathfrak{su}(2,2)$ (see Appendix \ref{Appendix Maximal}) can be explicitly realized by the following set of generators:
    \begin{enumerate}[leftmargin=0.8cm]
        \item{\textit{\textbf{$\mathfrak{su}(2)_L$ subalgebra:}} This acts on the $\mathfrak{a}$-oscillators and is generated by (the Schwinger realization \cite{Schwinger}):
        \begin{align}
            \begin{cases}
                J_+^{(L)} := E_1 = \mathfrak{a}^\ast_1 \mathfrak{a}^2\,, \\[0.2cm]
                J_-^{(L)} := F_1 = \mathfrak{a}^\ast_2 \mathfrak{a}^1\,, \\[0.2cm]
                J_3^{(L)} := \frac{1}{2} H_1 = \frac{1}{2} \left(\mathfrak{a}^\ast_1 \mathfrak{a}^1 - \mathfrak{a}^\ast_2 \mathfrak{a}^2 \right)\,.
            \end{cases}
        \end{align}
        These satisfy the standard $\mathfrak{su}(2)$ commutation relations:
        \begin{align}
            \big[J_3^{(L)}, J_\pm^{(L)}\big] = \pm J_\pm^{(L)}\,, \quad \big[J_+^{(L)}, J_-^{(L)}\big] = 2 J_3^{(L)}\,.
        \end{align}
        The corresponding quadratic Casimir form is:
        \begin{align}\label{CSU2-L}
            {\mathcal{C}}_2^{\mathfrak{su}(2)_L} &= {J^{(L)}_3}^2 + \frac{J^{(L)}_+J^{(L)}_- + J^{(L)}_-J^{(L)}_+ }{2} \nonumber\\[0.2cm]
            &= \frac{1}{4} \big(\mathfrak{a}^\ast_1 \mathfrak{a}^1 + \mathfrak{a}^\ast_2 \mathfrak{a}^2\big) \big(\mathfrak{a}^\ast_1 \mathfrak{a}^1 + \mathfrak{a}^\ast_2 \mathfrak{a}^2 + 2\big) \nonumber\\[0.2cm]
            &=: \texttt{j}_L (\texttt{j}_L + 1)\,,
        \end{align}
        where we have defined the number operator $N_\mathfrak{a} = \mathfrak{a}^\ast \mathfrak{a} = \mathfrak{a}^\ast_1 \mathfrak{a}^1 + \mathfrak{a}^\ast_2 \mathfrak{a}^2 = 2\texttt{j}_L$, with $\texttt{j}_L\in\tfrac{1}{2}\mathbb{Z}_{\geq 0}$.}
        
        \item{\textit{\textbf{$\mathfrak{su}(2)_R$ subalgebra:}} This acts on the $\mathfrak{b}$-oscillators and is generated by:
        \begin{align}
            \begin{cases}
                J_+^{(R)} := E_3 = - \mathfrak{b}^1 \mathfrak{b}^\ast_2\,, \\[0.2cm]
                J_-^{(R)} := F_3 = - \mathfrak{b}^2 \mathfrak{b}^\ast_1\,, \\[0.2cm]
                J_3^{(R)} := \frac{1}{2}H_3 = -\frac{1}{2} \left(\mathfrak{b}^1 \mathfrak{b}^\ast_1 - \mathfrak{b}^2 \mathfrak{b}^\ast_2 \right)\,.
            \end{cases}
        \end{align}
        These also satisfy the canonical $\mathfrak{su}(2)$ commutation relations:
        \begin{align}
            \big[J_3^{(R)}, J_\pm^{(R)}\big] = \pm J_\pm^{(R)}\,, \quad \big[J_+^{(R)}, J_-^{(R)}\big] = 2J_3^{(R)}\,.
        \end{align}
        The corresponding quadratic Casimir form is:
        \begin{align}\label{CSU2-R}
            {\mathcal{C}}_2^{\mathfrak{su}(2)_R} &= \frac{1}{4} \big(\mathfrak{b}^\ast_1 \mathfrak{b}^1 + \mathfrak{b}^\ast_2 \mathfrak{b}^2\big) \big(\mathfrak{b}^\ast_1 \mathfrak{b}^1 + \mathfrak{b}^\ast_2 \mathfrak{b}^2 + 2\big)\nonumber\\[0.2cm]
            &=: \texttt{j}_R (\texttt{j}_R + 1)\,,
        \end{align}
        where we have defined the number operator $N_\mathfrak{b} = \mathfrak{b}^\ast \mathfrak{b} = \mathfrak{b}^\ast_1 \mathfrak{b}^1 + \mathfrak{b}^\ast_2 \mathfrak{b}^2 = 2\texttt{j}_R$, with $\texttt{j}_R\in\tfrac{1}{2}\mathbb{Z}_{\geq 0}$.}
        
        \item{\textit{\textbf{$\mathfrak{u}(1)$ center:}} The $\mathfrak{u}(1)$ center of the maximal compact subalgebra is generated by the conformal Hamiltonian $H$ \eqref{2.16'}:
        \begin{align}\label{2.16''}
            H := \frac{1}{2} \big(H_1 + 2H_2 + H_3 \big) = \frac{1}{2} \big(\mathfrak{a}^\ast_1 \mathfrak{a}^1 + \mathfrak{a}^\ast_2 \mathfrak{a}^2 + \mathfrak{b}^1 \mathfrak{b}^\ast_1 + \mathfrak{b}^2 \mathfrak{b}^\ast_2 \big)\,.
        \end{align}
        This central element commutes with all the $\mathfrak{su}(2)_L$ and $\mathfrak{su}(2)_R$ generators:
        \begin{align}\label{Ghibli}
            \big[H, J_\pm^{(L)}\big] = \big[H, J_3^{(L)}\big] = \big[H, J_\pm^{(R)}\big] = \big[H, J_3^{(R)}\big] = 0\,.
        \end{align}}
    \end{enumerate}
}\end{remark}

\begin{remark}\label{Remark Raising and Lowering}{
    \emph{(interpretation of $E_a$ and $F_a$ as raising and lowering operators).} Before proceeding, let us clarify in what sense the operators $E_a$ and $F_a$ are called raising and lowering operators for representations of $\mathfrak{su}(2,2)$. Let $|\text{W}\rangle$ be a weight vector defined by:
    \begin{align}
        H_a |\text{W}\rangle = \text{W}_a\, |\text{W}\rangle\,,
    \end{align}
    with weight components $\text{W} = (\text{W}_1, \text{W}_2, \text{W}_3)$ corresponding to the Cartan generators $H_1, H_2$, and $H_3$; since the Cartan elements commute, they are simultaneously diagonalizable, and thus admit a common eigenbasis. Consider the action of $F_b$ on $|\text{W}\rangle$. Using the commutation relation $[H_a, F_b] = -A_{ab} F_b$, we find:
    \begin{align}
        H_a \big(F_b |\text{W}\rangle\big)
        &= \big([H_a, F_b] + F_b H_a\big) |\text{W}\rangle \nonumber\\[0.2cm]
        &= \big(-A_{ab} F_b + \text{W}_a F_b\big) |\text{W}\rangle \nonumber\\[0.2cm]
        &= \big(\text{W}_a - A_{ab}\big) F_b |\text{W}\rangle\,.
    \end{align}
    This shows that $F_b |\text{W}\rangle$ is again a weight vector, whose weight component is shifted as:
    \begin{align}
        \text{W}_a \quad\longmapsto\quad \text{W}_a - A_{ab}\,.
    \end{align}
    Equivalently, the total weight vector transforms as:
    \begin{align}
        F_b |\text{W}\rangle \quad\longmapsto\quad \bigl|\text{W} - \alpha_b\bigr\rangle\,,
    \end{align}
    where the simple root $\alpha_b$ is given by the $b$-th column of the Cartan matrix $A_{ab}$; explicitly, $\alpha_1 = (2, -1, 0)^\top$, $\alpha_2 = (-1, 2, -1)^\top$, and $\alpha_3 = (0, -1, 2)^\top$.

    Similarly, the operator $E_b$ acts as a raising operator, shifting the weight vector upward by the corresponding simple root:
    \begin{align}
        E_b |\text{W}\rangle \;\propto\; \bigl|\text{W} + \alpha_b\bigr\rangle\,,
    \end{align}
    where $\alpha_b$ is the simple root associated with $E_b$.
}\end{remark}

Given the above, we define the lowest weight state $|\mathrm{LW}_\lambda\rangle$ of the UPEIR of $\mathfrak{su}(2,2)$ (with a given helicity $\lambda$) as:
\begin{align}
    \label{LW ><0} |\mathrm{LW}_\lambda\rangle &\; \Big( =: |\mathrm{LW}_{\lambda\geq 0}\rangle \Big) = (\mathfrak{a}^\ast_2)^{2\lambda}\, |0^{\scriptscriptstyle{\mathrm{SO}(4)}}_{}\rangle \,, \quad \mbox{for} \quad \lambda\geq 0\,, \\[0.1cm]
    \label{LW ><0''} |\mathrm{LW}_\lambda\rangle &\; \Big( =: |\mathrm{LW}_{\lambda< 0}\rangle \Big) = (\mathfrak{b}^\ast_1)^{-2\lambda}\, |0^{\scriptscriptstyle{\mathrm{SO}(4)}}_{}\rangle\,, \quad \mbox{for} \quad \lambda < 0\,.
\end{align}
This state is annihilated by all lowering operators $F_a$ (see Eqs. \eqref{CRoperators} and \eqref{Vacdef}):
\begin{align}
    F_a\, |\mathrm{LW}_\lambda\rangle = 0\,, \quad \mbox{for}\quad a=1,2,3\,.
\end{align}
Moreover, $|\mathrm{LW}_\lambda\rangle$ is an eigenstate of the conformal Hamiltonian $H$ \eqref{2.16''} (see also \eqref{2.16'}) with eigenvalue:
\begin{align}\label{E100}
    \mathscr{E}_\circ := 1 + |\lambda|\,, \quad\mbox{with}\quad \lambda \in \frac{1}{2}\mathbb{Z}\,,
\end{align}
so that:
\begin{align}\label{2.166}
    (H-\mathscr{E}_\circ)|\mathrm{LW}_\lambda\rangle = 0 \,.
\end{align}
This eigenvalue — the ``lowest conformal energy'' (or conformal dimension) — serves to label the lowest weight state $|\mathrm{LW}_\lambda\rangle$ of the representation.\footnote{In the special case of a free massless scalar field (helicity $\lambda = 0$), the lowest weight state is simply denoted by $|\mathrm{LW}_{\lambda=0}\rangle =: |0^{\scriptscriptstyle{\mathrm{SO}(4)}}_{}\rangle$, and the corresponding eigenvalue equation \eqref{2.166} reduces to \eqref{2.16}, indicating that the lowest conformal energy (or dimension) is $\mathscr{E}_\circ = 1$.\label{footnoteCdimension}}

\begin{remark}\label{Remark Helicit+-}
    {\emph{(transformation properties of the lowest weight state under $\mathfrak{su}(2)_L \oplus \mathfrak{su}(2)_R$).} One verifies that the lowest weight state $|\mathrm{LW}_\lambda\rangle$ (see Eqs. \eqref{LW ><0} and \eqref{LW ><0''}) transforms in the spin-$s=\lambda$ representation of $\mathfrak{su}(2)_L$ and is a singlet under $\mathfrak{su}(2)_R$ for positive helicity $\lambda \geq 0$, whereas for negative helicity $\lambda < 0$ it transforms in the spin-$s=|\lambda|$ representation of $\mathfrak{su}(2)_R$ and is a singlet under $\mathfrak{su}(2)_L$; specifically:
    \begin{align}
        {\mathcal{C}}_2^{\mathfrak{su}(2)_L} |\mathrm{LW}_{\lambda\geq 0}\rangle = s(s+1)|\mathrm{LW}_{\lambda\geq 0}\rangle \,, \quad\mbox{and}\quad {\mathcal{C}}_2^{\mathfrak{su}(2)_R} |\mathrm{LW}_{\lambda\geq 0}\rangle &= 0\,,\nonumber\\[0.2cm]
        J_+^{(R)} |\mathrm{LW}_{\lambda\geq 0}\rangle &= 0\,,\nonumber\\[0.2cm]
        J_-^{(R)} |\mathrm{LW}_{\lambda\geq 0}\rangle &= 0\,,\nonumber\\[0.2cm]
        J_3^{(R)} |\mathrm{LW}_{\lambda\geq 0}\rangle &= 0\,,
    \end{align}
    while:
    \begin{align}
        {\mathcal{C}}_2^{\mathfrak{su}(2)_L} |\mathrm{LW}_{\lambda< 0}\rangle &= 0\,, \quad\mbox{and}\quad {\mathcal{C}}_2^{\mathfrak{su}(2)_R} |\mathrm{LW}_{\lambda< 0}\rangle = s(s+1) |\mathrm{LW}_{\lambda< 0}\rangle \,, \nonumber\\[0.2cm]
        J_+^{(L)} |\mathrm{LW}_{\lambda\geq 0}\rangle &= 0\,,\nonumber\\[0.2cm]
        J_-^{(L)} |\mathrm{LW}_{\lambda\geq 0}\rangle &= 0\,,\nonumber\\[0.2cm]
        J_3^{(L)} |\mathrm{LW}_{\lambda\geq 0}\rangle &= 0\,.
    \end{align}}
\end{remark}

The entire one-particle Hilbert space $\mathscr{H}^{(1)}_{\lambda}$ carrying the UPEIR of $\mathfrak{su}(2,2)$ is then generated by the successive action of the raising operators $E_a$ on this lowest weight state $|\mathrm{LW}_\lambda\rangle$, uniquely characterized (up to the helicity sign) by the lowest conformal energy $\mathscr{E}_\circ$. This iterative construction naturally endows $\mathscr{H}^{(1)}_{\lambda}$ with a hierarchical structure, since the action of the raising operators $E_a$ increases the eigenvalue of the conformal Hamiltonian $H$ in discrete steps (see Remark \ref{Remark Raising and Lowering}). Consequently, $\mathscr{H}^{(1)}_{\lambda}$ is foliated into the \textit{finite-dimensional}\footnote{See Eq. \eqref{d(Delta)}.} eigenspaces $\mathcal{H}^{(1)}_\mathscr{E} := \ker(H - \mathscr{E}) \,\big(\subset\mathscr{H}^{(1)}_{\lambda}\big)$ of the conformal Hamiltonian $H$, each satisfying:
\begin{align}\label{EDelta}
    (H - \mathscr{E}) |W\rangle = 0 \,, \quad \text{for all} \quad |W\rangle \in \mathcal{H}^{(1)}_\mathscr{E} \,,
\end{align}
where the spectrum of $H$ in this representation is given by:
\begin{align}\label{Delta}
    \mathscr{E} = \Big\{ |\lambda| + n \;;\; n = 1, 2, 3, \dots \Big\}\,,
\end{align}
with $\mathscr{E}_\circ = 1 + |\lambda|$ as the lowest value. This construction furnishes a transparent realization of the UPEIR in terms of conformal energy levels and, importantly, shows that every UPEIR of $\mathfrak{su}(2,2)$ — in particular the one-particle Hilbert space $\mathscr{H}^{(1)}_{\lambda}$ of a free massless conformal field of definite helicity $\lambda \in \tfrac{1}{2}\mathbb{Z}$ — is completely characterized by the lowest eigenvalue $\mathscr{E}_\circ$ of the conformal Hamiltonian (up to the sign of the helicity parameter $\lambda$), just as its lowest weight state $|\mathrm{LW}_\lambda\rangle$ is. This point is discussed further in Remark \ref{Remark parity start}.

\begin{remark}
    {\emph{(the energy eigenspaces as $\mathfrak{so}(4)$ modules).} It is important to emphasize that each eigenspace $\mathcal{H}^{(1)}_\mathscr{E} \,\big(\subset\mathscr{H}^{(1)}_{\lambda}\big)$ is not invariant under the full conformal Lie algebra $\mathfrak{su}(2,2)$; rather, it is invariant under the subalgebra $\mathfrak{so}(4)$. This is because the conformal Hamiltonian $H$ generates the $\mathfrak{u}(1)$ center of the maximal compact subalgebra $\mathfrak{s}\big(\mathfrak{u}(2) \oplus \mathfrak{u}(2)\big) \cong \mathfrak{so}(4) \oplus \mathfrak{u}(1)$, and therefore commutes with the $\mathfrak{so}(4)$ generators but not with the entire algebra $\mathfrak{su}(2,2)$. Consequently, the action of $\mathfrak{so}(4)$ preserves each eigenspace $\mathcal{H}^{(1)}_\mathscr{E}$ of $H$; strictly speaking, each $\mathcal{H}^{(1)}_\mathscr{E}$ carries a (generally reducible) finite-dimensional representation of $\mathfrak{so}(4)$. Proposition \ref{Proposition 3.3}, and in particular the third step of its proof, provides a more detailed analysis of this structure.}
\end{remark}

Accordingly, a discrete, compact basis can be constructed for the Hilbert space $\mathscr{H}^{(1)}_{\lambda}$ of massless one-particle states with helicity $\lambda$, consisting of simultaneous eigenstates of the conformal Hamiltonian $H$ \eqref{2.16''} and of the Casimir operator ${\mathcal{C}}_2^{\mathfrak{so}(4)}$ associated with the $\mathfrak{so}(4)$ Lie subalgebra. The physical, positive-definite Casimir invariant of $\mathfrak{so}(4)$ $\left(\subset \mathfrak{so}(4,1) \subset \mathfrak{so}(4,2)\right)$ is found to be:
\begin{align}\label{2.16666}
    {\mathcal{C}}_2^{\mathfrak{so}(4)} &= {\mathcal{C}}_2^{\mathfrak{su}(2)_L} + {\mathcal{C}}_2^{\mathfrak{su}(2)_R} \nonumber\\[0.2cm]
    &= \frac{1}{4}\mathfrak{a}^\ast \mathfrak{a} (\mathfrak{a}^\ast \mathfrak{a} + 2) + \frac{1}{4}\mathfrak{b}^\ast \mathfrak{b} (\mathfrak{b}^\ast \mathfrak{b} + 2) \nonumber\\[0.2cm]
    &= \frac{1}{8} \Big( \underbrace{(\mathfrak{a}^\ast \mathfrak{a} + \mathfrak{b} \mathfrak{b}^\ast)^2}_{=\,4H^2} + \underbrace{(\mathfrak{a}^\ast \mathfrak{a} - \mathfrak{b} \mathfrak{b}^\ast) (\mathfrak{a}^\ast \mathfrak{a} - \mathfrak{b} \mathfrak{b}^\ast + 4)}_{=\, \mathcal{C}_1(\mathcal{C}_1+4)} \Big) \nonumber\\[0.1cm]
    &= \frac{1}{2} \big( \mathscr{E}^2 + \lambda^2 - 1 \big) \,,
\end{align}
hence:
\begin{align}\label{-1}
    {\mathcal{C}}_2^{\mathfrak{so}(4)} &= \texttt{j}_L (\texttt{j}_L + 1) + \texttt{j}_R (\texttt{j}_R + 1)  \nonumber\\[0.2cm]
    &= \frac{1}{2} \big( \mathscr{E}^2 + \lambda^2 - 1 \big) \geq |\lambda| \left(|\lambda|+1\right) \,.
\end{align}
From the above relation, it follows that $\texttt{j}_L, \texttt{j}_R \,\big(\in\tfrac{1}{2}\mathbb{Z}_{\geq 0}\big)$ satisfy the bound (see Appendix \ref{Appen L,R E}): 
\begin{align}\label{-1'}
    \texttt{j}_L, \texttt{j}_R \leq \mathscr{E}-1 \,.
\end{align}
Note that the dimension of the eigenspace $\mathcal{H}^{(1)}_\mathscr{E} \subset {\mathscr{H}}^{(1)}_\lambda$ of $H$, which carries a representation of $\mathfrak{so}(4)$ and corresponds to the eigenvalue $\mathscr{E}$, is given by:
\begin{align}\label{d(Delta)}
    \texttt{d} \big(\mathcal{H}^{(1)}_\mathscr{E}\big) = (2\texttt{j}_L+1)(2\texttt{j}_R+1) \,.
\end{align}

\begin{remark}\label{Remark parity start} 
    {\emph{(classification of the $\mathfrak{su}(2,2)$ UPEIRs).} Taking all of the above into account, the states within each $\mathfrak{su}(2,2)$ UPEIR space $\mathscr{H}^{(1)}_{\lambda}$ are labeled by three (half-)integers: $\mathscr{E}$, which labels the $\mathfrak{so}(4)$-invariant eigenspaces $\mathcal{H}^{(1)}_\mathscr{E}$ of the conformal Hamiltonian $H$ that foliate the full Hilbert space $\mathscr{H}^{(1)}_{\lambda}$; and $\texttt{j}_L$ and $\texttt{j}_R$, which label the states within each such eigenspace $\mathcal{H}^{(1)}_\mathscr{E}$, corresponding to a finite-dimensional (generally reducible) representation of $\mathfrak{so}(4)$:
    \begin{align}
        \mathscr{H}^{(1)}_{\lambda} = \operatorname{Span}\,\Big\{ |\mathscr{E}, \texttt{j}_L, \texttt{j}_R \rangle \;;\quad &\mathscr{E} = |\lambda|+1, |\lambda|+2, \,\dots\,, \nonumber\\
        &\texttt{j}_L, \texttt{j}_R \in\tfrac{1}{2}\mathbb{Z}_{\geq 0} \;\leq\; \mathscr{E}-1\Big\} \,.
    \end{align}    
    In this context, the $\mathfrak{su}(2,2)$ UPEIR space $\mathscr{H}^{(1)}_{\lambda}$ is fully generated by repeated action of the raising operators $E_a$ on the lowest weight state:
    \begin{align}
        |\mathrm{LW}_\lambda\rangle = |\mathscr{E}_\circ, \texttt{j}_L, \texttt{j}_R \rangle\,, 
    \end{align}
    which has lowest conformal energy $\mathscr{E}_\circ = 1 + |\lambda|$, and is accordingly denoted by:
    \begin{align}
        \mathscr{H}^{(1)}_{\lambda} =: \mathscr{C}^>_{\mathscr{E}_\circ, \texttt{j}_L, \texttt{j}_R} \,.
    \end{align} 
    Note that the superscript `$^>$' indicates that the representations $\mathscr{C}^>_{\mathscr{E}_\circ, \texttt{j}_L, \texttt{j}_R}$ are characterized by the positivity of the conformal energy, meaning that the spectrum of $H$ is bounded from below (see Eq. \eqref{2.16}).

    Although this point has already been noted, it is appropriate to scrutinize, within the above labelling, that the lowest weight state $|\mathscr{E}_\circ, \texttt{j}_L, \texttt{j}_R \rangle$, which furnishes the full Hilbert space $\mathscr{H}^{(1)}_{\lambda}$, is itself uniquely specified by its conformal energy $\mathscr{E}_\circ$, up to the sign of the helicity parameter $\lambda$. In particular, for a given $\mathscr{E}_\circ$ (and hence a fixed $|\lambda|$), there exist two distinct configurations of $(\texttt{j}_L, \texttt{j}_R)$, each corresponding to one of the two possible helicity signs, which give rise to two inequivalent realizations of the corresponding UPEIR, distinguished as positive- and negative-helicity representations.
    
    To render this dependence explicit, we impose relation \eqref{-1} on the lowest weight state $|\mathscr{E}_\circ, \texttt{j}_L, \texttt{j}_R \rangle$, yielding:
    \begin{align}\label{-100}
        \texttt{j}_L (\texttt{j}_L + 1) + \texttt{j}_R (\texttt{j}_R + 1) = |\lambda| (|\lambda|+1) \,.
    \end{align}
    For a given $\mathscr{E}_\circ$ (and hence a fixed $|\lambda|$), this equation — consistent with Remark \ref{Remark Helicit+-} — admits only two possible choices for the lowest weight state $|\mathscr{E}_\circ,\texttt{j}_L, \texttt{j}_R\rangle$:
    \begin{align}\label{second choice}
        |\mathscr{E}_\circ, \texttt{j}_L , \texttt{j}_R \rangle = |\mathscr{E}_\circ, \mathscr{E}_\circ-1, 0 \rangle = \big||\lambda|+1, |\lambda|, 0\big\rangle \,,
    \end{align}
    \begin{align}\label{first choice}
         |\mathscr{E}_\circ, \texttt{j}_L , \texttt{j}_R \rangle = |\mathscr{E}_\circ, 0, \mathscr{E}_\circ-1\rangle = \big||\lambda|+1, 0, |\lambda|\big\rangle \,,
    \end{align}
    where:
    \begin{enumerate}[leftmargin=0.8cm]
        \item{The first choice \eqref{second choice} specifies the positive helicity UPEIR of $\mathfrak{su}(2,2)$ and is denoted by:
        \begin{align}
            \mathscr{C}^>_{|\lambda|+1, |\lambda|, 0} =: \mathscr{C}^>_{s+1, s, 0}\,, \quad\mbox{with}\quad |\lambda|=:s\,.
        \end{align}
        The positivity of the helicity follows from the fact that the lowest weight state $|s+1,s,0\rangle$ in this representation transforms in the spin-$s$ representation of $\mathfrak{su}(2)_L$ and is a singlet under $\mathfrak{su}(2)_R$ (see Remark \ref{Remark Helicit+-}).}

        \item{The second choice \eqref{first choice}, however, specifies the $\mathfrak{su}(2,2)$ UPEIR with negative helicity, in which the lowest weight state is annihilated by the same lowering operators, but with the roles of $\text{SU}(2)_L$ and $\text{SU}(2)_R$ exchanged. This representation is usually denoted as:
        \begin{align}
            \mathscr{C}^>_{|\lambda|+1, 0, |\lambda|} =: \mathscr{C}^>_{s+1, 0, s}\,, \quad\mbox{with}\quad |\lambda|=:s\,.
        \end{align}
        The negative helicity arises from the fact that the lowest weight state $|s+1, 0, s\rangle$ in this representation transforms in the spin-$s$ representation of $\mathfrak{su}(2)_R$ and is a singlet under $\mathfrak{su}(2)_L$ (see, again, Remark \ref{Remark Helicit+-}).}
    \end{enumerate}
    Thus, for each fixed lowest conformal energy $\mathscr{E}_\circ$, and hence for fixed $|\lambda| = s$, there is a unique positive-helicity $\mathfrak{su}(2,2)$ UPEIR $\mathscr{C}^>_{s+1, s, 0}$ and a corresponding negative-helicity UPEIR $\mathscr{C}^>_{s+1, 0, s}$, related by exchanging the $\mathfrak{su}(2)_L$ and $\mathfrak{su}(2)_R$ roles in the lowest weight state.
    
    With this labelling, in the zero-helicity ($\lambda=0$) case, the lowest weight state $|\mathrm{LW}_{\lambda=0}\rangle =: |0^{\scriptscriptstyle{\mathrm{SO}(4)}}_{}\rangle$ is identified with $|1,0,0\rangle$. It is perhaps worthwhile noting that the zero-helicity space $\mathscr{H}^{(1)}_{\lambda=0}$ — corresponding to the representation at the very bottom of the ladder representations, with $\texttt{j}_L = \texttt{j}_R = 0$ and $\lambda = 0$ — coincides with the state space of the non-relativistic, regularized hydrogen atom. This connection was first uncovered in the 1960s through the study of spectrum-generating algebras. 
    
    \textbf{\textit{Note}:} For a detailed account of how the conformal massless (ladder) UPEIR restricts to the dS case, see Appendix \ref{appendix: dS UIRs}.}
\end{remark}

\begin{remark}\label{Remark E=JL+JR}
    {\emph{(integer spectrum of $\mathscr{E} - \big(\texttt{j}_L + \texttt{j}_R\big)$).} From the construction above, one immediately observes a key identity that will prove useful in the sequel. Specifically, for any helicity parameter $\lambda$ in the lowest weight state, one has:
    \begin{align}
        \mathscr{E}_\circ - \big(\texttt{j}_L + \texttt{j}_R\big) = 1 \,.
    \end{align}
    This relation naturally extends to all descendant states as:
    \begin{align}
        \mathscr{E} - \big(\texttt{j}_L + \texttt{j}_R\big) \in \mathbb{Z}_{>0} \,.
    \end{align}
    Indeed, all descendant states are obtained by repeated action of the raising operators $E_a$ on the lowest weight state, thereby shifting the corresponding weights in discrete steps. Moreover, such actions leave the integer or half-integer nature of $\texttt{j}_L + \texttt{j}_R$ invariant. This latter observation can be viewed as a by-product of the discussion in Appendix \ref{appendix: dS UIRs}, in particular Eq. \eqref{begaee}.}
\end{remark}

\begin{remark}
    {\emph{(fixed-helicity $\mathfrak{u}(2,2)$ UPEIRs and their direct sum).} For each fixed helicity $\lambda\in\tfrac{1}{2}\mathbb{Z}$, the space $\mathscr{H}^{(1)}_{\lambda} =: \mathscr{C}^{>}_{\mathscr{E}_{\circ},\texttt{j}_{L},\texttt{j}_{R}}$ carries a UPEIR of $\mathfrak{su}(2,2)$. In the present oscillator realization, this representation extends to the full algebra $\mathfrak{u}(2,2)=\mathfrak{su}(2,2)\oplus\mathfrak{u}(1)$ by letting the central generator $C_1$ act as the scalar $2(\lambda-1)$. The extension remains irreducible, since any subspace invariant under $\mathfrak{u}(2,2)$ is also invariant under $\mathfrak{su}(2,2)$. Moreover, the conformal Hamiltonian and its positive-energy spectrum remain unchanged.

    Collecting all fixed-helicity sectors gives the direct sum:
    \begin{align}
        \mathscr{H}^{(1)}
        =
        \bigoplus_{\lambda\in\tfrac{1}{2}\mathbb{Z}}
        \mathscr{H}^{(1)}_{\lambda} \,.
    \end{align}
    Since $C_1$ commutes with all conformal generators, its eigenspaces $\mathscr{H}^{(1)}_{\lambda}$ are separately invariant under $\mathfrak{u}(2,2)$. The central generator therefore distinguishes the helicity sectors without mixing them. Accordingly, $\mathscr{H}^{(1)}$ carries a reducible positive-energy representation of $\mathfrak{u}(2,2)$, whose irreducible summands are the fixed-helicity spaces $\mathscr{H}^{(1)}_{\lambda}$.

    States in this direct-sum $1$-particle space are labelled by the helicity $\lambda$, which specifies the irreducible summand, together with the energy and spin quantum numbers of the preceding fixed-helicity construction:
    \begin{align}
        \mathscr{H}^{(1)}
        =
        \operatorname{Span}\,\Big\{
        |\lambda,\mathscr{E},\texttt{j}_{L},\texttt{j}_{R}\rangle
        \;;\quad
        &\lambda\in\tfrac{1}{2}\mathbb{Z},
        \nonumber\\[0.2cm]
        &\mathscr{E}=|\lambda|+1,|\lambda|+2,\dots,
        \nonumber\\[0.2cm]
        &\texttt{j}_{L},\texttt{j}_{R}
        \in\tfrac{1}{2}\mathbb{Z}_{\geq0},
        \quad
        \texttt{j}_{L},\texttt{j}_{R}\leq\mathscr{E}-1
        \Big\} \,.
    \end{align}
    Here, only the admissible spin pairs occurring in the corresponding fixed-helicity sector are included.}
\end{remark}

\subsection{An Alternative Realization of $\mathfrak{su}(2,2)$ in the Conformal Compact Basis}

The momentum generator $p^{}_{\dot{\mu}}$ ($\dot{\mu} = 0,1,2,3$) of the Poincar\'{e} subalgebra of $\mathfrak{u}(2,2)$ is verified to have zero square, thus corresponding to massless particles:
\begin{align}
    p^{}_{\dot{\mu}} := \frac{1}{2}\widetilde{\varphi}\, \gamma^{}_{\dot{\mu}} \left(\mathbbm{1} + \gamma^{}_{{4}}\right) \varphi \quad\Longrightarrow\quad p^{\dot{\mu}} p^{}_{\dot{\mu}} = ({\vec{p}})^2 - (p_0)^2 = 0\,,
\end{align}
where $\vec{p} := (p^{}_1 , p^{}_2 , p^{}_3)$ and $(\vec{p})^2 := (p^{}_1)^2 + (p^{}_2)^2 + (p^{}_3)^2$. This is proven in Ref. \cite{MT}, using a realization of $\gamma$ matrices with $\gamma^{}_4$ — rather than $\texttt{B}$ — diagonal. Here, we instead adopt a complex-variable formulation of $4$-dimensional conformal space, providing a unified treatment of both dS and Minkowski cases.

\begin{Proposition}\label{proposition T2=0}
    \textbf{(null structure of the complex translation and special conformal operators).} Let us introduce the complex translation operators $T_i$ and the special conformal generators $C_i$, both naturally suited to the ladder formalism and characterized by a vanishing four-square:
    \begin{align}\label{TTTTTi}
        T_i = \widetilde{\varphi}\, \mathrm{P}_+ \gamma^{}_i \mathrm{P}_- \varphi 
        = \mathfrak{a}^\ast\, {\mathcal{E}}_i\, \mathfrak{b}^\ast \,,
    \end{align}
    \begin{align}\label{C=TAST}
        C_i \; \big(= T_i^\ast \big) = \mathfrak{b} \,{\mathcal{E}}^\ast_i \,\mathfrak{a}\,,
    \end{align}
    where $i=1,2,3,4$ and the operators ${\mathcal{E}}_i$ are defined in Eq. \eqref{E}. They satisfy:
    \begin{align}\label{T2=0=C2}
        T^i T_i \; \left( = \sum_{i=1}^4 T^i T_i \right) = 0 = C^i C_i\,.
    \end{align}
\end{Proposition}

Note that this proposition will be complemented by Proposition \ref{Proposition 3.3}, which specifically justifies the naming of the complex translation operators $T_i$ and the special conformal generators $C_i$.

\begin{proof}{
    {\textbf{— Step I (the role of $\epsilon = \mathrm{i}\sigma_2$ as a spinor metric).} In $2$-component spinor formalism, the anti-symmetric tensor $\epsilon_{AB}$ ($\epsilon_{AB} = - \epsilon_{BA}$, with $A,B = 1,2$) — given explicitly by:
    \begin{align}\label{WEBziad}
        &\epsilon^{}_{AB} = \mathrm{i}\sigma_2 = 
        \begin{pmatrix}
        0 & 1 \\
        -1 & 0
        \end{pmatrix}\,,  \nonumber\\[0.2cm]
        &\epsilon^{AB} = (\epsilon^{}_{AB})^\ast = (\mathrm{i}\sigma_2)^\ast = 
        \begin{pmatrix}
        0 & -1 \\
        1 & 0
        \end{pmatrix}\,,  \nonumber\\[0.2cm]
        &\epsilon^{AC} \epsilon^{}_{CB} = \delta^A_B\,,
    \end{align}
    — plays the distinguished role of a spinor metric, enabling one to raise and lower spinor indices:
    \begin{align}
        \psi_A = \epsilon_{AB}\, \psi^B\,, \quad \psi^A = \epsilon^{AB} \,\psi_B\,.
    \end{align}
    Thus, if $\psi^A = \begin{pmatrix} \psi^1 \\ \psi^2 \end{pmatrix}$, then:
    \begin{align}
        \psi_A = \epsilon^{}_{AB} \psi^B = 
        \begin{pmatrix}
        0 & 1 \\ 
        -1 & 0
        \end{pmatrix} 
        \begin{pmatrix} 
        \psi^1 \\ 
        \psi^2 
        \end{pmatrix} = 
        \begin{pmatrix} \psi^2 \\ -\psi^1 \end{pmatrix}\,.
    \end{align}
    This structure is essential both algebraically and geometrically, for the following reasons:
    \begin{enumerate}[leftmargin=0.8cm]
        \item{\textbf{\textit{Anti-symmetry and spinor bilinears}:} The anti-symmetry of $\epsilon_{AB}$ guarantees that, for commuting spinor components, spinor contractions like $\psi^A \chi_A = \psi^A \,\epsilon_{AB}\, \chi^B$ are anti-symmetric under exchange, $\psi^A \chi_A = -\chi^A \psi_A$. For Grassmann-odd spinor fields, by contrast, the additional minus sign arising from the exchange of the spinor components compensates the anti-symmetry of $\epsilon_{AB}$, so that $\psi^A\chi_A=\chi^A\psi_A$. This distinction is essential when applying the spinor contraction to Grassmann-valued fields in quantum field theory.}
        
        \item{\textbf{\textit{Invariant tensor under $\mathrm{SL}(2,\mathbb{C})$}:} The matrix $\epsilon_{AB}$ is, up to a non-zero scalar multiple, the unique non-degenerate, anti-symmetric rank-2 tensor invariant under the action of $\mathrm{SL}(2,\mathbb{C})$ — the double cover of the proper, orthochronous Lorentz group. It satisfies the identity:
        \begin{align}
            (S \epsilon S^\intercal)_{AB} = \epsilon_{AB}\,, \quad \mbox{for all} \quad S \in \mathrm{SL}(2,\mathbb{C})\,,
        \end{align}
        thereby ensuring that spinor index raising and lowering, as well as the corresponding invariant spinor contractions, are preserved under transformations of the spinor representation.}
        
        \item{\textbf{\textit{Why $\sigma_1$ and $\sigma_3$ do not play this role}:} The other Pauli matrices, namely $\sigma_1 = \begin{pmatrix} 0 & 1 \\ 1 & 0 \end{pmatrix}$ and $\sigma_3 = \begin{pmatrix} 1 & 0 \\ 0 & -1 \end{pmatrix}$, are symmetric and therefore cannot serve as the canonical anti-symmetric spinor metric. For commuting spinor components, for instance:
        \begin{align}
            \psi^A (\sigma_1)_{AB} \chi^B = \chi^A (\sigma_1)_{AB} \psi^B\,,
        \end{align}
        reflecting the symmetry of $\sigma_1$. More fundamentally, neither $\sigma_1$ nor $\sigma_3$ is invariant under the full $\mathrm{SL}(2,\mathbb{C})$ action in the manner of $\epsilon_{AB}$, and hence neither defines the invariant anti-symmetric bilinear form required for spinor index raising and lowering.}
    \end{enumerate}}
    
    {\textbf{— Step II ($(\mathcal{E}_i)^A_B$ and $(\mathcal{E}_i)^{AB}$).} Recall that the Pauli matrices and $\sigma_0$ read as:
    \begin{align}\label{ramini}
        (\sigma_1)^A_B = 
        \begin{pmatrix}
        0 & 1 \\
        1 & 0
        \end{pmatrix}\,, &\quad 
        (\sigma_2)^A_B = 
        \begin{pmatrix}
        0 & -\mathrm{i} \\
        \mathrm{i} & 0
        \end{pmatrix}\,,  \nonumber\\[0.2cm]
        (\sigma_3)^A_B = 
        \begin{pmatrix}
        1 & 0 \\
        0 & -1
        \end{pmatrix}\,, &\quad 
        (\sigma_0)^A_B = 
        \begin{pmatrix}
        1 & 0 \\
        0 & 1
        \end{pmatrix}\,.
    \end{align}
    Then, the associated mixed-index tensors (``conjugate'' quaternions) $(\mathcal{E}_i)^A_B$ are:
    \begin{align}\label{ramin}
        (\mathcal{E}_i)^A_B := 
        \begin{cases}
        \mathrm{i} (\sigma_i)^A_B\,, &\quad\mbox{for}\quad i = 1,2,3\,, \\[2ex]
        (\sigma_0)^A_B\,, &\quad\mbox{for}\quad i = 4\,.
        \end{cases}
    \end{align}
    Given the above, the bispinor tensors $(\mathcal{E}_i)^{AB}$ are defined by:
    \begin{align}
        (\mathcal{E}_i)^{AB} = \epsilon^{AC} (\mathcal{E}_i)^B_C\,.
    \end{align}}
    
    {\textbf{— Step III (four essential expansions).} As a by-product of the proof of Proposition \ref{TheormCasimirs}, we put forward the following essential identity:
    \begin{align}\label{moosay1}
        ({\mathcal{E}}_i)^A_B \, ({\mathcal{E}}^{i\ast})^C_D = {\mathcal{A}}\, \delta^A_B\, \delta^C_D + {\mathcal{B}}\, \delta^A_D\, \delta^C_B\,,
    \end{align}
    where, again, ${\mathcal{A}}$ and ${\mathcal{B}}$ are constants to be determined. We take two traces of the above expansion:
    \begin{enumerate}[leftmargin=0.8cm]
        \item{Contract $A = D$, $B = C$:
            \begin{align}
                8 = 2{\mathcal{A}} + 4{\mathcal{B}} \,.
            \end{align}}
        \item{Contract $A = B$, $C = D$:
            \begin{align}
                4 = 4{\mathcal{A}} + 2{\mathcal{B}}\,.
            \end{align}}
    \end{enumerate}
    We obtain ${\mathcal{A}} = 0$ and ${\mathcal{B}}=2$, and therefore the expansion \eqref{moosay1} takes the form:
    \begin{align}\label{EE*}
        ({\mathcal{E}}_i)^A_B \, ({\mathcal{E}}^{i\ast})^C_D = 2\, \delta^A_D\, \delta^C_B\,.
    \end{align}
    Using the spinor metric $\epsilon$, we immediately deduce the second essential expansion:
    \begin{align}
        ({\mathcal{E}}_i)^{AB} \, ({\mathcal{E}}^{i\ast})_{CD} &= \left(\epsilon^{A\dot{B}}\,({\mathcal{E}}_i)^B_{\dot{B}} \right) \left( \epsilon_{C\dot{C}}\,({\mathcal{E}}^{i\ast})^{\dot{C}}_D \right) \nonumber\\[0.2cm]
        &= \left(\epsilon^{A\dot{B}}\, \epsilon_{C\dot{C}} \right) \left( 2\, \delta^B_D\, \delta^{\dot{C}}_{\dot{B}} \right) \nonumber\\[0.2cm]
        &= \left(\epsilon^{A\dot{B}}\, \epsilon_{C\dot{B}} \right) \left( 2\, \delta^B_D \right) \nonumber\\[0.2cm]
        &= \left( \delta^A_C \right) \left( 2\, \delta^B_D \right) = 2\, \delta^A_C\, \delta^B_D\,.
    \end{align}

    The third expansion can be viewed as:
    \begin{align}
        ({\mathcal{E}}_i)^A_B \, ({\mathcal{E}}^{i})^C_D = {\mathcal{A}}^\prime\, \delta^A_B\, \delta^C_D + {\mathcal{B}}^\prime\, \delta^A_D\, \delta^C_B\,,
    \end{align}
    where, again, ${\mathcal{A}}^\prime$ and ${\mathcal{B}}^\prime$ are constants to be determined. We consider two traces of the above expansion:
    \begin{enumerate}[leftmargin=0.8cm]
        \item{Contract $A = D$, $B = C$:
            \begin{align}
                -4 = 2{\mathcal{A}}^\prime + 4{\mathcal{B}}^\prime \,.
            \end{align}}
        \item{Contract $A = B$, $C = D$:
            \begin{align}
                4 = 4{\mathcal{A}}^\prime + 2{\mathcal{B}}^\prime \,.
            \end{align}}
    \end{enumerate}    
    We find ${\mathcal{A}}^\prime = -{\mathcal{B}}^\prime = 2$, and hence we have:
    \begin{align}
        ({\mathcal{E}}_i)^A_B \, ({\mathcal{E}}^{i})^C_D = 2\, \delta^A_B\, \delta^C_D - 2\, \delta^A_D\, \delta^C_B\,.
    \end{align} 
    Again employing the spinor metric $\epsilon$, we find the fourth expansion:
    \begin{align}\label{balabala}
        ({\mathcal{E}}_i)^{AB} \, ({\mathcal{E}}^{i})^{CD} &= \left(\epsilon^{A\dot{B}}\,({\mathcal{E}}_i)^B_{\dot{B}} \right) \left(\epsilon^{C\dot{C}}\, ({\mathcal{E}}^{i})^{D}_{\dot{C}} \right) \nonumber\\[0.2cm]
        &= \left(\epsilon^{A\dot{B}}\, \epsilon^{C\dot{C}} \right) \left( 2\, \delta^B_{\dot{B}}\, \delta^D_{\dot{C}} - 2\, \delta^B_{\dot{C}}\, \delta^D_{\dot{B}} \right) \nonumber\\[0.2cm]
        &= 2\, \epsilon^{AB}\, \epsilon^{CD} - 2\, \epsilon^{AD}\, \epsilon^{CB} \nonumber\\[0.2cm]
        &= 2\, \epsilon^{AC}\, \epsilon^{BD}\,.
    \end{align}
    The latter equality, namely $\epsilon^{AB}\, \epsilon^{CD} - \epsilon^{AD}\, \epsilon^{CB} = \epsilon^{AC}\, \epsilon^{BD}$, may require some clarification. Let us compute both sides for specific (non-trivial) values:
    \begin{enumerate}[leftmargin=0.8cm]
        \item{If $A = 2,\; B = 1,\; C = 1,\; D = 2$, then: \begin{align}
            \epsilon^{21}\, \epsilon^{12} - \underbrace{\epsilon^{22}\, \epsilon^{11}}_{=\; 0} = \epsilon^{21}\, \epsilon^{12} \,.
        \end{align}}
        
        \item{If $A = 2,\; B = 2,\; C = 1,\; D = 1$, then: 
        \begin{align}
            \underbrace{\epsilon^{22}\, \epsilon^{11}}_{=\; 0} - \;\epsilon^{21}\, \epsilon^{12} = \epsilon^{21}\, \epsilon^{21} \,.
        \end{align}}
    \end{enumerate}}

    {\textbf{— Step IV (proof of Eq. \eqref{T2=0=C2}).} Taking into account all preceding results — particularly Eq. \eqref{balabala} — we conclude that:
    \begin{align}
        T^i T_i &= \left( \mathfrak{a}^\ast\, {\mathcal{E}}^i\, \mathfrak{b}^\ast \right) \left( \mathfrak{a}^\ast\, {\mathcal{E}}_i\, \mathfrak{b}^\ast \right) \nonumber\\[0.2cm]
        &= \left( \mathfrak{a}^\ast_A\, ({\mathcal{E}}^i)^{AB}\, {\mathfrak{b}^\ast_B} \right) \left( \mathfrak{a}^\ast_C\, ({\mathcal{E}}_{i})^{CD}\, {\mathfrak{b}^\ast_D} \right) \nonumber\\[0.2cm]
        &= \left( \mathfrak{a}^\ast_A\, {\mathfrak{b}^\ast_B} \right) \left( \mathfrak{a}^\ast_C\, {\mathfrak{b}^\ast_D} \right) \left( 2\, \epsilon^{AC}\, \epsilon^{BD} \right) \nonumber\\[0.2cm]
        &= 2 \, \left( \mathfrak{a}^\ast_A\, \epsilon^{AC}\, \mathfrak{a}^\ast_C \right) \left( {\mathfrak{b}^\ast_B} \, \epsilon^{BD}\, {\mathfrak{b}^\ast_D} \right) \nonumber\\[0.2cm]
        &= 2 \, \underbrace{\left( \mathfrak{a}^\ast_1\, \epsilon^{11}\, \mathfrak{a}^\ast_1 + \mathfrak{a}^\ast_1\, \epsilon^{12}\, \mathfrak{a}^\ast_2 + \mathfrak{a}^\ast_2\, \epsilon^{21}\, \mathfrak{a}^\ast_1 + \mathfrak{a}^\ast_2\, \epsilon^{22}\, \mathfrak{a}^\ast_2 \right)}_{=\, 0} \underbrace{\left( \dots \right)}_{=\, 0} = 0 \,.
    \end{align}}

    Taking the Hermitian conjugate of $T^iT_i=0$ and using $C_i=T_i^\ast$, we likewise obtain:
    \begin{align}
        C^iC_i = 0\,.
    \end{align}
}\end{proof}

\begin{Proposition}\label{Proposition 3.3}
    \textbf{(realization of $\mathfrak{su} (2,2)$ in the conformal compact basis and its ladder structure).} Let $J_{\mu\nu} = \frac{1}{2} (\widetilde{\varphi}\, \gamma^{}_{\mu\nu}\, \varphi)$ denote the ``second-quantized'' generators of the Lie algebra $\mathfrak{su}(2,2)\cong\mathfrak{so}(4,2)$, as defined in Eqs. \eqref{123456}-\eqref{654321}. By taking suitable linear combinations of these generators, we construct an alternative realization of the Lie algebra that is naturally adapted to the so-called conformal compact basis:
    \begin{align} 
        \mbox{complex translations}\;;&\quad T_i := + \left(J_{i5} + \mathrm{i} J_{0i}\right) = \mathfrak{a}^\ast\, \mathcal{E}^{}_i \, \mathfrak{b}^\ast \,, \label{Taaa}\\[0.2cm]
        \mbox{special conformal transformations}\;;&\quad C_i := - \left(J_{i5} - \mathrm{i} J_{0i}\right) = \mathfrak{b}\, \mathcal{E}_i^\ast\, \mathfrak{a} \,, \label{Caaa}\\[0.2cm]
        \mbox{conformal Hamiltonian}\;;&\quad H := \mathrm{i} J_{05} = \frac{1}{2} \left( \mathfrak{a}^\ast \mathfrak{a} + \mathfrak{b}\, \mathfrak{b}^\ast \right) \,, \label{Haaa}\\[0.2cm]
        \mbox{angular momentum-like operators}\;;&\quad J_{ij} \,, \quad i,j = 1,2,3,4\,. \label{Sigmaab}
    \end{align}
    
    It is straightforward to verify the relation:
    \begin{align}
        C_i \; \big(= - (J_{i5} - \mathrm{i} J_{0i}) \big) \;\;=\;\; T_i^\ast \; \big(= (J_{i5} + \mathrm{i} J_{0i})^\ast \big) \,,
    \end{align}
    which follows directly from the anti-Hermitian property \eqref{anti-Hermitian J} of the generators $J^{}_{\mu\nu}$. This identity reflects the underlying symmetry \eqref{C=TAST} between the operators $C_i$ and $T_i$, and will play a key role in the structure of the corresponding Lie algebra representations.

    The above operators close under commutation and thereby furnish a realization of the Lie algebra $\mathfrak{su}(2,2) \cong \mathfrak{so}(4,2)$, satisfying the canonical commutation relations:
    \begin{align}
        \big[H, J_{ij}\big] &=0 \,, \label{lole}\\[0.1cm]
        \big[C_i, T_j\big] &= 2 H\, \delta_{ij} - 2 J_{ij} \,, \qquad \left( \Longrightarrow \quad H = \frac{1}{8} \sum_{i=1}^4 \; \big[C_i, T_i\big] \right)\,, \label{CaTb} \\[0.1cm]
        \big[H, T_i\big] &= T_i \,, \qquad \Big(\Longrightarrow \quad H \big(T_i |\text{W}\rangle\big) = T_i \big(H |\text{W}\rangle\big) + T_i |\text{W}\rangle\Big)\,, \label{lole11}\\[0.2cm]
        \big[H, C_i\big] &= -\,C_i \,, \qquad \Big(\Longrightarrow \quad H \big(C_i |\text{W}\rangle\big) = C_i \big(H |\text{W}\rangle\big) - C_i |\text{W}\rangle\Big)\,.\label{lole22}
    \end{align}
    
    In this representation of $\mathfrak{su}(2,2)\cong\mathfrak{so}(4,2)$, the generators $T_i$ and $C_i$ act by construction as ladder operators, raising and lowering the eigenvalues of the conformal Hamiltonian $H$, respectively. In particular, from the latter canonical commutation relation, it follows that the lowering operators $C_i$ annihilate the lowest weight state $|\mathrm{LW}_\lambda\rangle$, which is an eigenstate of $H$ with eigenvalue $\mathscr{E}_\circ:= 1+|\lambda|$ (see Eq. \eqref{2.166}):
    \begin{align}\label{boz}
        C_i|\mathrm{LW}_\lambda\rangle = 0 \,.
    \end{align}

    The commutation relation \eqref{lole} implies that the operators $J_{ij}$, which generate the compact subalgebra $\mathfrak{so}(4) \cong \mathfrak{su}(2)_L \oplus \mathfrak{su}(2)_R$, leave invariant the eigenspaces $\mathcal{H}^{(1)}_\mathscr{E}$ of the conformal Hamiltonian $H$, as noted before. In a particularly significant case, the operators $J_{ij}$ act on the lowest weight state in the zero-helicity ($\lambda=0$) case as follows:
    \begin{align}\label{whyFvacumm}
        J_{ij} |\mathrm{LW}_{\lambda=0}\rangle =: J_{ij}|0^{\scriptscriptstyle{\mathrm{SO}(4)}}_{}\rangle = 0\,,
    \end{align}
    for all $i,j=1,2,3,4$. This condition explicitly demonstrates that the zero-helicity lowest weight state $|\mathrm{LW}_{\lambda=0}\rangle =:|0^{\scriptscriptstyle{\mathrm{SO}(4)}}_{}\rangle$ transforms as a singlet under the compact subgroup $\mathrm{SO}(4)$, justifying its designation as a ``Fock vacuum'' with respect to this subgroup (see Eq. \eqref{2.16} and the accompanying discussion). Of course, it should not be confused with the unique conformally invariant zero-energy vacuum state $|0^{\scriptscriptstyle{\mathrm{U}(2,2)}}\rangle$, which is invariant under the full conformal group $\mathrm{U}(2,2)$ and will be introduced later in the context of the complete conformal field theory (CFT) construction.
\end{Proposition}

Note that the identities stated in the above proposition can be proven either within the oscillator (or twistor) realization \eqref{123456}-\eqref{654321} of the generators or independently, based on the intrinsic algebraic relations \eqref{Ghalee} among the $J_{\mu\nu}$ entities involved. For pedagogical clarity, we present below a proof that seeks to strike a balance between these two approaches, highlighting in each case the method that offers the more instructive perspective.

\begin{proof}{
    {\textbf{— Step I (proof of the commutation relations \eqref{CaTb}).} The commutation relations presented above can be readily derived from the known algebraic properties \eqref{Ghalee} of the generators $J_{\mu\nu}$. Nevertheless, it is instructive to verify them explicitly within the oscillator (or twistor) realization \eqref{123456}-\eqref{654321} of the generators. As an illustration, let us consider \eqref{CaTb}.
    
    By employing Eq. \eqref{LubiA} in combination with the previously established decompositions, we obtain the following commutation relations between the operators $C_i$ and $T_j$:
    \begin{align}\label{tttt}
        &\big[ C_i , T_j\big]  \nonumber\\[0.2cm]
        &= \big[ \mathfrak{b}^A \,({\mathcal{E}}^\ast_i)_{AB} \, \mathfrak{a}^B , \mathfrak{a}^\ast_E\, ({\mathcal{E}}^{}_j)^{ED}\, \mathfrak{b}^\ast_D \big] \nonumber\\[0.2cm]
        &= \mathfrak{b}^A \underbrace{\big[ ({\mathcal{E}}^\ast_i)_{AB} \, \mathfrak{a}^B , \mathfrak{a}^\ast_E\, ({\mathcal{E}}^{}_j)^{ED} \big]}_{=\, ({\mathcal{E}}^\ast_i)_{AB} \, \delta^B_E \, ({\mathcal{E}}^{}_j)^{ED}} \mathfrak{b}^\ast_D + \underbrace{\big[ \mathfrak{b}^A , \mathfrak{a}^\ast_E\, ({\mathcal{E}}^{}_j)^{ED} \big]}_{=\, 0} ({\mathcal{E}}^\ast_i)_{AB} \, \mathfrak{a}^B\, \mathfrak{b}^\ast_D \nonumber\\[0.2cm]
        &\quad\; + \mathfrak{a}^\ast_E\, ({\mathcal{E}}^{}_j)^{ED}\, \mathfrak{b}^A \underbrace{\big[ ({\mathcal{E}}^\ast_i)_{AB} \, \mathfrak{a}^B , \mathfrak{b}^\ast_D \big]}_{=\, 0} + \; \mathfrak{a}^\ast_E\, ({\mathcal{E}}^{}_j)^{ED} \underbrace{\big[ \mathfrak{b}^A , \mathfrak{b}^\ast_D \big]}_{=\, \delta^A_D} ({\mathcal{E}}^\ast_i)_{AB} \, \mathfrak{a}^B  \nonumber\\
        &= \mathfrak{b} \,{\mathcal{E}}^\ast_i  {\mathcal{E}}^{}_j \mathfrak{b}^\ast + \mathfrak{a}^\ast \,{\mathcal{E}}^{}_j  {\mathcal{E}}^\ast_i\, \mathfrak{a} \nonumber\\[0.2cm]
        &= 2 H \, \delta_{ij} - 2 J_{ij}\,.
    \end{align}
    Regarding the derivation of the last line, note that ($i,j=1,2,3,4$):
    \begin{align}
        {\mathcal{E}}^{}_j {\mathcal{E}}^\ast_i = {\mathcal{E}}^\ast_i  {\mathcal{E}}^{}_j = 1\,, &\quad\mbox{when}\quad i=j \,, \\[0.2cm]
        {\mathcal{E}}^{}_j {\mathcal{E}}^\ast_i = -{\mathcal{E}}^{}_i {\mathcal{E}}^\ast_j \,, &\quad\mbox{when}\quad i \neq j \,.
    \end{align}
    
    Proceeding analogously to the steps outlined above, with Eq. \eqref{LubiA} in mind, the remaining commutation relations \eqref{lole}, \eqref{lole11}, and \eqref{lole22} follow directly within the oscillator (or twistor) realization \eqref{123456}-\eqref{654321}.}

    {\textbf{— Step II (proof of \eqref{boz}).} Using a procedure analogous to that described in Remark \ref{Remark Raising and Lowering}, the proof of Eq. \eqref{boz} follows directly from the commutation relation \eqref{lole22}. On the other hand, employing the oscillator (or twistor) realization \eqref{123456}-\eqref{654321} of the generators, together with the definition \eqref{Vacdef}, immediately leads to Eq. \eqref{boz}.}

    {\textbf{— Step III (proof of \eqref{whyFvacumm}).} Eq. \eqref{whyFvacumm} can be straightforwardly verified through the (relevant) matrix realization \eqref{123456}-\eqref{654321}, in conjunction with \eqref{Vacdef}. We, however, choose to present an alternative — and arguably more illustrative — proof that relies purely on the algebraic relations intrinsic to the generators. 
    
    We begin with the conformal Hamiltonian $H$ and the angular momentum-like generators $J_{ij}$, which satisfy the commutation relation $[H, J_{ij}] = 0$ \eqref{lole}. This relation implies that each $J_{ij}$ preserves the eigenspaces $\mathcal{H}^{(1)}_\mathscr{E} \subset \mathscr{H}^{(1)}_\lambda$ of $H$ \eqref{EDelta}; that is:
    \begin{align}\label{SigmaOnE}
        J_{ij} \;;\quad \mathcal{H}^{(1)}_\mathscr{E} \;\longmapsto\; \mathcal{H}^{(1)}_\mathscr{E} \,.
    \end{align}
    In particular, this invariance extends to the zero-helicity, lowest-energy \linebreak eigenspace $\mathcal{H}^{(1)}_{\mathscr{E}_\circ}$, which contains solely (up to an overall scalar factor) the lowest weight state $|\mathrm{LW}_{\lambda=0}\rangle =: |0^{\scriptscriptstyle{\mathrm{SO}(4)}}_{}\rangle =: |1,0,0\rangle$. As a matter of fact, as can be deduced from Eq. \eqref{d(Delta)} for $\texttt{j}_L=0=\texttt{j}_R$, this eigenspace is $1$-dimensional. It then naturally follows that the action of $J_{ij}$ on this state must take the form:
    \begin{align}
        J_{ij} |\mathrm{LW}_{\lambda=0}\rangle = \alpha_{ij} |\mathrm{LW}_{\lambda=0}\rangle \,,
    \end{align}
    for some scalar coefficients $\alpha_{ij} \in \mathbb{C}$.

    The generators $J_{ij}$ span the compact subalgebra $\mathfrak{so}(4) \cong \mathfrak{su}(2)_L \oplus \mathfrak{su}(2)_R$, and therefore act as angular momentum operators. In any unitary representation, they are represented by traceless anti-Hermitian matrices. On the $1$-dimensional space $\mathcal{H}^{(1)}_{\mathscr{E}_\circ}$, the only possible linear operators are scalar multiples of the identity. Imposing tracelessness forces the scalar to vanish. Hence, the only consistent possibility is:
    \begin{align}
        \alpha_{ij} = 0 \quad\Longrightarrow\quad J_{ij} |\mathrm{LW}_{\lambda=0}\rangle =: J_{ij}|0^{\scriptscriptstyle{\mathrm{SO}(4)}}_{}\rangle = 0 \,,
    \end{align}
    which shows that $|\mathrm{LW}_{\lambda=0}\rangle=: J_{ij}|0^{\scriptscriptstyle{\mathrm{SO}(4)}}_{}\rangle$ transforms trivially — i.e., as a singlet — under the compact subgroup $\mathrm{SO}(4)$ generated by $J_{ij}$.}
}\end{proof}

\sectionmark{Complex Analytic Picture: Globally Conformal Invariant ...}
\section{Complex Analytic Picture: Globally Conformal Invariant Vertex Algebra}
\sectionmark{Complex Analytic Picture: Globally Conformal Invariant ...}

In terms of the homogeneous coordinates $\mathbb{R}^{4,2} \ni \zeta = (\zeta_\mu \;;\; \mu=0,5,1,2,3,4)$ over the Dirac quadric $Q$ \cite{D36}:
\begin{align}\label{2.36}
    Q \;;\quad &\zeta^i \zeta_i \left(:= \sum_{i=1}^4 (\zeta_i)^2 \right) = (\zeta_0)^2 + (\zeta_5)^2 \;>\; 0\,, \nonumber\\[0.2cm]
    &\zeta \sim \varrho\zeta\,, \quad\mbox{with}\quad \varrho\in\mathbb{R} \setminus \{0\}\,,
\end{align}
the complex coordinates $z_i$, on the compactified Minkowski space\footnote{The identification of the conformal compactification of Minkowski space with the Shilov boundary of the associated Hermitian symmetric domain is recalled in Appendix \ref{Appendix Shilov boundary}.}:
\begin{align}
    M^{\texttt{\tiny{compact}}}_{} \cong \frac{\mathbb{S}^3 \times \mathbb{S}^1}{\mathbb{Z}_2} \cong \mathrm{U}(2) \,,
\end{align} 
are defined by \cite{T86}:
\begin{align}\label{2.38-1}
    z_i = \frac{\zeta_i}{\zeta^5 - \mathrm{i}\zeta^0}\,, \quad i=1,\dots,4\,. 
\end{align}
One notes that:
\begin{align}
    \zeta_i = \zeta^i \quad \big(\Longrightarrow\; z_i = z^i\big)\,, \quad\mbox{for}\quad i=1,\dots,4\,,
\end{align}
\begin{align}
    \zeta_k = - \zeta^k\,, \quad\mbox{for}\quad k=0,5\,,
\end{align}
\begin{align}
    \frac{z_i}{z_j} = \frac{\zeta_i}{\zeta_j} \in \mathbb{R}\,,  
\end{align}
\begin{align}\label{zoverlinez=1}
    z \cdot\overline{z} \;\left( := \sum_{i=1}^4 z^i \overline{z}_i \right) = 1\,,
\end{align}
\begin{align}\label{2.37}
    (z)^2 = z\cdot z \;\left( := \sum_{i=1}^4 z^i {z}_i \right) &= \frac{\zeta^i\zeta_i}{(\zeta^5 - \mathrm{i}\zeta^0)(\zeta^5 - \mathrm{i}\zeta^0)}  \nonumber\\[0.2cm]
    &= \frac{(\zeta^5 + \mathrm{i}\zeta^0)(\zeta^5 - \mathrm{i}\zeta^0)}{(\zeta^5 - \mathrm{i}\zeta^0)(\zeta^5 - \mathrm{i}\zeta^0)} \nonumber\\[0.1cm]
    &= \frac{(\zeta^5 + \mathrm{i}\zeta^0)}{(\zeta^5 - \mathrm{i}\zeta^0)} =: e^{2\mathrm{i}\tau}\,,
\end{align}
where $\tau \in \mathbb{R}$ is called ``conformal time''. 

The condition $z \cdot \overline{z} = 1$ constrains $z$ to lie on the $3$-sphere $\mathbb{S}^3 \subset \mathbb{C}^4$, while the complex phase of $(z)^2 = z \cdot z = e^{2\mathrm{i}\tau}$ naturally parametrizes a circle $\mathbb{S}^1$. Together, these properties show that, locally, the complex coordinates $z$ describe the product space $\mathbb{S}^3 \times \mathbb{S}^1$. Globally, the projective identification $\zeta \sim -\zeta$ induces the equivalence $z \sim -z$, pairing antipodal points on this product. As a result, the compactified Minkowski space inherits the topology $\mathbb{S}^3 \times \mathbb{S}^1 / \mathbb{Z}_2 \cong \mathrm{U}(2)$, capturing both the spherical and circular structure together with the discrete $\mathbb{Z}_2$ identification arising from projectivity.

The complex variables $z_i$ are related to the real Minkowski spacetime $M\cong\mathbb{R}^{3,1}$ coordinates $x^{\dot{\mu}}$ ($\dot{\mu} = 0,1,2,3$) by a complex conformal transformation:
\begin{align}
    g_c \;;\quad M\cong\mathbb{R}^{3,1} \quad\longmapsto\quad M^{\texttt{\tiny{compact}}}_{} \cong \frac{\mathbb{S}^3 \times \mathbb{S}^1}{\mathbb{Z}_2} \cong \mathrm{U}(2)\,,
\end{align}
which explicitly reads as:
\begin{align}\label{2.38}
    g_c (x\in M) = z = \left( \vec{z} = \frac{\vec{x}}{\Omega(x)}\,, \; z_4 = \frac{1-(x)^2}{2\Omega(x)} \right) \in M^{\texttt{\tiny{compact}}}_{}\,,
\end{align}
where $\vec{z}=(z^1, z^2, z^3)$, $\vec{x}=(x^1,x^2,x^3)$, $(x)^2 = x^{\dot{\mu}}x_{\dot{\mu}} = (\vec{x})^2 - (x_0)^2$, and:
\begin{align}
    \Omega(x) = \frac{1+(x)^2}{2} - \mathrm{i}x^0 \,.
\end{align}
With the above definition, one checks that $z\cdot\overline{z}=1$ and $(z)^2=e^{2\mathrm{i}\tau}$. Moreover, we observe that:
\begin{align}\label{2.38'}
    (\mathrm{d}z)^2 = \left( := \sum_{i=1}^4 \mathrm{d}z^i\, \mathrm{d}z_i \right) = \frac{(\mathrm{d}x)^2}{\Omega^2(x)} \,,
\end{align}
\begin{align}\label{2.38''}
    (z-u)^2 = \big(\underbrace{g_c(x)}_{=:\,z} - \underbrace{g_c(y)}_{=:\,u}\big)^2 = \frac{(x-y)^2}{\Omega(x)\,\Omega(y)}\,.
\end{align}
Equation \eqref{2.38'} implies that the pullback of the flat metric in $z$-space is conformally equivalent to the Minkowski metric in $x$-space, thereby establishing conformality in the infinitesimal (local) sense. Hence, the map $g_c$ is locally conformal. On the other hand, Eq. \eqref{2.38''} demonstrates that causal structure and angular relationships are preserved up to a position-dependent scale factor, confirming that $g_c$ is globally conformal as well.

To give meaning to quantum fields depending on complex variables, we recall some properties of the vector-valued distribution $\hat{\phi}(x)|0^{\scriptscriptstyle{\mathrm{U}(2,2)}}\rangle$ which follow from the basic assumptions (Wightman assumptions) of relativistic QFT \cite{BLOT, SW}.

\begin{Proposition} 
    \textbf{(analytic continuation to the forward tube and its conformal compactification).} We recall:
    \begin{enumerate}[leftmargin=0.8cm]
        \item{Wightman axioms imply that the (ket-)vector-valued distribution \linebreak $\hat{\phi}(x)|0^{\scriptscriptstyle{\mathrm{U}(2,2)}}\rangle$, where $\hat{\phi}$ is a relativistic quantum field and $|0^{\scriptscriptstyle{\mathrm{U}(2,2)}}\rangle$ is the Poincar\'{e}-invariant vacuum vector, arises as the boundary value ($y \to 0$) — in the distributional sense — of an analytic (ket-)vector-valued function defined in the forward tube $\mathfrak{T}_+$:
        \begin{align}
            &\hat{\phi}(x+\mathrm{i}y) |0^{\scriptscriptstyle{\mathrm{U}(2,2)}}\rangle\,,  
        \end{align}
        for:
        \begin{align}
            (x+\mathrm{i}y) \in \mathfrak{T}_+ = \Big\{ x+\mathrm{i}y \;;\; x,y \in \mathbb{R}^{3,1}\,,\; y^0>|\vec{y}| \Big\} \,,
        \end{align}
        where $\vec{y} = (y^1, y^2, y^3)$ and $|\vec{y}| = \sqrt{(y^1)^2 + (y^2)^2 + (y^3)^2}$.
        
        \textbf{Note:} Throughout this work, we denote by $|0^{\scriptscriptstyle{\mathrm{U}(2,2)}}\rangle$ the unique vacuum state of the theory. In the setting of CFT, this vacuum is invariant under the full conformal group $\mathrm{U}(2,2)$, and in particular under its Poincar\'{e} subgroup. In the broader context of QFT, $|0^{\scriptscriptstyle{\mathrm{U}(2,2)}}\rangle$ refers to the Poincar\'{e}-invariant vacuum defined within the Wightman framework, characterized by translation invariance and the spectral condition. In theories with exact conformal symmetry, these two notions are compatible and in fact coincide. Accordingly, it is natural and consistent to use the same notation for both.}
        
        \item{The rational complex coordinate transformation $g_c$, mapping the complexified Minkowski space $M_{\mathbb{C}}\cong\mathbb{C}^{3,1}$ onto the complexified Euclidean space $E_{\mathbb{C}} \cong \mathbb{C}^4$: 
        \begin{align}
            g_c \;;\quad M_{\mathbb{C}}\cong\mathbb{C}^{3,1} \quad\longmapsto\quad E_{\mathbb{C}}\cong\mathbb{C}^4\,,
        \end{align}
        is regular (holomorphic) on the forward tube $\mathfrak{T}_+$ and maps it into:
        \begin{align}
            g_c \;;\quad \mathfrak{T}_+\subset M_{\mathbb{C}} \quad\longmapsto\quad \mathfrak{T}^{\texttt{\tiny{precompact}}}_+\subset E_{\mathbb{C}}\,,
        \end{align}
        which explicitly reads:
        \begin{align}\label{2.40}
            g_c(\mathfrak{T}_+) &= \mathfrak{T}^{\texttt{\tiny{precompact}}}_+ \nonumber\\[0.2cm]
            &= \left\{ z \in \mathbb{C}^4 \;;\; \left|(z)^2\right| < 1,\; z \cdot \overline{z} < \frac{1}{2} \left(1 + \left|(z)^2\right|^2 \right) \right\} \,.
        \end{align}
        This transformation preserves the form of Eqs. \eqref{2.38'} and \eqref{2.38''} under analytic continuation.

        \textbf{Note:} The domain $\mathfrak{T}^{\texttt{\tiny{precompact}}}_+$ is precompact (or relatively compact) in $E_{\mathbb{C}}$. Indeed, by construction, $\mathfrak{T}^{\texttt{\tiny{precompact}}}_+$ is an open and bounded subset of $\mathbb{C}^4$, contained in the open unit ball. Although it is therefore not itself compact, its closure is a closed and bounded subset of $\mathbb{C}^4$, hence compact. Thus, $\mathfrak{T}^{\texttt{\tiny{precompact}}}_+$ is precompact.}
        
        \item{The tube domain $\mathfrak{T}_+$ and its image $\mathfrak{T}^{\texttt{\tiny{precompact}}}_+$ are each invariant under real conformal transformations\footnote{Without assuming conformal invariance of underlying QFT (a fact noted earlier by V. Glaser, unpublished), $\mathfrak{T}^{\texttt{\tiny{precompact}}}_+$ is a homogeneous space of the conformal group; the stabilizer subgroup of $z=0\in\mathfrak{T}^{\texttt{\tiny{precompact}}}_+$ is the maximal compact subgroup $\mathrm{S}\big(\mathrm{U}(2)\times\mathrm{U}(2)\big) \subset \mathrm{SU}(2,2)$.} \cite{U}. The closure $M^{\texttt{\tiny{compact}}}$ of the precompact image of $M$ in $E_{\mathbb{C}}$ has the form:
        \begin{align}\label{r,e,tau}
            M^{\texttt{\tiny{compact}}} = \Bigg\{ z = r e^{\mathrm{i}\tau}\;;\;& r=(r_i\,, \; i=1,\dots,4)\in\mathbb{R}^4 \,, \; \tau\in\mathbb{R}\,,\;  \nonumber\\[0.2cm]
            & (r)^2 \left(= \sum_{i=1}^4 (r_i)^2\right)=1 \Bigg\} \;\cong\; \frac{\mathbb{S}^3\times\mathbb{S}^1}{\mathbb{Z}_2}\,,
        \end{align}
        such that $z$ remains invariant under the change of coordinates \linebreak $(r,\tau) \longmapsto (-r, \tau+\pi)$.}
    \end{enumerate}
\end{Proposition}

\begin{proof}
    {\textbf{— Step I (the proof of $1$).} The proof of (1) relies on the spectrum condition (energy positivity and Poincar\'{e} invariance) and standard properties of the Laplace transform of tempered distributions \cite{BLOT, SW}.

    By the Wightman axioms, the state $\hat{\phi}(x) |0^{\scriptscriptstyle{\mathrm{U}(2,2)}}\rangle$ defines a vector-valued tempered distribution on $\mathbb{R}^{3,1}$. In particular, its Fourier transform exists in the sense of tempered distributions and has support contained in the closed forward light-cone:
    \begin{align}
        V_+ := \Big\{ p \in \mathbb{R}^{3,1} \;;\; p^0 \geq 0\,,\; (p^{0})^2 - (\vec{p})^2\geq 0\Big\}\,.
    \end{align}
    This is the content of the spectrum condition, which follows from energy positivity and Poincar\'{e} invariance.

    We now apply the Laplace transform to this distribution. For any $y \in \mathbb{R}^{3,1}$, with $y^0 > |\vec{y}|$ (where $\vec{y} = (y^1, y^2, y^3)$ and $|\vec{y}| = \sqrt{(y^1)^2 + (y^2)^2 + (y^3)^2}$), we define the complex spacetime point as $z = x + \mathrm{i} y$, and consider the (ket-)vector-valued function as:
    \begin{align}
        \hat{\phi}(z) |0^{\scriptscriptstyle{\mathrm{U}(2,2)}}\rangle := \int_{V_+} e^{ -\mathrm{i} p_{\dot{\mu}} z^{\dot{\mu}} } \, \hat{\phi}(p) |0^{\scriptscriptstyle{\mathrm{U}(2,2)}}\rangle \, \mathrm{d}p \,,
    \end{align}
    where $\hat{\phi}(p) |0^{\scriptscriptstyle{\mathrm{U}(2,2)}}\rangle$ denotes the Fourier transform of $\hat{\phi}(x) |0^{\scriptscriptstyle{\mathrm{U}(2,2)}}\rangle$ and $\dot{\mu}=0,1,2,3$. Since $p^0 \geq |\vec{p}| \geq 0$ and $y^0 > |\vec{y}|$, and by Cauchy-Schwarz inequality:
    \begin{align}\label{Cauchy-Schwarz}
        -|\vec{p}|\, |\vec{y}| \;\leq\; \vec{p}\cdot\vec{y} = p^1y^1 +  p^2y^2 +  p^3y^3 \;\leq\; |\vec{p}|\, |\vec{y}| \,,
    \end{align}
    we have:
    \begin{align}
        \operatorname{Im}\left(p_{\dot{\mu}} z^{\dot{\mu}}\right) \;\leq\; - p^0 y^0 + |\vec{p}|\, |\vec{y}| \;\leq\; -\underbrace{p^0}_{\geq\; 0} \underbrace{\big( y^0 - |\vec{y}| \big)}_{>\; 0} \;\leq\; 0\,,
    \end{align}
    which ensures that the exponential factor $e^{-\mathrm{i} p_{\dot{\mu}} z^{\dot{\mu}}}=e^{-\mathrm{i} p_{\dot{\mu}} x^{\dot{\mu}}} e^{p_{\dot{\mu}} y^{\dot{\mu}}}$ decays exponentially in $p$. As a result, the integral converges for all $z \in \mathfrak{T}_+$, where $\mathfrak{T}_+ := \left\{ z \in \mathbb{C}^4 \;;\; \operatorname{Im} z = y,\, y^0 > |\vec{y}| \right\}$. Thus, $\hat{\phi}(z) |0^{\scriptscriptstyle{\mathrm{U}(2,2)}}\rangle$ defines a well-defined, analytic vector-valued function in the forward tube $\mathfrak{T}_+$.

    On the other hand, by standard results on the Laplace transform of tempered distributions with support in the forward light-cone (see, e.g., \cite{BLOT, SW}), the analytic function $\hat{\phi}(z) |0^{\scriptscriptstyle{\mathrm{U}(2,2)}}\rangle$ admits the original distribution $\hat{\phi}(x) |0^{\scriptscriptstyle{\mathrm{U}(2,2)}}\rangle$ as its boundary value in the distributional sense as $y \to 0$ within $\mathfrak{T}_+$. That is, for any smooth test function $f(x)$ with compact support in $M\cong\mathbb{R}^{3,1}$, one has:
    \begin{align}
        \lim_{y \to 0,\, y^0 > |\vec{y}|} \int_{\mathbb{R}^{3,1}}& f(x)\, \hat{\phi}(x + \mathrm{i} y) |0^{\scriptscriptstyle{\mathrm{U}(2,2)}}\rangle\, \mathrm{d}x  \nonumber\\[0.2cm]
        &\qquad= \int_{\mathbb{R}^{3,1}} f(x)\, \hat{\phi}(x) |0^{\scriptscriptstyle{\mathrm{U}(2,2)}}\rangle\, \mathrm{d}x \,.
    \end{align}
    This establishes that $\hat{\phi}(x) |0^{\scriptscriptstyle{\mathrm{U}(2,2)}}\rangle$ arises as the boundary value of the analytic function $\hat{\phi}(z) |0^{\scriptscriptstyle{\mathrm{U}(2,2)}}\rangle$ from the forward tube $\mathfrak{T}_+$. This completes the proof.}

    {\textbf{— Step II (the proof of $2$).} The proof of the second item follows from a direct calculation \cite{T86, N}. As an illustration, let us focus on the first part of the statement — namely, the regularity of the conformal map:
    \begin{align}
        g_c\;;\quad w \quad\longmapsto\quad z =&\, \left(\vec z=\frac{\vec w}{\Omega(w)}\,,\; z_4 = \frac{1-(w)^2}{2\Omega(w)}\right), \nonumber\\[0.2cm]
        \Omega(w) =&\, \frac{1+(w)^2}{2}-\mathrm{i}w^0\,,
    \end{align}
    on the forward tube $\mathfrak{T}_+= \big\{ w := x+\mathrm{i}y \;;\; x,y \in\mathbb{R}^{3,1}\,,\; y^0 > |\vec{y}| \big\}$. The map $g_c$ is regular (holomorphic) on $\mathfrak{T}_+$ provided that its denominator never vanishes there; in other words, regularity of $g_c$ on the forward tube requires that $\Omega(w) \neq 0$ for all $w \in\mathfrak{T}_+$.
    
    Suppose, for contradiction, that $\Omega(w)=0$ at some $w\in\mathfrak T_+$:
    \begin{align}
        \Omega(w) &= \frac{1+(w)^2}{2} - \mathrm{i}w^0 \nonumber\\[0.2cm]
        &= \frac{1+(x+\mathrm{i} y)^2}{2}-\mathrm{i} (x^0 + \mathrm{i}y^0) \nonumber\\[0.1cm]
        &= \frac{1+(x)^2-(y)^2}{2} + y^0 - \mathrm{i} (x^0 - x\cdot y) = 0\,.
    \end{align}
    This implies that:
    \begin{align}\label{595159}
        \operatorname{Re}\Omega(w) = \frac{1+(x)^2-(y)^2}{2} + y^0 = 0 \quad\Longrightarrow\quad (x)^2 &= -\big(1+y^0\big)^2 + |\vec{y}|^2 \,, \nonumber\\[0.4cm] 
        \big(\text{since}\; y^0 > |\vec{y}| \geq 0\big) \quad & \Downarrow \nonumber\\[0.4cm]
        (x^0)^2 &> |\vec{x}|^2 \geq 0 \,,
    \end{align}
    and:
    \begin{align}\label{4362}
        \operatorname{Im}\Omega(w) = x^0 - x\cdot y = 0 \quad\Longrightarrow\quad x^0 = x\cdot y =&\, -x^0 y^0 + \vec{x}\cdot\vec{y} \nonumber\\[0.2cm]
        & \Downarrow \nonumber\\[0.2cm]
        x^0 \big(1+ y^0\big) =&\, \vec{x}\cdot\vec{y} \nonumber\\[0.4cm]
        \big(\text{by Cauchy-Schwarz inequality \eqref{Cauchy-Schwarz}}\big) \quad & \Downarrow \nonumber\\[0.4cm]
        -|\vec{x}|\, |\vec{y}| \leq x^0 \big(& 1+ y^0\big) \leq |\vec{x}|\, |\vec{y}| \,.
    \end{align}
    The latter inequality evidently contradicts the conditions $y^0 > |\vec{y}| \geq 0$ and $(x^0)^2 > |\vec{x}|^2 \geq 0$. Hence, the assumption $\Omega(w) = 0$ is invalid, and the conformal transformation $g_c$ is therefore regular (holomorphic) over the entire forward tube $\mathfrak{T}_+$.}

    {\textbf{— Step III (the proof of $3$).} The map $g_c\,;\ x \;\longmapsto\; z = r e^{\mathrm{i}\tau}$, where $(r)^2=1$, was originally proposed \cite{T86, U} as a (generalized) Cayley map from the Lie algebra $\mathfrak{u}(2)$ to the group $\mathrm{U}(2)$:
    \begin{align}\label{2.42}
        g_c \;;\;\;\;& \mathfrak{r} = x^j{\mathcal{E}}^\ast_j + \mathrm{i}x^0 \in \mathfrak{su}(2)\oplus\mathfrak{u}(1) \cong \mathfrak{u}(2) \nonumber\\[0.2cm]
        &\qquad\qquad\longmapsto\quad z := z^j{\mathcal{E}}^\ast_j = \frac{1+\mathfrak{r}}{1-\mathfrak{r}} =: r e^{\mathrm{i}\tau} \in \mathrm{U}(2)\,,
    \end{align}
    where, in accord with the notational conventions established in Proposition \ref{proposition maghz}, the symbols ${\mathcal{E}}^\ast_j$ ($j=1,2,3$) denote the standard imaginary quaternionic units, which satisfy $({\mathcal{E}}^\ast_j)^2 = {\mathcal{E}}^\ast_1 {\mathcal{E}}^\ast_2 {\mathcal{E}}^\ast_3 = - {\mathcal{E}}^\ast_4 = -\mathbbm{1}$. The scalar $\mathrm{i}\in\mathbb{C}$, representing the complex imaginary unit with $\mathrm{i}^2=-1$, is assumed to commute with all ${\mathcal{E}}^\ast_j$ ($j=1,2,3,4$), i.e., $\big[\mathrm{i}, {\mathcal{E}}^\ast_j\big] = 0$. On this basis, by multiplying both the numerator and the denominator of $(1+\mathfrak{r})/(1-\mathfrak{r})$ by $1+\overline{\mathfrak{r}} = 1 + x^j{\mathcal{E}}^{\ast}_j - \mathrm{i}x^0$,\footnote{Note that $\mathfrak{r}\overline{\mathfrak{r}} = -(x)^2$.} one verifies that $z:= z^j{\mathcal{E}}^\ast_j = \frac{1+\mathfrak{r}}{1-\mathfrak{r}}$ is indeed equivalent to \eqref{2.38}. Keeping this equivalence in mind, one readily finds that $z \cdot \overline{z} = 1$ and $(z)^2 = e^{2\mathrm{i}\tau}$, which together justify $z:=\frac{1+\mathfrak{r}}{1-\mathfrak{r}} = r e^{\mathrm{i}\tau} \in \mathrm{U}(2)$.}
\end{proof}

\begin{remark}
    {\emph{(backward-tube analyticity of $\langle0^{\scriptscriptstyle{\mathrm{U}(2,2)}}| \hat{\phi}(x)$).} Importantly, the corresponding bra-vector-valued function $\langle0^{\scriptscriptstyle{\mathrm{U}(2,2)}}| \hat{\phi}(x)$ is analytic in the backward tube $\mathfrak{T}_-$ — complex conjugate to $\mathfrak{T}_+$; explicitly:
    \begin{align}
        \langle0^{\scriptscriptstyle{\mathrm{U}(2,2)}}| \hat{\phi}(x+\mathrm{i}y)
    \end{align}
    is analytic in:
    \begin{align}
        \mathfrak{T}_- \; \left( = \overline{\mathfrak{T}_+} \right)= \Big\{ x+\mathrm{i}y \;;\; x, y\in\mathbb{R}^4,\, y^0<-|\vec{y}| \Big\}\,.
    \end{align}
    The backward tube $\mathfrak{T}^{\texttt{\tiny{precompact}}}_-$ can also be defined in the $z$-picture as the image of $\mathfrak{T}_-$ under the map $g_c$ \eqref{2.38}: 
    \begin{align}
        g_c \;;\quad \mathfrak{T}_-\subset M_{\mathbb{C}}\cong\mathbb{C}^{3,1} \quad\longmapsto\quad \mathfrak{T}^{\texttt{\tiny{precompact}}}_-\subset E_{\mathbb{C}}\cong\mathbb{C}^4 \,.
    \end{align}}
\end{remark}

\begin{remark}
    {\emph{(distributional nature of the field and analyticity of the two-point function).} The bra and ket vectors $\langle 0^{\scriptscriptstyle{\mathrm{U}(2,2)}} | \hat{\phi}(x)$ and $\hat{\phi}(x) | 0^{\scriptscriptstyle{\mathrm{U}(2,2)}} \rangle$, respectively, are ill-defined for real spacetime points $x$. They are only meaningful as vector-valued distributions — that is, they become well-defined only after smearing $\hat{\phi}(x)$ with a suitable test function. In particular, this distributional nature is reflected in vacuum expectation values of field products, such as the Wightman two-point function:
    \begin{align}
        W(z_1, z_2) = \langle 0^{\scriptscriptstyle{\mathrm{U}(2,2)}} | \hat{\phi}(z_1) \hat{\phi}(z_2) | 0^{\scriptscriptstyle{\mathrm{U}(2,2)}} \rangle \,,
    \end{align}
    where $z_j = x_j + \mathrm{i}y_j$, with $j=1,2$. The two-point function $W(z_1, z_2)$ is naturally defined on complexified spacetime and possesses different (complex conjugate) analyticity domains in its first ($z_1$) and second ($z_2$) arguments.}
\end{remark}

\begin{Proposition}
    \textbf{(conformal conjugation exchanging the forward and backward tubes).} A conformal conjugation in $E_{\mathbb{C}}\cong\mathbb{C}^4$ is defined by:
    \begin{align}\label{conjugate}
        z \quad\longmapsto\quad z^\ast := \frac{\overline{z}}{(\overline{z})^2}\,.
    \end{align}
    This conformal conjugation interchanges the precompact forward and backward tubes:
    \begin{align}\label{2.44}
        z \in \mathfrak{T}^{\texttt{\tiny{precompact}}}_\pm \quad\Longleftrightarrow\quad z^\ast \in \mathfrak{T}^{\texttt{\tiny{precompact}}}_\mp\,,
    \end{align}
    and leaves the compactified Minkowski space $M^{\texttt{\tiny{compact}}}$ pointwise invariant:
    \begin{align}
        z^\ast = z\,, \quad\mbox{for any}\quad z \in M^{\texttt{\tiny{compact}}}\,.
    \end{align}
    In fact, the conformal conjugation acts as the identity on $M^{\texttt{\tiny{compact}}}$.
\end{Proposition}

\begin{proof}{
    {\textbf{— Step I (interchanging the tubes).} Let $z \in \mathfrak{T}^{\texttt{\tiny{precompact}}}_+$, i.e., there exists $x + \mathrm{i} y \in \mathfrak{T}_+$ (with $x, y \in {M\cong\mathbb{R}^{3,1}}$ and $y^0 > |\vec{y}|$), such that:
    \begin{align}
        z = g_c(x + \mathrm{i} y)\,.
    \end{align}
    We need to show:
    \begin{align}\label{eq:goal}
        z^\ast = \frac{\overline{g_c(x + \mathrm{i} y)}}{\left(\overline{g_c(x + \mathrm{i} y)}\right)^2} = g_c(x - \mathrm{i} y)\,,
    \end{align}
    where, by construction, $(x - \mathrm{i} y) \in \overline{{\mathfrak{T}}_+} = {\mathfrak{T}}_-$, and hence $g_c(x - \mathrm{i} y) \in \mathfrak{T}^{\texttt{\tiny{precompact}}}_-$. This follows from a direct computation using the explicit form:
    \begin{align}
        g_c(x \pm \mathrm{i} y) = \left( \frac{\vec{x} \pm \mathrm{i} \vec{y}}{\Omega_\pm},\; \frac{1 - (x \pm \mathrm{i} y)^2}{2 \Omega_\pm} \right)\,,
    \end{align}
    where $\Omega_\pm := \frac{1 + (x \pm \mathrm{i} y)^2}{2} - \mathrm{i}(x^0 \pm \mathrm{i} y^0)$.
    The converse direction follows similarly by reversing the sign of $\mathrm{i} y$.}
 
    {\textbf{— Step II ($z^\ast = z$ iff $z \in M^{\texttt{\tiny{compact}}}$).} This follows immediately from the realization \eqref{r,e,tau}, valid for all $z \in M^{\texttt{\tiny{compact}}}$.}
}\end{proof}

\begin{remark}\label{Remark z to x}
    {\emph{(Euclidean interpretation of the conformal map $g_c$ and its inverse).} If we define $x_4 := -\mathrm{i} x^0$ and reinterpret $g_c$ as a map between two Euclidean spaces, i.e., $(\vec{x}, x_4) \longmapsto (\vec{z}, z_4)$, then the inverse map $z \longmapsto x$ retains the same functional form as \eqref{2.38}:
    \begin{align}\label{2.38888888}
        g^{-1}_c(z) = x = \left( \vec{x} = \frac{\vec{z}}{\Omega(z)} \,, \; x_4 = \frac{1-(z)^2}{2\Omega(z)} \right) \,,
    \end{align}
    where $\Omega(z) = \frac{1+(z)^2}{2} + z_4$. We actually have:
    \begin{align}
        g^{-1}_c \big(g^{}_c (x) \big) = x \,.
    \end{align}}
\end{remark}

\subsection{Free Massless Scalar Field}

We now reformulate the theory of the free massless scalar field $\hat{\phi}(x)$, whose two-point function in Minkowski spacetime $M\cong\mathbb{R}^{3,1}$ is given by:
\begin{align}\label{2.46}
    \langle \hat{\phi}(x_1)\, \hat{\phi}(x_2) \rangle_0 &= \frac{1}{(2\pi)^2 \left((\vec{x}^{}_{12})^2 - (x_{12}^0 - \mathrm{i}0)^2\right)} \nonumber\\[0.2cm]
    &= \int \frac{e^{-{\mathrm{i}}\big(\vec{p}\cdot \vec{x}^{}_{12} - |\vec{p}|x^0_{12}\big)}}{(2\pi)^3 \;2 |\vec{p}|} \, \mathrm{d}^3\vec{p}\,,
\end{align}
where $x^{}_{12} = x^{}_1 - x^{}_2$ (that is, $x^0_{12} = x_1^0 - x_2^0$ and $\vec{x}^{}_{12} = \vec{x}^{}_1 - \vec{x}^{}_2$) and $p_0 = |\vec{p}| = \sqrt{p_1^2 + p_2^2 + p_3^2}$ is the energy of a free massless scalar field. We reinterpret this theory in terms of a conformal vertex operator:
\begin{align}\label{ConformalVertexOperator}
    \hat{\phi}(z) := \frac{2\pi}{\Omega(z)} \, \hat{\phi}\big(x=g_c(z)\big)\,,
\end{align}
which defines a scalar field of \textit{conformal dimension} $\mathscr{E}=1$ on the conformally compactified Minkowski space $M^{\texttt{\tiny{compact}}}$. 

\begin{remark}\label{footnote cTransform}
    {\emph{(conformal dimension and field rescaling).}The ``conformal dimension'' (also called ``scaling dimension'') $\mathscr{E}$ of a field $\phi(x)$ is defined by its transformation under dilations (scale transformations) of the spacetime coordinates:
    \begin{align}
        &x^{\dot{\mu}} \quad\longmapsto\quad x^{\prime\dot{\mu}} = \varrho x^{\dot{\mu}}\,, \quad \varrho>0 \,,\nonumber\\[0.2cm]
        &\qquad\qquad\Longrightarrow\qquad \phi(x) \quad\longmapsto\quad \phi^\prime(x^\prime) \sim \phi(x^\prime) = \varrho^{-\mathscr{E}} \phi(x) \,.
    \end{align}
    Note that, in the last expression, the ``prime'' on the field is omitted by convention ($\phi^\prime(x^\prime) \sim \phi(x^\prime)$), and the field is understood to be the transformed one, identified by its transformed coordinate. This transformation rule ensures that correlation functions and the action remain consistent with scale invariance. For a free, massless scalar field in $4$-dimensional spacetime, conformal invariance of the kinetic term $\int d^4x\, (\partial_{\dot{\mu}} \hat{\phi})^2$ fixes $\mathscr{E} = 1$; see also footnote \ref{footnoteCdimension}. In our case, the field $\hat{\phi}(z) := \frac{2\pi}{\Omega(z)}\, \hat{\phi}\big(g_c(z)\big)$ is defined as a conformally rescaled version of the free scalar field $\hat{\phi}(x)$, where $x = g_c(z)$ is a conformal map from the compactified coordinates $z$ to Minkowski coordinates $x$. The factor $\frac{1}{\Omega(z)}$ ensures that $\hat{\phi}(z)$ transforms as a primary scalar field of conformal dimension $\mathscr{E} = 1$ under conformal transformations in the compactified coordinates.}
\end{remark}

Remarkably, the two-point function of the transformed field $\hat{\phi}(z)$ has essentially the same form as \eqref{2.46}, up to a conformally invariant factor. Specifically, by starting from the original two-point function \eqref{2.46} and considering the conformal transformation \eqref{2.38888888} together with the corresponding conformal distance relation (analogous to \eqref{2.38''}), we obtain:
\begin{align}
    \big\langle \hat{\phi}\big(x_1=g_c(z)\big)\, \hat{\phi}\big(x_2=g_c(u)\big) \big\rangle_0 &= \frac{1}{(2\pi)^2} \frac{1}{\big( g_c(z) - g_c(u) \big)^2} \nonumber\\[0.2cm]
    &= \frac{1}{(2\pi)^2} \frac{\Omega(z)\, \Omega(u)}{(z - u)^2}\,.
\end{align}
Then, to obtain the two-point function of the rescaled field \linebreak $\hat{\phi}(z) = \frac{2\pi}{\Omega(z)}\, \hat{\phi}\big(g_c(z)\big)$, we multiply the above expression by the corresponding factors:
\begin{align}\label{2.47}
    \boldsymbol{W}(z,u) :=&\, \langle 0^{\scriptscriptstyle{\mathrm{U}(2,2)}} | \hat{\phi}(z)\, \hat{\phi}(u) | 0^{\scriptscriptstyle{\mathrm{U}(2,2)}} \rangle \nonumber\\[0.2cm]
    =&\, \left( \frac{2\pi}{\Omega(z)} \right) \left( \frac{2\pi}{\Omega(u)} \right) \big\langle \hat{\phi}\big(g_c(z)\big)\, \hat{\phi}\big(g_c(u)\big) \big\rangle_0 \nonumber\\[0.1cm]
    =&\, \frac{1}{(z - u)^2}\,.
\end{align}
One notes that $\boldsymbol{W}(z,u)$ incorporates the conformal weight, thereby ensuring the correct transformation behavior under conformal mappings.

It is important to emphasize that the asymmetry between the two factors $\hat{\phi}(x_1)$ and $\hat{\phi}(x_2)$, manifested through the presence of $(\dots-\mathrm{i}0)^2$ (or, equivalently, through the support of the Fourier transform) in \eqref{2.46}, persists in the analytic ($z$-)picture — though it is no longer apparent in the notation used in \eqref{2.47}. Strictly speaking, the function $\boldsymbol{W}(z,u)$ is analytic for $u \in \mathfrak{T}^{\texttt{\tiny{precompact}}}_+$ and $z \in \mathfrak{T}^{\texttt{\tiny{precompact}}}_-$, where $\mathfrak{T}^{\texttt{\tiny{precompact}}}_+$ is the bounded domain defined in \eqref{2.40} that contains the origin ($u = 0$), and $\mathfrak{T}^{\texttt{\tiny{precompact}}}_-$ — implicitly defined by Eqs. \eqref{conjugate} and \eqref{2.44} — is an unbounded domain surrounding the point at infinity.

This reformulation sets the stage for introducing and illustrating the fundamental structures of a higher-dimensional conformal vertex algebra. Through explicit computations, it reveals several unexpected features of the familiar free scalar field theory.

\begin{Proposition}
    \textbf{(extension of conformal conjugation to field operators and one-particle states).} The conformal conjugation map \eqref{conjugate} extends naturally to field operators and to vectors in the Hilbert space ${\mathscr{H}}^{(1)}_{\lambda=0}$ of zero-helicity, massless one-particle states, by defining:
    \begin{align}\label{2.48}
        \big( \hat{\phi}(z) | 0^{\scriptscriptstyle{\mathrm{U}(2,2)}} \rangle \big)^\ast = \langle 0^{\scriptscriptstyle{\mathrm{U}(2,2)}} | \big(\hat{\phi}(z)\big)^\ast \,, \quad\mbox{where}\quad \big(\hat{\phi}(z)\big)^\ast = \frac{1}{(\overline{z})^2}\, \hat{\phi}(z^\ast) \,.
    \end{align}
\end{Proposition}

\begin{proof}{
    To derive \eqref{2.48}, we interpret the conjugated operator $\big(\hat{\phi}(z)\big)^\ast$ as arising from $\hat{\phi}(\overline{z})$ via the conformal inversion.
    
    \textbf{— Step I (conformal inversion and Weyl scaling).} Conformal inversion reads as: 
    \begin{align}\label{inversion}
        z \quad\longmapsto\quad \breve{z} := \frac{z}{(z)^2} \,.
    \end{align}
    The conformal conjugation map \eqref{conjugate} can thus be understood as the composition of complex conjugation with conformal inversion:    
    \begin{align}
        z^\ast = \breve{\overline{z}} = \overline{\breve{z}} = \frac{\overline{z}}{(\overline{z})^2} \,.
    \end{align}
    
    The inversion map \eqref{inversion} is conformal in the sense that the squared line element $(\mathrm{d}z)^2 = \sum_{i=1}^4 \mathrm{d}z^i\, \mathrm{d}z_i$ is rescaled by a Weyl factor. This can be seen by explicitly computing the differential of $\breve{z}_i$:
    \begin{align}\label{diffbrevez}
        \breve{z}_i = \frac{z_i}{(z)^2} \quad \Longrightarrow \quad \mathrm{d} \breve{z}_i = \frac{\mathrm{d}z_i}{(z)^2} - \frac{z_i\, \mathrm{d}(z)^2}{\left((z)^2\right)^2}\,,
    \end{align}
    where $i=1,2,3,4$. Since $(z)^2 = \sum_{j=1}^4 z^j z_j$, it follows that:
    \begin{align}
        \mathrm{d}(z)^2 = 2 \sum_{j=1}^4 z^j\, \mathrm{d}z_j \,.
    \end{align}
    Substituting this into \eqref{diffbrevez} yields:
    \begin{align}
        \mathrm{d} \breve{z}_i &= \frac{\mathrm{d}z_i}{(z)^2} - \frac{2 z_i \,\sum_{j=1}^4 z^j\, \mathrm{d}z_j}{\left((z)^2\right)^2} \nonumber\\[0.2cm]
        &= \sum_{j=1}^4 \left( \frac{\delta_i^j}{(z)^2} - \frac{2 z_i z^j}{\left((z)^2\right)^2} \right) \mathrm{d}z_j =: \sum_{j=1}^4 \tau_i^j \;\mathrm{d}z_j\,,
    \end{align}
    where:
    \begin{align}
        \tau_i^j := \frac{\delta_i^j}{(z)^2} - \frac{2 z_i z^j}{\left((z)^2\right)^2}\,,
    \end{align}
    satisfying:
    \begin{align}
        \sum_{k=1}^4 \;\tau^j_k \tau^k_i = \frac{\delta^j_i}{\left((z)^2\right)^2}\,.
    \end{align}
    Using the above, we compute the transformed line element:
    \begin{align}\label{weylfactor}
        (\mathrm{d} \breve{z} )^2 &= \sum_{i=1}^4 \mathrm{d}\breve{z}^i\, \mathrm{d}\breve{z}_i \nonumber\\[0.2cm]
        &= \sum_{i,j,k=1}^4 
        \tau^i_j\, \tau_i^k\, \mathrm{d}z^j \,\mathrm{d}z_k \nonumber\\
        &= \frac{1}{\left((z)^2\right)^2}\,\sum_{j,k} \delta_j^k\, \mathrm{d}z^j\, \mathrm{d}z_k = \frac{(\mathrm{d}z)^2}{\left((z)^2\right)^2}\,.
    \end{align}
    Thus, inversion rescales the metric by the Weyl factor $1/\left((z)^2\right)^2$, confirming that $z \longmapsto \breve{z}$ is indeed a conformal transformation.
    
    \textbf{— Step II (application to the field operator).} Under the inversion map \eqref{inversion}, a scalar primary field of conformal dimension $\mathscr{E} = 1$ transforms as (see Remark \ref{footnote cTransform}):
    \begin{align}
        \hat{\phi}(z) \quad\longmapsto\quad \hat{\phi}(\breve{z}) = (z)^2\, \hat{\phi}(z)\,.
    \end{align}
    Equivalently:
    \begin{align}
        \hat{\phi}(z) = \frac{1}{(z)^2}\, \hat{\phi}(\breve{z}) \,.
    \end{align}
    To define the conjugated field operator $\big(\hat{\phi}(z)\big)^\ast$, we apply the above rule to the complex conjugated coordinate $\overline{z} \,\longmapsto\, \breve{\overline{z}} = z^\ast = {\overline{z}}/{(\overline{z})^2}$, yielding \eqref{2.48}:
    \begin{align}\label{dimi}
        \frac{1}{(\overline{z})^2}\, \hat{\phi}(z^\ast) =: \big(\hat{\phi}(z)\big)^\ast\,.
    \end{align}
    This expression reflects both the geometric action of the conformal conjugation map and the conformal transformation law of the scalar field, with the prefactor accounting for the field's conformal weight $\mathscr{E} = 1$.
}\end{proof}

Using \eqref{2.47} and \eqref{2.48}, one readily verifies that the vector-valued function $\hat{\phi}(z) |0^{\scriptscriptstyle{\mathrm{U}(2,2)}}\rangle$ has a finite positive norm squared for $z\in{\mathfrak{T}}^{\texttt{\tiny{precompact}}}_+$:
\begin{align}\label{2.50}
    \big\| \hat{\phi}(z) |0^{\scriptscriptstyle{\mathrm{U}(2,2)}}\rangle \big\|^2 &= \frac{1}{(\overline{z})^2} \langle0^{\scriptscriptstyle{\mathrm{U}(2,2)}} | \hat{\phi}(z^\ast)\, \hat{\phi}(z) |0^{\scriptscriptstyle{\mathrm{U}(2,2)}}\rangle \nonumber\\[0.2cm]
    &= \frac{1}{1 - 2z\cdot \overline{z} + (z)^2 (\overline{z})^2} < \infty\,.
\end{align}
Note that, for all $z\in{\mathfrak{T}}^{\texttt{\tiny{precompact}}}_+$, one has $2z\cdot \overline{z} < 1 + (z)^2 (\overline{z})^2$ (see Eq. \eqref{2.40}). The quantity on the right-hand side of \eqref{2.50} is the diagonal value of the reproducing kernel of the scalar one-particle Hilbert space in the analytic realization on ${\mathfrak{T}}^{\texttt{\tiny{precompact}}}_+$.

\section{$4$-dimensional Conformal Vertex Algebra}

QFT is primarily concerned with the vacuum representation of products of quantum fields. The algebra of globally conformally invariant fields — such as free massless fields and their operator product expansions — fits naturally into the framework of a $4$-dimensional (or higher) vertex algebra, as developed in Ref. \cite{NT01} and further summarized in Section 4.3 of Ref. \cite{NT}. This formalism not only generalizes the structure of field products beyond perturbation theory but also offers a more natural and conceptually unified perspective than the explicit classification of individual UPEIRs. We now
proceed to outline its key features.

The states of the theory span a \textit{(pre-)Hilbert space} $\mathscr{H}$ that carries a (reducible) unitary representation $U$ of the conformal group $\mathrm{U}(2,2)$, featuring a unique conformally invariant vacuum state $|0^{\scriptscriptstyle{\mathrm{U}(2,2)}}\rangle$. This conformally invariant vacuum is the unique normalized lowest weight state in $\mathscr{H}$ that is invariant under the full conformal group $\mathrm{U}(2,2)$. Explicitly, the vacuum satisfies:
\begin{align}\label{CIvac}
    U(g) |0^{\scriptscriptstyle{\mathrm{U}(2,2)}}\rangle = |0^{\scriptscriptstyle{\mathrm{U}(2,2)}}\rangle\,, \quad\mbox{for all}\quad g \in \mathrm{U}(2,2)\,.
\end{align}

\textbf{\emph{Note:}} A ``pre-Hilbert space'' is a complex vector space equipped with a positive-definite inner product, but not necessarily complete with respect to the induced norm. Its completion yields a Hilbert space. In many QFT settings — especially CFT — the natural space of states is first defined as a dense subspace of a Hilbert space, typically consisting of finite-energy or polynomially bounded states. This subspace suffices for algebraic constructions, such as operator products and the action of symmetry generators, and is often referred to as a pre-Hilbert space. The completion is then required to rigorously establish limits and convergence. Thus, the term ``pre-Hilbert space'' reflects whether one is working directly with the dense, physically meaningful subspace or with its Hilbert space completion, depending on the level of mathematical rigor required.

\begin{Proposition} \label{proposition moham}
    \textbf{(spectral decomposition of the conformal Hamiltonian).} Let $\mathscr{H}$ be a (pre-)Hilbert space carrying a (possibly reducible) unitary positive-energy representation of the conformal group $\mathrm{U}(2,2)$, and let $H := \mathrm{i}J_{05}$ \eqref{Haaa} (see also \eqref{CaTb}) denote the conformal Hamiltonian, which is essentially self-adjoint and bounded from below. Assume, in addition, that each energy eigenspace is finite-dimensional and that the zero-energy eigenspace is spanned by the conformally invariant vacuum $|0^{\scriptscriptstyle{\mathrm{U}(2,2)}}\rangle$. Then:
    \begin{enumerate}[leftmargin=0.8cm]
        \item{The spectrum of $H$ is purely discrete and non-negative:
        \begin{align}
            \operatorname{Spec}(H) \subseteq \Big\{ \mathscr{E} \in \mathbb{R}_{\geq 0} \;;\; \mathscr{E} = 0, 1, \tfrac{3}{2}, 2, \dots \Big\}\,,
        \end{align}
        where each eigenvalue $\mathscr{E}$ has finite multiplicity.}
        \item{The space $\mathscr{H}$ decomposes into a direct sum of (finite-dimensional) eigenspaces $\mathcal{H}^{}_\mathscr{E}$ of the conformal Hamiltonian $H$:
        \begin{align}
            \mathscr{H} = \bigoplus_{\mathscr{E} \in \operatorname{Spec}(H)} \mathcal{H}^{}_\mathscr{E} \,, 
        \end{align}
        such that:
        \begin{align}            
            (H - \mathscr{E}) \mathcal{H}^{}_\mathscr{E} = 0 \,.
        \end{align}}
        
        \item{The conformally invariant vacuum $|0^{\scriptscriptstyle{\mathrm{U}(2,2)}}\rangle$ spans the lowest-energy ($1$-dimensional) eigenspace:
        \begin{align}
            \mathcal{H}^{}_{\mathscr{E}=0} = \mathbb{C} \, |0^{\scriptscriptstyle{\mathrm{U}(2,2)}}\rangle \,.
        \end{align}}
        
        \item{Each eigenspace $\mathcal{H}^{}_\mathscr{E}$ is generally not invariant under the full conformal group $\mathrm{U}(2,2)$; rather, it is preserved only by the compact subgroup $\mathrm{Spin}(4)$. This can be seen as a consequence of Proposition \ref{Proposition 3.3}, particularly as established in Step III of its proof. More precisely, each $\mathcal{H}^{}_\mathscr{E}$ carries a (generally reducible) finite-dimensional representation of $\mathrm{Spin}(4) \cong \mathrm{SO}(4)$. The central element $-\mathbbm{1}_4 \in \mathrm{Spin}(4)$ acts on states in $\mathcal{H}^{}_\mathscr{E}$ by multiplication with their parity, given by: 
        \begin{align}
            (-1)^{2(\texttt{j}_L + \texttt{j}_R)} = (-1)^{2\mathscr{E}} \;\in\; \big\{-1, +1\big\} \,.
        \end{align}
        This identity follows directly from Remark \ref{Remark E=JL+JR}.}
    \end{enumerate}
\end{Proposition}

In this construction, the fields $\hat{\phi}$ of the theory are represented by infinite power series in $(z)^2 \; \left(:=(z_1)^2 + (z_2)^2 + (z_3)^2 + (z_4)^2\right)$ with harmonic polynomial coefficients:
\begin{align}\label{fieldnm}
    \hat{\phi}(z) = \sum_{n\in\mathbb{Z}} \;\sum_{m\in\mathbb{Z}_{\geq 0}} \big((z)^2\big)^n \,\hat{\phi}_{n,m}(z) \,,
\end{align}
such that:
\begin{align}
    \label{fieldnm1} \hat{\phi}_{n,m}(\varrho z) =&\, \varrho^m \hat{\phi}_{n,m}(z) \,, \quad\mbox{for}\quad \varrho >0 \,,\\[0.2cm]
    \label{fieldnm2} \Delta \hat{\phi}_{n,m}(z) :=&\, \sum_{i=1}^4 \frac{\partial^2}{\partial z_i^2} \; \hat{\phi}_{n,m}(z) = 0 \,.
\end{align}

Restricting attention to quantum fields $\hat{\phi}$ and $\hat{\psi}$ of definite parity, we assign to each field a parity degree $\mathscr{P}_{\hat{\phi}}, \mathscr{P}_{\hat{\psi}} \in \{0,1\}$, defined by:
\begin{align}
    (-1)^\mathscr{P} \,\hat{\phi}(z) |0^{\scriptscriptstyle{\mathrm{U}(2,2)}}\rangle :=&\, (-1)^{2H} \,\hat{\phi}(z) |0^{\scriptscriptstyle{\mathrm{U}(2,2)}}\rangle \nonumber\\[0.2cm]
    =&\, (-1)^{2\mathscr{E}} \,\hat{\phi}(z) |0^{\scriptscriptstyle{\mathrm{U}(2,2)}}\rangle \nonumber\\[0.2cm]
    =&\, \pm\, \hat{\phi}(z) |0^{\scriptscriptstyle{\mathrm{U}(2,2)}}\rangle\,.
\end{align}
This reflects the field's transformation under the central element $-\mathbbm{1}_4 \in \mathrm{Spin}(4)$. The sign `$\pm$' determines the intrinsic parity of the field; even-parity (bosonic) fields correspond to $\mathscr{P} = 0$, and odd-parity (fermionic) fields to $\mathscr{P} = 1$. As discussed below, this assignment endows the fields with a $\mathbb{Z}_2$ grading that enters the strong (Huygens) locality condition, governing whether fields commute or anti-commute in the graded operator algebra.

The strong (Huygens) locality condition takes the form:
\begin{align}\label{SLC}
    \left( (z - u)^2 \right)^N \left( \hat{\phi}(z) \,\hat{\psi}(u) \;-\; (-1)^{\mathscr{P}_{\hat{\phi}} \,\mathscr{P}_{\hat{\psi}}}\, \hat{\psi}(u)\,\hat{\phi}(z) \right) = 0 \,,
\end{align}
for some sufficiently large integer $N$. This relation is understood as an identity of formal distributions; multiplication by $\left((z-u)^2\right)^N$ annihilates the graded commutator. In the corresponding real-spacetime formulation, the graded commutator vanishes at every non-null separation, so its distributional support is confined to null-separated pairs, including coincident points. The condition thus expresses strong locality, rather than a prescription for taking the coincidence limit. For an example of the associated light-cone singularity, see Eq. \eqref{2.47}.

At non-null separation, the parity degrees $\mathscr{P}_{\hat{\phi}}$ and $\mathscr{P}_{\hat{\psi}}$ determine the corresponding exchange symmetry:
\begin{enumerate}[leftmargin=0.8cm]
    \item{If both fields are bosonic ($\mathscr{P}_{\hat{\phi}} = \mathscr{P}_{\hat{\psi}} = 0$): $(-1)^{0 \times 0} = +1$, then the fields commute.}
    \item{If both fields are fermionic ($\mathscr{P}_{\hat{\phi}} = \mathscr{P}_{\hat{\psi}} = 1$): $(-1)^{1 \times 1} = -1$, then the fields anti-commute.}
    \item{If one field is bosonic and the other fermionic ($\mathscr{P}_{\hat{\phi}} = 0, \mathscr{P}_{\hat{\psi}} = 1$ or vice versa): $(-1)^{0 \times 1} = (-1)^{1 \times 0} = +1$, then the fields commute.}
\end{enumerate}

\begin{Proposition}
    \textbf{(conformal interpretation of strong Huygens locality).} The strong (Huygens) locality condition \eqref{SLC} reflects the following geometric property: any ordered pair of non-null separated points of real compactified Minkowski space can be mapped to any other such pair by a single conformal transformation in $\mathrm{SU}(2,2)$ acting simultaneously on both points. In the $z$-picture, non-null separation is expressed by $(z-u)^2\neq0$, wherever these coordinates are defined. Under this global action, the null/non-null distinction is invariant, whereas the spacelike/timelike distinction within a chosen Minkowski patch is not. 
\end{Proposition}

Before turning to the proof, let us clarify the distinction between ordinary microcausality and strong (Huygens) locality. In Minkowski spacetime, microcausality requires fields to commute or anti-commute at spacelike separation, according to their grading. Huygens locality is stronger; it requires the same vanishing at every non-null separation, including timelike separation.

Under the global conformal-invariance and QFT assumptions of Ref. \cite{NT01}, ordinary locality extends to this stronger condition. The geometric ingredient is that ordered pairs of non-null separated points of compactified Minkowski space form a single conformal orbit. Global covariance consequently extends the vanishing of graded commutators from spacelike separated pairs to all non-null separated pairs.

Restriction to a Minkowski patch preserves Huygens locality. Thus, the global condition implies ordinary microcausality, while retaining the additional requirement of vanishing graded commutators at timelike separation.

\begin{proof}{
    We explain the geometric mechanism underlying the proposition, using the global transitivity result and the QFT consequences established in Ref. \cite{NT01}.

    {\textbf{— Step I (transitivity within fixed causal classes in Minkowski spacetime).} Let $M\cong\mathbb{R}^{3,1}$ denote flat Minkowski spacetime. Translations, proper orthochronous Lorentz transformations, and positive dilations act transitively on ordered pairs with spacelike separation, and separately on ordered pairs with future-directed or past-directed timelike separation.

    Let $(x,y)$ and $(x^\prime,y^\prime)$ be two ordered pairs with non-null separation vectors $r=y-x$ and $r^\prime=y^\prime-x^\prime$ of the same causal type. For timelike separations, assume also that $r$ and $r^\prime$ have the same time orientation. We construct a conformal transformation $\tilde{g}_c$ mapping $x\longmapsto x^\prime$ and $y\longmapsto y^\prime$ as follows: 
    \begin{enumerate}[leftmargin=0.8cm]
        \item{Translate $x$ to the origin by $T_x(A)=A-x$, mapping $(x,y)$ to $(0,r)$.}

        \item{Similarly, the translation $T_{x^\prime}(A)=A-x^\prime$ maps the target pair $(x^\prime,y^\prime)$ to $(0,r^\prime)$.}

        \item{Apply the dilation $D_\varrho(A)=\varrho A$, with $\varrho=\sqrt{\left| \frac{(r^\prime)^2}{r^2}\right|}>0$. Since the separation vectors have the same causal type, $(\varrho r)^2=(r^\prime)^2$.}

        \item{Choose a proper orthochronous Lorentz transformation $\Lambda\in\mathrm{SO}_0(3,1)$ such that $\Lambda(\varrho r)=r^\prime$. Its existence follows from transitivity on spacelike vectors of fixed squared norm, and separately on each time-oriented sheet of timelike vectors of fixed squared norm.}

        \item{Undo the target translation by applying $T_{x^\prime}^{-1}(A)=A+x^\prime$.}
    \end{enumerate}
    The resulting transformation is:
    \begin{align}
        \tilde{g}_c(A)
        &=
        \big(
        T_{x^\prime}^{-1}\circ\Lambda
        \circ D_\varrho\circ T_x
        \big)(A)
        \nonumber\\[0.2cm]
        &=
        \Lambda\big(\varrho(A-x)\big)+x^\prime,
    \end{align}
    and satisfies:
    \begin{align}
        \tilde{g}_c(x)=x^\prime,
        \quad
        \tilde{g}_c(y)=y^\prime.
    \end{align}
    This establishes transitivity within each of the specified causal classes in the Minkowski patch.}

    {\textbf{— Step II (extension to compactified Minkowski space).} Minkowski spacetime is identified with an open dense patch of $M^{\texttt{\tiny{compact}}}$. The global statement requires transformations that need not preserve this patch. We therefore use the geometric result of Ref. \cite{NT01}: the connected conformal group, represented here by $\mathrm{SU}(2,2)$, acts transitively on ordered pairs of non-null separated points of real compactified Minkowski space. In particular, any such pair can be mapped to a spacelike separated pair in a Minkowski patch, where ordinary microcausality guarantees the vanishing of the graded commutator.

    To express this geometry in the $z$-picture, let $z=g_c(x)$ and $u=g_c(y)$ wherever the coordinate transformation \eqref{2.38} is defined. Equation \eqref{2.38''} gives:
    \begin{align}
        (z-u)^2
        =
        \frac{(x-y)^2}{\Omega(x)\,\Omega(y)}.
    \end{align}
    On the common coordinate domain, the Weyl factors are non-zero, so:
    \begin{align}
        (x-y)^2\neq0
        \quad\iff\quad
        (z-u)^2\neq0.
    \end{align}
    This identity expresses the preservation of null separation under the change of coordinates. The generally complex values of the Weyl factors do not, by themselves, establish transitivity between different causal classes. That conclusion follows from the global conformal action on the compactification, as established in Ref. \cite{NT01}.}

    {\textbf{— Step III (implications for strong locality).} Ordinary microcausality implies that graded commutators vanish at spacelike separation. Under the global conformal-invariance and QFT assumptions of Ref. \cite{NT01}, covariance and the transitivity established in Step II extend this vanishing to every non-null separated pair. Thus, in real spacetime, the distributional support of the graded commutator is confined to null separation.

    In the conformal vertex-algebra formulation, the corresponding strong locality condition is:
    \begin{align}
        \left((z-u)^2\right)^N
        \left(
            \hat{\phi}(z)\hat{\psi}(u)
            -
            (-1)^{\mathscr{P}_{\hat{\phi}}
                     \mathscr{P}_{\hat{\psi}}}
            \hat{\psi}(u)\hat{\phi}(z)
        \right)
        =0 \,,
    \end{align}
    for some sufficiently large integer $N$. The finite-power statement uses the analytic and algebraic properties of the globally conformally invariant theory, not merely the transitivity of the conformal action.

    The prefactor annihilates the graded commutator as a formal distribution; it is not a prescription for taking the coincidence limit. Restriction to a Minkowski patch preserves vanishing at timelike as well as spacelike non-null separation. Hence, Huygens locality remains stronger than ordinary microcausality.}
}\end{proof}

\begin{Proposition}\label{proposition HVY}
    \textbf{(fundamental structural statements of the conformal vertex algebra).} The essence of the theory of conformal vertex algebra is contained in the following statements \cite{NT}:
    \begin{enumerate}[leftmargin=0.8cm]
        \item{The expansion $\hat{\phi}(z) |0^{\scriptscriptstyle{\mathrm{U}(2,2)}}\rangle$ contains no negative powers of $(z)^2$. That is, any term in the field $\hat{\phi}(z)$ (see Eq. \eqref{fieldnm}) involving negative powers of $(z)^2$ — and hence singular at the origin — must annihilate the conformally invariant vacuum $|0^{\scriptscriptstyle{\mathrm{U}(2,2)}}\rangle$.}
        
        \item{Every local component $\hat{\phi}_i(z)$ is uniquely determined by the vector $|v_i\rangle$, defined as:
        \begin{align}
            |v_i\rangle := \hat{\phi}_i(0)\, |0^{\scriptscriptstyle{\mathrm{U}(2,2)}}\rangle \,,
        \end{align}
        such that:
        \begin{align}
            (-1)^{\mathscr{P}_{\hat{\phi}}} |v_i\rangle = (-1)^{2H} |v_i\rangle \,.
        \end{align}}

        \item{If $|v\rangle := \hat{\phi}(0)\, |0^{\scriptscriptstyle{\mathrm{U}(2,2)}}\rangle$ is an eigenstate of the conformal Hamiltonian $H$ with eigenvalue $\mathscr{E}$, i.e., $(H - \mathscr{E})\,|v\rangle = 0$, then $\mathscr{E}$ is identified with the scaling dimension of the field $\hat{\phi}$. In the case of free (massless) fields, $\mathscr{E}$ coincides with the canonical dimension of the field: $\mathscr{E} = 1$ for a massless scalar field (see footnote \ref{footnoteCdimension} and Remark \ref{footnote cTransform}); $\mathscr{E} = \tfrac{3}{2}$ for a Weyl spinor field with helicity $\lambda = \pm \tfrac{1}{2}$; more generally, $\mathscr{E}$ increases with spin, following the free-field scaling dimension. Conversely, a field with $\mathscr{E} = \tfrac{1}{2}$ would represent a subcanonical scaling dimension, incompatible with unitarity and in violation of Wightman's Hilbert space positivity condition.}

        \item{If $|v\rangle := \hat{\phi}(0)\, |0^{\scriptscriptstyle{\mathrm{U}(2,2)}}\rangle$ is the lowest-energy state within an irreducible subrepresentation of the conformal group (i.e., a lowest weight state), then it is annihilated by all special conformal generators $C_i\, |v\rangle = 0$ (see Proposition \ref{Proposition 3.3}). Such states are referred to as quasiprimary in Ref. \cite{NT}, adopting terminology from $2$-dimensional CFT, where the notion originates.}
        
        \item{For every vector $|v\rangle$ arising as $|v\rangle := \hat{\phi}(0)\, |0^{\scriptscriptstyle{\mathrm{U}(2,2)}}\rangle$ and having definite parity $\mathscr{P}$, there exists a unique local field $\textbf{Y}\big(|v\rangle,z\big)$ of parity $\mathscr{P}$ such that:
        \begin{align}\label{w3w4w}
            \textbf{Y}\big(|v\rangle, z\big) |0^{\scriptscriptstyle{\mathrm{U}(2,2)}}\rangle &= e^{z\cdot T} |v\rangle \nonumber\\[0.2cm]
            &\quad\Longrightarrow\quad \textbf{Y}\big(|v\rangle, z=0\big)\, |0^{\scriptscriptstyle{\mathrm{U}(2,2)}}\rangle = |v\rangle \,, 
        \end{align}
        where $z\cdot T := \sum_{i=1}^4 z^i T_i$, and $\textbf{Y}\big(|v\rangle, z\big)$ obeys the conformal covariance conditions:
        \begin{align}\label{TVertex}
            \big[T_i , \textbf{Y}\big(|v\rangle, z\big)\big] = T^{(z)}_i\, \textbf{Y}\big(|v\rangle, z\big) \,,
        \end{align}
        \begin{align}
            \big[C_i , \textbf{Y}\big(|v\rangle, z\big)\big] =&\, C^{(z)}_i \textbf{Y}\big(|v\rangle, z\big) + 2z_i \textbf{Y}\big(H|v\rangle, z\big) \nonumber\\[0.2cm]
            &\, - 2z^j \textbf{Y}\big(J_{ij}|v\rangle, z\big) + \textbf{Y}\big(C_i|v\rangle, z\big)\,,
        \end{align}
        \begin{align}\label{HVertex}
            \big[H , \textbf{Y}\big(|v\rangle, z\big)\big] = -H^{(z)} \textbf{Y}\big(|v\rangle, z\big) + \textbf{Y}\big(H|v\rangle, z\big)\,,
        \end{align}
        \begin{align}\label{JVertex}
            \big[J_{ij} , \textbf{Y}\big(|v\rangle, z\big)\big] = - J^{(z)}_{ij} \textbf{Y}\big(|v\rangle, z\big) + \textbf{Y}\big(J_{ij}|v\rangle, z\big)\,,
        \end{align}
        where: 
        \begin{align}
            T^{(z)}_i &= \partial_i := \tfrac{\partial}{\partial z^i} \,, \nonumber\\[0.2cm] 
            C^{(z)}_i &= -(z)^2\, \partial_i + 2z_i\, (z\cdot\partial) \,, \nonumber\\[0.2cm]
            H^{(z)} &= -z\cdot\partial \,, \nonumber\\[0.2cm]
            J^{(z)}_{ij} &= z_i\partial_j - z_j\partial_i 
        \end{align}
        are the conformal Lie algebra generators realized as differential operators acting on compactified Minkowski space $M^{\texttt{\tiny{compact}}} \cong (\mathbb{S}^3 \times \mathbb{S}^1)/\mathbb{Z}_2 \cong \mathrm{U}(2)$ (see Appendix \ref{Appendix z-picture}).}
    \end{enumerate}    
\end{Proposition}

\begin{proof}{\textbf{(consistency check)}\\
    {\textbf{— Step I (the absence of negative powers in $\hat{\phi}(z)\, |0^{\scriptscriptstyle{\mathrm{U}(2,2)}}\rangle$).} From the vacuum-creation relation \eqref{w3w4w}, we have:
    \begin{align}
        \hat{\phi}(z)\, 
        |0^{\scriptscriptstyle{\mathrm{U}(2,2)}}\rangle
        &=
        \textbf{Y}\big(|v\rangle,z\big)\,
        |0^{\scriptscriptstyle{\mathrm{U}(2,2)}}\rangle
        \nonumber\\[0.2cm]
        &=
        e^{z\cdot T}|v\rangle
        =
        \sum_{r=0}^{\infty}
        \frac{(z\cdot T)^r}{r!}|v\rangle \,.
    \end{align}
    Each term in this series is a homogeneous polynomial of degree $r$ in $z$. Every such polynomial can be written uniquely as:
    \begin{align}
        P_r(z)
        =
        \sum_{n=0}^{\lfloor r/2\rfloor}
        \big((z)^2\big)^n h_{r-2n}(z) \,,
    \end{align}
    where $h_{r-2n}(z)$ is a homogeneous harmonic polynomial of degree $r-2n$. This also holds for polynomials whose coefficients are vectors.

    Since $n\geq0$ in this decomposition, rewriting the series in harmonic form introduces no negative powers of $(z)^2$. Comparing with \eqref{fieldnm}, we therefore obtain:
    \begin{align}
        \hat{\phi}_{n,m}(z)\,
        |0^{\scriptscriptstyle{\mathrm{U}(2,2)}}\rangle
        = 0 \,,
        \quad n<0 \,, 
        \quad
        m\in\mathbb{Z}_{\geq0} \,.
    \end{align}
    Thus, the absence of negative powers follows from \eqref{w3w4w} for every $|v\rangle$, without requiring $|v\rangle$ to be a lowest-weight state.}

    {\textbf{— Step II (verification of conformal covariance).} We now examine whether the vertex operators introduced in Eqs. \eqref{TVertex}-\eqref{JVertex} are consistent with the canonical commutation relation of the conformal Lie algebra $\mathfrak{su}(2,2)\cong\mathfrak{so}(4,2)$ (see Eqs. \eqref{lole}-\eqref{lole22}):
    \begin{align}
        &\big[ \big[H , T_i \big] , \textbf{Y}\big(|v\rangle, z\big) \big] \nonumber\\[0.2cm]
        &= \big[ H , \big[T_i , \textbf{Y}\big(|v\rangle, z\big)\big] \big] - \big[ T_i , \big[H , \textbf{Y}\big(|v\rangle, z\big)\big] \big] \nonumber\\[0.2cm]
        &= \big[ H , T^{(z)}_i \textbf{Y}\big(|v\rangle, z\big) \big] - \big[ T_i , -H^{(z)} \textbf{Y}\big(|v\rangle, z\big) + \textbf{Y}\big(H|v\rangle, z\big) \big] \nonumber\\[0.2cm]
        &= T^{(z)}_i \big[ H \,,\, \textbf{Y}\big(|v\rangle, z\big) \big] + H^{(z)} \big[ T_i , \textbf{Y}\big(|v\rangle, z\big) \big] - \big[ T_i , \textbf{Y}\big(H|v\rangle, z\big) \big] \nonumber\\[0.2cm]
        &= T^{(z)}_i \Big( -H^{(z)} \textbf{Y}\big(|v\rangle, z\big) + \textbf{Y}\big(H|v\rangle, z\big) \Big) \nonumber\\[0.2cm]
        &\quad + H^{(z)} T^{(z)}_i \textbf{Y}\big(|v\rangle, z\big) - T^{(z)}_i \textbf{Y}\big(H|v\rangle, z\big) \nonumber\\[0.2cm]
        &= \big[H^{(z)} , T^{(z)}_i \big] \textbf{Y}\big(|v\rangle, z\big) \nonumber\\[0.2cm]
        &= T^{(z)}_i \textbf{Y}\big(|v\rangle, z\big) \nonumber\\[0.2cm]
        &= \big[ T_i, \textbf{Y}\big(|v\rangle, z\big) \big] \,.
    \end{align}
    To derive this relation, we have used the Jacobi identity:
    \begin{align}
        \big[[A,B],C\big] = \big[A,[B,C]\big] - \big[B,[A,C]\big] \,.
    \end{align}
    
    Proceeding along the same path, we obtain:
    \begin{align}
        &\big[ \big[H , C_i \big] , \textbf{Y}\big(|v\rangle, z\big) \big] \nonumber\\[0.2cm]
        &= \big[ H , \big[C_i , \textbf{Y}\big(|v\rangle, z\big)\big] \big] - \big[ C_i , \big[H , \textbf{Y}\big(|v\rangle, z\big)\big] \big] \nonumber\\[0.2cm]
        &= \big[ H , \,C^{(z)}_i \textbf{Y}\big(|v\rangle, z\big) + 2z_i \textbf{Y}\big(H|v\rangle, z\big)  \nonumber\\[0.2cm]
        &\hspace{2cm} - 2z^j \textbf{Y}\big(J_{ij}|v\rangle, z\big) + \textbf{Y}\big(C_i|v\rangle, z\big) \big] \nonumber\\[0.2cm]
        &\;\;\; - \big[ C_i , -H^{(z)} \textbf{Y}\big(|v\rangle, z\big) + \textbf{Y}\big(H|v\rangle, z\big)\big] \nonumber\\[0.2cm]
        &= \,C^{(z)}_i  \big[ H , \textbf{Y}\big(|v\rangle, z\big) \big] + 2z_i \big[ H , \textbf{Y}\big(H|v\rangle, z\big) \big] \nonumber\\[0.2cm]
        & \;\;\; - 2z^j \big[ H , \textbf{Y}\big(J_{ij}|v\rangle, z\big)\big] + \big[ H , \textbf{Y}\big(C_i|v\rangle, z\big)\big]\nonumber\\[0.2cm]
        &\;\;\; + H^{(z)} \big[ C_i , \textbf{Y}\big(|v\rangle, z\big)\big] - \big[ C_i , \textbf{Y}\big(H|v\rangle, z\big)\big]  \nonumber\\[0.2cm]
        &= \,C^{(z)}_i \Big( -H^{(z)} \textbf{Y}\big(|v\rangle, z\big) + \textbf{Y}\big(H|v\rangle, z\big) \Big) \nonumber\\[0.2cm]
        &\;\;\; + 2z_i \Big( -H^{(z)} \textbf{Y}\big(H|v\rangle, z\big) + \textbf{Y}\big(H^2|v\rangle, z\big) \Big) \nonumber\\[0.2cm]
        & \;\;\; - 2z^j \Big( -H^{(z)} \textbf{Y}\big(J_{ij}|v\rangle, z\big) + \textbf{Y}\big(HJ_{ij}|v\rangle, z\big) \Big) \nonumber\\[0.2cm]
        &\;\;\; + \Big( -H^{(z)} \textbf{Y}\big(C_i|v\rangle, z\big) + \textbf{Y}\big(HC_i|v\rangle, z\big) \Big)\nonumber\\[0.2cm]
        &\;\;\; + H^{(z)} \Big( C^{(z)}_i \textbf{Y}\big(|v\rangle, z\big) + 2z_i \textbf{Y}\big(H|v\rangle, z\big) \nonumber\\[0.2cm]
        &\hspace{2cm} - 2z^j \textbf{Y}\big(J_{ij}|v\rangle, z\big) + \textbf{Y}\big(C_i|v\rangle, z\big) \Big) \nonumber\\[0.2cm]
        &\;\;\; - \Big( C^{(z)}_i \textbf{Y}\big(H|v\rangle, z\big) + 2z_i \textbf{Y}\big(H^2|v\rangle, z\big) \nonumber\\[0.2cm]
        &\hspace{2cm} - 2z^j \textbf{Y}\big(J_{ij}H|v\rangle, z\big) + \textbf{Y}\big(C_i H|v\rangle, z\big) \Big)  \nonumber\\[0.2cm]
        &= -\big[ C^{(z)}_i , H^{(z)} \big] \textbf{Y}\big(|v\rangle, z\big) - \big[ 2z_i , H^{(z)} \big] \textbf{Y}\big(H|v\rangle, z\big) \nonumber\\[0.2cm]
        &\hspace{2cm} + \big[ 2z^j , H^{(z)} \big] \textbf{Y}\big(J_{ij}|v\rangle, z\big) + \textbf{Y}\big([ H , C_i ]|v\rangle, z\big) \nonumber\\[0.2cm]
        &= -\big[ C_i, \textbf{Y}\big(|v\rangle, z\big) \big] \,,
    \end{align}
    and similarly, one also verifies that:
    \begin{align}\label{cxz}
        \big[\big[C_i, T_j\big], \textbf{Y}\big(|v\rangle, z\big)\big] &= 2\delta_{ij} \big[H, \textbf{Y}\big(|v\rangle, z\big) \big] - 2 \big[J_{ij} , \textbf{Y}\big(|v\rangle, z\big) \big] \,.
    \end{align}}
}\end{proof}

\sectionmark{Appendix: Upper Bound on $\mathfrak{su}(2)$ Spins from ...}
\section{Appendix: Upper Bound on $\mathfrak{su}(2)$ Spins from the $\mathfrak{so}(4)$ Casimir} \label{Appen L,R E}
\sectionmark{Appendix: Upper Bound on $\mathfrak{su}(2)$ Spins from ...}

\begin{Proposition}
    \textbf{($\mathfrak{so}(4)$ spin bound in terms of the conformal energy level).} Let $\mathcal{C}_2^{\mathfrak{so}(4)}$ be the quadratic Casimir of $\mathfrak{so}(4) \cong \mathfrak{su}(2)_L \oplus \mathfrak{su}(2)_R$:
    \begin{align}
        \mathcal{C}_2^{\mathfrak{so}(4)} = \texttt{j}_L (\texttt{j}_L + 1) + \texttt{j}_R (\texttt{j}_R + 1)\,, \quad \texttt{j}_L, \texttt{j}_R \in \frac{1}{2} \mathbb{Z}_{\geq 0} \,.
    \end{align}
    Assume:
    \begin{align}
        &\mathcal{C}_2^{\mathfrak{so}(4)} = \frac{1}{2} \big( \mathscr{E}^2 + \lambda^2 - 1 \big) \,, \nonumber\\[0.2cm] 
        &\mathscr{E} = |\lambda| + n \; \big(\geq 1\big)\,, \nonumber\\[0.2cm] 
        &\lambda \in \frac{1}{2}\mathbb{Z}\,, \quad\mbox{and}\quad n \in \mathbb{Z}_{>0} \,.
    \end{align}
    Then, we have:
    \begin{align}\label{-11111}
        \texttt{j}_L, \texttt{j}_R \leq \mathscr{E}-1\,.
    \end{align}
\end{Proposition}

\begin{figure}[t]
\centering
\begin{tikzpicture}[
    scale=1.08,
    >=latex,
    every node/.style={align=center},
    dot/.style={circle, fill=black, inner sep=2.5pt}
]

\def\E{6}

\fill[gray!16] (0,0) -- (0,\E) -- (\E,0) -- cycle;

\draw[very thick,->] (0,0) -- (\E+1.5,0) node[below right=2pt] {$\texttt{j}_L$};
\draw[very thick,->] (0,0) -- (0,\E+1.05) node[above left=2pt] {$\texttt{j}_R$};

\draw[very thick] (0,\E) -- (\E,0);
\draw[thick] (0,0) -- (0,\E);
\draw[thick] (0,0) -- (\E,0);

\foreach \x in {0.5,1,...,6} {
    \foreach \y in {0.5,1,...,6} {
        \pgfmathparse{\x+\y<=\E+0.001 ? 1 : 0}
        \ifnum\pgfmathresult=1
            \fill (\x,\y) circle (2.55pt);
        \fi
    }
}

\node at (4.55,4.2) {$\texttt{j}_L,\texttt{j}_R \leq \mathscr{E}-1$};

\end{tikzpicture}
\caption{Lattice region in the $(\texttt{j}_L,\texttt{j}_R)$-plane.}
\end{figure}

\begin{proof}{
    Using $\texttt{j}_R (\texttt{j}_R + 1) \geq 0$, we have:
    \begin{align}
        &\texttt{j}_L (\texttt{j}_L + 1) \leq \frac{1}{2} (\mathscr{E}^2 + \lambda^2 - 1) \nonumber\\[0.2cm] 
        &\qquad \Longrightarrow \quad \texttt{j}_L \leq \frac{-1 + \sqrt{2\mathscr{E}^2 + 2\lambda^2 - 1}}{2} \,.
    \end{align}
    Since $\mathscr{E} \geq 1$, one checks:
    \begin{align}
        \frac{-1 + \sqrt{2\mathscr{E}^2 + 2\lambda^2 - 1}}{2} \leq \mathscr{E}-1 \,,
    \end{align}
    because squaring gives the identity $0 \leq 2\big((\mathscr{E}-1)^2 - \lambda^2\big)$. The same argument holds for $\texttt{j}_R$.    
}\end{proof} 

\sectionmark{Appendix: de Sitter (dS) Massless Unitary Irreducible Representations}
\section{Appendix: de Sitter (dS) Massless Unitary Irreducible Representations (UIRs)} \label{appendix: dS UIRs}
\sectionmark{Appendix: de Sitter (dS) Massless Unitary Irreducible Representations}

We begin again with the conformal Lie algebra $\mathfrak{su}(2,2)\cong\mathfrak{so}(4,2)$ and its ``second-quantized'' ladder generators $J_{\mu\nu} = \tfrac{1}{2} (\widetilde{\varphi}\, \gamma^{}_{\mu\nu}\, \varphi)$, with $\mu,\nu=0,1,2,3,4,5$ (see Eq. \eqref{fJf6}). In this context, one may obtain the dS algebra by restricting attention to the generators $J_{\alpha\beta}$, with $\alpha,\beta=0,1,2,3,4$, and discarding the generators $J_{\alpha5}$; this yields a $10$-dimensional subalgebra isomorphic to $\mathfrak{sp}(2,2)\cong\mathfrak{so}(4,1)$. In particular, the conformal Hamiltonian $H = \mathrm{i} J_{05}$ does not belong to the dS subalgebra and therefore cannot serve as a global energy operator or as an energy label $\mathscr{E}$ for dS UIRs; dS representations are not classified by a global energy eigenvalue in the same way as anti-dS representations (see Ref. \cite{Gazeau2022}).

The maximal compact subalgebra of the dS algebra is $\mathfrak{so}(4) \cong \mathfrak{su}(2)_L \oplus \mathfrak{su}(2)_R$, generated by $J_{ij}$ with $i,j=1,2,3,4$. Every state in a dS UIR can be labeled by the compact quantum numbers $(\texttt{j}_L, \texttt{j}_R)$, where again $\texttt{j}_L$ and $\texttt{j}_R$ are the familiar $\mathrm{S\mathrm{U}(2)}$ spins associated with the left and right $\mathrm{SU}(2)$ factors \cite{Dixmier}. We now proceed to elaborate on the structure and interpretation of these labels.

A multiplet will denote the collection of states $|\texttt{j}_L, \texttt{j}_R;\, m_L, m_R\rangle$ transforming irreducibly under $S\mathrm{U}(2)_L \times S\mathrm{U}(2)_R$ with fixed $(\texttt{j}_L,\texttt{j}_R)$, where $m_L=-\texttt{j}_L,\dots,\texttt{j}_L$ and $m_R=-\texttt{j}_R,\dots,\texttt{j}_R$ label the magnetic quantum numbers. The dimension $\texttt{d}$ of such a multiplet is $(2\texttt{j}_L+1)(2\texttt{j}_R+1)$ (see Eq. \eqref{d(Delta)}). The compact generators $J_{ij}$ act within each multiplet, mixing the magnetic quantum numbers while preserving the total spins $(\texttt{j}_L,\texttt{j}_R)$. 

The remaining generators $J_{0i}$ (with $i=1,\dots,4$), which span the non-compact directions of $\mathfrak{so}(4,1)$, act between multiplets. Their commutation relations with the compact generators are:
\begin{align}
    \big[J_{ij}, J_{0k}\big] = \big(\eta_{jk} J_{0i} - \eta_{ik} J_{0j}\big)\,,
\end{align}
which shows that the operators $J_{0i}$ transform in the vector ($4$-dimensional) representation of the compact group $\mathrm{SO}(4)$. Upon lifting to the double cover $\mathrm{Spin}(4)\cong \mathrm{SU}(2)_L\times\mathrm{SU}(2)_R$, this $4$-vector corresponds to the representation $\big(\tfrac{1}{2},\tfrac{1}{2}\big)$, with dimension $\texttt{d}\big(\tfrac{1}{2},\tfrac{1}{2}\big) = (1+1)(1+1)=4$.

When $J_{0i}$ acts on a compact multiplet $(\texttt{j}_L,\texttt{j}_R)$ it therefore produces states in the tensor product:
\begin{align}\label{begaee}
    (\texttt{j}_L,\texttt{j}_R)\otimes\big(\tfrac{1}{2},\tfrac{1}{2}\big) \;\;=\;\;  \big(\texttt{j}_L+\tfrac{1}{2},\texttt{j}_R+\tfrac{1}{2}\big) &\oplus \big(\texttt{j}_L+\tfrac{1}{2},\texttt{j}_R-\tfrac{1}{2}\big) \nonumber\\[0.2cm]
    \oplus \big(\texttt{j}_L-\tfrac{1}{2},\texttt{j}_R+\tfrac{1}{2}\big) &\oplus \big(\texttt{j}_L-\tfrac{1}{2},\texttt{j}_R-\tfrac{1}{2}\big)\,,
\end{align}
with the obvious convention that any term with a negative spin is omitted. Thus, the non-compact generators $J_{0i}$ (with $i=1,\dots,4$) connect up to four neighboring $\mathrm{SO}(4)$ multiplets; repeated application of the non-compact generators builds a ``tower of multiplets'' which, together with the action of the compact algebra, spans the full UIR Hilbert space.

In summary, a dS UIR decomposes into a (generally countable) collection of $\mathrm{SO}(4)$ multiplets $(\texttt{j}_L,\texttt{j}_R)$. The compact generators $J_{ij}$ act within each multiplet, whereas the non-compact generators $J_{0i}$ link multiplets according to the Clebsch-Gordan decomposition with $\big(\tfrac{1}{2}, \tfrac{1}{2}\big)$; there is no analogue of a globally conserved energy label as in anti-dS space.

\begin{remark}
    {\emph{(compact-multiplet decomposition).} Every dS UIR, whether it belongs to the discrete, principal, or complementary series, decomposes under the compact subgroup $\mathrm{Spin}(4)$ into a Hilbert direct sum of finite-dimensional multiplets labelled by $(\texttt{j}_L,\texttt{j}_R)$. These compact-spin labels are discrete in all three cases. The continuous parameters of the principal and complementary series distinguish different dS representations; they do not make the compact-spin labels continuous within any one representation.

    The particular multiplets that occur, and the action of the non-compact generators between them, depend on the chosen UIR. Equation \eqref{begaee} gives the possible target multiplets when $J_{0i}$ acts on a multiplet $(\texttt{j}_L,\texttt{j}_R)$. Some of these transitions may be absent because the target multiplet does not occur in the UIR or because the corresponding matrix element vanishes.}
\end{remark}

For the dS massless UIRs obtained by restricting the fixed-helicity conformal ladder representations constructed above, the full Hilbert space is generated, upon Hilbert-space completion, from a seed multiplet:
\begin{align}
    (\texttt{j}_L,\texttt{j}_R) = (s, 0) \quad\mbox{or}\quad (0, s)\,, \quad\mbox{with}\quad s := |\lambda| = 0, \tfrac{1}{2}, 1, \dots\,.
\end{align}
For $s>0$, the helicity sign is encoded in the chirality of the seed multiplet (see Remark \ref{Remark Helicit+-}):
\begin{align}
    (s, 0) \quad&\longrightarrow\quad (\mbox{helicity}\; +s)\,, \nonumber\\[0.2cm]
    (0, s) \quad&\longrightarrow\quad (\mbox{helicity}\; -s) \,.
\end{align}

For $s>0$, the corresponding massless dS representations are written, in Dixmier's notation \cite{Dixmier}, as:
\begin{align}
    \Pi^+_{s,s} \quad\mbox{and}\quad \Pi^-_{s,s}\,,
\end{align}
respectively. For this massless family, the identification $p=q=s$ agrees with the compact-spin data of the seed multiplet:
\begin{align}
    p=q=s
    &=
    \big|\texttt{j}_L-\texttt{j}_R\big|
    =
    \texttt{j}_L+\texttt{j}_R\,,
    \nonumber\\[0.2cm]
    (\texttt{j}_L,\texttt{j}_R) &=(s,0)
    \quad\mbox{or}\quad (0,s)\,.
\end{align}
Here, $(\texttt{j}_L,\texttt{j}_R)$ denotes the compact spins of the seed multiplet. These equalities express the identification within the non-zero-helicity massless family, rather than a general definition of the Dixmier parameters through extrema over the compact spectrum. The superscript `$\pm$' records the helicity sign associated with the chirality of the seed multiplet; see also Sect. \ref{Sect. Brief on dS}.

For $s=0$, the two seed multiplets coincide with $(0,0)$, and there is no helicity-sign distinction. The corresponding massless scalar dS representation is the complementary-series UIR $U^{\mathrm{cs}}_{0,\frac12}$, with Dixmier parameters $(p,q)=(0,1)$, as recalled in Sect. \ref{Sect. Brief on dS}.

\section{Appendix: Compactified Minkowski Space as a Shilov Boundary}\label{Appendix Shilov boundary}

The conformal compactification of $4$-dimensional Minkowski space:
\begin{align}
    M^{\texttt{\tiny{compact}}}_{} \cong \frac{\mathbb{S}^{3}\times \mathbb{S}^{1}}{\mathbb{Z}_{2}}
\end{align}
(with the $3$‑sphere factor being the familiar compact space sometimes noted for its exceptional rigidity in higher‑dimensional calibrated geometry — see the classical real $\mathbb{S}^3\subset \mathbb{S}^7$ problem about minimal associative $3$‑spheres in the $7$‑sphere \cite{Gazeau-Language}) admits a natural interpretation as the Shilov boundary of a classical bounded symmetric domain associated with the conformal group \cite{Helgason00, Penrose00, Faraut00}.

Indeed, the conformal group $\mathrm{SO}(4,2)$ is locally isomorphic to $\mathrm{SU}(2,2)$, and there are two equivalent realizations of the corresponding Hermitian symmetric space. On the one hand, $\mathrm{SU}(2,2)$ acts transitively on the Cartan domain of type I:
\begin{align}
    D_{2,2}^{\mathrm{I}} = \frac{\mathrm{SU}(2,2)}{\mathrm{S}\big(\mathrm{U}(2)\times \mathrm{U}(2)\big)} \,,
\end{align}
whose Shilov boundary is canonically identified with $\mathrm{U}(2)$. Using the standard topological identification:
\begin{align}
    \mathrm{U}(2) \cong\frac{\mathbb{S}^{3}\times \mathbb{S}^{1}}{\mathbb{Z}_{2}}\,,
\end{align}
one recovers $M^{\texttt{\tiny{compact}}}_{}$ as this Shilov boundary.

On the other hand, via the double-covering homomorphism $\mathrm{SU}(2,2)\cong\mathrm{SO}_0(4,2)$ (with the notational convention adopted above), the same symmetric space can be realized as the Cartan domain of type IV (the Lie ball):
\begin{align}
    D^{\mathrm{IV}}_{4} = \frac{\mathrm{SO}(4,2)}{\big(\mathrm{SO}(2)\times\mathrm{SO}(4)\big)} \,,
\end{align}
whose Shilov boundary is the projectivized null cone in $\mathbb{R}^{4,2}$. This boundary is precisely the conformal compactification of Minkowski space.

Thus, depending on the chosen realization, $M^{\texttt{\tiny{compact}}}_{}$ appears either as:
\begin{enumerate}[leftmargin=0.8cm]
    \item{The Shilov boundary of the type-I domain $D^{\mathrm{I}}_{2,2}$, or equivalently}

    \item{The Shilov boundary of the type-IV domain $D^{\mathrm{IV}}_{4}$.}
\end{enumerate}
Both descriptions are canonically equivalent and reflect the dual realizations of the conformal symmetric space associated with $\mathrm{SO}(4,2)$.

\sectionmark{Appendix: The Lie Algebra $\mathfrak{so}(4,2)$ and its Euclidean $\mathfrak{so}(5,1)$ ...}
\section{Appendix: The Lie Algebra $\mathfrak{so}(4,2)$ and its Euclidean $\mathfrak{so}(5,1)$ Counterpart; the $z$-Picture of the $\mathfrak{so}(4,2)$ Generators}\label{Appendix z-picture}
\sectionmark{Appendix: The Lie Algebra $\mathfrak{so}(4,2)$ and its Euclidean $\mathfrak{so}(5,1)$ ...}

In the context of the present study, it is natural to introduce the $z$-picture realization of the conformal generators acting on compactified Minkowski space $M^{\texttt{\tiny{compact}}} \cong (\mathbb{S}^3 \times \mathbb{S}^1)/\mathbb{Z}_2 \cong \mathrm{U}(2)$. Specifically, we consider the infinitesimal translation generators $T_i$ \eqref{Taaa}, the special conformal transformations $C_i$ \eqref{Caaa}, and the conformal Hamiltonian $H$ \eqref{Haaa} (see also \eqref{CaTb}), as realized via differential operators in the complex coordinate variables $z^i = z_i \in M^{\texttt{\tiny{compact}}}$, with $i=1,2,3,4$, introduced in Eqs. \eqref{2.38-1}-\eqref{2.38}.

To achieve this goal, we consider functions defined on Dirac's projective quadric $\mathbbm{P}Q$, which arises by projectivizing $\mathbbm{P}$ the null cone $Q$ \eqref{2.36} in $\mathbb{R}^{4,2}$:
\begin{align}
    \mathbbm{P}Q := Q / \mathbb{R}^{\times} \,,
\end{align}
with:
\begin{align}
    Q := \Big\{ \zeta \in \mathbb{R}^{4,2}\setminus{0} \;\;;\;\; (\zeta)^2 := \eta^{}_{\mu\nu} \zeta^\mu \zeta^\nu = 0 \Big\} = \mbox{Eq. \eqref{2.36}} \,,
\end{align}
where $\mathbb{R}^{\times} := \mathbb{R}\setminus\{0\}$, $\mu, \nu = 0,5,1,2,3,4$, and the ambient metric $\eta^{}_{\mu\nu}$ on $\mathbb{R}^{4,2}$ has signature $(-,-,+,+,+,+)$.

Some fundamental properties of Dirac's projective quadric $\mathbbm{P}Q$ and related structures are as follows:
\begin{enumerate}[leftmargin=0.8cm]
    \item{\textbf{\textit{Null directions in $\mathbb{R}^{4,2}$}:} $\mathbbm{P}Q$ consists of all null directions (rays) through the origin in $\mathbb{R}^{4,2}$. Concretely, two non-zero null vectors $\zeta$ and $\zeta^\prime$ are considered equivalent if they differ by a non-zero real scalar:
    \begin{align}
        \zeta \sim \varrho \zeta\,, \quad\mbox{where}\quad \varrho \in \mathbb{R}^{\times} \,.
    \end{align}
    Each equivalence class corresponds to a single point on the projectivized light-cone, capturing the direction of the vector while ignoring its length.}  
    
    \item{\textbf{\textit{Topology and conformal compactification}:} The resulting space $\mathbbm{P}Q$ is a compact, $4$-dimensional manifold. Topologically, it is homeomorphic to $(\mathbb{S}^3 \times \mathbb{S}^1)/\mathbb{Z}_2$. This identification arises from the null condition $\zeta^2 = 0$, which imposes that the spatial components $(\zeta_1, \dots, \zeta_4)$ and the temporal components $(\zeta_0, \zeta_5)$ lie on spheres of equal radius — that is, $\zeta_1^2 + \dots + \zeta_4^2 = \zeta_0^2 + \zeta_5^2$ — thus defining $\mathbb{S}^3$ and $\mathbb{S}^1$, respectively. The $\mathbb{Z}_2$ quotient reflects the projectivization of null directions, identifying antipodal points along each ray via $\zeta \sim -\zeta$. Consequently, $\mathbbm{P}Q$ serves as the conformal compactification of Minkowski spacetime:
    \begin{align}
        \mathbbm{P}Q \quad\cong\quad \underbrace{(\mathbb{S}^3 \times \mathbb{S}^1)/\mathbb{Z}_2}_{\cong\, \mathrm{U}(2)} \quad\cong\quad M^{\texttt{\tiny{compact}}} \,.
    \end{align}}

    \item{\textbf{\textit{Embedding of Minkowski space}:} To complete the foregoing, Minkowski spacetime $M\cong\mathbb{R}^{3,1}$ can be embedded in $\mathbbm{P}Q$ in terms of the (dimensionless, homogeneous) coordinates as follows:
    \begin{align}\label{B.4}
        &x^{\dot{\mu}} = \frac{\zeta^{\dot{\mu}}}{\zeta^5 + \zeta^4}\,, \quad \dot{\mu} = 0,1,2,3 \nonumber\\
        &\Longrightarrow\quad (x)^2 := x^{\dot{\mu}} x_{\dot{\mu}} = \frac{\zeta^{\dot{\mu}}\zeta_{\dot{\mu}}}{\left(\zeta^5 + \zeta^4\right)^2} = \frac{(\zeta^5)^2 - (\zeta^4)^2}{\left(\zeta^5 + \zeta^4\right)^2} = \frac{\zeta^5 - \zeta^4}{\zeta^5 + \zeta^4} \,. 
    \end{align}
    The consistency of the above relations is verified by substituting $x^{\dot{\mu}}$ and $(x)^2$, as given above, into Eq. \eqref{2.38}, which yields Eq. \eqref{2.38-1}.}

    \item{\textbf{\textit{Action of the conformal group $\mathrm{SO}(4,2)$}:} The conformal group acts linearly on the ambient space $\mathbb{R}^{4,2}$, preserving the null cone $Q$. This action descends naturally to the projectivized cone $\mathbbm{P}Q$, inducing a well-defined projective action on compactified Minkowski space. Concretely, the Lie algebra $\mathfrak{so}(4,2)$ is realized on Dirac's projective quadric by first-order differential operators of the form:
    \begin{align}\label{B.2}
        X_{\mu\nu} \big( = - X_{\nu\mu} \big) = \zeta_{\mu}\partial^{}_\nu - \zeta_{\nu}\partial^{}_\mu \,, \quad \partial^{}_\mu:=\frac{\partial}{\partial\zeta^\mu} \,,
    \end{align}
    which satisfy the commutation relations \eqref{Ghalee}. In this way, Dirac's projective quadric furnishes a natural geometric stage for the global action of the conformal group, making manifest the realization of conformal symmetry on compactified Minkowski space.}

    \item{\textbf{\textit{Generators of translations and special conformal transformations}:} The Lie algebra $\mathfrak{so}(4,2)$, defined via differential operators acting on smooth functions over Minkowski spacetime, is generated by the infinitesimal translations and special conformal transformations. Specifically, the infinitesimal translations are given by:
    \begin{align}
        \mathrm{i}p_{\dot{\mu}} = \partial_{\dot{\mu}} := \frac{\partial}{\partial x^{\dot{\mu}}}\,,
    \end{align}
    where $p_{\dot{\mu}}$ denotes the $4$-momentum. The special conformal generators are defined as:
    \begin{align}
        \mathrm{i}k_{\dot{\mu}} = (x)^2\, \partial_{\dot{\mu}} - 2x_{\dot{\mu}}\, (x \cdot \partial) \,.
    \end{align}
    The Lie brackets of these generators generate the full conformal Lie algebra, including the Lorentz transformations $x_{\dot{\mu}}\partial_{\dot{\nu}} - x_{\dot{\nu}}\partial_{\dot{\mu}}$ and the dilatation operator $x \cdot \partial$. In particular, the commutator between translation and special conformal generators reads:
    \begin{align}\label{B.5}
        \big[\mathrm{i}k_{\dot{\nu}} , \mathrm{i}p_{\dot{\mu}}\big] = 2\,\eta_{\dot{\mu}\dot{\nu}} (x \cdot \partial) - 2\left(x_{\dot{\mu}}\partial_{\dot{\nu}} - x_{\dot{\nu}}\partial_{\dot{\mu}}\right) \,.
    \end{align}
    We recall our mostly `$+$' convention for the metric signature in $\mathbb{R}^{4,2}$; $\eta^{}_{11} = \dots = \eta^{}_{44} = +1$, while $\eta^{}_{00} = \eta^{}_{55} =-1$. The vector fields $\mathrm{i}p_{\dot{\mu}}$ and $\mathrm{i}k_{\dot{\mu}}$ are related to the generators $X_{\mu\nu}$ \eqref{B.2} by:
    \begin{align}
        \mathrm{i}p_{\dot{\mu}} = X_{\dot{\mu}5} - X_{\dot{\mu}4}\,, \quad \mathrm{i}k_{\dot{\mu}} = X_{\dot{\mu}5} + X_{\dot{\mu}4} \,.
    \end{align}
    This follows from the relations (for $x^{\dot{\nu}}$ given by \eqref{B.4}):
    \begin{align}
        X_{\dot{\mu}5}\, x^{\dot{\nu}} = \frac{\zeta^5 \,\delta^{\dot{\nu}}_{\dot{\mu}}}{\zeta^4+\zeta^5} - \frac{\zeta_{\dot{\mu}}\zeta^{\dot{\nu}}}{\left(\zeta^4+\zeta^5\right)^2} = \frac{1+(x)^2}{2} \, \delta^{\dot{\nu}}_{\dot{\mu}} - x_{\dot{\mu}} \,x^{\dot{\nu}} \,,
    \end{align}
    \begin{align}
        X_{\dot{\mu}4}\, x^{\dot{\nu}} = -\frac{\zeta^4 \,\delta^{\dot{\nu}}_{\dot{\mu}}}{\zeta^4+\zeta^5} - \frac{\zeta_{\dot{\mu}}\zeta^{\dot{\nu}}}{\left(\zeta^4+\zeta^5\right)^2} = \frac{(x)^2-1}{2} \, \delta^{\dot{\nu}}_{\dot{\mu}} - x_{\dot{\mu}} \,x^{\dot{\nu}} \,.
    \end{align}
    \textbf{\textit{Note}:} Regarding the above calculations, it is worthwhile to highlight a subtle point here:
    \begin{align}
        \partial_{\dot{\mu}} x^{\dot{\nu}} = \frac{\partial}{\partial x^{\dot{\mu}}} x^{\dot{\nu}} = \delta^{\dot{\nu}}_{\dot{\mu}}\,,
    \end{align}
    \begin{align}
        \partial_{\dot{\mu}} x_{\dot{\nu}} = \frac{\partial}{\partial x^{\dot{\mu}}} x_{\dot{\nu}} = \frac{\partial}{\partial x^{\dot{\mu}}} \big( \eta_{\dot{\nu}\dot{\rho}} x^{\dot{\rho}} \big) = \eta_{\dot{\nu}\dot{\rho}} \delta^{\dot{\rho}}_{\dot{\mu}} = \eta_{\dot{\nu}\dot{\mu}} = \eta_{\dot{\mu}\dot{\nu}}\,.
    \end{align}}

    \item{\textbf{\textit{Huygens principle and propagation along null directions}:} In four-dimensional Minkowski spacetime, the massless scalar wave equation reads:
    \begin{align}
        \Box \phi(x)
        =
        \eta^{\dot{\mu}\dot{\nu}}\,
        \partial_{\dot{\mu}}\partial_{\dot{\nu}}\phi(x)
        =
        0 \,.
    \end{align}
    The Huygens principle states that the retarded fundamental solution of this wave operator is supported on the future light-cone. Thus, the retarded response to a point source at $x_\circ$ is confined to $(x-x_\circ)^2=0$, with $x^0\geq x_\circ^0$, and has no tail in the timelike interior of that cone. For spatially extended initial data, however, the resulting field need not vanish inside a single light cone; its value at a given event depends on the initial data where the past light cone of that event intersects the initial hypersurface.

    Within the projective framework, null separation has a natural ambient description. Each point $[\zeta]\in \mathbbm{P}Q$ represents a point of compactified Minkowski spacetime, with $\zeta\in Q$ a non-zero null representative. The condition $\zeta^2=0$ therefore specifies the ambient cone itself; it does not impose null separation from a chosen spacetime point. For two points represented by $\zeta$ and $\zeta_\circ$, the homogeneous embedding \eqref{B.4} instead gives, within a common Minkowski chart:
    \begin{align}
        (x-x_\circ)^2=0
        \quad\Longleftrightarrow\quad
        \zeta\cdot\zeta_\circ=0 \,.
    \end{align}
    This condition is independent of the non-zero rescalings of either representative and is preserved by the ambient conformal-group action. It therefore provides a global description of null separation on compactified Minkowski spacetime, expressing geometrically the light-cone support underlying the Huygens principle.}
\end{enumerate}

With the above material in mind, we now return to our main objective: introducing the $z$-picture of the $\mathfrak{so}(4,2)$ generators. The variables $z_i\; \big(=z^i\big)$ \eqref{2.38-1} are covariantly transformed by the generators $\mathrm{i}X_{0\rho}$ and $X_{\rho\rho^\prime}$, with $\rho,\rho^\prime = 1,2,3,4,5$, which span a real form $\mathrm{SO}(5,1)$ of the complexification of $\mathrm{SO}(4,2)$. Explicitly, we have:
\begin{align}
    T^{(z)}_i &= + \left(X_{i5} + \mathrm{i}X_{0i}\right) = \partial_i \,, \quad \partial_i := \frac{\partial}{\partial z^i} \,, \label{TAAAbb} \\[0.2cm]
    C^{(z)}_i &= -\left(X_{i5} - \mathrm{i}X_{0i}\right) = -(z)^2\, \partial_i + 2z_i\, (z\cdot\partial) \,, \label{cAAAbb}
\end{align}
since:
\begin{align}
    X_{i5}\, z^j = \frac{\zeta^5 \,\delta^j_i}{\zeta^5-\mathrm{i}\zeta^0} - \frac{\zeta_i \zeta^j}{\left(\zeta^5-\mathrm{i}\zeta^0\right)^2} = \frac{1+(z)^2}{2} \, \delta^j_i - z_i \,z^j \,,
\end{align}
\begin{align}
    \mathrm{i}X_{0i}\, z^j = -\frac{\mathrm{i}\zeta^0 \,\delta^j_i}{\zeta^5-\mathrm{i}\zeta^0} + \frac{\zeta_i \zeta^j}{\left(\zeta^5-\mathrm{i}\zeta^0\right)^2} = \frac{1-(z)^2}{2} \, \delta^j_i + z_i \,z^j \,,
\end{align}
with $i,j=1,2,3,4$. Moreover, we have:
\begin{align}\label{H z-picture}
    H^{(z)} = \mathrm{i} X_{05} = -z\cdot \partial\,, \quad\mbox{since}\quad \mathrm{i} X_{05}\, z^j = \frac{-\zeta^j}{\zeta^5 - \mathrm{i}\zeta^0} = -z^j\,,
\end{align}
\begin{align}
    J^{(z)}_{ij} = X_{ij} = z_i\partial_j - z_j\partial_i \,.
\end{align}
In direct analogy with the commutation relations given in Eqs. \eqref{lole}, \eqref{CaTb}, \eqref{lole11}, and \eqref{lole22}, it is observed that:
\begin{align}
    \big[H^{(z)} , J^{(z)}_{ij}\big] &= 0 \,,\\[0.2cm]
    \big[C^{(z)}_i, T^{(z)}_j\big] &= 2H^{(z)}\, \delta_{ij} - 2J^{(z)}_{ij} \,, \\[0.2cm]
    \big[H^{(z)} , T^{(z)}_i\big] &= T^{(z)}_i \,, \\[0.2cm]
    \big[H^{(z)} , C^{(z)}_i\big] &= -\,C^{(z)}_i \,.
\end{align}
	\renewcommand*\vec{\mathaccent"017E\relax}

\setcounter{equation}{0} 

\chaptermark{Conformal Massless Low-Helicity Fields and their de Sitter ...}
\chapter{Conformal Massless Low-Helicity Fields and their de Sitter (dS) Spacetime Restriction}\label{Chapter 4}
\chaptermark{Conformal Massless Low-Helicity Fields and their de Sitter ...}

\section*{Chapter Overview}

This chapter presents a systematic construction of low-helicity conformal massless fields, with particular emphasis on the zero-helicity case in de Sitter (dS) spacetime. It begins by formulating the massless scalar field in terms of ladder-type vertex operators built from Bose creation and annihilation operators, supplemented by zero-mode operators that ensure the correct lowest-energy states and reproduce the canonical two-point functions. The analytic structure of these operators is established within a precompact tube domain, ensuring convergence of the corresponding harmonic polynomial expansions.

Subsequently, it describes the restriction of conformal massless fields from the $6$-dimensional conformal space to $4$-dimensional dS spacetime, emphasizing the role of homogeneous functions on the light-cone and the associated gauge freedom. The construction elucidates the relation between Euclidean coordinates in the analytically continued half-space and Minkowski-type coordinates, highlighting the Weyl rescaling that connects the conformally flat metric to the standard dS metric.

Finally, the chapter demonstrates how the $6$-dimensional conformally invariant two-point functions project naturally onto dS spacetime, yielding the standard dS two-point functions for massless scalar fields. These results provide a rigorous framework for understanding the interplay between conformal symmetry, low-helicity fields, and their geometric realization in curved spacetimes.

\section{Ladder Vertex Operator Construction of the Scalar Field}

In order to write down the free massless scalar field $\hat{\phi}(z)$ in terms of Bose creation and annihilation operators \eqref{Vacdef}, we need yet another pair of zero-mode operators $\chi$ and $\chi^\ast$ intertwining $|0^{\scriptscriptstyle{\mathrm{U}(2,2)}}_{}\rangle$ and $|\mathrm{LW}_{\lambda=0}\rangle =: |0^{\scriptscriptstyle{\mathrm{SO}(4)}}_{}\rangle = |1,0,0\rangle$, such that: 
\begin{align}
    \chi |0^{\scriptscriptstyle{\mathrm{U}(2,2)}}_{}\rangle = 0 = \langle 0^{\scriptscriptstyle{\mathrm{U}(2,2)}}_{} | \chi^\ast \,,
\end{align}
\begin{align}
    \chi^\ast |0^{\scriptscriptstyle{\mathrm{U}(2,2)}}_{}\rangle = |1,0,0\rangle \,, \quad\mbox{and}\quad \langle 1,0,0| = \langle 0^{\scriptscriptstyle{\mathrm{U}(2,2)}}_{} | \chi\,, 
\end{align}
\begin{align}
    \big[\chi , \chi^\ast\big] = 1 \,, \quad\mbox{and}\quad \big[\chi^{(\ast)} , \varphi^a\big] = 0 = \big[\chi^{(\ast)} , \widetilde{\varphi}_b\big] \,, 
\end{align}
for $\chi^{(\ast)} = \chi$ or $\chi^\ast$.

\begin{Proposition}\label{proposition 3.1}
    \textbf{(vertex-operator realization of the free massless scalar field).} The vertex operator:
    \begin{align} \label{3.2}
        \hat{\phi}(z) = \boldsymbol{Y}\big(|1,0,0\rangle , z \big) = e^{z\cdot T} \left( \chi^\ast + \frac{\chi}{(z)^2} \right) \, e^{\breve{z}\cdot C}\,,
    \end{align}
    where $\breve{z}_i := {z_i}/{(z)^2}$, $T_i = \mathfrak{a}^\ast \mathcal{E}^{}_i \mathfrak{b}^\ast$, and $C_i = (T_i)^\ast = \mathfrak{b}\, \mathcal{E}^{\ast}_i \mathfrak{a}$, with $i=1,\dots, 4$, satisfies all properties of a free massless scalar field and reproduces the two-point function \eqref{2.47}. In practice, the absence of negative powers of $(z)^2$ in the power series expansion of $\hat{\phi}(z) |0^{\scriptscriptstyle{\mathrm{U}(2,2)}}_{}\rangle$ (asserted by the first statement of Proposition \ref{proposition HVY}) is guaranteed by:
    \begin{enumerate}[leftmargin=0.8cm]
        \item{The presence of the zero-modes $\chi$ and $\chi^\ast$ in \eqref{3.2}.}
        
        \item{The fact that $|1,0,0\rangle$ is a minimal-energy state in the conformally invariant zero-helicity subspace $\mathscr{H}^{(1)}_{\lambda=0} \subset \mathscr{H}$; in particular, $\breve{z}\cdot C|1,0,0\rangle = 0$, or, more precisely, $C_i|1,0,0\rangle = 0$.}
    \end{enumerate}
\end{Proposition}

\begin{remark}
    {\emph{(conformal invariance of the vacuum state).} In this context, before proceeding with the proof, it is important to emphasize that the vacuum state $|0^{\scriptscriptstyle{\mathrm{U}(2,2)}}_{}\rangle$ is conformally invariant, meaning that it is annihilated by all infinitesimal generators of the conformal Lie algebra $\mathfrak{u}(2,2)$; the explicit form of these generators can be found in Chap. \ref{Chapter 3}. Explicitly:
    \begin{align}
        H|0^{\scriptscriptstyle{\mathrm{U}(2,2)}}_{}\rangle = T_i |0^{\scriptscriptstyle{\mathrm{U}(2,2)}}_{}\rangle = C_i |0^{\scriptscriptstyle{\mathrm{U}(2,2)}}_{}\rangle = J_{ij}|0^{\scriptscriptstyle{\mathrm{U}(2,2)}}_{}\rangle = {\mathcal{C}}_1 |0^{\scriptscriptstyle{\mathrm{U}(2,2)}}_{}\rangle = 0\,,
    \end{align}
    where $i,j=1,\dots, 4$.}
\end{remark}

\begin{proof}{ 
    We begin with some basic steps, illustrating along the way the advantages of the analytic $z$-picture:
    {\begin{enumerate}[leftmargin=0.8cm]
        \item{While the basic properties of the $x$-space two-point function and the vector-valued distribution $\hat{\phi}(x) |0^{\scriptscriptstyle{\mathrm{U}(2,2)}}_{} \rangle$ require using a singular Fourier transformation — involving $\theta(p_0) \delta(p^2)$ (see Eq. \eqref{2.46}), $\hat{\phi}(z) |0^{\scriptscriptstyle{\mathrm{U}(2,2)}}_{} \rangle$ is an analytic vector-valued function in $\mathfrak{T}^{\texttt{\tiny{precompact}}}_+$, whose (discrete!) energy spectrum is displayed by its (convergent) expression in homogeneous harmonic polynomials:
        \begin{align}\label{woody}
            \hat{\phi}(z) |0^{\scriptscriptstyle{\mathrm{U}(2,2)}}_{} \rangle &= \textbf{Y}\big(|1,0,0\rangle , z \big) |0^{\scriptscriptstyle{\mathrm{U}(2,2)}}_{} \rangle \nonumber\\[0.2cm]
            &= e^{z\cdot T} \left( \chi^\ast + \frac{\chi}{(z)^2} \right) \, e^{\breve{z}\cdot C} |0^{\scriptscriptstyle{\mathrm{U}(2,2)}}_{} \rangle \nonumber\\[0.2cm]
            &= e^{z\cdot T} |1,0,0\rangle = \sum_{n=0}^\infty \frac{(z\cdot T)^n}{n!} \;|1,0,0\rangle \,,
        \end{align}
        where:
        \begin{align}
            \Delta\, (z\cdot T)^n &= \sum_{i=1}^4 \frac{\partial^2}{\partial z^2_i} (z\cdot T)^n \nonumber\\[0.2cm] 
            &= n (T\cdot\partial) (T\cdot z)^{n-1} \nonumber\\[0.2cm] 
            &= n(n-1) \,(T)^2\, (T\cdot  z)^{n-2} = 0\,,
        \end{align}
        since $(T)^2 = \sum_{i=1}^4 T^i T_i = 0$ (see Proposition \ref{proposition T2=0}).}
        \item{The expected two-point function \eqref{2.47} is harmonic in both $z$ and $u$ and can be expressed in terms of the generating function of the Gegenbauer polynomials $\texttt{C}^1_n$:
        \begin{align}\label{4.7}
            \frac{1}{(z-u)^2} = \frac{1}{(z)^2} \frac{1}{1 + (\breve{z})^2 (u)^2 - 2\breve{z}\cdot u} = \frac{1}{(z)^2} \sum_{n=0}^{\infty} \texttt{H}_n (\breve{z}, u) \,,
        \end{align}
        where:
        \begin{align}
            \texttt{H}_n (\breve{z}, u) &= \sum_{k=0}^{[\frac{n}{2}]} \begin{pmatrix}
            n-k\\ k \end{pmatrix} (2\breve{z}\cdot u)^{n-2k} \left( -(\breve{z})^2 (u)^2 \right)^k \nonumber\\[0.2cm] 
            &= \left( (\breve{z})^2 (u)^2 \right)^\frac{n}{2} \texttt{C}^1_n \left( \frac{\breve{z}\cdot u}{\sqrt{(\breve{z})^2 (u)^2}} \right) \,.
        \end{align}}
        \item{Applying the commutation relation \eqref{cxz}, we find:
        \begin{align}
            \big[C_i, T_j\big] \hat{\phi}(u) \; \Big( = \big[C_i, T_j\big] \textbf{Y}\big(|1,0,0\rangle , u \big) \Big) = 2\big( H\,\delta_{ij} - J_{ij} \big) \hat{\phi}(u) \,.
        \end{align}
        Moreover, from Proposition \ref{Proposition 3.3}, we have: 
            \begin{align}
                H|1,0,0\rangle = |1,0,0\rangle\,, \quad J_{ij}|1,0,0\rangle = 0 \,.
            \end{align}
        Accordingly, we obtain:
        \begin{align}\label{qweewq}
            \langle1,0,0|\; \frac{(\breve{z}\cdot C)^n}{n!} \frac{(u\cdot T)^n}{n!} \; |1,0,0\rangle = \texttt{H}_n (\breve{z}, u) \,.
        \end{align}
        Let us examine the above identity for some small values of $n$:
        \begin{enumerate}[leftmargin=0.8cm]
            \item{If $n=0$, then:
            \begin{align}
                \langle1,0,0|1,0,0\rangle = 1 = \texttt{H}_0 (\breve{z}, u) \,.
            \end{align}}
            \item{If $n=1$, then:
            \begin{align}
                \breve{z}^i u^j\langle1,0,0| \, C_i\, T_j \, |1,0,0\rangle &= \breve{z}^i u^j \underbrace{\langle1,0,0| \, [C_i , T_j] \, |1,0,0\rangle}_{=\, 2\delta_{ij}} \nonumber\\[0.1cm]
                &\qquad + \breve{z}^i u^j\langle1,0,0| \, T_j \underbrace{C_i \, |1,0,0\rangle}_{=\, 0} \nonumber\\
                &= 2 (\breve{z}\cdot u) = \texttt{H}_1 (\breve{z}, u) \,.
            \end{align}}
        \end{enumerate}}
    \end{enumerate}}
    
    {Now, equipped with the vertex operator expression \eqref{3.2}, we proceed to verify that it reproduces the correct two-point function explicitly:
    \begin{align}
        \langle 0^{\scriptscriptstyle{\mathrm{U}(2,2)}}_{} | \hat{\phi}(z)\, \hat{\phi}(u)\, |0^{\scriptscriptstyle{\mathrm{U}(2,2)}}_{} \rangle = \frac{1}{(z-u)^2} \,.
    \end{align}
    First, recalling Eq. \eqref{woody}, we have:
    \begin{align}
        \langle 0^{\scriptscriptstyle{\mathrm{U}(2,2)}}_{} | \hat{\phi}(z)\, \hat{\phi}(u) | 0^{\scriptscriptstyle{\mathrm{U}(2,2)}}_{} \rangle = \langle 0^{\scriptscriptstyle{\mathrm{U}(2,2)}}_{} | \hat{\phi}(z)\, e^{u \cdot T} |1,0,0\rangle \,.
    \end{align}
    Next, we plug in the vertex operator $\hat{\phi}(z)$ and obtain:
    \begin{align}
        &\langle 0^{\scriptscriptstyle{\mathrm{U}(2,2)}}_{}  | \hat{\phi}(z)\, \hat{\phi}(u) |0^{\scriptscriptstyle{\mathrm{U}(2,2)}}_{} \rangle \nonumber\\[0.2cm] 
        &= \langle 0^{\scriptscriptstyle{\mathrm{U}(2,2)}}_{}  | e^{z \cdot T} \left( \chi^\ast + \frac{\chi}{(z)^2} \right) e^{\breve{z} \cdot C} e^{u \cdot T} |1,0,0\rangle \nonumber\\[0.1cm]
        &= \frac{1}{(z)^2} \langle 1,0,0 |\, e^{\breve{z} \cdot C} e^{u \cdot T} |1,0,0\rangle \nonumber\\[0.1cm]
        &= \frac{1}{(z)^2} \sum_{n,m=0}^{\infty} \frac{1}{n! m!} \langle 1,0,0 | (\breve{z} \cdot C)^n (u \cdot T)^m |1,0,0\rangle \,.
    \end{align}
    The above expression is non-vanishing only when $n=m$. In that case, using the basic materials derived previously, we immediately find:
    \begin{align}
        &\langle 0^{\scriptscriptstyle{\mathrm{U}(2,2)}}_{}  | \hat{\phi}(z)\, \hat{\phi}(u) |0^{\scriptscriptstyle{\mathrm{U}(2,2)}}_{} \rangle \nonumber\\[0.2cm] 
        &= \frac{1}{(z)^2} \sum_{n=0}^\infty \frac{1}{(n!)^2} \langle 1,0,0 | (\breve{z} \cdot C)^n (u \cdot T)^n |1,0,0\rangle \nonumber\\[0.1cm]
        &= \frac{1}{(z)^2} \sum_{n=0}^{\infty} \texttt{H}_n (\breve{z}, u) = \frac{1}{(z - u)^2} \,.
    \end{align}}
}\end{proof}

\begin{remark}
    {\emph{(positivity of $\texttt{H}_n$).} It follows from the Hilbert space (or Wightman) positivity condition \eqref{2.50} (see also Eq. \eqref{4.7}) that $\texttt{H}_n(\overline{z}, z) >0$, for $(z)^2 \neq 0$ and for all $n=0,1,2, \dots\;$. Moreover, one can write:
    \begin{align}
        \texttt{H}_n(\overline{z} , z) = \big|(z)^2\big|^n \; \texttt{C}^1_n\left( \frac{z\cdot\overline{z}}{\big|(z)^2\big|} \right) \geq \big|(z)^2\big|^n \; \texttt{C}^1_n\left( 1 \right) = \big|(z)^2\big|^n \,(n+1) >0 \,,
    \end{align}
    since $\texttt{C}^1_n(x) \geq \texttt{C}^1_n (1)$, for $x>1$.}
\end{remark}

\begin{remark}
    {\emph{(unique harmonic continuation).} The biharmonic polynomial $\texttt{H}_n(z,u)$ (that is, $\Delta_z\, \texttt{H}_n(z,u)=0=\Delta_u\, \texttt{H}_n(z,u)$) is the unique harmonic in $z$ extension of the monomial $(2z\cdot u)^n$ originally defined (and harmonic in $u$) for $(z)^2=0$ (see Example 2.6 in Ref. \cite{BT}).}
\end{remark}

\begin{remark}
    {\emph{(boundary extraction of the lowest one-particle state).} It follows from \eqref{2.48} and \eqref{2.50} that the bra lowest energy one-particle (zero helicity) state $\langle 1,0,0 |$ is:
    \begin{align}       
        &\lim_{\overline{z}\to 0} \; \frac{1}{(\overline{z})^2} \left(\langle 0^{\scriptscriptstyle{\mathrm{U}(2,2)}}_{} | \, \hat{\phi}(z^\ast) \right) \nonumber\\[0.2cm] 
        &= \lim_{\overline{z}\to 0} \; \frac{1}{(\overline{z})^2} \left(\langle 0^{\scriptscriptstyle{\mathrm{U}(2,2)}}_{} |\, e^{z^\ast\cdot T} \left( \chi^\ast + \frac{\chi}{(z^\ast)^2} \right) \, e^{\breve{z}^\ast\cdot C} \right) \nonumber\\[0.2cm]
        &= \lim_{\overline{z}\to 0} \; \frac{1}{(\overline{z})^2} \left(\langle 1,0,0 |\, \frac{1}{(z^\ast)^2} \, e^{\overline{z}\cdot C} \right) \nonumber\\[0.2cm]
        &= \lim_{\overline{z}\to 0} \; \frac{1}{(\overline{z})^2} \sum_{n=0}^\infty \left(\langle 1,0,0 |\, \frac{\big((\overline{z})^2\big)^2}{(\overline{z})^2} \, \frac{(\overline{z}\cdot C)^n}{n!} \right) \nonumber\\[0.2cm]
        &= \langle 1,0,0 | \,,
    \end{align}
    or, alternatively, upon replacing $z$ by $z^\ast = \breve{\overline{z}} = \frac{\overline{z}}{(\overline{z})^2}$:
    \begin{align}
        &\lim_{z\to\infty} \; (z)^2 \left(\langle 0^{\scriptscriptstyle{\mathrm{U}(2,2)}}_{} |\, \hat{\phi}(z) \right) \nonumber\\[0.2cm] 
        &= \lim_{z\to\infty} \; (z)^2 \left(\langle 0^{\scriptscriptstyle{\mathrm{U}(2,2)}}_{} |\, e^{z\cdot T} \left( \chi^\ast + \frac{\chi}{(z)^2} \right) \, e^{\breve{z}\cdot C} \right) \nonumber\\[0.2cm]
        &= \lim_{z\to\infty} \; (z)^2 \left(\langle 1,0,0 |\, \frac{1}{(z)^2}\, e^{\breve{z}\cdot C} \right) \nonumber\\[0.2cm]
        &= \lim_{z\to\infty} \; (z)^2 \sum_{n=0}^\infty \left(\langle 1,0,0 |\, \frac{1}{(z)^2}\, \frac{1}{n!} \Big(\frac{z}{(z)^2}\cdot C\Big)^n \right) \nonumber\\[0.2cm]
        &= \langle 1,0,0 | \,.
    \end{align}}
\end{remark}

\section{Conformal Massless Scalar Field in de Sitter (dS) Space}

The $4$-dimensional dS spacetime \eqref{dS-M_R} can be embedded in the Dirac quadric $Q$ \eqref{2.36} as its intersection with the hyperplane $\zeta^5 = R$:
\begin{align}\label{3.17}
    \text{dS} = Q \;\bigcap\; \Big\{ \zeta^5=R>0 \Big\} = \Bigg\{ \zeta^0 = - \zeta_0\,, \; \zeta^i=\zeta_i\,,\; i=1,2,3,4 \;;\;&\nonumber\\[0.2cm] 
    (\zeta)^2 = \sum_{i=1}^4 (\zeta_i)^2 - (\zeta_0)^2 = R^2&\Bigg\} \,.
\end{align}
Accordingly, the Euclidean de Sitter (EdS) half-space is defined as follows:
\begin{align}\label{EdS}
    \text{EdS} &= \left\{ \zeta^0 = \mathrm{i}\zeta_\beta\,,\; \zeta_\beta>0\,, \;\sum_{i=1}^4 (\zeta_i)^2 + (\zeta_\beta)^2 = R^2 \right\}\nonumber\\
    &= \Big\{ \zeta=\big(\zeta_\beta >0 ,\, \,\zeta_1,\,\dots,\, \zeta_4\big) \in \mathbb{S}_R^4 \Big\}\,.
\end{align}

\begin{remark}{
\emph{(EdS-dS coordinate correspondence and conformal factors).} In this context, a few points merit further clarification:
\begin{enumerate}[leftmargin=0.8cm]
    \item{The real (dimensionless) Cartesian coordinates of the EdS half-space, defined in Eq. \eqref{EdS}, are connected to the analytically continued versions of the coordinates $z_i$ introduced in Eq. \eqref{2.38-1}. Specifically, setting:
    \begin{align}\label{convention50}
        \zeta^5 = R \; >0 \,, \quad\mbox{and}\quad \zeta_\beta = R \tanh{\beta} \; >0 \quad \big(\Longrightarrow \;\; \beta>0 \big)\,,
    \end{align}
    we have:
    \begin{align}\label{3.20}
        z_i = \frac{\zeta_i}{R + \zeta_\beta} = \frac{\kappa \zeta_i}{1 + \kappa \zeta_\beta}\,, 
    \end{align}
    where $\kappa=\frac{1}{R}$ denotes the inverse radius of $\text{EdS} \cong \mathbb{S}^4_R$. It then follows that $(z)^2$ is obtained as an analytic continuation of \eqref{2.37}:
    \begin{align}\label{2.37cont}
        (z)^2 \; \left( := \sum_{i=1}^4 z^i z_i \right) = \frac{R-\zeta_\beta}{R+\zeta_\beta} = \frac{1-\kappa\zeta_\beta}{1+\kappa\zeta_\beta} = e^{-2\beta}\,. 
    \end{align}
    Note that $\kappa\zeta_\beta = \tanh{\beta}$, with $\beta>0$. From this realization, one can observe that:
    \begin{enumerate}[leftmargin=0.8cm]
        \item{Given the equation above and the strict positivity of $\beta$, it follows immediately that:
        \begin{align}
            \left|(z)^2\right| \; \Big(= z \cdot \overline{z} \Big) = e^{-2\beta} < 1\,,
        \end{align}
        and also one can readily check that:
        \begin{align}
            z \cdot \overline{z} < \frac{1}{2} \left(1 + \left|(z)^2\right|^2 \right)\,,
        \end{align}
        This demonstrates that EdS half-space \eqref{EdS} lies within the intersection of the hyperplane $\zeta^5=R$ with the tube domain $\mathfrak{T}^{\texttt{\tiny{precompact}}}_+$ \eqref{2.40}, where the vector-valued function $\hat{\phi}\big(z(\zeta)\big)|0^{\scriptscriptstyle{\mathrm{U}(2,2)}}_{}\rangle$ is analytic.}
        
        \item{The label $\beta$ in the imaginary time coordinate $\mathrm{i}\zeta_\beta$ is chosen to evoke the Boltzmann factor $\beta = \frac{1}{KT}$, where $T$ denotes the absolute temperature and $K$ is the Boltzmann constant. This choice reflects the fact that the analytically continued two-point function of the field $\hat{\phi}(\zeta)|0\rangle$ admits an interpretation as a thermal correlation function (see \cite{NT}, particularly Sect. 7). In this context, $\beta$ can be understood as a real continuation of the imaginary proper time $\mathrm{i}\tau$, where $\tau$ is the physical proper time (compare Eqs. \eqref{2.37} and \eqref{2.37cont}).} 
	\end{enumerate}}
    
    \item{The metric form $(\mathrm{d} z)^2$ is related by a (Weyl) conformal factor to the manifestly $\mathrm{O}(5)$-invariant quadratic $(\mathrm{d} \zeta)^2 = (\mathrm{d}\zeta_i)^2 + (\mathrm{d}\zeta_\beta)^2$:
    \begin{align}\label{3.22}
        &(\mathrm{d} z)^2 \nonumber\\[0.2cm] 
        &= \mathrm{d} \left( \frac{\zeta^i}{R + \zeta_\beta} \right) \mathrm{d} \left( \frac{\zeta_i}{R + \zeta_\beta} \right) \nonumber\\[0.2cm]
        &= \left( \frac{\mathrm{d}\zeta^i}{R + \zeta_\beta} - \frac{(\mathrm{d}\zeta_\beta)\,\zeta^i}{(R + \zeta_\beta)^2}\right) \left( \frac{\mathrm{d}\zeta_i}{R + \zeta_\beta} - \frac{(\mathrm{d}\zeta_\beta)\,\zeta_i}{(R + \zeta_\beta)^2}\right) \nonumber\\[0.2cm]
        &= \frac{1}{\left(R+\zeta_\beta\right)^2} \left( (\mathrm{d}\zeta_i)^2 + \frac{\zeta^i\zeta_i \; \big(=R^2-(\zeta_\beta)^2\big)}{(R+\zeta_\beta)^2}\, (\mathrm{d}\zeta_\beta)^2 - \frac{2\zeta^i\mathrm{d}\zeta_i}{R+\zeta_\beta} \, \mathrm{d} \zeta_\beta \right) \nonumber\\[0.2cm]
        &= \frac{1}{\left(R+\zeta_\beta\right)^2} \left( (\mathrm{d}\zeta_i)^2 + \frac{R-\zeta_\beta}{R+\zeta_\beta}\, (\mathrm{d}\zeta_\beta)^2 - \frac{\mathrm{d}(\zeta_i)^2\; \big( = \mathrm{d}\big(R^2-(\zeta_\beta)^2\big) \big)}{R+\zeta_\beta} \, \mathrm{d} \zeta_\beta \right) \nonumber\\[0.2cm]
        &= \frac{1}{\left( R+\zeta_\beta \right)^2} \left( (\mathrm{d}\zeta_i)^2 + (\mathrm{d}\zeta_\beta)^2 \right) = \frac{1}{\left( R+\zeta_\beta \right)^2} \; (\mathrm{d} \zeta)^2\,,
    \end{align}        
    and hence:
    \begin{align}\label{vvvvvob}
        &(\mathrm{d}\zeta)^2 \; \Big(= (\mathrm{d}\zeta_i)^2 + (\mathrm{d}\zeta_\beta)^2\Big) = \Omega^2(z)\, (\mathrm{d}z)^2\,, \nonumber\\[0.2cm]
        &\quad\mbox{with}\quad\Omega(z) := R + \zeta_\beta = \frac{2R}{1+(z)^2}\,.
    \end{align}
    Note that summation over the repeated index `$i$' from $1$ to $4$ is assumed throughout.}

    \item{Coming back to the real dS spacetime with Lorentzian signature, we can express the Minkowski coordinate $x^{\dot{\mu}}$ (of dimension of length) in terms of the $\mathrm{O}(4,1)$ covariant vector $\zeta$ \eqref{3.17} by the following relations:
    \begin{align}\label{bbbbob}
        &\kappa x^{\dot{\mu}} = \frac{\zeta^{\dot{\mu}}}{R+\zeta_4} = \frac{\kappa\zeta^{\dot{\mu}}}{1+\kappa\zeta_4} \nonumber\\[0.2cm]
        &\quad\Longrightarrow\quad \kappa^2 (x)^2 \;\Big(:= \kappa^2 x^{\dot{\mu}}x_{\dot{\mu}}\Big) = \frac{1-\kappa\zeta_4}{1+\kappa\zeta_4}\,,
    \end{align}
    where, again, $\dot{\mu}=0,1,2,3$. Then, in complex analogy with \eqref{3.22}, we have:
    \begin{align}\label{3.25}
        &(\mathrm{d}\zeta)^2 \; \Big( = (\mathrm{d} \zeta_{\dot{\mu}})^2 + (\mathrm{d} \zeta_4)^2 \Big) = \Omega^2(x) \,(\mathrm{d}x_{\dot{\mu}})^2 \,, \nonumber\\[0.2cm]
        &\quad\mbox{with}\quad \Omega(x) \; := 1+\kappa\zeta_4 = \frac{2}{1+\kappa^2(x)^2}\,.
    \end{align}

    \textbf{\textit{Note}:} The Weyl factor $\Omega_{\text{FHH}}(x) = \left( 1+\frac{\kappa^2}{4}(x)^2 \right)^{-1}$ of Ref. \cite{FHH} differs from ours by the factor $\frac{1}{4}$ in front of $\kappa^2 (x)^2$; it can be reproduced by substituting \eqref{bbbbob} by $x^{\dot{\mu}} = \frac{2\zeta^{\dot{\mu}}}{1+\kappa\zeta_4}$.}
    
    \item{The expression \eqref{3.22} and \eqref{3.25} involve different small curvature ($R\to\infty$ or $\kappa\to 0$) limits:
    \begin{enumerate}[leftmargin=0.8cm]
        \item{\textbf{\textit{EdS case}:} For $\zeta_\beta \approx R$ (or $\kappa\zeta_\beta \approx 1$) and hence $|\zeta_i|\ll R$, we have:
        \begin{enumerate}[leftmargin=0.8cm]
            \item{Eqs. \eqref{3.20} and \eqref{2.37cont}, respectively, imply that $|z_i|$ and $(z)^2\ll 1$.}
            \item{From Eqs. \eqref{bbbbob}, \eqref{3.20}, and \eqref{2.37cont}, we obtain:
            \begin{align}
                \kappa^2(x)^2 = \frac{1-\kappa\zeta_4}{1+\kappa\zeta_4} = \frac{1+(z)^2-2z_4}{1+(z)^2+2z_4} \approx 1-4z_4 \approx 1 \,.
            \end{align}}
            \item{From Eqs. \eqref{vvvvvob}, \eqref{3.20}, and \eqref{2.37cont}, the Euclidean $\mathrm{O}(5)$-invariant metric $(\mathrm{d}\zeta)^2$ approaches:
            \begin{align}
                (\mathrm{d}\zeta)^2 \; \Big(:= (\mathrm{d}\zeta_i)^2 + (\mathrm{d}\zeta_\beta)^2\Big) \quad\longrightarrow\quad 4R^2 (\mathrm{d}z_i)^2 \gg (\mathrm{d}z_i)^2\,.
            \end{align}}
        \end{enumerate}}

        \item{\textbf{\textit{dS case}:} For $\zeta_4\approx R$ (or $\kappa\zeta_4 \approx 1$) and hence $|\zeta^{\dot{\mu}}|\ll R$, we have:
        \begin{enumerate}[leftmargin=0.8cm]
            \item{Eq. \eqref{bbbbob} implies that $|x^{\dot{\mu}}|$ and $(x)^2\ll 1$.}
            \item{From Eqs. \eqref{2.37cont} and \eqref{bbbbob}, we obtain:
            \begin{align}
                (z)^2 = \frac{1-\kappa\zeta_\beta}{1+\kappa\zeta_\beta} = \frac{1+\kappa^2(x)^2-2\kappa x_\beta}{1+\kappa^2(x)^2+2\kappa x_\beta} \approx 1-4\kappa x_\beta \approx 1 \,,
            \end{align}
            where $x^0 := \mathrm{i} x_\beta$.}
            \item{From Eqs. \eqref{3.25} and \eqref{bbbbob}, the Euclidean $\mathrm{O}(5)$-invariant metric $(\mathrm{d}\zeta)^2$ approaches:
            \begin{align}
                (\mathrm{d}\zeta)^2 \; \Big( := (\mathrm{d} \zeta_{\dot{\mu}})^2 + (\mathrm{d} \zeta_4)^2 \Big) \quad\longrightarrow\quad 4 (\mathrm{d}x_{\dot{\mu}})^2\,.
            \end{align}}
        \end{enumerate}}
    \end{enumerate}}
\end{enumerate}
}\end{remark}



We now proceed to write down the field (wave) equation and the two-point function for a conformal (massless) scalar field in (E)dS space. 

Working in conformal space — without imposing the dS constraint $\zeta^5 = R$ — Dirac \cite{D36} demonstrated that solutions $\phi(\zeta)$ to the manifestly invariant $6$-dimensional d'Alembert equation:
\begin{align}\label{dAlembert equation}
    \square_6 \phi(\zeta) = 0 \,, 
\end{align}
with:
\begin{align}    
    \square_6 = \sum_{i=1}^{4} \frac{\partial^2}{\partial \zeta_i^2} - \frac{\partial^2}{\partial \zeta_0^2} - \frac{\partial^2}{\partial \zeta_5^2} \,,
\end{align}
are only well-defined on the light-cone $(\zeta)^2 = 0$, if $\phi$ is homogeneous of degree $-1$. Specifically, if $\phi(\zeta)$ satisfies the homogeneity condition:
\begin{align}
    \phi(\varrho \zeta) = \varrho^d \phi(\zeta)\,, \quad\text{for all}\quad \varrho > 0 \,,
\end{align}
then the equation $\square_6 \phi = 0$ defines a conformally invariant field only on the quadric $(\zeta)^2 = 0$ provided that $d = -1$.

This requirement arises from a subtle invariance: 
\begin{enumerate}[leftmargin=0.8cm]
    \item{Since the physical field $\phi(\zeta)$ is defined only on the cone $(\zeta)^2 = 0$, it can be freely modified by adding terms proportional to $(\zeta)^2$, without affecting its restriction to the cone; that is:
    \begin{align}
        \phi(\zeta) \quad\longmapsto\quad \phi(\zeta) + (\zeta)^2 \boldsymbol{\phi}(\zeta) \,,
    \end{align}
    where the new field $\boldsymbol{\phi}(\zeta)$ is assumed to be homogeneous of degree $d - 2$:
    \begin{align}
        \boldsymbol{\phi}(\varrho \zeta) = \varrho^{d-2} \boldsymbol{\phi}(\zeta) \,, \quad\text{for all}\quad \varrho > 0 \,.
    \end{align}}

    \item{To ensure that the gauge-transformed function $\phi(\zeta) + (\zeta)^2 \boldsymbol{\phi}(\zeta)$ still satisfies d'Alembert equation \eqref{dAlembert equation} on the cone, we require that:
    \begin{align}
        \square_6 \left( (\zeta)^2 \boldsymbol{\phi}(\zeta) \right)\Big|_{(\zeta)^2 = 0} = 0 \,.
    \end{align}
    This condition holds only if $d = -1$. The reason follows from the commutator (see Appendix \ref{Appen commutator}):
    \begin{align}\label{commutator}
        \big[ \square_n, (\zeta)^2 \big] = 2\left( n + 2 \zeta^\mu \frac{\partial}{\partial \zeta^\mu}\right) =: 2\left(n + 2  \zeta\cdot\partial\right) \,,
    \end{align}
    which yields (see derivation in Appendix \ref{Appen 4.43}):
    \begin{align}
        \square_{n=6} \left( (\zeta)^2 \boldsymbol{\phi} \right)\Big|_{(\zeta)^2 = 0} = 2\big(6 + 2(d - 2)\big) \boldsymbol{\phi} = 0 \quad\Longleftrightarrow\quad d = -1 \,.
    \end{align}}
\end{enumerate}

For $\zeta^{}_{\text{D}} = (\zeta^{}_{\text{dS}}, \zeta^5_{\text{D}} = R)$ and $\zeta_{\text{D}}^\prime = (\zeta_{\text{dS}}^\prime, \zeta^{\prime 5}_{\text{D}} = R)$, where $(\zeta^{}_{\text{dS}})^2 = R^2 = (\zeta_{\text{dS}}^\prime)^2$ and $(\zeta^{}_{\text{D}})^2 = 0 = (\zeta_{\text{D}}^\prime)^2$, it is natural to expect that a conformally invariant two-point function on dS spacetime (see, for instance, Ref. \cite{Gazeaus1}\footnote{In Ref. \cite{Gazeaus1}, the metric signature is chosen to be mostly negative, in contrast to the convention adopted in this manuscript. As a result, certain terms in the corresponding two-point function differ from ours by a sign.}) can be recovered from the manifestly invariant one defined on the $6$-dimensional light-cone. In this setting, the corresponding two-point function takes the form:
\begin{align}
    W(\zeta^{}_{\text{dS}}, \zeta_{\text{dS}}^\prime) :=&\, W(\zeta^{}_{\text{D}}, \zeta_{\text{D}}^\prime) \Big|_{\zeta^5_{\text{D}} = R = \zeta^{\prime 5}_{\text{D}}} \nonumber\\[0.2cm]
    =&\, \frac{1}{(2\pi)^2 \left(\zeta^{}_{\text{D}} - \zeta_{\text{D}}^\prime\right)^2}\Bigg|_{\zeta^5_{\text{D}} = R = \zeta^{\prime 5}_{\text{D}}} = \frac{1}{8\pi^2 \left(R^2 - \zeta^{}_{\text{dS}} \cdot \zeta_{\text{dS}}^\prime\right)}\,.
\end{align}
This is simply the pullback of the conformally invariant $6$-dimensional two-point function onto the $5$-dimensional dS hyperboloid $(\zeta^{}_{\text{dS}})^2 = R^2$. Note that, the two-point function $W(\zeta^{}_{\text{D}}, \zeta_{\text{D}}^\prime)$ is manifestly homogeneous of degree $-1$ in $\zeta^{\prime}_{\text{D}}$; that is:
\begin{align}
    W(\zeta^{}_{\text{D}}, \varrho \zeta_{\text{D}}^\prime) &= \frac{1}{(2\pi)^2 \left(\zeta^{}_{\text{D}} - \varrho \zeta_{\text{D}}^\prime\right)^2} \nonumber\\[0.2cm]
    &= \frac{1}{(2\pi)^2 \left(-2\varrho\, \zeta^{}_{\text{D}}\cdot\zeta_{\text{D}}^\prime\right)} = \varrho^{-1} W(\zeta^{}_{\text{D}}, \zeta_{\text{D}}^\prime)\,,
\end{align}
and, by symmetry, it is also homogeneous of degree $-1$ in $\zeta^{}_{\text{D}}$. Moreover, it satisfies the $6$-dimensional wave equation:
\begin{align}
    \square_6\, W(\zeta^{}_{\text{D}}, \zeta_{\text{D}}^\prime) = 0\,.
\end{align}

\section{Appendix: On the Commutator \eqref{commutator}}\label{Appen commutator}

We compute the action of the commutator \eqref{commutator} on a test function $f(\zeta)$:
\begin{align}
    \big[ \square_n, (\zeta)^2 \big] f &= \square_n \big( (\zeta)^2 f \big) - (\zeta)^2 \square_n f \nonumber\\[0.2cm]
    &= \eta^{\mu\nu} \partial_\mu \partial_\nu \big( (\zeta)^2 f \big) - (\zeta)^2 \square_n f\,,
\end{align}
where $\mu, \nu = {0,5,1,2,3,4}$ and the metric $\eta_{\mu\nu}$ has the signature \linebreak $(-,-,+,+,+,+)$. First, we employ the product rule:
\begin{align}
    \partial_\nu \big( (\zeta)^2 f \big) =&\, \big(\partial_\nu (\zeta)^2\big) f + (\zeta)^2 \partial_\nu f \,, 
\end{align}
and subsequently:
\begin{align}
    \partial_\mu \partial_\nu \big( (\zeta)^2 f \big) =&\, \partial_\mu \Big( \big(\partial_\nu (\zeta)^2\big) f + (\zeta)^2 \partial_\nu f \Big) \nonumber\\[0.2cm]
    =&\, \big(\partial_\mu \partial_\nu (\zeta)^2\big) f + \big(\partial_\nu (\zeta)^2\big) \partial_\mu f \nonumber\\[0.2cm] 
    & +\big(\partial_\mu (\zeta)^2\big) \partial_\nu f + (\zeta)^2 \partial_\mu \partial_\nu f \,,
\end{align}
where $\partial_\mu := \tfrac{\partial}{\partial\zeta^\mu}$, and then, we use $\partial_\mu (\zeta)^2 = 2 \eta_{\mu\varsigma} \zeta^\varsigma$ and $\partial_\mu \partial_\nu (\zeta)^2 = 2 \eta_{\mu\nu}$. Substituting into the full expression yields:
\begin{align}
    &\eta^{\mu\nu} \partial_\mu \partial_\nu \big( (\zeta)^2 f \big) \nonumber\\[0.2cm] 
    &= \eta^{\mu\nu} \Big( 2 \eta_{\mu\nu} f + 2 \eta_{\nu\varsigma} \zeta^\varsigma \partial_\mu f + 2 \eta_{\mu\varsigma} \zeta^\varsigma \partial_\nu f + (\zeta)^2 \partial_\mu \partial_\nu f \Big) \nonumber\\[0.2cm]
    &= 2n f + 4 \zeta^\mu \partial_\mu f + (\zeta)^2 \square_n f \,.
\end{align}
Thus, we obtain the commutator:
\begin{align}
    \left[ \square_n, (\zeta)^2 \right] f = 2n f + 4 \zeta \cdot \partial f \,.
\end{align}

\section{Appendix: Derivation of the Homogeneity Degree $d=-1$} \label{Appen 4.43}

For any sufficiently smooth function $f(\zeta)$, from the commutator \eqref{commutator}, we have: 
\begin{align}
    \square_n \left( (\zeta)^2 f \right) = 2n f + 4 \zeta^\mu \frac{\partial f}{\partial \zeta^\mu} + (\zeta)^2 \square_n f\,,
\end{align}
so that, on the cone $(\zeta)^2 = 0$, it simplifies to:
\begin{align}
    \square_n \left( (\zeta)^2 f \right)\Big|_{(\zeta)^2 = 0} = 2n f + 4 \zeta\cdot\partial f\,.
\end{align}
Now, if $f = \boldsymbol{\phi}$ is homogeneous of degree $d - 2$, then $\zeta\cdot\partial\boldsymbol{\phi} = (d - 2) \boldsymbol{\phi}$, and thus:
\begin{align}
    \square_n \left( (\zeta)^2 \boldsymbol{\phi} \right)\Big|_{(\zeta)^2 = 0} = \big(2n + 4(d - 2)\big) \boldsymbol{\phi}\,.
\end{align}
In $n = 6$ dimensions, this vanishes only when $d = -1$.
	
	\chapter*{Glossary}

Glossary of Clifford and Lie algebras used in the monograph, with precise nesting and isomorphisms:
\begin{enumerate}[leftmargin=0.8cm]
    \item{$\mathfrak{cl}(4,2)$ is the $64$-dimensional real associative Clifford algebra generated by $m^{}_\mu$, $\mu=0,5,1,2,3,4$, as:
    \begin{align*}
        \mathfrak{cl}(4,2) = \mathrm{Span}\Big\{& \mathbbm{1},\, m^{}_\mu,\, m^{}_{\mu\nu} := m^{}_\mu \,m^{}_\nu \;(\mu<\nu),\, \nonumber\\[0.2cm] 
        &\; m^{}_\mu \,m^{}_\nu \,m^{}_\rho \;(\mu<\nu<\rho),\, \dots,\,
        E := m^{}_0 \,m^{}_5 \,m^{}_1 \,\dots \,m_4 \Big\}\,,
    \end{align*}
    subject to the Clifford relations:
    \begin{align*}
        &\big\{ m^{}_\mu, m^{}_\nu \big\} = m^{}_\mu m^{}_\nu + m^{}_\nu m^{}_\mu = 2 \eta_{\mu\nu} \mathbbm{1}\,, \nonumber\\[0.2cm] 
        &\eta_{\mu\nu}=\mathrm{diag}(-1,-1,+1,+1,+1,+1)\,.
    \end{align*}}

    \item{The even Clifford algebra $\mathfrak{cl}^{\text{even}}(4,2) \; \big( \cong \mathfrak{cl}(4,1) \big)$ is:
    \begin{align*}
        \mathfrak{cl}^{\text{even}}(4,2) = \mathrm{Span}\Big\{ \mathbbm{1},\, m^{}_{\mu\nu} := m^{}_\mu \,m^{}_\nu \;(\mu<\nu), \qquad\; &\,\nonumber\\[0.2cm] 
        \; m^{}_\mu \,m^{}_\nu \,m^{}_\rho \,m^{}_\sigma \;(\mu<\nu<\rho<\sigma), &\,\nonumber\\[0.2cm] 
        E := m^{}_0 \,m^{}_5 \,m^{}_1 \,\dots \,m_4 &\Big\}\,,
    \end{align*}
    i.e., the $32$-dimensional real subalgebra spanned by all products of an even number of the generators $m^{}_\mu$.}

    \item{The bivectors $m^{}_{\mu\nu}$ are closed under commutator and provide a spinorial realization of the conformal Lie algebra:
    \begin{align*}
        \mathfrak{su}(2,2) \; \big(\cong\mathfrak{so}(4,2)\big) = \mathrm{Span} \Big\{ m^{}_{\mu\nu} := m^{}_\mu \,m^{}_\nu \Big\} \;\subset\; \mathfrak{cl}^{\text{even}}(4,2)\,,
    \end{align*} 
    with the Lie bracket given by $\big[m^{}_{\mu\nu}, m^{}_{\rho\sigma}\big]$ computed inside $\mathfrak{cl}^{\text{even}}(4,2)$.}

    \item{The pseudoscalar $E := m^{}_0 \,m^{}_5 \,m^{}_1 \,\dots \,m_4$ commutes with all bivectors $m^{}_{\mu\nu}$ generating $\mathfrak{su}(2,2) \; \big(\cong\mathfrak{so}(4,2)\big)$. Hence, $E$ spans the central $\mathfrak{u}(1)$ in the extended conformal Lie algebra:
    \begin{align*}
        \mathfrak{u}(2,2) \cong \mathfrak{su}(2,2) \; \big(\cong\mathfrak{so}(4,2)\big)\oplus \mathfrak{u}(1) \; \subset \; \mathfrak{cl}^{\text{even}}(4,2) \,.
    \end{align*}}

    \item{The de Sitter (dS) algebra is realized as the subalgebra of $\mathfrak{su}(2,2) \; \big(\cong \mathfrak{so}(4,2)\big)$ that preserves a fixed ``$5$-direction'' in the $6$-dimensional space. Specifically:
    \begin{align*}
        \mathfrak{so}(4,1) \; \big(\cong\mathfrak{sp}(2,2)\big) &= \mathrm{Span}\Big\{ m^{}_{\alpha\beta} := m^{}_\alpha m^{}_\beta \;;\; \alpha, \beta = 0,1,2,3,4 \Big\} \nonumber\\[0.2cm]
        &\qquad\qquad \subset\; \mathfrak{su}(2,2) \;\subset\; \mathfrak{cl}^{\mathrm{even}}(4,2)\,,
    \end{align*}
    with commutator taken inside $\mathfrak{cl}^{\mathrm{even}}(4,2)$. These bivectors $m^{}_{\alpha\beta}$ satisfy the $\mathfrak{so}(4,1)$ Lie algebra relations and provide a spinorial realization of the corresponding dS group $\mathrm{Spin}(4,1) \cong \mathrm{Sp}(2,2)$.}
\end{enumerate}
	
	\backmatter
	\begin{theindex}

\item{\textbf{A}}

\item{Algebra}
    \subitem{Alternative}
    \subitem{Clifford}
    \subitem{Composition}
    \subitem{Vertex}

\vspace{0.5cm}

\item{\textbf{B}}

\item{Basis}
    \subitem{Cartan}
    \subitem{Chiral}
    \subitem{Majorana}
    \subitem{Real}

\vspace{0.5cm}

\item{\textbf{C}}

\item{Conjugate}
    \subitem{Cayley-Dickson}
    \subitem{Charge}
    \subitem{Complex}
    \subitem{Hermitian}
    \subitem{Octonionic}

\item{Cosmological Constant}

\vspace{0.5cm}
\item{\textbf{E}}

\item{Elementary System}
    \subitem{Massless}

\item{Energy}
    \subitem{Conformal}

\vspace{0.5cm}
\item{\textbf{F}}

\item{Field}
    \subitem{Massless}
    \subitem{Massless Scalar}

\item{Form}
    \subitem{Bilinear}
    \subitem{Volume}

\vspace{0.5cm}
\item{\textbf{G}}

\item{Group}
    \subitem{Poincar\'{e}}
    \subitem{Conformal}
    \subitem{de Sitter}

\vspace{0.5cm}
\item{\textbf{H}}

\item{Hamiltonian}
    \subitem{Conformal}

\item{Helicity}

\vspace{0.5cm}
\item{\textbf{O}}  

\item{Octonion}
    \subitem{Split}

\item{Operator}
    \subitem{Annihilation}
    \subitem{Casimir}
    \subitem{Creation}
    \subitem{Ladder}
    \subitem{Lowering}
    \subitem{Raising}
    \subitem{Self-adjoint}
    \subitem{Vertex}

\vspace{0.5cm}
\item{\textbf{P}}

\item{Parity}

\vspace{0.5cm}
\item{\textbf{Q}}

\item{Quaternion}

\vspace{0.5cm}
\item{\textbf{R}}

\item{Representation}
    \subitem{Ladder}
    \subitem{Massless}

\vspace{0.5cm}
\item{\textbf{S}}

\item{Space}
    \subitem{Conformal}
    \subitem{de Sitter}
    \subitem{Euclidean}
    \subitem{Hilbert}
    \subitem{Minkowski}
    \subitem{Compactified Minkowski}

\item{Spin}

\item{Spinor}
    \subitem{Chiral}
    \subitem{Majorana}

\item{Subalgebra}
    \subitem{Maximal Compact}

\vspace{0.5cm}
\item{\textbf{T}}

\item{Two-point function}

\item{Transformation}
    \subitem{Similarity}
    \subitem{Unitary}

\vspace{0.5cm}
\item{\textbf{V}}

\item{Vacuum}
    \subitem{Fock}

\end{theindex}

\end{document}